\documentclass[12pt,psamsfonts]{amsart}
\usepackage[foot]{amsaddr}
\usepackage{amsmath,amsthm,amsfonts,amssymb}
\usepackage{eucal}
\usepackage{graphicx}
\usepackage{caption}
\usepackage{indentfirst}
\usepackage{bbm}
\usepackage{grffile}
\usepackage[all,knot]{xy}
\usepackage{chngcntr}
\usepackage{floatrow}
\usepackage{subfig}
\usepackage[colorlinks, citecolor=blue, linkcolor=blue]{hyperref}
\usepackage[usenames, dvipsnames]{color}
\usepackage{enumitem,kantlipsum}
\usepackage{amsbsy}
\usepackage{multicol}
\usepackage{multirow}
\usepackage{mathrsfs}
\counterwithin{figure}{section}
\xyoption{arc}

\newcommand{\A}{\mathcal{A}}
\newcommand{\M}{\mathcal{M}}
\newcommand{\mN}{\mathcal{N}}
\newcommand{\D}{\mathcal{D}}

\newcommand{\CC}{\mathcal{C}}
\newcommand{\ZZ}{\mathcal{Z}}

\newcommand{\Rep}{{\rm Rep}}

\newcommand{\SUM}{\text{SUM}}

\DeclareMathOperator{\FPdim}{FPdim}

 \DeclareMathOperator{\Dim}{Dim} 
\DeclareMathOperator{\End}{End}  
\DeclareMathOperator{\Hom}{Hom} 
 
\DeclareMathOperator{\Span}{Span} 
\DeclareMathOperator{\SU}{SU}

\DeclareMathOperator{\Fun}{Fun}
\DeclareMathOperator{\Obj}{Obj}
\DeclareMathOperator{\Tr}{Tr}
\newcommand{\one}{\mathbf{1}}
\newcommand{\C}{\mathbb C}
\newcommand{\mC}{\mathcal{C}}
\newcommand{\mE}{\mathcal{E}}

\newcommand{\mY}{\mathcal{Y}}
\newcommand{\mZ}{\mathcal{Z}}
\newcommand{\mQ}{\mathcal{Q}}
\newcommand{\mD}{\mathcal{D}}
\newcommand{\mW}{\mathcal{W}}

\newcommand{\msC}{\mathscr{C}}
\newcommand{\msD}{\mathscr{D}}
\newcommand{\mfL}{\mathfrak{L}}
\newcommand{\mfC}{\mathfrak{C}}
\newcommand{\mfB}{\mathfrak{B}}
\newcommand{\mfr}{\mathfrak{r}}
\newcommand{\mfh}{\mathfrak{h}}
\newcommand{\mfD}{\mathfrak{D}}
\newcommand{\Z}{\mathbb Z}

\newcommand{\B}{\mathcal{B}}

\newcommand{\comments}[1]{}

\renewcommand{\vec}[1]{{\mathbf #1}}
\newcommand{\ket}[1]{|#1\rangle}

\renewcommand{\one}{\mathbf{1}}

\renewcommand{\CC}{\mathcal{C}}
\renewcommand{\D}{\mathcal{D}}
\renewcommand{\Z}{\mathbb{Z}}

\newcommand{\overbar}[1]{\mkern 2.3mu\overline{\mkern-2.3mu#1\mkern-2.3mu}\mkern 2.3mu}

\numberwithin{equation}{section}

\newtheorem{theorem}{Theorem}[section]

\newtheorem{corollary}[theorem]{Corollary}
\newtheorem{lemma}[theorem]{Lemma}

\newtheorem{prop}[theorem]{Proposition}
\theoremstyle{definition}

\newtheorem{remark}[theorem]{Remark}

\newtheorem{definition}[theorem]{Definition}

\begin{document}

\title{Topological Quantum Computation with Gapped Boundaries}
\author{Iris Cong$^{1,4}$}
\email{$^1$irisycong@engineering.ucla.edu}
\address{$^1$Dept of Computer Science, University of California\\
Los Angeles, CA 90095\\
U.S.A.}

\author{Meng Cheng$^{2,4}$}
\email{$^2$m.cheng@yale.edu}
\address{$^2$Dept of Physics\\
	Yale University\\
	New Haven, CT 06520-8120\\
    U.S.A.}

\author{Zhenghan Wang$^{3,4}$}
\email{$^4$zhenghwa@microsoft.com}
\address{$^3$Dept of Mathematics\\
    University of California\\
    Santa Barbara, CA 93106-6105\\
    U.S.A.}
\address{$^4$Microsoft Station Q\\
    University of California\\
    Santa Barbara, CA 93106-6105\\
    U.S.A.}

\begin{abstract}
This paper studies fault-tolerant quantum computation with gapped boundaries. We first introduce gapped boundaries of Kitaev's quantum double models for Dijkgraaf-Witten theories using their Hamiltonian realizations. We classify the elementary excitations on the boundary, and systematically describe the bulk-to-boundary condensation procedure. We also provide a commuting Hamiltonian to realize defects between boundaries in any quantum double model. Next, we present the algebraic/categorical structure of gapped boundaries and boundary defects, which will be used to describe topologically protected operations and obtain quantum gates. To demonstrate a potential physical realization, we provide quantum circuits for surface codes that can perform all basic operations on gapped boundaries. Finally, we show how gapped boundaries of the abelian theory $\mfD(\Z_3)$ can be used to perform universal quantum computation.
\end{abstract}

\maketitle
{\hypersetup{linkcolor=black}
\tableofcontents}

\clearpage

\section{Introduction}
\label{sec:intro}

\subsection{Motivations}
\label{sec:motivations}

The quantum model of computation strikes a delicate balance between classical digital and analog computing models, as its stability lies closer to digital models, while its computational power is closer to analog ones.  Still, a major obstacle to developing quantum computers lies in the susceptibility of qubits to decoherence. One elegant theoretical solution to this problem is to perform quantum computation topologically \cite{Free98, Kitaev97, FKLW}.

Topological quantum computation (TQC) is a paradigm that information is encoded in topological degrees of freedom of certain quantum systems, thereby protected from local decoherence.  The standard implementation is through anyons in topological phases of matter, where qubits or qudits are built out of degenerate ground states of many anyon systems, and braiding matrices of anyons are used as quantum gates \cite{Nayak08, W10}.  Recent studies in topological phases of matter revealed that certain topological phases of matter also support gapped boundaries.  Therefore, it is natural to ask if these cousins of anyons can be employed for quantum information processing.  This is indeed the case.  In this paper, we develop an exactly solvable lattice Hamiltonian theory for gapped boundaries in the Dijkgraaf-Witten topological quantum field theories (TQFTs) by modifying Kitaev's quantum double model. We systematically investigate the extra computational power provided by gapped boundaries, and the resulting enhancements to anyonic computation.  Our study of gapped boundaries is through an interplay between their Hamiltonian realization and an algebraic model using category theory. Using a very simple picture of triangles and a ribbon ring around a hole as (Fig. \ref{fig:main-results}), we provide an insightful interpretation of this interplay in the context of many physical processes such as bulk-to-boundary condensation of anyons.

A topological phase of matter $\mathcal{H}=\{H\}$ is an equivalence class of gapped Hamiltonians $H$ which realizes a TQFT at low energy.  Elementary excitations in a topological phase of matter $\mathcal{H}$ are point-like anyons.  Anyons can be modeled algebraically as simple objects in a unitary modular tensor category (UMTC) $\mathcal{B}$, which will be referred to as the topological order of the topological phase $\mathcal{H}$.  A salient feature of anyons for application to TQC is the topological ground state degeneracy, which can arise either from non-trivial topology of the space, or from non-abelian anyons even without topology (anyons with quantum dimension $>1$). For abelian anyons in the plane, the ground state manifold is non-degenerate.  However, when the topological phase of matter $\mathcal{H}$ supports gapped boundaries, new topological degeneracies can arise, even for abelian anyons in the plane. This is because a gapped boundary is essentially a coherent superposition of anyons, and hence behaves like a non-abelian anyon.

We consider only topological phases of matter $\mathcal{H}$ that can be represented by fixed-point gapped Hamiltonians $H$ of the form $H=-\sum_{i}H_i$ such that all local terms $H_i$ are commuting Hermitian projectors.  Two general classes of such Hamiltonians are the Kitaev quantum double model for Dijkgraaf-Witten TQFTs and the Levin-Wen model for Turaev-Viro-Barrett-Westbury TQFTs.  Their input data, finite groups $G$ and unitary fusion categories $\mathcal{C}$ respectively,  are dual to each other. When the Kitaev model is extended from finite groups to connected $*$-quantum groupoids \cite{Chang14}, the two models are equivalent because they both realize the same topological orders---the representation categories $\mathfrak{D}(G)$ of quantum doubles $D(G)$ or Drinfeld centers $\mathcal{Z}(\mathcal{C})$ of the input categories $\mathcal{C}$.  For such topological phases of matter, a gapped boundary is an equivalence class of extensions of the ideal gapped Hamiltonian from a closed surface to a local commuting Hamiltonian on the surface with a boundary.  We classify gapped boundaries by the maximal collection of bulk anyons that can be condensed to the boundary. While our theory works for any surface with boundaries, we will mainly focus on a planar region $\Lambda$ with many holes $\mathfrak{h_i}$, which are small rectangles removed from $\Lambda$ (see Fig. \ref{fig:boundary} for an example).  We generally imagine the holes $\mathfrak{h_i}$ as small disks or rectangles, but in order to achieve topological protection of the ground state degeneracy, their sizes cannot be too small. When gapped boundaries become too small, they decohere into single anyons.

In the UMTC model of a 2D doubled topological order $\mathcal{B}=\mZ(\mathcal{C})$, a stable gapped boundary or gapped hole is modeled by a Lagrangian algebra $\mathcal{A}$ in $\mathcal{B}$.\footnote{We will use the terms gapped boundary, gapped hole and hole interchangeably.}  The Lagrangian algebra $\mathcal{A}$ consists of a collection of bulk bosonic anyons that can be condensed to vacuum at the boundary, and the corresponding gapped boundary is a condensate of those anyons which behaves as a non-abelian anyon of quantum dimension $d_{\mathcal{A}}$.  Lagrangian algebras in $\mathcal{B}=\mZ(\mathcal{C})$ are in one-to-one correspondence with indecomposable module categories $\mathcal{M}$  over $\mathcal{C}$, which can also be used to label gapped boundaries.

A route to creating, manipulating, and measuring topological degeneracy for gapped boundaries in $\mfD(\Z_3)$ in bilayer fractional quantum Hall states coupled to superconductors  has been presented \cite{Bark16}.   Other experimentally reasonable designs proposed for realizing topological degeneracy from gapped boundaries in abelian fractional quantum Hall states \cite{clarke2013, cheng2012, lindner2012, Bark16, GGGG}. Moreover, a linear array of $9$ qubits \cite{Kelly15} and a square of $4$ qubits \cite{Corcoles15} on the $\Z_2$ surface code code gapped boundaries are also being developed to experimentally realize gapped boundaries of the toric code with superconducting integrated circuits.

\subsection{Main Results}

Due to the length of this paper, we will provide a summary of our main results in this section of the Introduction.

The body of our paper is divided into five major chapters: In Chapter \ref{sec:hamiltonian}, we present the Hamiltonian realization of gapped boundaries, boundary defects, and bulk-to-boundary condensation. Chapter \ref{sec:algebraic} presents algebraic models for these same physical processes using category theory. Chapter \ref{sec:circuits} demonstrates a potential implementation using surface codes. Finally, in Chapter \ref{sec:operations}, we present the topologically protected operations on gapped boundaries, and in Chapter \ref{sec:uqc}, we show how to use these operations to perform universal quantum computation.

\subsubsection{Hamiltonian realizations}

Suppose we are given a finite group $G$.  Consider a large rectangle $\Lambda$ of the square lattice $\mathbb{Z}^2$  in the plane $\mathcal{R}$.  Let $V(\Lambda), E(\Lambda), F(\Lambda)$ be the set of vertices (sites), edges (bonds or links), and faces (plaquettes) of $\Lambda$, respectively.  We attach a qudit or spin in $\mathbb{C}[G]$ to each edge $e\in E(\Lambda)$, so the local Hilbert space for the quantum system is $\mathcal{L}=\otimes_{e\in E(\Lambda)} \mathbb{C}[G]$.  The Kitaev Hamiltonian $H=-\sum_{v}A(v)-\sum_{p}B(p)$  for the Dijkgraaf-Witten theory based on the finite group $G$ consists of two kinds of terms: the vertex term $A(v)$ at each vertex $v$, which enforces a Gaussian law, and the plaquette term $B(p)$ at each plaquette $p$, which enforces the zero-flux condition.  This Kitaev quantum double Hamiltonian is a discrete gauge theory based on a finite group $G$ with topological charges labeled by pairs $(C_g,\pi)$, where $C_g$ is a conjugacy class and $\pi$ is an irreducible representation of the centralizer $E(C_g)$ of a representative $g\in C_g$.  

To generate a hole (with a gapped boundary), we modify Kitaev's Hamiltonian using the two-parameter Hamiltonians presented by Bombin and Martin-Delgado in \cite{Bombin08}.  Our resulting Hamiltonian is different from the one presented by Beigi et al. in Ref. \cite{Beigi11}, because raw data qudits still exist beyond the boundary in our model, and it is different from the situation discussed by Bombin and Martin-Delgado in Ref. \cite{Bombin11} because we have given an explicit construction of the boundary region and all Hamiltonian terms that act on it.  In general, irreducible hole types can be labeled by Lagrangian algebras of the Drinfeld center $\mfD(G) = \mZ(\textrm{Vec}_G)=\mZ(\textrm{Rep}(G)) = \Rep(D(G))$ (or equivalently, indecomposable module categories of $\textrm{Vec}_G$). In the Dijkgraaf-Witten theory for a finite group $G$, the different irreducible hole types are parameterized by subgroups $K\subseteq G$ up to conjugation, which can be directly used to construct indecomposable module categories over $\textrm{Vec}_G$.  

A hole $\mathfrak{h}$ is a small rectangle inside the rectangle $\Lambda$, which separates $\Lambda$ into two parts (see Fig.  \ref{fig:main-results}): the small rectangle $\mathfrak{h}$ and its outside $\mathfrak{B}$, which is considered to be the bulk of the topological phase of matter.  The vertices, edges, and faces of $\Lambda$ are also divided into two disjoint subsets: $V(\Lambda)=V(\mathfrak{B})\sqcup V(\mathfrak{h}), E(\Lambda)=E(\mathfrak{B})\sqcup E(\mathfrak{h}),  F(\Lambda)=F(\mathfrak{B})\sqcup F(\mathfrak{h})$. The subset $V(\mathfrak{h})$ consists of all vertices in the hole $\mathfrak{h}$ and those on its boundary; the subset $E(\mathfrak{h})$ consists of all edges in the hole, but not those on its boundary.  There is no confusion as to whether a face of $\Lambda$ is in  $\mathfrak{B}$ or $\mathfrak{h}$.

\begin{figure}
\centering
\includegraphics[width = 0.9\textwidth]{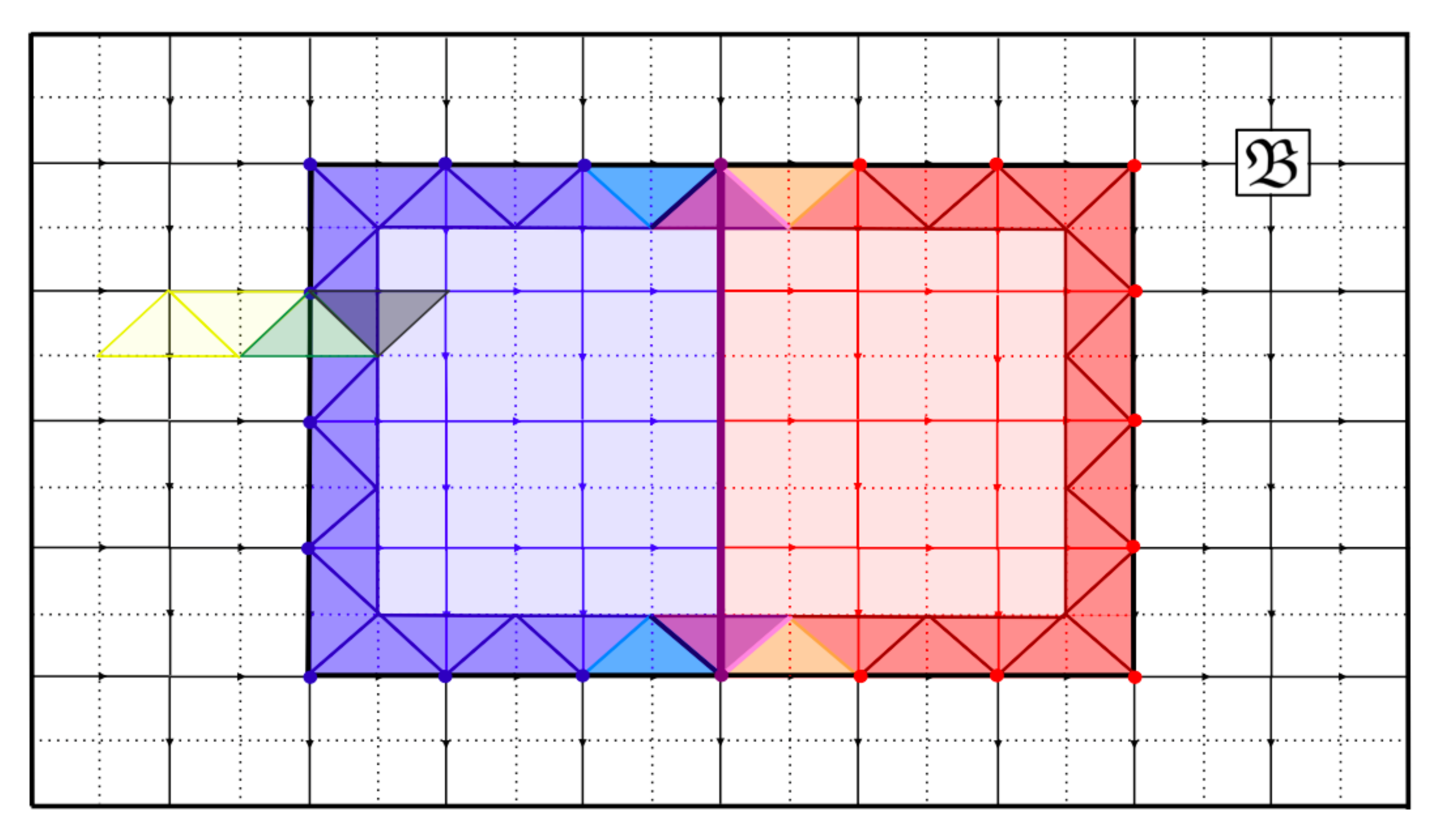}
\caption{Pictorial summary of the Hamiltonian realization and algebraic model of gapped boundaries. In this picture, the hole $\mathfrak{h}$ is the inner boldfaced and multi-colored rectangle.} 
\label{fig:main-results}
\end{figure}

Our Hamiltonian consists of two parts, $H_{(G,1)}$ and $H_{(G,1)}^{(K,1)}$, which act on the bulk and hole parts of the local Hilbert space $\mathcal{L}$ as follows:

\begin{equation}
H_{\text{G.B.}} = H_{(G,1)} (\mfB) + H_{(G,1)}^{(K,1)} (\mfh)
\end{equation}

\noindent
The bulk Hamiltonian is the same as the Kitaev Hamiltonian and acts on $\mathfrak{B}$; the hole Hamiltonian is given by $H_{(G,1)}^{(K,1)}=-\sum_{v\in V(\mathfrak{h})}A_v^K-\sum_{p\in F(\mathfrak{h})}B_{p}^K-\sum_{e\in E(\mathfrak{h})} (L^{K}(e)+T^{K}(e))$ and acts on $\mathfrak{h}$. $H_{(G,1)}^{(K,1)}$ has $4$ kinds of terms: the vertex term $A^{K}(v)$ is similar to Kitaev's vertex term, but is generalized to project the vertex to a trivial representation of $G$ when restricted to $K$; the plaquette term $B^K(p)$ is extended to include all fluxes $k\in K$; the two kinds of single qudit edge terms $L^K(e)$ and $T^K(e)$ are Zeeman-like, and explicitly break the gauge group from $G$ to $K$.  All local terms commute with each other.  The new edge terms $L^K(e)$ and $T^K(e)$ confine all anyons in the discrete gauge theory with gauge group $K$ in the hole.

The input finite group $G$ for the Kitaev quantum double model should be regarded as the fusion category $\textrm{Vec}_G$ of $G$-graded vector spaces.  Then, the topological order of the topological phase of matter represented by this Hamiltonian is the Drinfeld center $\mfD(G)$, which corresponds to the Dijkgraaf-Witten TQFT.  The basic property of a gapped boundary is that a collection of bulk anyons can condense to (or be created out of) the gapped boundary at zero energy cost.  This collection of anyons forms a Lagrangian algebra in the topological order $\mathcal{Z}(\textrm{Vec}_G)$.  New topological degeneracy arises in the presence of gapped boundaries, because certain local degrees of freedom are now protected by boundary Hamiltonians.  The energy splitting of the topological degeneracy scales as $e^{-\frac{L}{\xi}}$ when the holes are separated far apart, where $L$ is the size of the hole and $\xi$ is the correlation length.

The main technical tool to analyze the Hamiltonian is the ribbon operators.  To define ribbon operators both in the bulk and on the boundary precisely, we introduce the following terminologies.  The lattice $\Lambda$ (solid lines in Fig. \ref{fig:main-results}, a.k.a. the {\it direct} lattice) is subdivided by taking its intersection with its {\it dual} lattice $\hat{\Lambda}$ (dotted lines in Fig. \ref{fig:main-results}) in the plane.  A {\it cilium} is a pair $(v,p)$, where $p$ is a plaquette of $\Lambda$ and $v$ is a vertex on $p$. A cilium is illustrated as a line segment from $v$ to the center of $p$ (the colored diagonal line segments in Fig. \ref{fig:main-results}). A {\it triangle} is formed by two cilia from the same (direct or dual) vertex joined with an edge from either $\Lambda$ or $\hat{\Lambda}$ (see Fig. \ref{fig:main-results}).  There are two kinds of triangles in the hybrid lattice $\Lambda\cup \hat{\Lambda}$: a triangle is direct (dual) if its horizontal or vertical edge is anchored on the direct (dual) lattice. Each triangle supports a qudit---the qudit on the horizontal or vertical edge of the triangle. These definitions are presented in more detail in Section \ref{sec:ribbon-operators}.

A ribbon is a chain of triangles, alternating direct/dual, from one cilium to another. A ribbon operator is an operator supported on a ribbon---acting trivially on the qudits not on the ribbon---that commutes with the Hamiltonian except at the two end cilia.  Ribbon operators become string operators if they act trivially on qudits supported by all direct (or all dual) triangles. In the Kitaev model, a ribbon operator $F^{(h,g)}$ along a ribbon creates a pair of magnetic fluxes $(h,h^{-1})$ at the two end cilia and an electric flux $g$ along the ribbon. In general, the  excitations at the two cilia do not have well-defined topological charges, but a beautiful Fourier transform presented in \cite{Bombin08} expresses the topological charges as superpositions of the ribbon operators $F^{(h,g)}$ (see Section \ref{sec:ribbon-operators} for details).  

To extend ribbon operators to the boundary, we define the {\it boundary ribbon} of a hole $\mathfrak{h}$ to be the closed ribbon consisting of all direct triangles inside $\mathfrak{h}$ anchored on the border of $\mathfrak{h}$ and dual triangles anchored on the smaller square in the dual lattice (blue and red triangles in Fig. \ref{fig:main-results}).  The different colored triangles represent various functors in the algebraic model.

Our first technical result is a generalization of ribbon operators to the boundary. Boundary ribbon operators are operators supported on a ribbon that commute with all vertex and plaquette terms of $H^{(K,1)}_{(G,1)}$ except at the two end cilia, and can create all excitations that can result from dragging a bulk anyon to the boundary. As in the bulk case, they do not necessarily have definite \lq\lq topological charges".  Motivated by the formula in \cite{Bombin08} for bulk excitations, we develop a Fourier transform to express the irreducible types of boundary excitations. This parametrizes the boundary elementary excitation types as

\begin{equation}
\{(T,R): \text{ } T = K r_T K \in K\backslash G/K, \text{ } R \in (K^{r_T})_{\text{ir}}\},
\end{equation}

\noindent
where $K^{r_T}=K\cap r_TKr_T^{-1}$ is a stabilizer group. The quantum dimensions of these excitations are given by

\begin{equation}
\textrm{FPdim}(T,R)=\frac{|K|\textrm{dim}(R)}{|K^{r_T}|}.
\end{equation}

\noindent
We provide a simple and systematic method to determine how a bulk anyon $(C,\pi)$ condenses into the boundary into $(T,R)$'s and vice versa. These details are presented in Sections \ref{sec:bd-hamiltonian}-\ref{sec:bd-excitations}.

One of our most interesting contributions is a microscopic theory for boundary defects between different boundary types (i.e. different subgroups $K_1, K_2 \subseteq G$), which behave like non-abelian objects such as the Majorana and para-fermion zero modes. We design an exactly solvable Hamiltonian to create these defects in the quantum double model. As before, this Hamiltonian is a combination of $H_{(G,1)}$ and $H_{(G,1)}^{(K,1)}$:

\begin{equation}
H_{\text{dft}} = H_{(G,1)} (\mfB) + H^{(K_1,1)}_{(G,1)} (\mfr_1) + H^{(K_2,1)}_{(G,1)} (\mfr_2) + H^{(K_1 \cap K_2,1)}_{(G,1)} (\mfL).
\end{equation}

\noindent
Here, $\mfr_1$, $\mfr_2$ are the blue and red regions in Fig. \ref{fig:main-results}, respectively, and $\mfL$ is the purple line dividing them. As before, the bulk $\mfB$ consists of everything else in the lattice.

We next analyze the topological properties of boundary defects. We find that the simple defect types are parametrized by

\begin{equation}
\{
(T,R): T \in K_1 \backslash G / K_2, \text{ } R \in ((K_1, K_2)^{r_T})_{\text{ir}} 
\},
\end{equation}

\noindent
where $(K_1, K_2)^{r_T} = K_1 \cap r_T K_2 r_T{-1}$ for some representative $r_T \in T$, and their quantum dimensions are given by

\begin{equation}
\FPdim(T,R) = \frac{\sqrt{|K_1| |K_2|}}{|(K_1, K_2)^{r_T}|} \cdot \Dim(R).
\end{equation}

\noindent
These defects also generate topological degeneracy and can be used for topological quantum computation.  An important class of such defects are equivalent to genons in bilayer systems \cite{Bark13a,Bark13b,Bark13c}, and our Hamiltonian generates such genons in bilayer $\mathfrak{D}(G)$ theories.  Genons in bilayer Ising theory can be used to provide the missing $\frac{\pi}{8}$-gate in Ising theory, and make bilayer Ising theory universal \cite{Barkeshli16}.  A similar protocol makes the doubled Ising theory universal when it is enhanced with gapped boundaries \cite{Barkeshli16}.

With the Hamiltonians, we can derive many properties of gapped boundaries and boundary defects. Gapped boundaries can be created or annihilated one at a time, unlike anyons (which have to be created in pairs). We believe that boundary defects share many properties with bulk anyons, and so must also be created in pairs. Using adiabatic Hamiltonian tuning, both defects and gapped boundaries can be moved, and therefore both may be braided. We do not know if it is possible or how to fuse different holes.

In principle, we can derive all properties for the gapped boundaries and boundary defects from the microscopic Hamiltonians. However, due to the high dimension of the corresponding Hilbert spaces, this is not easy in practice (just as in the case of the bulk). Hence, we develop an algebraic theory, where the main tool is the extension of the modular tensor category formalism to holes and their boundaries.  Category theory is well suited for studying topological properties of quantum systems without local states, since it is a formulation of set theory without elements. Hence, this formalism will be rigorously developed in Chapter \ref{sec:algebraic} and is outlined below.

\subsubsection{Algebraic theory}

In the categorical formalism, the bulk of a TQFT is given by a modular tensor category $\B = \mZ(\mC)$ for some unitary fusion category $\mC$, and a (gapped) hole is a Lagrangian algebra $\mathcal{A}=\oplus_{a}n_a a$ in $\B$. In the case of Dijkgraaf-Witten theories, we have $\mC = \text{Vec}_G$. For most purposes, $\A$ can be regarded as a (composite) non-abelian anyon of quantum dimension $d_{\mathcal{A}}$. Gapped boundaries are in one-to-one correspondence to indecomposable module categories $\mathcal{M}_i$ over $\mC$. Then, elementary excitations on $\mathcal{M}_i$ are the simple objects in the functor fusion category $\mC_{ii} = \textrm{Fun}_{\mC}(\mathcal{M}_i, \mathcal{M}_i)$, and simple
boundary defects between two gapped boundaries $\mathcal{M}_i, \mathcal{M}_j$ are the simple objects in the bimodule category  $\mC_{ij} = \textrm{Fun}_{\mC}(\mathcal{M}_i, \mathcal{M}_j)$.  In this formalism, the condensation functor is a tensor functor from $\mathcal{Z}(\mC)$ to $\mC_{ii}$. The collections of fusion categories $\mC_{ii}$ and their bimodule categories $\mC_{ij}$ form a multi-fusion category $\mathfrak{C}$. From this multi-fusion category, we can find quantum dimensions of both boundary excitations and the defects between gapped boundaries. We also find that the fusion of boundary defects is given by the sequence

\begin{equation}
\mC_{ij} \otimes \mC_{jk} \rightarrow (\mC_{ij} \boxtimes \mC_{ij}^{\text{op}}) \otimes (\mC_{jk} \boxtimes \mC_{jk}^{\text{op}})
\rightarrow
\mZ(\mfC)^{\otimes 2}
\rightarrow
\mZ(\mfC)
\rightarrow
\mC_{ik} \boxtimes \mC_{ik}^{\text{op}} 
\rightarrow 
\mC_{ik}
\end{equation}

\noindent
Here, each arrow represents a functor between categories. Hence, fusion of defects mainly occurs in the doubled category $\mZ(\mfC)$, which is equivalent to the Drinfeld center of the original input category $\mC$ \cite{Chang15}. This allows defects to be braided among themselves and with the bulk anyons. If $x \in \mC_{ij}$ and $y \in \mC_{ji}$ are boundary defects, the topological degeneracy in the fusion of $x$ and $y$ is given by the degeneracy of the hom-space $\Hom(\one_{\M_i}, x \otimes y)$, where $\one_{\M_i}$ is the tensor unit of the fusion category $\mC_{ii}$. This is easily generalized to the case where there are $n$ boundary defects.

Similarly, topological degeneracies in the presence of holes $\mathfrak{h_i}$ labeled by $\mathcal{A}_i$ are described by the morphism space $\textrm{Hom}(\mathcal{A_\infty}, \otimes_i\mathcal{A}_i)$, where the outermost boundary is labeled by $\mathcal{A_\infty}$. $\mathcal{A_\infty}$ can be either an anyon type or a boundary type.

\subsubsection{Surface code implementation}

Inspired by the surface code approach to fault-tolerant quantum computation, we turn to the quantum computing side of gapped boundaries.  We treat the edge qudit in $\mathbb{C}[G]$ as a data qudit and the vertex and plaquette terms of the Hamiltonian as syndrome qudits on the vertices and plaquettes.  The syndrome and data qudits are considered as local physical qudits, while qudits encoded in ground states of gapped boundaries are logical ones.  We are interested in which quantum gates on the logical qudits can be realized by low depth and efficient quantum circuits on the physical ones.  

\subsubsection{Topologically protected operations and universal quantum computation}

Finally, in Chapters \ref{sec:operations} and \ref{sec:uqc}, we discuss topological quantum computation using gapped boundaries. 

Examples of topologically protected operations such as tunneling and loop operators can be computed using all data of the modular tensor category $\B$ and the indecomposable modules $\M_i$, as discussed in Chapter \ref{sec:operations}. Furthermore, we can braid holes around each other to obtain a representation of the pure braid group. In general, these braids may be used to produce two- or multiple-qudit entangling gates.

In the topological degeneracy manifold $\textrm{Hom}(\mathcal{A_\infty}, \otimes_i\mathcal{A}_i)$, there are no natural tensor structures to encode qudits for quantum computing.  For gapped boundaries in our models, they also behave like integral non-abelian anyons. By the property F conjecture \cite{Naidu09}, braidings alone probably would not be sufficient to achieve universal quantum computing. Therefore, to achieve universality, we must supplement braiding with extra topological operations such as topological charge measurement.

In Chapter \ref{sec:uqc}, we analyze two concrete examples, namely $\mathfrak{D}(S_3)$ and $\mathfrak{D}(\Z_3)$. There are several existing schemes to make $\mfD(S_3)$ universal. We believe that introducing gapped boundaries to the $\mfD(S_3)$ anyon theory would provide an elegant improvement over the existing universal gate set. Finally, we are able to produce a universal qutrit gate set using purely gapped boundaries of $\mathfrak{D}(\Z_3)$. This theory is particularly interesting, because the gapped boundaries can be potentially realized in bilayer fractional quantum Hall states \cite{Bark16}. The accomplishment is especially significant, as $\mathfrak{D}(\Z_3)$ is an abelian theory, so braidings of anyons in the plane are all projectively trivial. In fact, this is the first purely topological method (i.e. it does not use external high-fidelity state injection) to obtain a universal quantum computation model using only an abelian theory.

\subsection{Previous Works}

The first example of gapped boundaries appeared in \cite{Bravyi98} as the smooth and rough boundaries of the $\Z_2$ toric code. Boundaries of Kitaev's quantum double model for general finite groups $G$ were studied by Beigi et al. in \cite{Beigi11}. In that work, they generated gapped boundaries with a different Hamiltonian and described condensations to vacuum, but left the description of the boundary excitations as an important open problem.  In 2009, Kitaev contemplated the categorical formulation that a gapped boundary is modeled by a condensable Frobenius algebra \cite{Kitaev09}. Later, a related categorical description with some details is outlined in \cite{KitaevKong, Fuchs2014, Kapustin1410b}.  Further clarifications appeared in \cite{Kong} on boundary excitations, but no explicit general solvable Hamiltonian is presented.  The mathematics of such a theory is in \cite{KO,Davydov12}. 

A physical theory of gapped boundaries related to defects for abelian topological phases of matter is developed in \cite{ Levin13, Bark13a,Bark13b,Bark13c, Kapustin14}. For recent works on a physical understanding of gapped boundaries in more general topological phases and in the closely related topic of anyon condensation, see \cite{Kapustin10, Bais09, Kong13, Eliens13, LWW, Kong15, Neupert16a, Neupert16b, Wan16}. Ref. \cite{Bark13a} gives a universal quantum computing gate set from braiding anyons in the bilayer Ising theory supplemented with genons (a special case of boundary defects).  Topological degeneracy has been studied using various techniques in \cite{Kapustin14, Bark13b, LWW,WW,HW}. Surface code implementations of gapped boundaries of the $\Z_2$ toric code have been studied by \cite{Dennis02,Fowler12}, especially for applications to
quantum information processing.
Braiding of gapped boundaries has been used to produce quantum gates in \cite{Bombin11,Fowler12,Raussendorf03}.
Topological charge projection has been introduced recently to produce more topological quantum gates, and can be used to produce a universal gate set based on the doubled Ising theory \cite{Barkeshli16}.

\subsection{Notations}

The notations we adopt throughout the paper are presented in Appendix \ref{sec:notations}.  

Throughout the paper, all algebras and tensor categories are over the complex numbers $\C$.  All fusion and modular tensor categories are unitary.  Unitary fusion categories are spherical.

\subsection{Acknowledgment}
The authors thank Maissam Barkeshli, Shawn Cui, and Cesar Galindo for answering many questions.  We thank Alexei Davydov for pointing out the example that two different Lagrangian algebras can have the same underlying object. I.C. would like to thank Michael Freedman and Microsoft Station Q for hospitality in hosting the summer internship and visits during which this work was done. M.C. thanks Chao-Ming Jian for collaborations on related topics. Z.W. is partially supported by NSF grants DMS-1108736 and DMS-1411212.

\vspace{2mm}
\section{Hamiltonian realization of gapped boundaries}
\label{sec:hamiltonian}

In this chapter, we present a Hamiltonian realization of gapped boundaries in any Kitaev quantum double model for the (untwisted) Dijkgraaf-Witten theory based on a finite group $G$. Sections \ref{sec:kitaev-hamiltonian} and \ref{sec:ribbon-operators} review existing works on the Hamiltonian for the standard Kitaev model and its corresponding ribbon operators. In Section \ref{sec:bd-hamiltonian}, we review existing works on Hamiltonians for Kitaev models with boundaries, and develop our own Hamiltonian that is best suited towards topological quantum computation with gapped boundaries. Section \ref{sec:hamiltonian-gsd-condensation} presents the ground state degeneracy for this model and describes our topological qudit encoding. Section \ref{sec:bd-excitations} classifies the elementary excitations on the boundary and systematically describes the bulk-to-boundary condensation procedure. In Section \ref{sec:defect-hamiltonian}, we further generalize the Hamiltonians presented in this chapter, to consider cases where two distinct boundaries of the Kitaev model meet and form a defect. We analyze topological properties of defects such as the simple defect types and their quantum dimensions. Finally, in Sections \ref{sec:tc-hamiltonian-example} and \ref{sec:ds3-hamiltonian-example}, we provide concrete examples for all of the theory by considering the toric code and $\mfD(S_3)$. Section \ref{sec:genon-hamiltonian} discusses a particular example of boundary defects, namely genons in the bilayer theory $\mfD(G\times G)$ for any finite group $G$.

\subsection{Kitaev quantum double models}
\label{sec:kitaev-hamiltonian}

Kitaev's famous toric code paper \cite{Kitaev97} presents a model for topological quantum computation on a general lattice based on any finite group $G$. For simplicity of illustration and calculation, we will assume throughout our paper that the lattice is the square lattice in the plane; however, it is clear that all of the developed theory here extends to arbitrary lattices. In this model, a data qudit is placed on each edge of the lattice, as shown in Fig. \ref{fig:kitaev}. The Hilbert space for each qudit has an orthonormal basis given by $\{\ket{g}: g \in G\}$, so the total Hilbert space is $\mathcal{L}=\otimes_{e}\mathbb{C}[G]$.

\begin{figure}
\centering
\includegraphics[width = 0.65\textwidth]{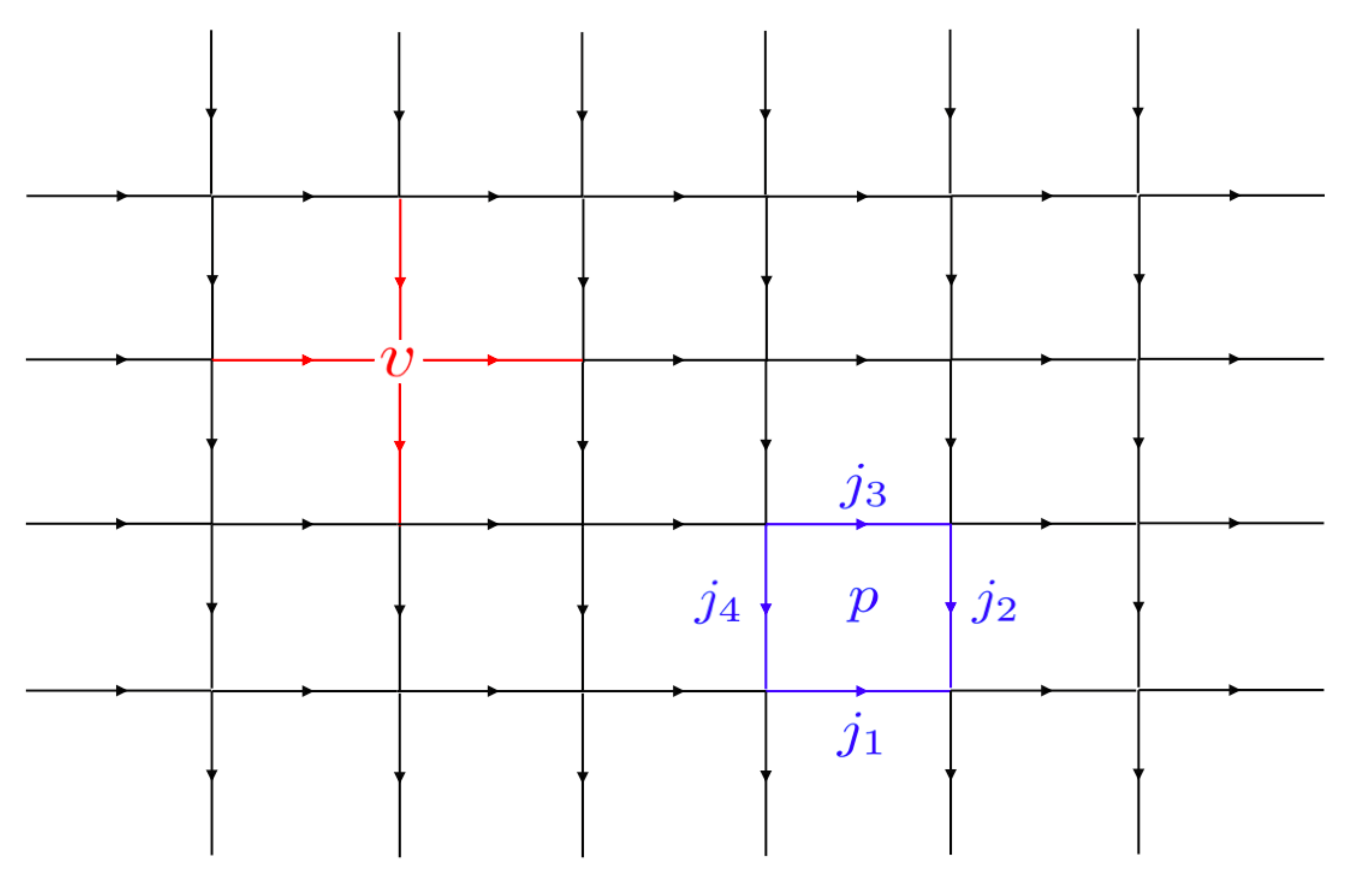}
\caption{Lattice for the Kitaev model. For simplicity of illustration and calculation, we use a square lattice, but in general, one can use an arbitrary lattice. If the group $G$ is nonabelian, it is necessary to define orientations on edges, as we have shown here. The edges $j$ and $j_1,...j_m$, used to obtain $A^g(v)$ and $B^h(p)$, are illustrated for this example of $v,p$.}
\label{fig:kitaev}
\end{figure}

As discussed in Ref. \cite{Kitaev97}, a Hamiltonian is used to transform the high-dimensional Hilbert space of all data qudits into a topological encoding. This Hamiltonian is built from several basic operators on a single data qudit:

\begin{equation}
L^{g_0}_+ \ket{g} = \ket{g_0g}
\end{equation}
\begin{equation}
L^{g_0}_- \ket{g} = \ket{gg_0^{-1}}
\end{equation}
\begin{equation}
T^{h_0}_+ \ket{h} = \delta_{h_0,h}\ket{h}
\end{equation}
\begin{equation}
T^{h_0}_- \ket{h} = \delta_{h_0^{-1},h}\ket{h}
\end{equation}

\noindent
where $\delta_{i,j}$ is the Kronecker delta function. These operators are defined for all elements $g_0,h_0 \in G$, and provide a faithful representation of the left/right multiplication and comultiplication in the Hopf algebra $\C[G]$. Using these operators, local gauge transformations and magnetic charge operators are defined as follows, on each vertex $v$ and plaquette $p$ \cite{Kitaev97}:

\begin{equation}
\label{eq:kitaev-vertex-g-term}
A^{g}(v,p) = A^{g}(v) = \prod_{j \in \text{star}(v)} L^g(j,v)
\end{equation}

\begin{equation}
\label{eq:kitaev-plaquette-h-term}
B^h(v,p) = \sum_{h_1 \cdots h_k = h} \prod_{m=1}^k T^{h_m}(j_m, p)
\end{equation}

Here, $j_1, ..., j_k$ are the boundary edges of the plaquette $p$ in counterclockwise order (see Fig. \ref{fig:kitaev}), and $L^g$ and $T^h$ are defined as follows: if $v$ is the origin of the directed edge $j$, $L^g(j,v) = L^g_-(j)$, otherwise $L^g(j,v) = L^g_+(j)$; if $p$ is on the left (right) of the directed edge $j$, $T^h(j,p) = T^h_-(j)$ ($T^h_+(j)$) \cite{Kitaev97}.

Note that since the $A^g(v)$ satisfy $A^g(v) A^{g'}(v) = A^{gg'}(v)$, the set of all $A^g(v)$ (for fixed $v$) form a representation of $G$ on the entire Hilbert space $\mathcal{L}=\otimes_{e}\mathbb{C}[G]$ of all data qudits \cite{Bombin08}.

In fact, we can define operators

\begin{equation}
\label{eq:bulk-local-operators}
D^{(h,g)} (v,p) = B^h (v,p) A^g (v,p)
\end{equation}

\noindent
that act on a {\it cilium} $s = (v,p)$, where $v$ is a vertex of $p$. These operators act locally, and they form the basis of a quasi-triangular Hopf algebra $\mathcal{D} = \Span\{D^{(h,g)}\}$, the {\it quantum double} $D(G)$ of the group $G$. The specific multiplication, comultiplication, and antipode for the Hopf algebra are presented in Ref. \cite{Kitaev97}. As vector spaces, we have

\begin{equation}
D(G) = F[G] \otimes \C[G],
\end{equation}
where $F[G]$ are complex functions on $G$.

In the next sections, we will see the importance of these local operators in determining the topological properties of excitations in this group model.

Finally, two more linear combinations of these $A^g$ and $B^h$ operators are required to define the Hamiltonian:

\begin{equation}
\label{eq:kitaev-vertex-term}
A(v) = \frac{1}{|G|} \sum_{g \in G} A^g(v,p) 
\end{equation}

\begin{equation}
\label{eq:kitaev-plaquette-term}
B(p) = B^1(v,p).
\end{equation}

\noindent
The Hamiltonian\footnote{Note: We call this Hamiltonian $H_{(G,1)}$, as this model is the Dijkgraaf-Witten theory with trivial cocycle (twist). In general, this Hamiltonian may be twisted by a 3-cocycle $\omega \in H^3(G,\C^\times)$, and may be written as $H_{(G,\omega)}$.} for the Kitaev model is then defined as\footnote{For a physical implementation of this Hamiltonian using quantum circuits, see Section \ref{sec:circuits}.}

\begin{equation}
\label{eq:kitaev-hamiltonian}
H_{(G,1)} = \sum_v (1-A(v)) + \sum_p (1-B(p))
\end{equation}

It is important to note that all terms in the Hamiltonian $H_{(G,1)}$ commute with each other. By the spectral theorem, this means that these operators share simultaneous eigenspaces. Each individual operator is a projector and has eigenvalues $\lambda = 0,1$. (Specifically, the $A(v)$ terms project onto the trivial representation, and the $B(p)$ terms project onto trivial flux \cite{Bombin08}.) The ground state of the Hamiltonian corresponds to the eigenspace with overall eigenvalue (energy) $\lambda = 0$, and states with excitations will have positive energy. Here, an excitation or a quasi-particle is defined so that exactly one of terms $(1-A(v))$ and one of the terms $(1-B(p))$ is in the $\lambda = 1$ eigenstate; we say the quasi-particle is located at the cilium $s=(v,p)$. The resulting quantum encoding is hence ``topological'': regardless of how densely we place the data qudits, there will always be a constant energy gap between the ground state and the first excited state, and between each excited state.

\subsection{Ribbon operators}
\label{sec:ribbon-operators}

In this section, we review the algebra of bulk ribbon operators for the Kitaev models as presented in Refs. \cite{Bombin08,Kitaev97}. These definitions play a crucial role in this chapter, as a major contribution of this chapter will be the presentation of boundary ribbon operators in Sections \ref{sec:bd-hamiltonian} and \ref{sec:bd-excitations}.

\subsubsection{Basic definitions}

Before we proceed to define the algebra of ribbon operators, let us first review the following basic definitions. In what follows, the {\it direct} lattice will denote the original lattice of the Kitaev model (cf. the {\it dual} lattice, in which vertices and plaquettes of the direct lattice are switched). Both lattices are shown in Fig. \ref{fig:ribbon-defs}.

\begin{figure}
\centering
\includegraphics[width = 0.65\textwidth]{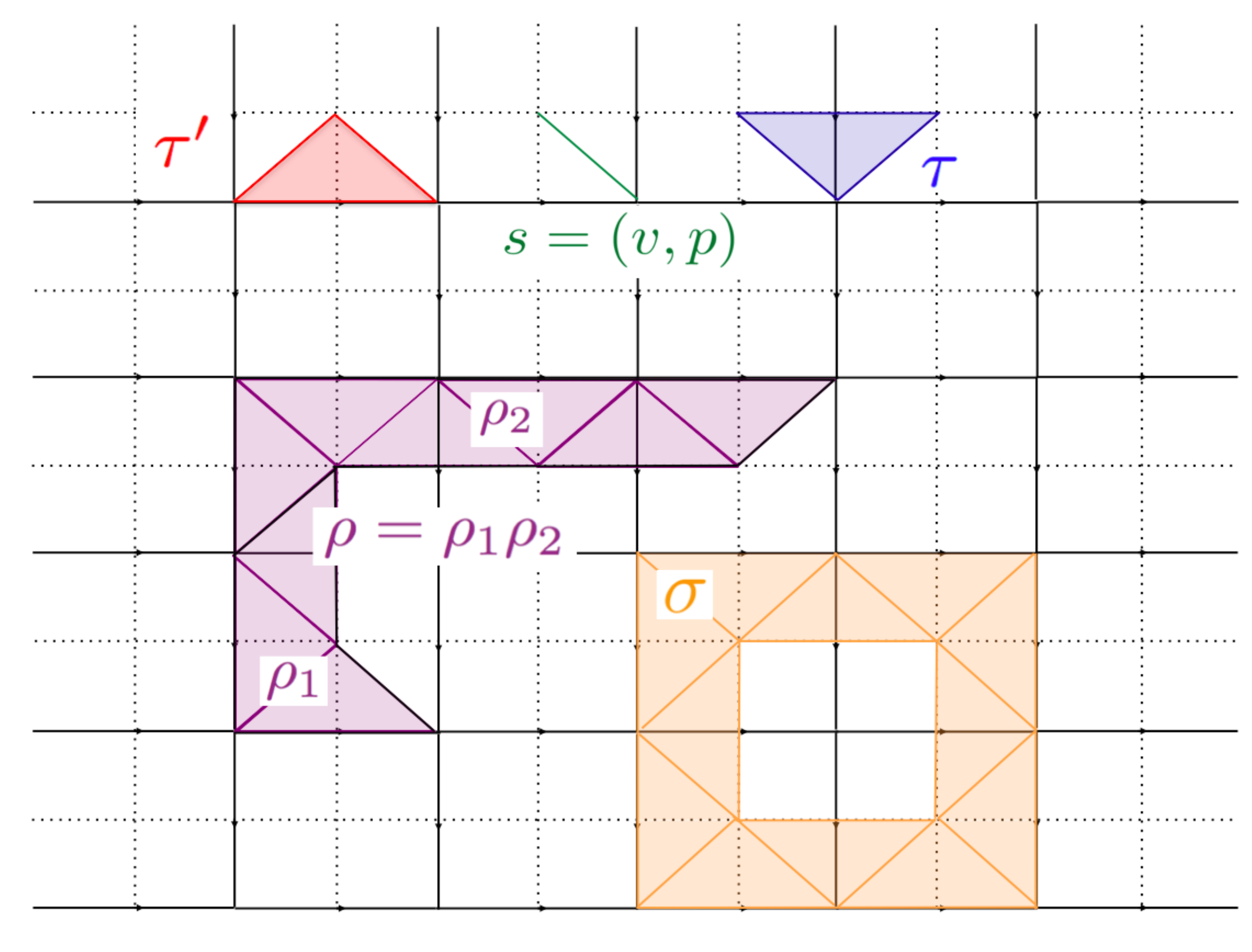}
\caption{Illustration of Definitions \ref{cilium-def}-\ref{ribbon-def}. The direct lattice is shown as before, and the dual lattice is shown in dotted lines. $s=(v,p)$ is a cilium. $\tau$ is a dual triangle, and $\tau'$ is a direct triangle. $\rho = \rho_1 \rho_2$ is a composite ribbon, formed by gluing the last site of $\rho_1$ to the first site of $\rho_2$. $\rho$ is an open ribbon, and $\sigma$ is a closed ribbon.}
\label{fig:ribbon-defs}
\end{figure}

\begin{definition}
\label{cilium-def}
A {\it cilium} is a pair $s = (v,p)$, where $p$ is a plaquette in the lattice, and $v$ is a vertex of $p$. These are visualized (e.g. in Fig. \ref{fig:ribbon-defs}) as lines connecting $v$ to the center of $p$ (i.e. the dual vertex corresponding to $p$).
\end{definition}

\begin{definition}
\label{triangle-def}
A {\it direct (dual) triangle} $\tau$ consists of two adjacent cilia $s_0,s_1$ connected via an edge $e$ on the direct (dual) lattice, as shown in Fig. \ref{fig:ribbon-defs}. We write $\tau = (s_0, s_1, e) = (\partial_0 \tau, \partial_1 \tau, e)$, listing sides in counterclockwise order if $\tau$ is direct, and clockwise order if $\tau$ is dual. Throughout the chapter, $\tau$ will be used to denote a dual triangle, and $\tau'$ a direct triangle.
\end{definition}

\begin{definition}
\label{ribbon-def}
A {\it ribbon} $\rho$ is an oriented strip of triangles $\tau_1, ... \tau_n$, alternating direct/dual, such that $\partial_1 \tau_i = \partial_0 \tau_{i+1}$ for each $i = 1,2,...n-1$, and the intersection $\tau_i \cap \tau_j$ has zero area if $i \neq j$ (i.e. $\rho$ does not intersect itself).

$\rho$ is said to be {\it closed} if $\partial_1 \tau_n = \partial_0 \tau_1$. $\rho$ is {\it open} if it is not closed. Examples of closed and open ribbons are shown in Fig. \ref{fig:ribbon-defs}.
\end{definition}

\begin{definition}
\label{open-ribbon-operator-def}
Let $\rho$ be an open ribbon, with endpoint cilia $s_0 = (v_0, p_0)$, $s_1 = (v_1, p_1)$. A {\it ribbon operator} on $\rho$ is an operator $F_\rho$ that commutes with all terms of the Hamiltonian (\ref{eq:kitaev-hamiltonian}) except possibly the terms corresponding to $v_0,p_0,v_1,$ or $p_1$.
\end{definition}

The goal of this section is hence to determine the algebra $\mathcal{F}$ of ribbon operators in the bulk of the Kitaev model.

\subsubsection{Triangle operators and the gluing relation}

The ribbon operators are defined recursively \cite{Bombin08,Kitaev97}. As discussed in Refs. \cite{Bombin08,Kitaev97}, given a ribbon $\rho$, the set of ribbon operators on $\rho$ has basis elements $F^{(h,g)}_\rho$ indexed by two elements of $G$. The simplest ribbon is the empty ribbon $\epsilon$, for which the ribbon operators are given by

\begin{equation}
F^{(h,g)}_\epsilon = \delta_{1,g}.
\end{equation}

The next simple case is when $\rho$ is a single triangle. Let $\tau = (s_0,s_1,e)$ be any dual triangle, and let $\tau'=(s_0',s_1',e')$ be any direct triangle. The ribbon operators are defined as follows:

\begin{equation}
\label{eq:triangle-operator-def}
F^{(h,g)}_\tau := \delta_{1,g} L^h (e), \qquad F^{(h,g)}_{\tau'} := T^g (e')
\end{equation}

\noindent
In this definition, the choice of $+$ or $-$ for the $L,T$ operators is determined by the orientation of the edge on each triangle.

Finally, we define a ``gluing relation'' on ribbon operators. Let $\rho = \rho_1 \rho_2$ be the ribbon formed by gluing the last cilium of $\rho_1$ to the first cilium of $\rho_2$ (see Fig. \ref{fig:ribbon-defs}). We define the ribbon operator on this composite ribbon to be

\begin{equation}
\label{eq:ribbon-gluing}
F^{(h,g)}_\rho := \sum_{k \in G} F^{(h,k)}_{\rho_1} F^{(k^{-1}hk,k^{-1}g)}_{\rho_2}.
\end{equation}

It is simple to check that this definition makes $F^{(h,g)}_\rho$ independent of the particular choice of $\rho_1,\rho_2$ \cite{Bombin08}.

The operators $F^{(h,g)}$ also form a basis for a quasi-triangular Hopf algebra $\mathcal{F}$, as shown in Ref. \cite{Kitaev97}. In fact, Ref. \cite{Kitaev97} also shows that the algebra $\mathcal{F}$ is precisely the dual Hopf algebra to the quantum double $\mathcal{D} = D(G)$.

\subsubsection{Elementary excitations in the Kitaev model}

One of the most important applications of ribbon operators is to classify the elementary excitations\footnote{There are many terms in the literature that all refer to essentially the same thing: an elementary excitation, a simple quasi-particle, an anyon, or a simple object of $\mathcal{Z}(\Rep(G))$.  A topological charge or a superselection sector is an isomorphism class of all the above. We will alternate in our use of these terms in different sections of the paper, to best match the current literatures in the corresponding fields (physics, mathematics, or computer science).} in the Kitaev model. As shown in the previous section, the operators $F^{(h,g)}_\rho$ create a pair of excitations at the endpoints of the ribbon $\rho$. However, these excitations may be a superposition of ``elementary'' excitations, so they are not stable and may easily decohere. Let us formally define these elementary excitations as follows:

\begin{definition}
Let $\mathcal{E}$ denote the space of excitations that can be created at any cilium $s = (v,p)$ by applying a linear combination of the operators $F^{(h,g)}_\rho$ to some ribbon $\rho$ terminating at $s$. An {\it elementary excitation}
or {\it simple quasi-particle}
is given by a subspace of $\mathcal{E}$ that is preserved under the action of local operators $D^{(h,g)}(s)$ 
(defined in Eq. (\ref{eq:bulk-local-operators})),
that cannot be further decomposed (non-trivially) into the direct sum of such subspaces.
\end{definition}

Since these subspaces cannot be modified by local operators, they determine the ``topological charge'' of the excitation. On the other hand, the degrees of freedom within this subspace are purely local properties of the excitation. We define the quantum dimension of the excitation to be the square root of the dimension of this subspace.

It turns out that the basis $F^{(h,g)}$ for the algebra $\mathcal{F}$ is not useful in classifying elementary excitations. Instead, by Ref. \cite{Bombin08}, we have the following Theorem:

\begin{theorem}
\label{bulk-anyon-types}
The elementary excitations of the Kitaev model with group $G$ are given by pairs $(C,\pi)$, where $C$ is a conjugacy class of $G$ and $\pi$ is an irreducible representation of the centralizer $E(C)$ of $C$.
\end{theorem}

\begin{proof}
To prove this theorem, we construct a change-of-basis for the ribbon operator algebra $\mathcal{F}$. Following Ref. \cite{Bombin08}, let us construct a new basis as follows:

\begin{enumerate}
\item
Choose an arbitrary element $r_C \in G$ and form its conjugacy class $C = \{gr_Cg^{-1}: g \in G\}$. Index the elements of $C$ so that $C = \{ c_i \}_{i=1}^{|C|}$.
\item
Form the centralizer $E(C) = \{ g \in G: gr_C = r_Cg \}$.\footnote{It is not hard to show that, up to conjugation, $E(C)$ depends only on $C$ and not on $r_C \in C$.}
\item
Form a set of representatives $P(C) = \{ p_i \}_{i=1}^{|C|}$ of $G/E(C)$, so that $c_i = p_i r_C p_i^{-1}$.
\item
Choose a basis for each irreducible representation $\pi$ of $E(C)$. Let $\Gamma_\pi(k)$ denote the corresponding unitary matrix for the representation of $k \in G$.
\item
The new basis is
\begin{multline}
\label{eq:elementary-ribbon-basis}
\{F^{(C,\pi);({\bf u,v})}_\rho: \text{ } C \text{ a conjugacy class of }G, \text{ } \pi \in (E(C))_{\text{ir}},
\\{\bf u} = (i,j), {\bf v} = (i', j'), 1 \leq i, i' \leq |C|, 1 \leq j,j', \leq \dim(\pi)\},
\end{multline}
where $E(C)_{\text{ir}}$ denotes the irreducible representations of $E(C)$, and each $F^{(C,\pi);({\bf u,v})}_\rho$ is given by
\begin{equation}
\label{eq:bulk-ribbon-FT}
F^{(C,\pi);({\bf u,v})}_\rho := \frac{\dim(\pi)}{|E(C)|}
\sum_{k \in E(C)} \left(\Gamma_\pi^{-1}(k)\right)_{jj'}F^{(c_i^{-1}, p_i k p_{i'}^{-1})}.
\end{equation}
\end{enumerate}

We can also construct the inverse change of basis \cite{Bombin08}. Suppose we are given $g,h \in G$. Then:

\begin{enumerate}
\item
Let $C$ be the conjugacy class of $h^{-1}$. Index the elements of $C$ so that $C = \{ c_i \}_{i=1}^{|C|}$.
\item
Let $E(C)$ be the centralizer of $C$ as above.
\item
Form a set of representatives $P(C)$ of $G/E(C)$ as above.
\item
Any $g \in G$ has a unique decomposition $g = p_i k$ s.t. $p_i \in P(C)$ and $k \in E(C)$.
\item
For each $g \in G$, let $i(g)$, $k(g)$ denote the index functions to obtain $p_i$ and $k$ from (4).
\item
Let $k_{(h,g)} = \left(p_{i(h^{-1})}\right)^{-1}g p_{i(g^{-1}h^{-1}g)}$.
\item
The inverse change-of-basis is
\begin{equation}
F^{(h,g)}_\rho =
\sum_{\pi \in E(C)_{\text{ir}}}\sum_{j,j' = 1}^{\dim(\pi)}
\left(\Gamma_\pi(k_{(h,g)})\right)_{jj'} F^{(C,\pi);(\bf{u,v})}_\rho
\end{equation}
where ${\bf u} = (i(h^{-1}), j)$, ${\bf v} = (i(g^{-1}h^{-1}g), j')$.
\end{enumerate}

The basis (\ref{eq:elementary-ribbon-basis}) is particularly useful because the parameters $(C,\pi)$ completely encode the global degrees of freedom of the particles created, and the $(\bf{u,v})$ completely encode the local degrees of freedom. Specifically, as shown in Ref. \cite{Bombin08}, different operators $F^{(C,\pi);(\bf{u,v})}_\rho$ with the same $(C,\pi)$ but different $(\bf{u,v})$ may be changed into one another by applying the local operators $D^{(h,g)}$ at the two endpoints $s_0,s_1$ of $\rho$. Similarly, if two ribbon operators in this new basis have different $(C,\pi)$ pairs, any operator that can change one to another must have support that connects $s_0$ and $s_1$. It follows that the elementary excitations of the Kitaev model are described precisely by pairs $(C,\pi)$, where $C$ is a conjugacy class of the original group $G$, and $\pi$ is an irreducible representation of the centralizer of $C$.

\end{proof}

Physically, in the basis (\ref{eq:elementary-ribbon-basis}), $C$ represents magnetic charge, and $\pi$ represents electric charge. The quantum dimension of an elementary excitation $(C,\pi)$ is given by the square root of the dimension of the subalgebra spanned by all $F^{(C,\pi);({\bf u,v})}_\rho$, or

\begin{equation}
\FPdim(C,\pi) = |C|\Dim(\pi).
\end{equation}

As a special case, the simple particle $(C,\pi)$, where $C = \{1\}$ is the conjugacy class of the identity element and $\pi$ is the trivial representation, is the vacuum particle (i.e. absence of excitation). The vacuum particle always has a quantum dimension of 1.

When two anyons (given by pairs $(C_1,\pi_1)$ and $(C_2, \pi_2)$, say) are brought by ribbon operators to the same cilium on the lattice, one can essentially consider them as one composite anyon. Specifically, one can again consider the local operators $D^{(h,g)}$ acting on this cilium, which will determine new sets of local/global degrees of freedom on the new composite anyon. This process is known as {\it anyon fusion}. It can be shown that anyon fusion in this group model is described by the fusion rules of the unitary modular tensor category $\mZ(\Rep(G))$.

We would like to note that the change of basis to (\ref{eq:elementary-ribbon-basis}) and its inverse is essentially a general Fourier transform and its inverse. However, this Fourier transform acts on the {\it operator algebra} $\mathcal{F}$, not on the vectors themselves. In fact, this is because each pair $(C,\pi)$ corresponds to an irreducible representation of the quantum double $D(G)$ \cite{Kitaev97}. For general (non-abelian) groups, the Fourier basis is precisely given by matrix elements of the irreducible representations \cite{Moore06}.

\subsection{Quantum double models with boundaries}
\label{sec:bd-hamiltonian}

In previous sections, we have defined the Kitaev quantum double model on a sphere, or an infinitely large lattice on the plane. However, it is also important to consider the case where the lattice has boundaries/holes (e.g. in Fig. \ref{fig:boundary}), as this is a powerful model with degeneracy that will allow us to achieve universal quantum computation. In this section, we present the Hamiltonian and ribbon operators for the Kitaev model with boundary. The Hamiltonian will be adapted from previous works on gapped boundaries and domain walls by Beigi et al. \cite{Beigi11} and Bombin and Martin-Delgado \cite{Bombin08}.

\begin{figure}
\centering
\includegraphics[width = 0.65\textwidth]{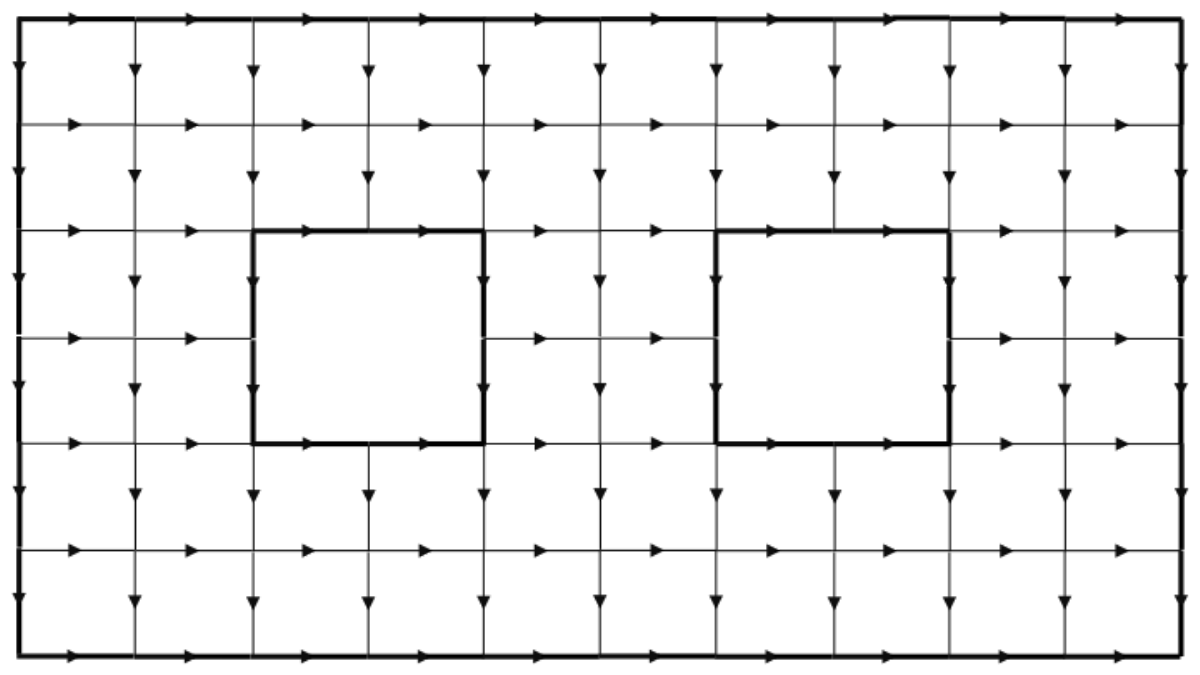}
\caption{Lattice for the Kitaev model with boundary. For any fixed group $G$, there can be multiple ways to define projection operators at the boundary such that all terms in the new Hamiltonian still commute. These are studied in Section \ref{sec:bd-hamiltonian}.}
\label{fig:boundary}
\end{figure}

\subsubsection{Hamiltonians for quantum double models with boundaries}

We will consider the model in which a gapped boundary is determined by a subgroup $K \subseteq G$ up to conjugation. In general, as shown in Ref. \cite{Beigi11}, a boundary is determined by both $K$ and a 2-cocycle $\phi \in H^2(K,\C^\times)$, and it is straightforward to generalize our results. Before we define the Hamiltonian, let us first define some new projector terms, as in Ref. \cite{Bombin08}:

\begin{equation}
A^K(v) := \frac{1}{|K|} \sum_{k \in K} A^k(v)
\end{equation}

\begin{equation}
B^K(p) := \sum_{k \in K} B^k(p)
\end{equation}

\begin{equation}
L^K(e) := \frac{1}{|K|} \sum_{k \in K} L^k(e)
\end{equation}

\begin{equation}
T^K(e) := \sum_{k \in K} T^k(e)
\end{equation}

Here, $e$ is an edge on the lattice and $A^k,B^k,L^k,T^k$ are the operators defined in Section \ref{sec:kitaev-hamiltonian}. In this context, we see that $A^k$ and $L^k$ are now different faithful representations of the multiplication in the Hopf sub-algebra $\C[K] \subseteq \C[G]$, and $B^k$ and $T^k$ are representations of the comultiplication in the subalgebra. The new projectors $A^K$ ($L^K$) now project vertices (edges) to a trivial sector of the representation of $K \subseteq G$, where the representation matrices are again given by $A^g$ ($L^g$) as noted in Section \ref{sec:kitaev-hamiltonian}. Similarly, the new projectors $B^K$ and $T^K$ now restrict the flux through a plaquette/on an edge to an element of $K$.

Following Ref. \cite{Bombin08}, we can now define the following Hamiltonian\footnote{As before, we write $H^{(K,1)}_{(G,1)}$ to leave room for the generalized version, where a boundary depends also on a 2-cocycle $\phi$ of $K$.}: 

\begin{equation}
\label{eq:bd-hamiltonian-K}
H^{(K,1)}_{(G,1)} = \sum_v (1-A^K(v)) + \sum_p (1 - B^K(p)) + \sum_e ((1-T^K(e)) + (1-L^K(e))
\end{equation}

It is important to note that as in the Hamiltonian (\ref{eq:kitaev-hamiltonian}), all terms in this Hamiltonian commute with each other. Hence $H^{(K,1)}_{(G,1)}$ is also gapped.

\begin{figure}
\centering
\includegraphics[width = 0.65\textwidth]{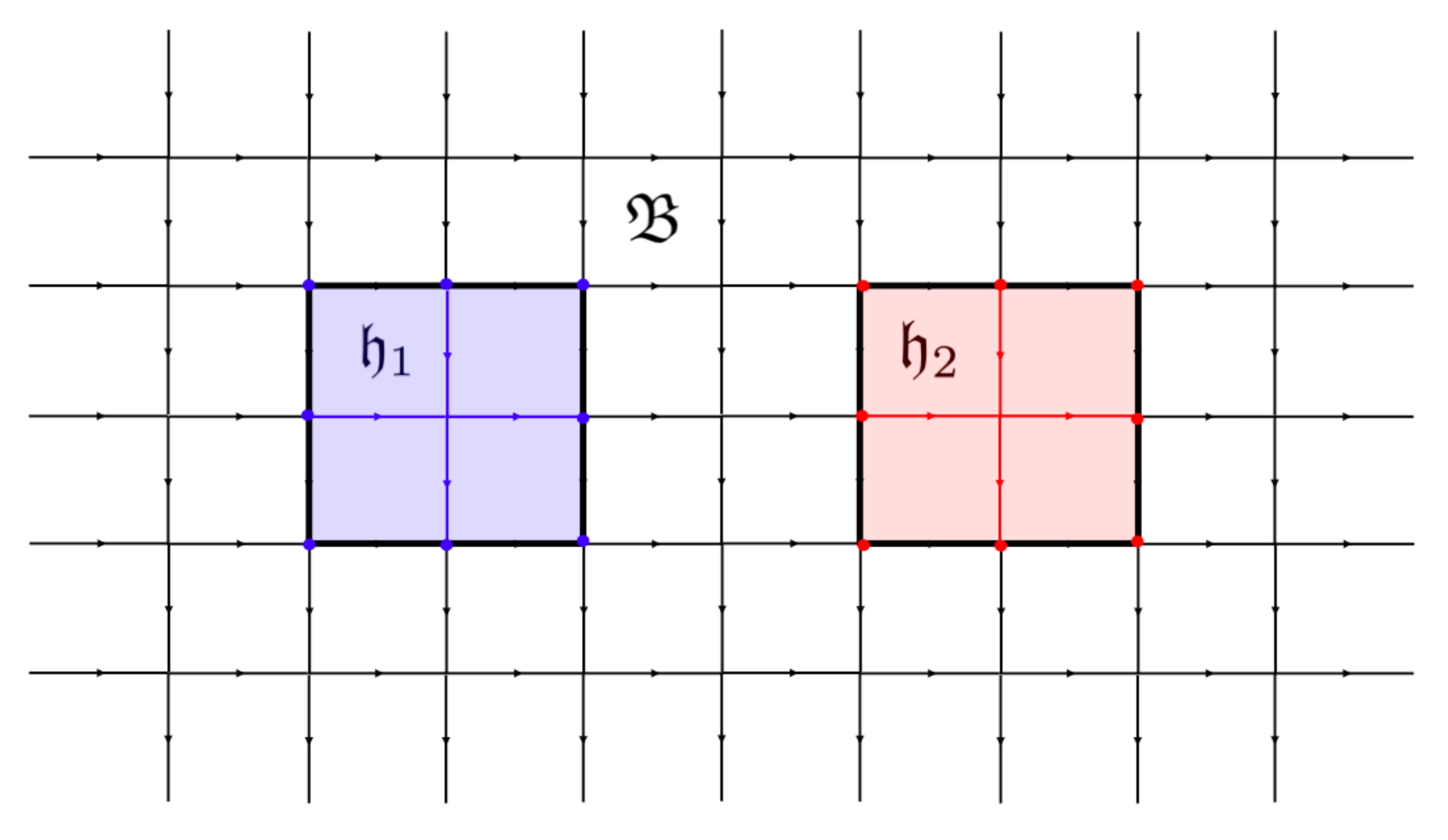}
\caption{Example: defining the Hamiltonian (\ref{eq:gapped-bds-hamiltonian}), in the case of two holes on an infinite lattice. The new Hamiltonians $H^{(K_1,1)}_{(G,1)}$ ($H^{(K_2,1)}_{(G,1)}$) are applied to all vertices, plaquettes, and edges within the blue (red) shaded region, and all vertices (blue or red dots) on the boldfaced lines. Specifically, we note that the vertices on the boldfaced lines are part of the holes, while the edges are not (black lines vs. blue or red dots). The bulk Hamiltonian $H_{(G,1)}$ is applied to all other vertices and plaquettes (white region).}
\label{fig:boundary-hamiltonian}
\end{figure}

We wish to take the standard Kitaev model, but modify the Hamiltonian in the presence of $n$ holes $\mathfrak{h}_1, ... \mathfrak{h}_n$ in the lattice, given by subgroups $K_1,...K_n$, respectively. Each hole is defined to contain all vertices, plaquettes, and edges within its border, and all (direct or dual) vertices on its border. We specifically note that edges on the border are {\it not} a part of the hole. Let $\mathfrak{B}$ denote the bulk, i.e. the complement of $\cup_i \mathfrak{h}_i$. The situation is shown in Fig. \ref{fig:boundary-hamiltonian}. The new Hamiltonian for this gapped boundary model will be defined as follows:

\begin{equation}
\label{eq:gapped-bds-hamiltonian}
H_{\text{G.B.}} = H_{(G,1)}(\mathfrak{B}) + \sum_{i=1}^{n} H^{(K_i,1)}_{(G,1)}(\mathfrak{h}_i).
\end{equation}

Here, $H^{(K_i,1)}_{(G,1)}(\mathfrak{h}_i)$ indicates that the Hamiltonian $H^{(K_i,1)}_{(G,1)}$ is acting on all edges, vertices, and plaquettes of the hole $\mathfrak{h}_i$, and similarly for $H_{(G,1)}(\mathfrak{B})$. As in the cases of (\ref{eq:kitaev-hamiltonian}) and (\ref{eq:bd-hamiltonian-K}), all terms in the Hamiltonian commute with each other, and $H_{\text{G.B.}}$ is also gapped.

\begin{remark}
\label{bd-ribbon-def}
As discussed in Ref. \cite{Bombin08}, the Hamiltonian $H^{(K,1)}_{(G,1)}$, $K \subseteq G$, reduces the gauge symmetry of the original Hamiltonian $H_{(G,1)}$ to the trivial one (equivalent to vacuum) in all areas to which it is applied. Hence, if it is preferable, we may simply have $H^{(K,1)}_{(G,1)}$ act on a border of the hole with a width of a single plaquette, and empty space beyond it. Because of this, the term ``boundary'' will henceforth be used to refer to the ribbon that runs along the line dividing two different Hamiltonians and lies within the region of $H^{(K,1)}_{(G,1)}$, as anything beyond this boundary ribbon is essentially vacuum. Similarly, boundary cilia will be cilia along this ribbon. This configuration is illustrated in Fig. \ref{fig:boundary-hamiltonian-2}.

Similarly, one can consider the case where the lattice has an external boundary given by subgroup $K_0$. In this case, it is not practical or necessary to have $H^{(K_0,1)}_{(G,1)}$ act on all (i.e. infinitely many) data qudits outside the original lattice. Instead, we will simply have the Hamiltonian $H^{(K_0,1)}_{(G,1)}$ act on a border of the entire lattice with a width of a single plaquette. Anything beyond the border may then be regarded as empty space. This is also illustrated in Fig. \ref{fig:boundary-hamiltonian-2}.
\end{remark}

\begin{remark}
\label{single-bd-creation}
We would like to note that the Hamiltonian $H_{\text{G.B.}}$ can be used to create just a single gapped boundary, unlike anyons in the bulk, which must be created in pairs.

We also note that a gapped boundary may be moved via adiabatic Hamiltonian tuning of the Hamiltonian $H_{\text{G.B.}}$, to enlarge or shrink the hole. This becomes very important in the context of Chapter \ref{sec:operations}, where we would like to braid gapped boundaries around each other to obtain quantum gates.
\end{remark}

\begin{figure}
\centering
\includegraphics[width = 0.7\textwidth]{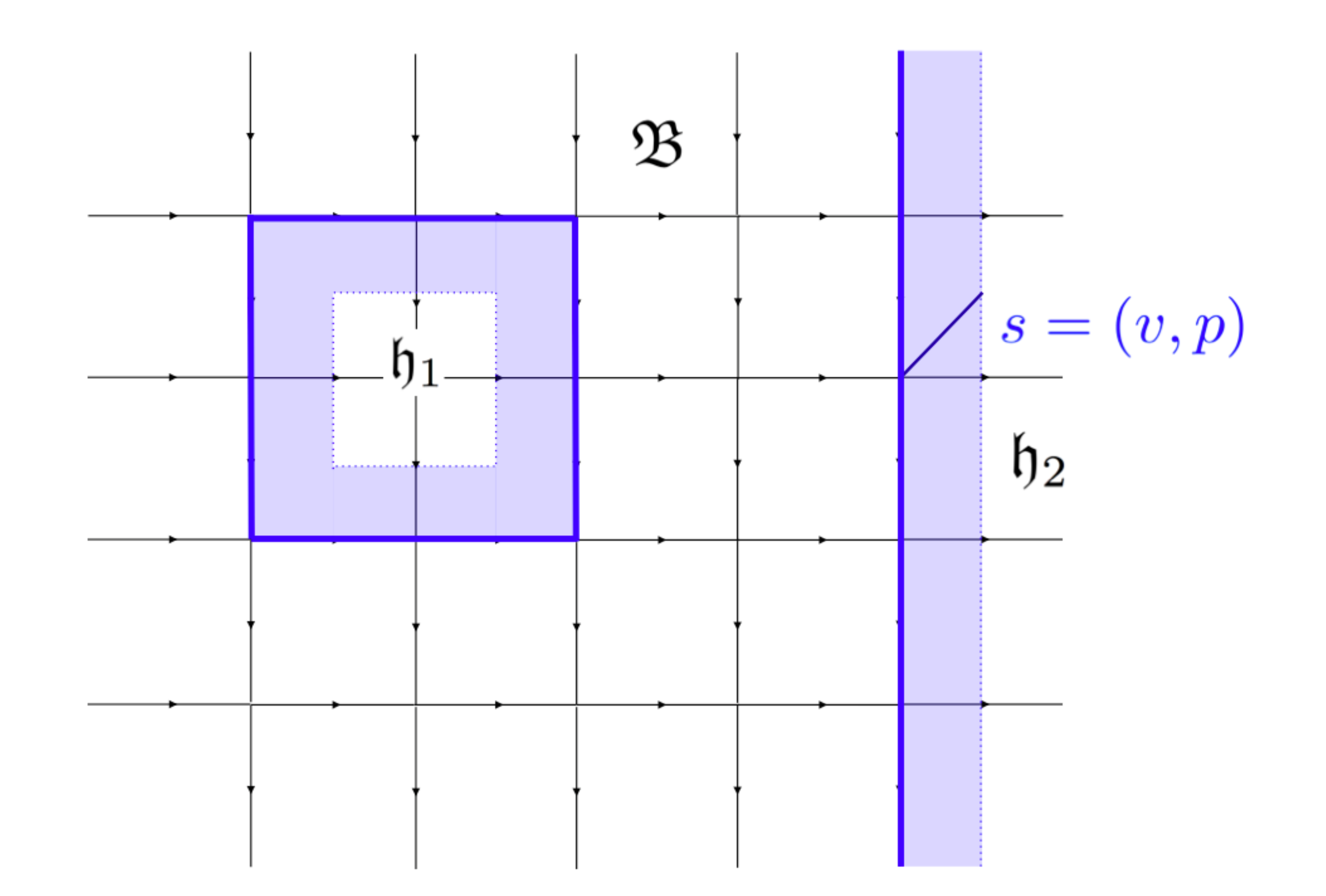}
\caption{Definitions of boundary ribbon/cilia. Given boundary lines (boldfaced) dividing regions of bulk/boundary Hamiltonians as shown, the boundary ribbons are the shaded ribbons. A cilium on the boundary ribbon (e.g. $s$ in the figure) is said to be a boundary cilium. Since the Hamiltonians $H^{(K,1)}_{(G,1)}$ break all gauge symmetries, anything beyond the boundary ribbon in the region of $H^{(K,1)}_{(G,1)}$ is essentially vacuum; one may ignore all Hamiltonian terms there, if desired.}
\label{fig:boundary-hamiltonian-2}
\end{figure}

\begin{figure}
\centering
\includegraphics[width = 0.7\textwidth]{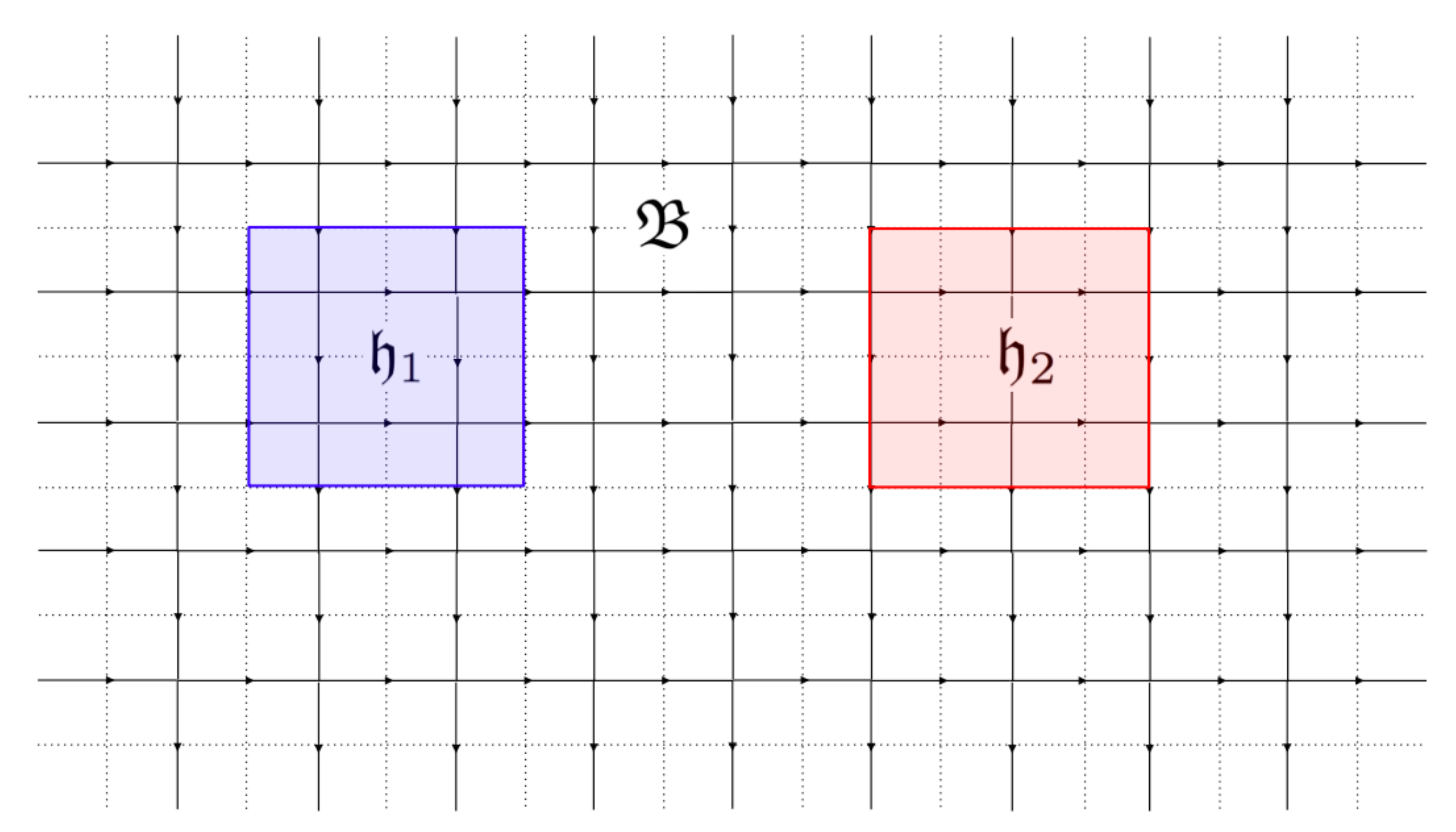}
\caption{Definition of the Hamiltonian $H_{\text{G.B.}}$, in cases where some (e.g. $\mathfrak{h}_2$) or all (e.g. $\mathfrak{h}_1$) of the hole's sides lie on the dual lattice. In this case, we note that the vertices and dual vertices on the boldfaced lines are part of the holes, while the edges are not.}
\label{fig:boundary-hamiltonian-3}
\end{figure}

\begin{remark}
\label{dual-bd-rmk}
In this section, we have defined the Hamiltonian $H_{\text{G.B.}}$ that can create holes in the lattice, whose sides lie on the direct lattice (see Fig. \ref{fig:boundary-hamiltonian}). More generally, we can create holes where some or all of the sides lie on the dual lattice, such as $\mathfrak{h}_1$ and $\mathfrak{h}_2$, respectively, in Fig. \ref{fig:boundary-hamiltonian-3}. For both cases, we say that the Hamiltonian $H^{(K_i,1)}_{(G,1)}$ acts on vertices, plaquettes, and edges within the shaded square, and on the vertices and plaquettes of the boldfaced boundary of the square. Note that as before, $H^{(K_i,1)}_{(G,1)}$ does not act on the edges of the boldfaced square boundary. In general, the properties of holes completely on the dual lattice such as $\mathfrak{h}_1$ are almost the same as those of holes on the direct lattice; the only difference is up to an electromagnetic symmetry in the model. However, holes such as $\mathfrak{h}_2$ that are partially on the dual lattice are of more interest. These holes were briefly considered by Fowler et al. in Ref. \cite{Fowler12} in the special case of the toric code; we consider them in generality in Section \ref{sec:defect-hamiltonian}.
\end{remark}

\subsubsection{Local operators}
\label{sec:bd-local-operators}

We now introduce a {\it quasi-Hopf algebra} $\mathcal{Z} = Z(G,1,K,1)$ which describes the local operators that act on a cilium $s = (v,p)$ on the boundary. We call this the {\it group-theoretical quasi-Hopf algebra} based on the group $G$, subgroup $K$, and trivial cocycles $\omega \in H^3(G, \C^\times)$ and $\psi \in H^2(K, \C^\times)$, following the definition of the group-theoretical category $\mC(G,\omega,K,\psi)$ in Ref. \cite{Etingof05}. The algebra was first constructed by Zhu in Ref. \cite{Zhu01}, and is also discussed in detail in Ref. \cite{Schauenburg02}. The following operators form a basis for $\mathcal{Z}$:

\begin{equation}
\label{eq:bd-local-operators}
Z^{(hK,k)}(v,p) = B^{hK}(v,p)A^k(v,p)
\end{equation}

\noindent
where we have defined

\begin{equation}
B^{hK}(v,p) = \sum_{j \in hK} B^j(v,p).
\end{equation}

\noindent
Here, $k \in K$ is an element of the subgroup, and $hK = \{hk: k \in K\}$ is a left coset. HHere, we only need to consider the local vertex operators $A^k$ where $k \in K$, as the actions of $A^g(v,p)$, $g \notin K$ on the representations of $K$ at the vertex $v$ and edges in $\text{star}(v)$ are linear combinations of the actions of $A^k(v,p)$. Similarly, since the $B^K$ terms of the Hamiltonian (\ref{eq:bd-hamiltonian-K}) project onto flux in the subgroup $K$, we only need to consider the generalization where the flux lies within a left coset (instead of restricting to a particular group element). Hence, as vector spaces, we have

\begin{equation}
Z(G,1,K,1) = F[G/K] \otimes \C[K],
\end{equation}

\noindent
i.e. $Z(G,1,K,1)$ is the tensor product of the algebra of complex functions over the left cosets of $K$ and the group algebra $\C[K]$. The multiplication, comultiplication and antipode for the quasi-Hopf algebra $\mathcal{Z}$ are presented in Refs. \cite{Zhu01,Schauenburg02}, where they show that they satisfy all the quasi-Hopf algebra axioms. We present them in the context of local operators in Appendix \ref{sec:quasi-hopf-algebra}. In Section \ref{sec:algebraic-condensation}, we show how the representation category of $Z(G,1,K,1)$ is the group-theoretical category $\mC(G,1,K,1)$ as defined in \cite{Etingof05}.

\subsubsection{Ribbon operators}
\label{sec:bd-ribbon-operators}

We will now describe the coquasi-Hopf algebra $\mathcal{Y} = Y(G,K)$ of ribbon operators which create all possible excited states on the boundary that may result from pushing a bulk anyon into the boundary (see Fig. \ref{fig:condensation}). In this case, a {\it boundary ribbon operator} is defined as an operator that is supported on a boundary ribbon (and acts trivially elsewhere) and commutes with all vertex and plaquette terms in the Hamiltonian $H_{\text{G.B.}}$ except at the two end cilia. These are defined recursively, as in the case of the bulk ribbon operators. As with the local operators $\mZ$, these will also be indexed by a pair $(hK,k)$, $hK$ a left coset, and $k \in K$. As before, we have the following definition for the trivial ribbon:

\begin{equation}
Y^{(hK,k)}_\epsilon := \delta_{1,k}
\end{equation}

Similarly, let $\tau = (s_0,s_1,e)$ be any dual triangle, and let $\tau'=(s_0',s_1',e')$ be any direct triangle. We define

\begin{equation}
\label{eq:bd-triangle-operator-def}
Y^{(hK,k)}_\tau := \delta_{1,k} L^{hK} (e), \qquad Y^{(hK,k)}_{\tau'} := T^k (e')
\end{equation}

\noindent
where we have defined

\begin{equation}
L^{hK}(e) := \frac{1}{|hK|} \sum_{j \in hK} L^j(e).
\end{equation}

Finally, we have the following gluing relation: If $\rho = \rho_1 \rho_2$ is a composite ribbon on the boundary, then

\begin{equation}
\label{eq:Y-gluing-formula}
Y^{(hK,k)}_\rho = \sum_{j \in K} Y^{(hK,j)}_{\rho_1} Y^{(j^{-1}hjK,j^{-1}k)}_{\rho_2}.
\end{equation}

Simple group theory manipulations show that for any $h,j \in K$ the left coset $j^{-1}hjK$ depends only on the left coset $hK$, and not on the particular representative $h$. By the same argument as in Ref. \cite{Bombin08}, the operator $Y^{(hK,k)}_\rho$ is independent of the particular choice of $\rho_1,\rho_2$. Hence, Eq. (\ref{eq:Y-gluing-formula}) is well-defined.

Looking at Eq. (\ref{eq:bd-triangle-operator-def}), we see that as vector spaces,

\begin{equation}
\label{eq:Y-vector-space}
Y(G,1,K,1) = \C[G/K] \otimes F[K]
\end{equation}

\noindent
so it is clear that $\mathcal{Y}$ is dual to $\mathcal{Z}$. The multiplication, comultiplication and antipode for the coquasi-Hopf algebra $\mathcal{Y}$ are presented in Refs. \cite{Zhu01,Schauenburg02}. These structures are discussed in the context of ribbon operators in Appendix \ref{sec:quasi-hopf-algebra}. We note now that the gluing formula (\ref{eq:Y-gluing-formula}) is in fact the comultiplication of $\mathcal{Y}$ (and also corresponds to the multiplication of $\mathcal{Z}$ by duality).

This gluing procedure is very similar to the movement of a bulk anyon via bulk ribbon operators and the gluing formula of Eq. (\ref{eq:ribbon-gluing}). However, there is one very important difference: After applying the gluing formula (\ref{eq:ribbon-gluing}) to move a bulk anyon from the endpoint $s_1$ of the ribbon $\rho_1$ to the endpoint $s_2$ of $\rho_2$, the resulting operator $F^{(h,g)}_\rho$ on the ribbon $\rho = \rho_1 \rho_2$ now commutes with the terms in the Hamiltonian $H_{(G,1)}$ of Eq. (\ref{eq:kitaev-hamiltonian}) at the cilium $s_1$. Instead, the only places where $F^{(h,g)}_\rho$ does not commute with the terms in $H_{(G,1)}$ are $s_2$ and the other endpoint cilium of $\rho_1$. In this new case, applying the gluing formula (\ref{eq:Y-gluing-formula}) on such a ribbon $\rho = \rho_1 \rho_2$ on the boundary allows the new operator $Y^{(hK,k)}_\rho$ to commute with the $A^K(v)$ and $B^K(p)$ terms at $s_1$, but it still does not commute with the edge terms $L^K(e), T^K(e)$ surrounding $s_1$. Hence, the excitations in the boundary are {\it confined}: the energy required to move an excitation along a boundary ribbon $\rho$ is linearly proportional to the length of $\rho$ (measured in the number of dual triangles). Physically, this means that $L^K, T^K$ in the Hamiltonian (\ref{eq:bd-hamiltonian-K}) represent string tension terms, which break all gauge symmetries past the boundary.

\begin{remark}
We would like to note that the above definition of boundary ribbon operators cannot create all types of excitations on the boundary at the end cilia of the ribbon; instead, it can only create vertex and plaquette excitations at these cilia. In general, an excitation of $H^{(K,1)}_{(G,1)}$ can be given by a vertex, plaquette, and edge excitation simultaneously. However, the boundary ribbon operators discussed here create all excitations that can be formed by the condensation a bulk anyon to the boundary (discussed in the next two sections).

Furthermore, by the detailed analysis of Ref. \cite{Bombin08}, all excitations within the region of the Hamiltonian $H^{(K,1)}_{(G,1)}$ (vertex, plaquette, edge, or any combination) are confined, and no particles are deconfined. This means any definition of excitation-creating operators and the gluing relation will always make the energy cost to move an excitation linear in the ribbon length. For instance, another set of excitation-creating operators for $H^{(K,1)}_{(G,1)}$ would be those that act only on edges, and not on triangles. In this case, the vertex and plaquette terms would serve to confine particles, instead of the edge terms $L^K, T^K$. 
\end{remark}

\subsection{Degeneracy and condensations to vacuum}
\label{sec:hamiltonian-gsd-condensation}

The standard Kitaev model has no ground state degeneracy on a surface with trivial topology like the infinite plane or the sphere, as the only operators that commute with the Hamiltonian $H_{(G,1)}$ of (\ref{eq:kitaev-hamiltonian}) are closed ribbon operators on contractible loops, which can be expressed as a linear combination of products of $A(v)$ or $B(p)$ \cite{Kitaev97}. However, once boundaries are introduced, one can construct operators that commute with the Hamiltonian $H_{\text{G.B.}}$ of (\ref{eq:gapped-bds-hamiltonian}) that cannot be expressed as such a product.

\begin{figure}
\centering
\includegraphics[width = 0.65\textwidth]{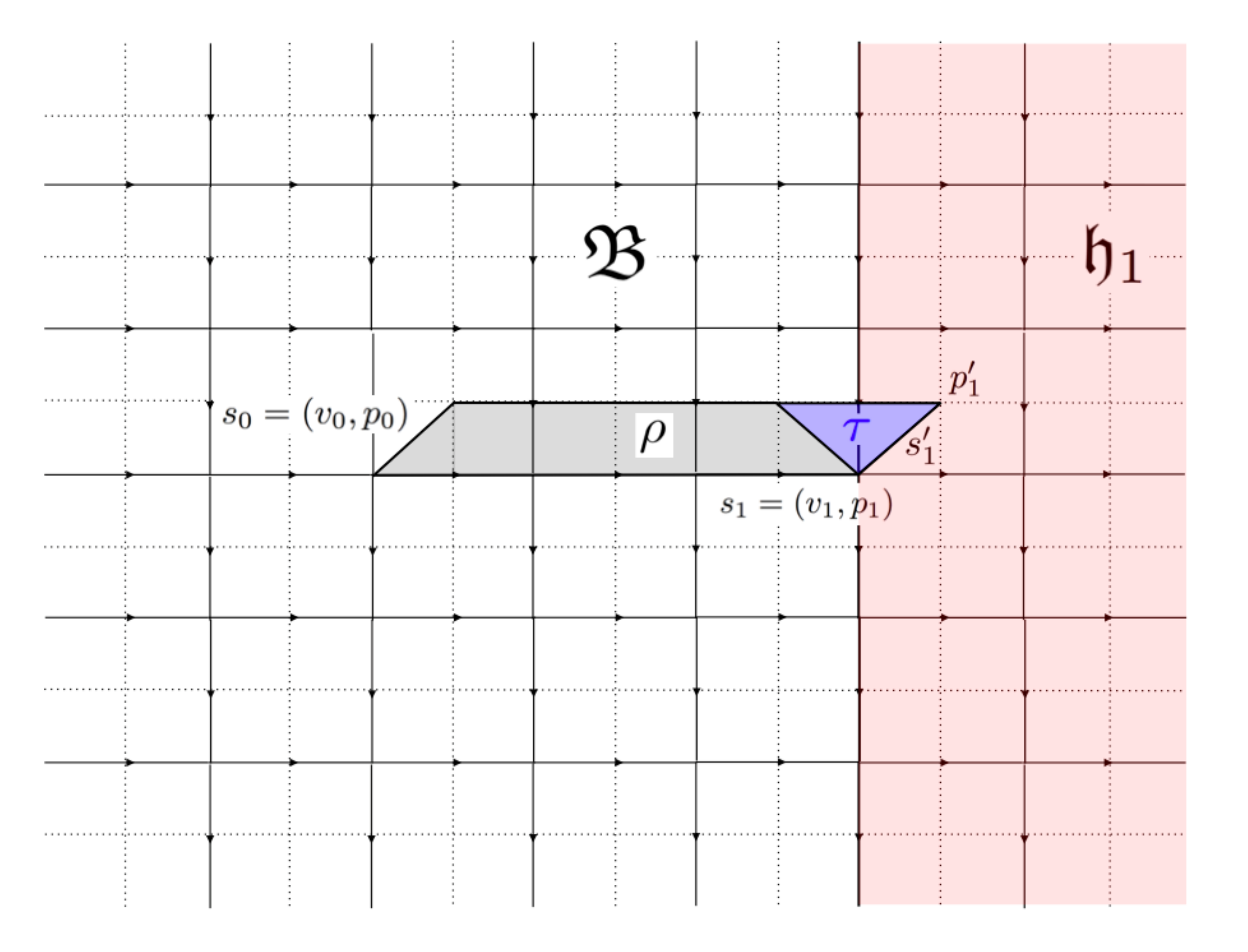}
\caption{Illustration of the bulk-to-boundary condensation procedure. If the new ribbon operator $F_{\rho \cup \tau}$ commutes with the Hamiltonian terms $A_{v_1}$ and $B_{p_1'}$  at the cilium $s_1'$, we say that the anyon $a$ has condensed to vacuum on the boundary.}
\label{fig:condensation}
\end{figure}

Let us first consider a scenario where we create a pair of anyons $a,\overbar{a}$ in the bulk from vacuum, by applying a ribbon operator $F^{(C,\pi);(\bf{u,v})}_\rho$. Without loss of generality, we suppose we have chosen $\rho$ with endpoints $s_0 = (v_0,p_0)$ and $s_1 = (v_1,p_1)$ such that $a$ is located at $s_1$ and is as close to the boundary as possible, as shown in Fig. \ref{fig:condensation}. Specifically, $s_1 = (v_1,p_1)$, where $v$ already lies on the line separating two regions with different Hamiltonians.

Suppose we would like to extend $\rho$ to $\rho \cup \tau$ and push the anyon $a$ into the boundary. We can apply the gluing formula (\ref{eq:ribbon-gluing}) as we would in the bulk, and the original excitation is pushed to the boundary cilium $s_1'=(v_1,p_1')$. However, since the Hamiltonian terms at $s_1'$ are different from those at $s_1$, it is possible that the new ribbon operator now commutes with all Hamiltonian terms in the vicinity of $s_1'$. The only terms that do not commute with $F^{(C,\pi);(\bf{u,v})}_\rho$ are now the terms corresponding to $s_0$. In this case, we say the anyon $a$ has {\it condensed to vacuum} in the boundary.

\begin{figure}
\centering
\includegraphics[width = 0.65\textwidth]{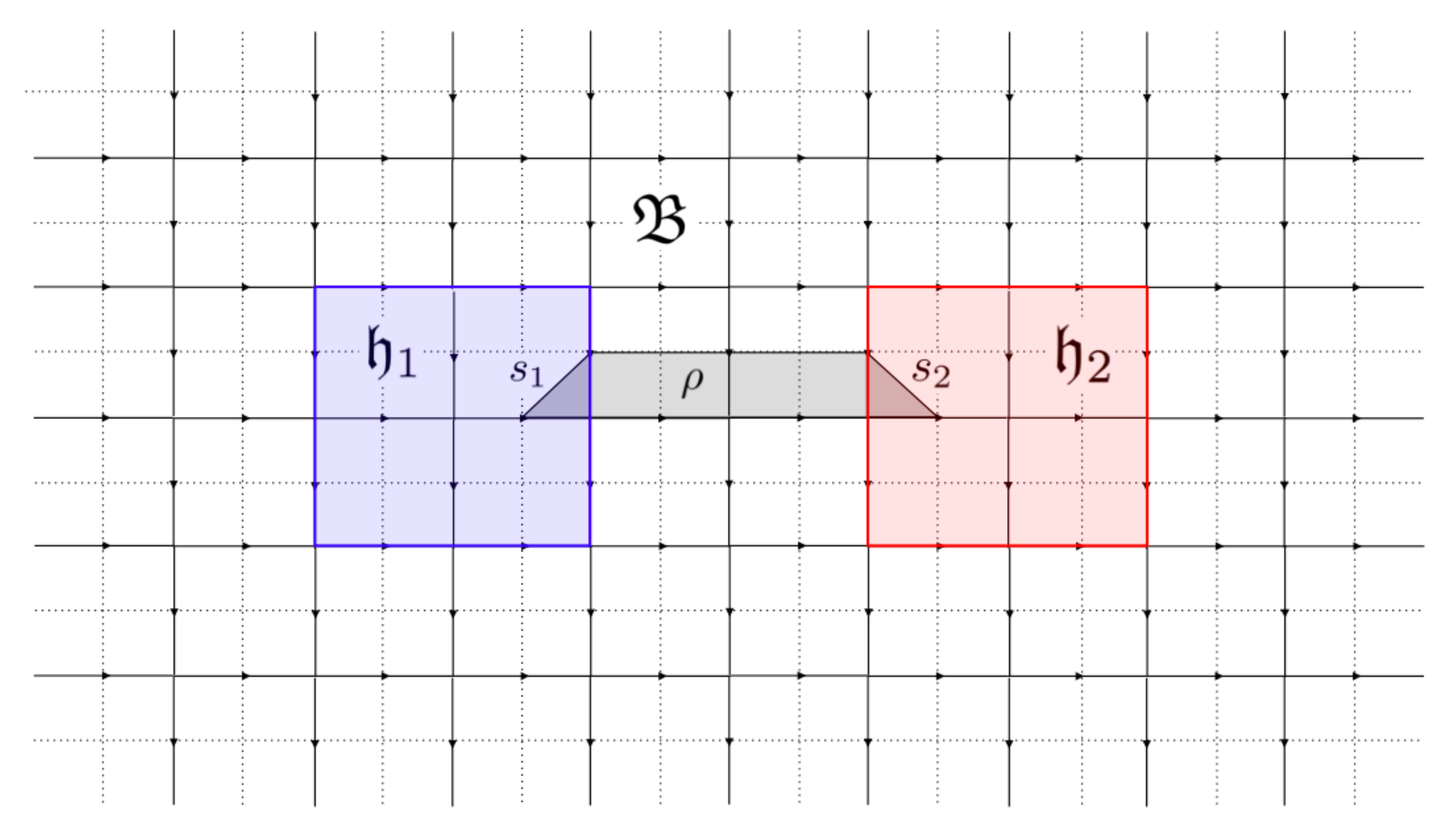}
\caption{Ground state degeneracy in the Kitaev model with boundary. By definition, a ribbon operator $F_\rho$ on $\rho$ commutes with all Hamiltonian terms in the bulk. Because the cilia $s_1$ and $s_2$ lie in areas with ribbon operators, it is possible that $F_\rho$ may also commute with the Hamiltonian terms at these cilia. The algebra of such operators $F_\rho$ will form the degenerate ground state of this system.}
\label{fig:degeneracy}
\end{figure}

Suppose we now have two holes $\mathfrak{h_1},\mathfrak{h_2}$ in the lattice, as shown in Fig. \ref{fig:degeneracy}. $\mathfrak{h_1},\mathfrak{h_2}$ are in the ground state of the Hamiltonians $H^{(K_1,1)}_{(G,1)}$, $H^{(K_2,1)}_{(G,1)}$, respectively. Again, let us consider a ribbon operator $F^{(C,\pi);(\bf{u,v})}_\rho$ which creates anyons $a,\overbar{a}$ in the bulk; it may now be possible be possible to condense $a$ to vacuum along the boundary of $\mathfrak{h_1}$, and $\overbar{a}$ to vacuum along the boundary of $\mathfrak{h_2}$, if the ribbon operator commutes with the boundary Hamiltonians in $\mathfrak{h_1}$, $\mathfrak{h_2}$.

It is now clear that the ground state of the Hamiltonian (\ref{eq:gapped-bds-hamiltonian}) may have nontrivial degeneracy. Specifically, each such operator $F^{(C,\pi);(\bf{u,v})}_\rho$ now corresponds to an operator $W_{(C,\pi);(\bf{u,v})}$ that commutes with the Hamiltonian. Given any ground state $\ket{0}$, the set of states $W_{(C,\pi);(\bf{u,v})}\ket{0}$ now form a basis for the Hilbert space of ground states. This Hilbert space is topologically protected, and will be used as our qudit for topological quantum computation.

\subsection{Excitations on the boundary}
\label{sec:bd-excitations}

\subsubsection{Elementary excitations}

In Section \ref{sec:bd-ribbon-operators}, we defined a basis for the coquasi-Hopf algebra $\mathcal{Y}$ of ribbon operators that create arbitrary excited states on the boundary of the Kitaev model given by subgroup $K$ that can result from bulk-to-boundary condensation. However, as in the case of bulk ribbon operators, we would like to classify the elementary excitations on the boundary. This is described in the following Theorem:

\begin{theorem}
\label{bd-anyon-types}
The elementary excitations on a subgroup $K$ boundary of the Kitaev model with group $G$ are given by pairs $(T,R)$, where $T \in K\backslash G / K$ is a double coset, and $R$ is an irreducible representation of the stabilizer $K^{r_T} = K \cap r_T K r_T^{-1}$ ($r_T \in T$ is any representative of the double coset).
\end{theorem}

\begin{proof}
As before, we must perform a Fourier transform on the group-theoretical co-quasi-Hopf algebra to obtain a new basis. This change-of-basis formula is constructed as follows:

\begin{enumerate}
\item
Choose a representative $r_T \in G$ and construct the corresponding double coset $T = K r_T K \in K\backslash G/K$.
\item
Construct the subgroup $K^{r_T} = K \cap r_T K r_T^{-1}$.
\item
Construct a set of representatives $Q$ of $K/K^{r_T}$. Label the elements of $Q$ as $Q = \{ q_i \}_{i=1}^{|Q|}$.
\item
Choose an irreducible representation $R$ of the subgroup $K^{r_T}$. Choose a basis for $R$ and denote the resulting unitary matrix representations $\Gamma_R(k)$ for $k \in K^{r_T}$.
\item
For each $i = 1,2,...|Q|$, let $s_i = q_i r_T q_i^{-1}$. Construct the set of right cosets $S_R(T) = \{ K s_i\}_{i=1}^{|Q|}$. Simple group theory shows that the set $S_R(T)$ forms a partition of $T$. Similarly, the set of left cosets $S_L(T) = \{ s_i^{-1}K\}_{i=1}^{|Q|}$ is also a partition of $T$.
\item
The new basis is
\begin{multline}
\label{eq:elementary-bd-ribbon-basis}
\{Y^{(T,R);({\bf u,v})}_\rho: \text{ } T = K r_T K \in K\backslash G/K, \text{ } R \in (K^{r_T})_{\text{ir}},
\\{\bf u} = (i,j), {\bf v} = (i', j'), 1 \leq i, i' \leq |Q|, 1 \leq j,j', \leq \dim(R)\},
\end{multline}
where each $Y^{(T,R);({\bf u,v})}_\rho$ is given by
\begin{equation}
\label{eq:bd-ribbon-FT}
Y^{(T,R);({\bf u,v})}_\rho := \frac{\dim(R)}{|K^{r_T}|}
\sum_{k \in K^{r_T}} \left(\Gamma_R^{-1}(k)\right)_{jj'}Y^{(s_i^{-1}K, q_i k q_{i'}^{-1})}.
\end{equation}
\end{enumerate}

As before, this Fourier basis for $\mathcal{Y}$ completely separates the topological and local degrees of freedom in the created excitations. It is straightforward to show that linear combinations of the local operators $Z^{(hK,k)}$ at the endpoints $s_0,s_1$ of $\rho$ may be used to transform any $Y^{(T,R);({\bf u,v})}_\rho$ into another basis operator that differs in only the pair $({\bf u,v})$. Similarly, if two operators in the basis (\ref{eq:elementary-bd-ribbon-basis}) have different pairs $(T,R)$, any operator that can change one to another must have support that connects $s_0$ and $s_1$. We can now conclude that the elementary excitations on the boundary of the Kitaev model are described precisely by pairs $(T,R)$, where $T= K r_T K$ is a double coset, and $R$ is an irreducible representation of the group $K^{r_T}$. 
\end{proof}

In the basis (\ref{eq:elementary-bd-ribbon-basis}), the quantum dimension of $(T,R)$ is given by the square root of the dimension of the subalgebra spanned by all $Y^{(T,R);({\bf u,v})}_\rho$, or

\begin{equation}
\FPdim(T,R) = |Q|\Dim(R) = \frac{|K|}{|K^{r_T}|} \Dim(R).
\end{equation}

As a special case, the simple particle $(T,R)$, where $T = K1K = K$ is the double coset of the identity element and $R$ is the trivial representation, is the vacuum particle (i.e. absence of excitation). The vacuum particle always has a quantum dimension of 1.

As in the case of the bulk, one can also fuse boundary excitations by bringing two excitations to the same boundary cilium via boundary ribbon operators, and consider the local operators $Z^{(hK,k)}$ acting on the new composite excitation. One can then show that the excitations on the boundary have a \lq\lq topological order" given by a unitary fusion category, as we will discuss in Chapter \ref{sec:algebraic}.  This new kind of boundary topological order exists only in the presence of bulk, and we will refer to as \lq\lq bordered topological order".  In particular, the fusion category is the representation category of the group-theoretical quasi-Hopf algebra $\mathcal{Z}$ introduced in Section \ref{sec:bd-local-operators} (or equivalently, the representation category of the coquasi-Hopf algebra $\mathcal{Y}$). In fact, this category is Morita equivalent to the representation category $\Rep(G)$; its Drinfeld center is indeed the same as $\mZ(\Rep(G))$.

\subsubsection{Products of bulk-to-boundary condensation}

In Section \ref{sec:hamiltonian-gsd-condensation}, we informally described how a ground state degeneracy can result from the ability for certain bulk particles to condense to vacuum on the boundary. Now that we have formally defined and classified the elementary excitations of the boundary, we can provide a formal classification of these special bulk particles. More generally, given any elementary excitation $(C,\pi)$ in the bulk, we present a way to determine the products $(T,R)$ that are formed by condensation to the boundary. 

Suppose we have a boundary given by subgroup $K$, and a bulk anyon $a = (C,\pi)$ to condense to the boundary. In terms of ribbon (triangle) operators, if the border line between the bulk Hamiltonian $H_{(G,1)}$ and the boundary Hamiltonian $H^{(K,1)}_{(G,1)}$ lies on the direct lattice, the condensation procedure is always described by a dual triangle operator on a triangle such as the triangle $\tau$ in Fig. \ref{fig:condensation}. To bring $a$ to the boundary, we simply apply one of the operators $F^{(C,\pi);(\bf{u,v})}_\tau$. So far, this movement operator is the same as moving the anyon to anywhere else in the bulk.

The difference arises when $a$ crosses the boundary. Once this happens, $a$ may no longer be an elementary excitation: instead, it could be a superposition of the elementary excitations of the boundary that we classified earlier. Mathematically, this corresponds to the fact that the ribbon operators $F^{(C,\pi);(\bf{u,v})}_\tau$ no longer form a basis for the triangle operators in the boundary, so we must express them as a linear combination of the basis operators $Y^{(T,R);({\bf u,v})}_\tau$. This linear combination is constructed as follows:

Since $\tau$ is a dual triangle, by Equations (\ref{eq:triangle-operator-def}) and (\ref{eq:bd-triangle-operator-def}), we have (before the Fourier transform)
\begin{equation}
F^{(h,g)}_\tau = \delta_{1,g} L^h(e) \qquad Y^{(hK,k)}_\tau = \delta_{1,k} L^{hK}(e)
\end{equation}

By Equations (\ref{eq:bulk-ribbon-FT}) and (\ref{eq:bd-ribbon-FT}), we have (after the Fourier transform)
\begin{multline}
\begin{aligned}[t]
F^{(C,\pi);(\bf{u,v})}_\tau & := \frac{\dim(\pi)}{|E(C)|}
\sum_{k \in E(C)} \left(\Gamma_\pi^{-1}(k)\right)_{jj'}\delta_{1,p_i k p_{i'}^{-1}} L^{c_i^{-1}}(e)
\\ &= \frac{\dim(\pi)}{|E(C)|}\left(\Gamma_\pi^{-1}(p_i^{-1} p_{i'})\right)_{jj'}L^{c_i^{-1}}(e).
\end{aligned}
\end{multline}
\begin{multline}
\begin{aligned}[t]
Y^{(T,R);(\bf{u,v})}_\tau & := \frac{\dim(R)}{|K^{r_T}|}
\sum_{k \in K^{r_T}} \left(\Gamma_R^{-1}(k)\right)_{jj'}\delta_{1,q_i k q_{i'}^{-1}} L^{s_i^{-1}K}(e)
\\ &= \frac{\dim(R)}{|K^{r_T}|}\left(\Gamma_R^{-1}(q_i^{-1} q_{i'})\right)_{jj'}L^{s_i^{-1}K}(e).
\end{aligned}
\end{multline}
(In both cases, it is possible that $p_i^{-1} p_{i'} \notin C$ or $q_i^{-1} q_{i'} \notin K^{r_T}$; if that happens, we simply have $F^{(C,\pi);(\bf{u,v})}_\tau = 0$ or $Y^{(T,R);(\bf{u,v})}_\tau = 0$).

The following theorem governs the products of condensation:

\begin{theorem}
\label{condensation-products}
Let $(T,R)$ and $(C,\pi)$ be given elementary excitations in the boundary and bulk, respectively. The term $Y^{(T,R);(\bf{u_2,v_2})}_\tau$ has a nonzero coefficient in the decomposition of $F^{(C,\pi);(\bf{u_1,v_1})}_\tau$ (for some quadruple $(\bf{u_1,v_1,u_2,v_2})$) if and only if the following two conditions hold:
\begin{enumerate}
\item
The intersection $C \cap T$ is nonempty.
\item
Let $\rho_\pi$ be the (possibly reducible) representation of the subgroup $E(C) \cap K^{r_T}$ resulting from the restriction of $\pi$ to $E(C) \cap K^{r_T}$; let $\rho_R$ be the representation of the same subgroup formed by restricting $R$. Decompose $\rho_\pi$, $\rho_R$ into irreducible representations of $E(C) \cap K^{r_T}$:
\begin{equation}
\rho_\pi = \oplus_\sigma n^{\rho_\pi}_{\sigma} \sigma
\end{equation}
\begin{equation}
\rho_R = \oplus_\sigma n^{\rho_R}_{\sigma} \sigma
\end{equation}
\noindent
There must exist some irreducible representation $\sigma$ of $E(C) \cap K^{r_T}$ such that $n^{\rho_\pi}_{\sigma} \neq 0$ and $n^{\rho_R}_{\sigma} \neq 0$.
\end{enumerate}
In particular, for given $(C,\pi), (T,R)$ let us write the decomposition after condensation as
\begin{equation}
\label{eq:condensation}
(C,\pi) = \oplus n_{(T,R)}^{(C,\pi)} (T,R).
\end{equation}
Then, we have 
\begin{equation}
\label{eq:condensation-coefficients}
n_{(T,R)}^{(C,\pi)} = n_{C,T} n_{\pi,R},
\end{equation}
\noindent
where we define
\begin{equation}
n_{C,T} = |\{ i: s_i^{-1} K = c_i^{-1} K\}|
\end{equation}
\begin{equation}
n_{\pi,R} = \sum_{\sigma \in (E(C) \cap K^{r_T})_{\text{ir}}} n^{\rho_\pi}_{\sigma}  n^{\rho_R}_{\sigma}.
\end{equation}
\noindent
($c_i$ and $s_i$ have been defined in the proofs of Theorems \ref{bulk-anyon-types} and \ref{bd-anyon-types}, respectively).

Furthermore, these coefficients imply that the two sides of Eq. (\ref{eq:condensation}) will always have the same quantum dimensions.
\end{theorem}

Similarly, we may also consider the process of pulling a boundary excitation back into the bulk. The situation here is exactly the inverse of the above: we wish to write the $Y^{(T,R);(\bf{u,v})}_\tau$ as a linear combination of $F^{(C,\pi);(\bf{u,v})}_\tau$. Hence, by the same reasoning as above, we have the following theorem:

\begin{theorem}
\label{inverse-condensation-products}
Let $(T,R)$ and $(C,\pi)$ be given elementary excitations in the boundary and bulk, respectively. The term $F^{(C,\pi);(\bf{u_1,v_1})}_\tau$ has a nonzero coefficient in the decomposition of $Y^{(T,R);(\bf{u_2,v_2})}_\tau$ (for some quadruple $(\bf{u_1,v_1,u_2,v_2})$) if and only if the conditions (1) and (2) of Theorem \ref{condensation-products} hold.

In particular, let us write the decomposition of the simple boundary excitation as
\begin{equation}
\label{eq:inverse-condensation}
(T,R) = \oplus n_{(C,\pi)}^{(T,R)} (C,\pi).
\end{equation}
Then, we have 
\begin{equation}
n_{(C,\pi)}^{(T,R)} = n^{(C,\pi)}_{(T,R)},
\end{equation}
where $n^{(C,\pi)}_{(T,R)}$ is defined as in Theorem \ref{condensation-products}. Furthermore, these coefficients imply that the quantum dimension of the right hand side of Eq. (\ref{eq:inverse-condensation}) will always be $|G|$ times that of the left hand side.
\end{theorem}

Theorem \ref{inverse-condensation-products} gives us a straightforward way to determine which quasi-particles $(C,\pi)$ may be condensed to vacuum on a given boundary based on subgroup $K$: we can simply find all quasi-particles appearing with nonzero coefficient in the decomposition (\ref{eq:inverse-condensation}) corresponding to $(T,R)$ trivial. In general, we will use these anyon types to label the corresponding gapped boundary.

We would like to note that the above two theorems are exactly consistent with the mathematical results presented by Schauenburg for group-theoretical categories in Ref. \cite{Schauenburg15}.

\begin{figure}
\centering
\includegraphics[width = 0.65\textwidth]{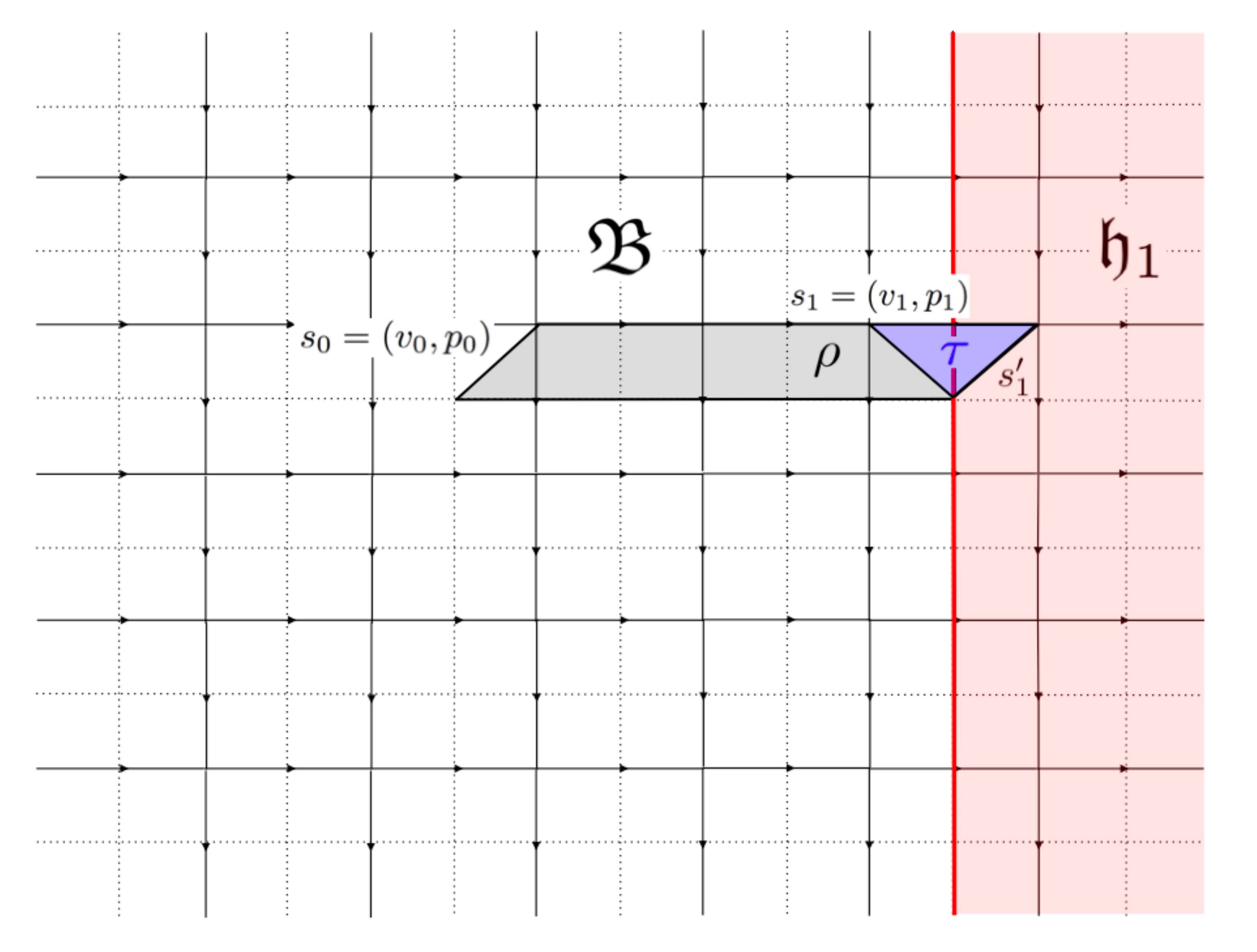}
\caption{Illustration of the bulk-to-boundary condensation procedure. If the new ribbon operator $F_{\rho \cup \tau}$ commutes with the Hamiltonian terms at the cilium $s_1$, we say that the particle $a$ has condensed to vacuum on the boundary.}
\label{fig:condensation-2}
\end{figure}

\begin{remark}
\label{dual-bd-rmk-2}
In Remark \ref{dual-bd-rmk}, we noted that it is also possible to create a boundary line on the dual lattice. In this case, the ``condensation triangle'' $\tau$ of Fig. \ref{fig:condensation} is now a direct triangle instead of a dual triangle (see Fig. \ref{fig:condensation-2}). For this case, there are analogous results to Theorems \ref{condensation-products} and \ref{inverse-condensation-products}, which are obtained using ribbon operators on the direct triangle. In general, these two methods of creating boundaries with the same subgroup $K$ can result in different boundary types (i.e. different condensation formulas as in Equations (\ref{eq:condensation}) and (\ref{eq:inverse-condensation})). However, it is straightforward to show that the boundary type corresponding to the dual lattice boundary may also be created using a different subgroup on a direct lattice boundary.
\end{remark}

\subsubsection{Multiple condensation channels}
\label{sec:multiple-condensation-channels}

As seen in Theorems \ref{condensation-products} and \ref{inverse-condensation-products}, it is possible (e.g. in the case of $G = S_3$ which we study in Section \ref{sec:ds3-hamiltonian-example}) that one can have condensation multiplicities $n^{(C,\pi)}_{(T,R)}$ greater than 1. In this case, we say that there are multiple {\it condensation channels}. Physically, this is very similar to multiple fusion channels in the bulk (such as in the UMTC given by $\text{SU}(3)_3$), where we have fusion rule coefficients greater than 1. 

The origin of the condensation multiplicity can be traced back to the local degrees of freedom in the definition of the (bulk) ribbon operators. Recall that a ribbon operator $F^{(C, \pi);(\vec{u},\vec{v})}$ has local degrees of freedom indexed by $\vec{u}$ (or $\vec{v}$) at the two ends of the ribbon, resulting in quantum dimension $|C|\text{Dim}(\pi)$. Application of local operators $D^{(h,g)}$ at the ends can mix the local states completely. However, if we move one end of the ribbon to a gapped boundary, to distinguish the local degrees of freedom we can only apply operators that commute with the boundary Hamiltonian. Therefore, on the boundary, without creating additional excitations, we may not be able to distinguish all the local degrees of freedom of the ribbon operator completely; the remaining degeneracy becomes the condensation multiplicity.

The most common situation where such multiplicity arises in quantum double models is when $K=\{1\}$. In this case, it is easy to see that all gauge charges condense to vacuum on the boundary, and the multiplicity is given by $n^{(\{1\}, \pi)}_{(\{1\},1)}=\mathrm{Dim}(\pi)$. Let us understand the multiplicity in this example more concretely in terms of ribbon operators. Recall that for a gauge charge corresponding to an irreducible representation $\pi$ of $G$, the ribbon operators read
\begin{equation}
	F_\rho^{(\{1\},\pi);(\bf{u,v})}=\mathrm{Dim}(\pi)\sum_{g\in G} \big(\Gamma_\pi^{-1}(g)\big)_{jj'}F^{(1, g)}_\rho.
	\label{}
\end{equation}
Suppose we now use the above ribbon operator to create a pair of charges in the bulk, and then move one of them, say the end corresponding to the $\vec{u}$ index, to the boundary.
In the bulk, one can easily show that applying $A^h$ at the end of the ribbon mixes the local indices. However, on the boundary $A^h$ do not commute with the $T^K$ boundary terms, unless one applies the product of all such $A$'s of the entire boundary. Therefore, different indices $\vec{u}$ are now locally indistinguishable, which explains the origin of the multiplicity.

From this example we also see that condensation channels are topologically protected. In order to change from one channel to another, without leaving any trace of excitation, one must apply a boundary ribbon operator that completely encircles the boundary. Since boundary particles are confined, such an operator would require energy input proportional to the perimeter of the boundary.

\subsection{Defects between different boundary types}
\label{sec:defect-hamiltonian}

Thus far, this chapter has presented general Hamiltonians to create arbitrary gapped boundaries (holes) in the Kitaev model for the Dijkgraaf-Witten theory based on group $G$. In these cases, each hole has been associated with a single subgroup $K \subseteq G$ (i.e. a single boundary type). In this section, we will examine cases in which one hole can be associated with multiple boundary types. When this happens, we say that a {\it defect} is located where two distinct boundary types meet. Although we will not use these defects for quantum computation in our paper, such defects actually have many interesting topological properties, and can also be used to encode qubits and qudits \cite{Bark13a,Bark13b,Bark13c}. Hence, we will write down explicit Hamiltonians to create defects between boundaries, and analyze some of their simple topological properties. We hope that this will provide a good foundation that future papers can use to study topological quantum computation.

\subsubsection{Hamiltonians for defects}

In general, there are multiple ways to create defects between different boundary types. The first and simplest way is a direct generalization of Fig. 3 of Ref. \cite{Fowler12}. This was illustrated as the hole $\mathfrak{h}_2$ in Fig. \ref{fig:boundary-hamiltonian-3}: as seen in Remark \ref{dual-bd-rmk-2}, in general, using the same subgroup $K$ to create boundaries on the direct/dual lattices will result in two different boundary types. As a result, one way to create defects between boundary types is to apply our Hamiltonian $H_{\text{G.B.}}$ to a lattice such as the one in Fig. \ref{fig:boundary-hamiltonian-3}: specifically, this can create defects between these two boundary types at each corner of the hole where a direct lattice boundary line meets a dual lattice boundary line. In the case of the hole $\mathfrak{h}_2$, this occurs at all four corners.

However, the above procedure only allows us to create very special kinds of defects: the two boundary types involved must be the direct/dual lattice boundary types corresponding to a common subgroup $K \subseteq G$. More generally, we would like to create defects between arbitrary boundary types. This can be done by defining a new commuting Hamiltonian $H_{\text{dft}}$.

\begin{figure}
\centering
\includegraphics[width = 0.78\textwidth]{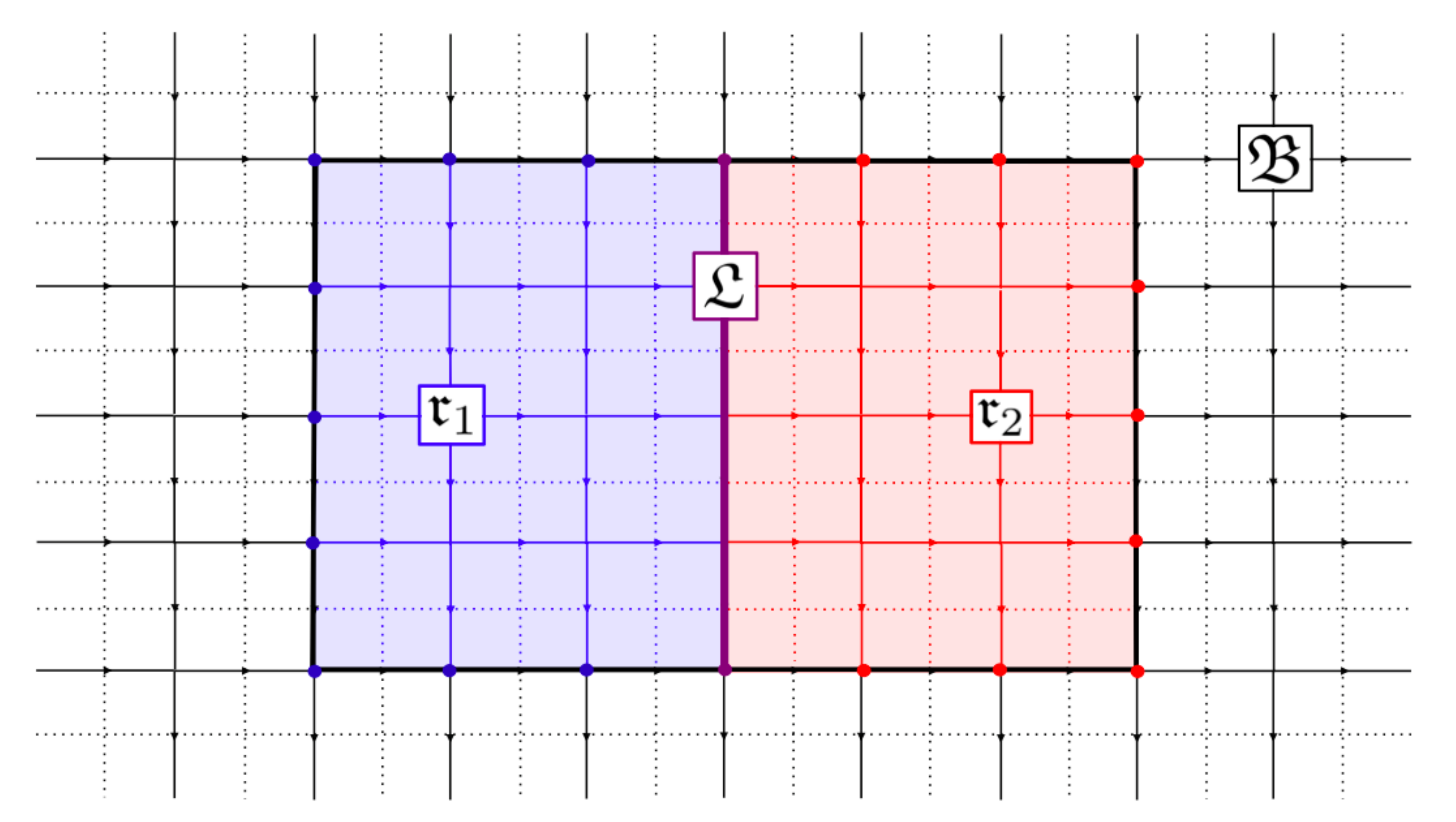}
\caption{Definition of the Hamiltonian $H_{\text{dft}}$. The line $\mathfrak{L}$ consists of all vertices and edges along the purple line which divides the blue and red regions. The region $\mathfrak{r}_1$ consists of all vertices, plaquettes, and edges within the blue shaded rectangle, and all blue vertices on the border of the hole. $\mathfrak{r}_2$ is defined similarly. As before, the bulk $\mfB$ will denote all other vertices, edges, and plaquettes (i.e. the black and white region). Specifically, we would like to note that edges on the border of the hole are considered as part of the bulk $\mfB$.} 
\label{fig:defect-hamiltonian}
\end{figure}

To define this new Hamiltonian, let us first consider a picture such as Fig. \ref{fig:defect-hamiltonian}. Suppose we would like to create two defects, one at each endpoint of the line $\mfL$. Specifically, we would like the boundary type to the left of $\mfL$ (i.e. the blue portion) to be the one given by subgroup $K_1$, and the boundary type to the right of $\mfL$ (red portion) to be given by $K_2$. In Section \ref{sec:bd-hamiltonian}, we defined the Hamiltonian $H^{(K,1)}_{(G,1)}$ for each subgroup $K \in G$, and combined this with the original Kitaev Hamiltonian to define the Hamiltonian $H_{\text{G.B.}}$ for the lattice with arbitrary holes and boundary types. Following that model, we will now combine several Hamiltonians $H^{(K,1)}_{(G,1)}$ and the original $H_{(G,1)}$ to form $H_{\text{dft}}$:

\begin{equation}
\label{eq:defect-hamiltonian-1}
H_{\text{dft}} = H_{(G,1)} (\mfB) + H^{(K_1,1)}_{(G,1)} (\mfr_1) + H^{(K_2,1)}_{(G,1)} (\mfr_2) + H^{(K_1 \cap K_2,1)}_{(G,1)} (\mfL).
\end{equation}

\noindent
Here, the region $\mathfrak{r}_1$ consists of all vertices, plaquettes, and edges within the blue shaded rectangle, and all blue vertices on the border of the hole. $\mathfrak{r}_2$ is defined similarly. The line $\mfL$ consists of all vertices and edges along the purple line that divides $\mfr_1$ from $\mfr_2$. As before, the bulk $\mfB$ will denote all other vertices, edges, and plaquettes (i.e. the black and white region). Again, we would like to note that edges on the border are considered as part of the bulk $\mfB$.

It is simple to check that the above Hamiltonian is also a commuting Hamiltonian.

\begin{figure}
\centering
\includegraphics[width = 0.78\textwidth]{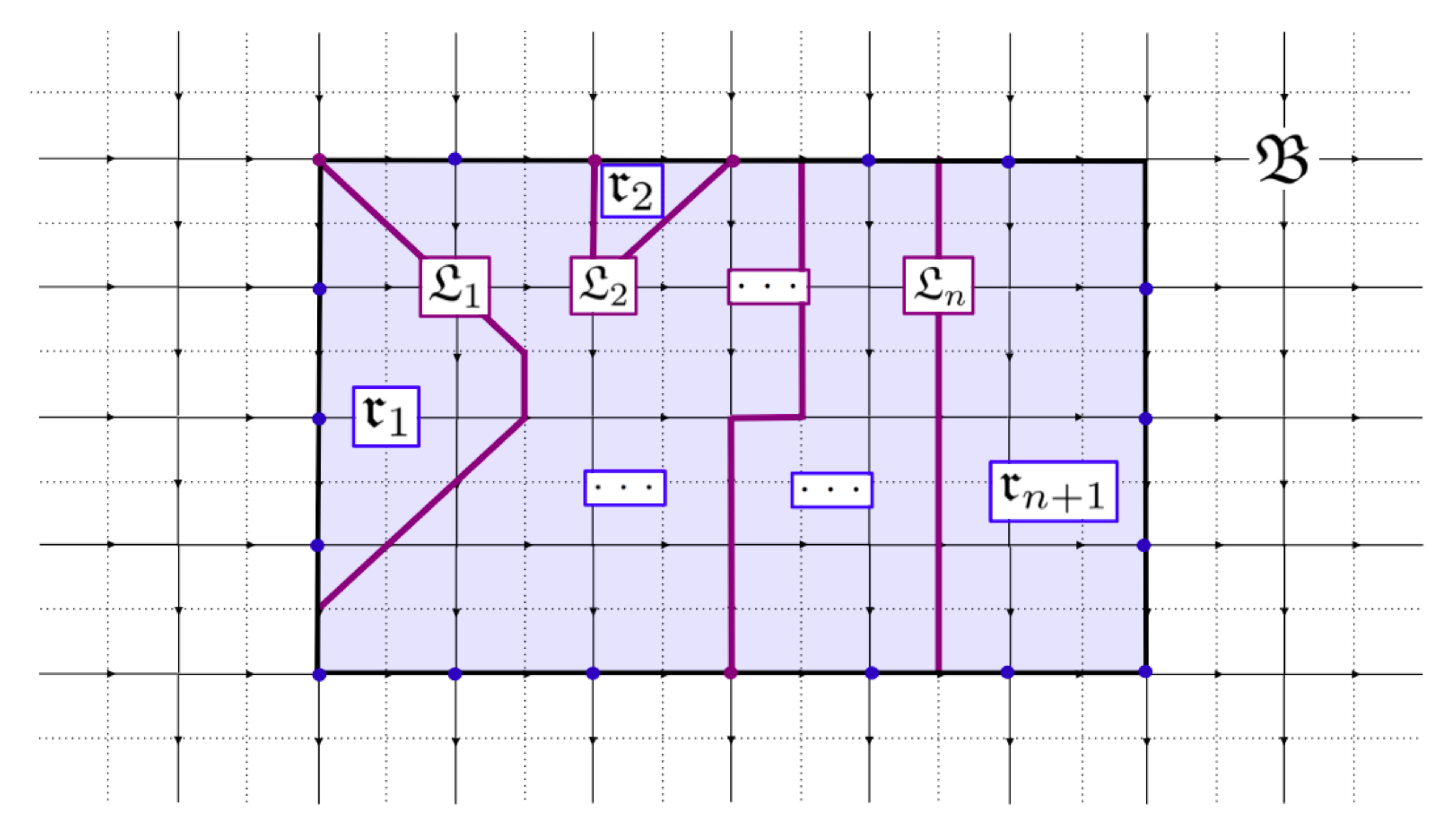}
\caption{More general definition of the Hamiltonian $H_{\text{dft}}$. There are $n$ non-intersecting piecewise-linear borders $\mfL_1, ... \mfL_n$, which separate the hole into $n+1$ regions $\mfr_1, ... \mfr_{n+1}$. Each $\mfL_i$ is formed from a chain of cilia and (direct or dual) edges. A total of $2n$ defects are created, one at each endpoint of each border $\mfL_i$. Each border $\mfL_i$ consists of all vertices, plaquettes, and edges that the corresponding purple line(s) cross, and includes the direct or dual vertex on the boundary of the hole if applicable (purple dots). Each region $\mathfrak{r}_j$ consists of all vertices, plaquettes, and edges within the corresponding blue shaded region, including the direct or dual vertices along the border of the hole (blue dots), but not along the purple lines. The bulk $\mfB$ will denote all other vertices, edges, and plaquettes (i.e. the black and white region). Again, we would like to note that edges on the border of the hole are considered as part of the bulk $\mfB$.} 
\label{fig:defect-hamiltonian-2}
\end{figure}

More generally, we may create as many defects as we like on the boundary of a hole, and the border lines between different boundary types need not be straight. This is illustrated in Fig. \ref{fig:defect-hamiltonian-2}. In this case, for each $i = 1, 2, ... n$, let us define $a(i), b(i) \in \{1,2,...n+1\}$ to be the numbers such that the regions on the two sides of $\mfL_i$ are $\mfr_{a(i)}$ and $\mfr_{b(i)}$ (it does not matter which is which). Suppose each region $\mfr_j$ is given by the subgroup $K_j \subseteq G$. The general Hamiltonian $H_{\text{dft}}$ is then defined as follows:

\begin{equation}
\label{eq:defect-hamiltonian-2}
H_{\text{dft}} = H_{(G,1)} (\mfB) + \sum_{j=1}^{n+1} H^{(K_j,1)}_{(G,1)} (\mfr_j) + \sum_{i = 1}^{n} H^{(K_{a(i)} \cap K_{b(i)},1)}_{(G,1)} (\mfL_i).
\end{equation}

\noindent
Here, the Hamiltonian $H^{(K_{a(i)} \cap K_{b(i)},1)}_{(G,1)}$ is applied to all vertices, plaquettes, and edges that the corresponding purple line(s) of $\mfL_i$ cross, including the direct or dual vertex on the boundary of the hole if applicable (purple dots in Fig. \ref{fig:defect-hamiltonian-2}). The Hamiltonian $H^{(K_j,1)}_{(G,1)}$ is applied to all vertices, plaquettes, and edges within the corresponding blue shaded region, including the direct or dual vertices along the border of the hole (blue dots in Fig. \ref{fig:defect-hamiltonian-2}), but not those along the purple lines. The bulk Hamiltonian is applied to all other vertices, edges, and plaquettes (i.e. the black and white region). Again, we would like to note that edges on the border of the hole are considered as part of the bulk $\mfB$.

Of course, it is very simple to generalize this to the case with many holes.

\begin{remark}
We would like to note that the above Hamiltonian creates boundary defects in pairs (corresponding to the two endpoints of each purple line of Fig. \ref{fig:defect-hamiltonian-2}). In general, defects are quite similar to anyons, in the sense that they live on cilia (as opposed to gapped boundaries/holes). Hence, we believe that it is not possible to create a single boundary defect, unlike gapped boundaries (see Remark \ref{single-bd-creation}).
\end{remark}

\subsubsection{Topological properties of defects}

We will now examine some topological properties of the defects between different boundary types, as these are properties that must be harnessed in order to perform topological quantum computation. Let us consider again the simple case of two defects, as shown in Fig. \ref{fig:defect-ribbons} (the generalizations are obvious). As before, the two regions $\mfr_1$, $\mfr_2$ are given by subgroups $K_1, K_2 \subseteq G$, respectively. 

\begin{figure}
\centering
\includegraphics[width = 0.78\textwidth]{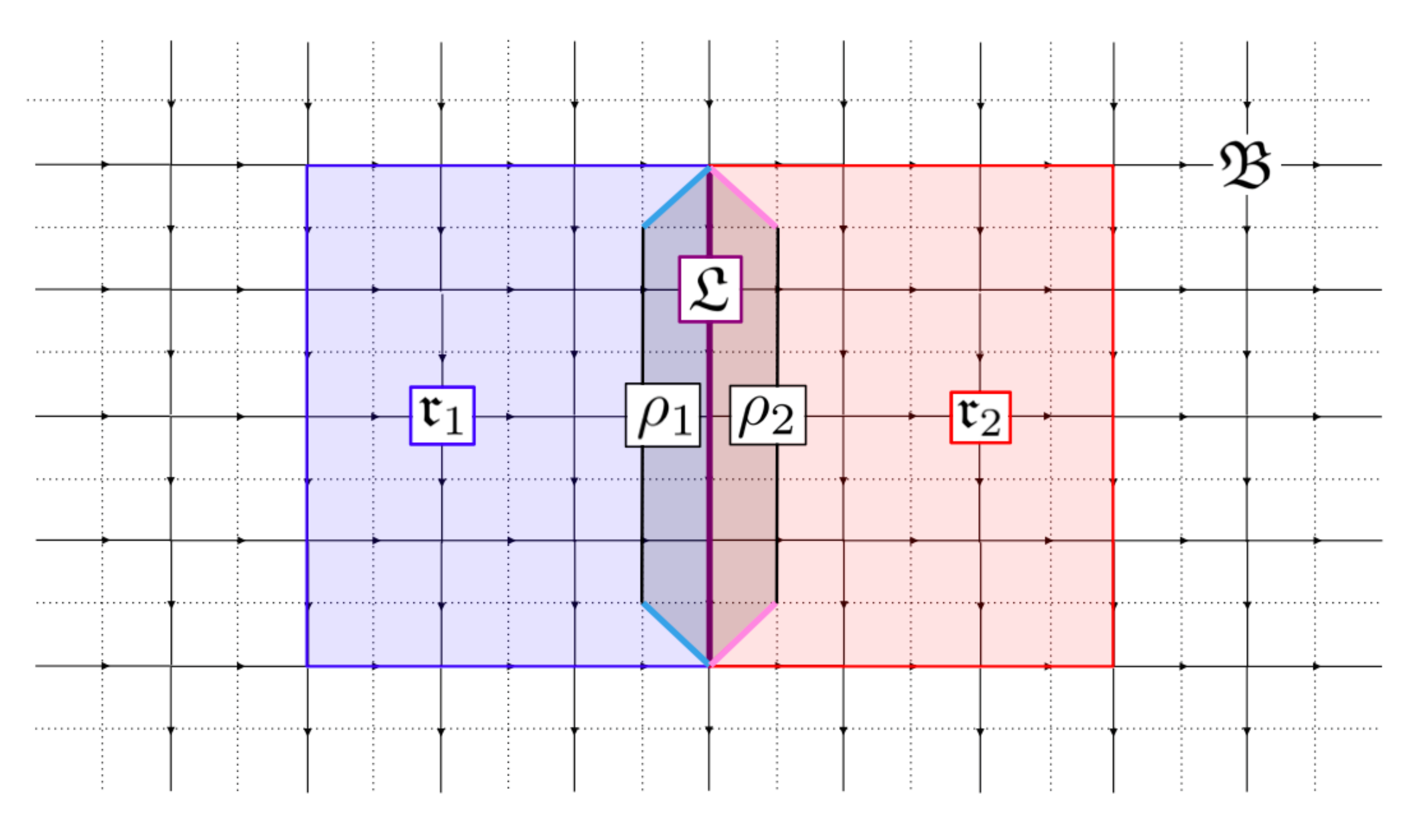}
\caption{Ribbons along the defect border. In principle, it is possible to define ribbon operators along the special ribbons $\rho_1$ and $\rho_2$, to analyze the topological properties of the defects. To be precise, each defect is a cilium, and these properties depend on whether we choose the cilia as endpoints of $\rho_1$ (light blue) or $\rho_2$ (pink).} 
\label{fig:defect-ribbons}
\end{figure}

In Sections \ref{sec:ribbon-operators}, \ref{sec:bd-hamiltonian}, and \ref{sec:bd-excitations}, we have defined ribbon operators for the bulk and boundary Hamiltonians and used them to analyze topological properties bulk and boundary excitations. In principle, one may do the same with the new Hamiltonians $H_{\text{dft}}$, where we would consider ribbons along $\rho_1$ and $\rho_2$ in Fig. \ref{fig:defect-ribbons}. For the sake of space and relevance, we will not present the details here; we will simply state the results for the topological parameterizations of defect types. Furthermore, as we will discuss soon, these operators are not of practical use, due to the energy costs of quasi-particle confinement with a hole.

As we saw in Section \ref{sec:bd-excitations}, by performing a Fourier transform on the ribbon operator algebra $Y(G,1,K,1)$, we showed that the topological labels of the elementary excitations on a subgroup $K$ boundary are given by pairs $(T,R)$, where $T \in K \backslash G / K$ is a double coset, and $R$ is an irreducible representation of the stabilizer $K^{r_T} = r_T K r_T^{-1}$ for some representative $r_T \in T$. Furthermore, if no ribbon operator was applied, we obtained the trivial excitation, which was given by the pair $(T = K, \text{ trivial representation})$. Similarly, in this situation, we may write down the ribbon operator algebra for all ribbons along $\rho_1$ or $\rho_2$, and perform a Fourier transform to obtain the corresponding simple defect types. In doing so, we see that if we consider the two defects as the end cilia of $\rho_1$ (light blue segments in Fig. \ref{fig:defect-ribbons}), the simple defect types are given by the following set:

\begin{equation}
\{
(T,R): T \in K_1 \backslash G / K_2, \text{ } R \in (K_1, K_2)^{r_T}_{\text{ir}} 
\}.
\end{equation}

\noindent
As always, $r_T \in T$ is a representative of the double coset. Here, we define the subgroup $(K_1, K_2)^{r_T} = K_1 \cap r_T K_2 r_T^{-1}$ to be the generalized stabilizer group.  Likewise, if we consider the two defects as the end cilia of $\rho_2$ (pink segments in Fig. \ref{fig:defect-ribbons}), the simple defect types are given by the set

\begin{equation}
\label{eq:simple-defect-types}
\{
(T,R): T \in K_2 \backslash G / K_1, \text{ } R \in (K_2, K_1)^{r_T}_{\text{ir}} 
\}.
\end{equation}

For the former case, the quantum dimension of the simple defect is given by:

\begin{equation}
\label{eq:simple-defect-qdim}
\FPdim(T,R) = \frac{\sqrt{|K_1| |K_2|}}{|K_1 \cap r_T K_2 r_T^{-1}|} \cdot \Dim(R).
\end{equation}

\noindent
The formula for the latter case is obtained by swapping $K_1$ and $K_2$. For the rest of the chapter, without loss of generality, we will consider only the first case.

When we create a pair of defects by applying the Hamiltonian $H_{\text{dft}}$, both defects are of the type $(T,R)$, where $T = K_1 1 K_2$ is the trivial double coset and $R$ is the trivial representation. In general, by applying a ribbon operator along one of the ribbons $\rho_i$, one may change the defect into a different defect type. However, this is a highly impratical procedure: as before, all excitations within the hole are still confined, so the amount of energy required to apply such a ribbon operator would be linearly proportional to the length of $\rho_i$, which is typically very large for the purposes of topological protection. Hence, in general, it is only of interest to consider the defect type $(T,R) = (K_1 1 K_2, \text{ trivial})$.

We note that in the case where $K_1 = K_2 = K$, we obtain all the relations and properties for elementary excitations in the boundary, which we studied in Section \ref{sec:bd-excitations}.

\begin{remark}
In this analysis, we have only considered defect types created using our Hamiltonian $H_{\text{dft}}$ for holes completely along the direct lattice. If we have a corner defect as in the case of $\mfh_2$ of Fig. \ref{fig:boundary-hamiltonian-3}, one should first determine the subgroups that would give the corresponding boundary types if all boundary lines were on the direct lattice (see Remark \ref{dual-bd-rmk-2}) and then apply these formulas.
\end{remark}

\subsubsection{Fusion and braiding of boundary defects}

Boundary defects are very similar to anyons, as both live on cilia. Hence, it is reasonable to expect that boundary defects can also be moved. This is certainly true, as one can move a defect along the boundary of a hole, simply by adiabatically tuning the Hamiltonian $H_{\text{dft}}$. 

Motivated by the topological operations that can be obtained by anyon braiding (e.g. in \cite{Cui15-m,Cui15}), we would also like to braid defects with each other. Suppose we have two boundary types as shown in Fig. \ref{fig:defect-ribbons}, where the blue and red regions $\mfr_1$, $\mfr_2$ have boundary types given by subgroups $K_1, K_2 \subseteq G$, respectively. The effective ``braiding'' of boundary defects occurs when two of them are moved very close to each other and they fuse. In general, fusion of boundary defects produces a degeneracy given by symmetries of the Hamiltonian $H_{\text{dft}}$. The braid relation will result from looking at the fusion of the defects from the perspective of the light blue cilium on the top fusing with the pink cilium on the bottom, vs. the pink cilium on the top fusion with the light blue cilium on the bottom (see Fig. \ref{fig:defect-ribbons}). These two distinct tensor products are not necessarily related to each other trivially, and hence gives a braiding.

\subsection{Example: The Toric Code}
\label{sec:tc-hamiltonian-example}

In this section, we present the toric code as an example to illustrate the theory we have developed in this chapter.

\subsubsection{Toric code Hamiltonian}

The toric code is a special case of the Kitaev model, where $G = \Z_2$. In this case, a standard qubit is attached to each edge of the lattice in Fig. \ref{fig:kitaev}. Since $G$ is abelian, it is not necessary to define orientations for edges. By definition, we have  

\begin{equation}
L^0_- \ket{g} = L^0_+ \ket{g} = \ket{g}, \qquad
L^1_- \ket{g} = L^1_+ \ket{g} = \sigma^x \ket{g}.
\end{equation}

\noindent
where $\sigma^x$ is the Pauli $x$ operator. Hence,

\begin{equation}
1-A(v) = \frac{1}{2}\left( 1-  \prod_{j \in \text{star}(v)} \sigma^x_j \right).
\end{equation}

Similarly, one can show that 

\begin{equation}
1-B(p) = \frac{1}{2} \left( 1- \prod_{j \in \text{boundary}(p)} \sigma^z_j \right).
\end{equation}

This means that, up to rescaling and constant shift, the Hamiltonian for the toric code is given by 

\begin{equation}
\label{eq:tc-hamiltonian}
H_0 = - \sum_v \prod_{j \in \text{star}(v)} \sigma^x_j - \sum_p \prod_{j \in \text{boundary}(p)} \sigma^z_j.
\end{equation}

\subsubsection{Elementary excitations}

We can now examine the elementary excitations of the toric code. By the discussions of this chapter, each elementary excitation is a pair $(C,\pi)$, where $C$ is a conjugacy class of $\Z_2$, and $\pi$ is an irreducible representation of $E(C)$, the centralizer of an element of $C$.

The conjugacy classes $C$ of $\Z_2$, the corresponding centralizers $E(C)$ and their irreducible representations are:

\begin{enumerate}
\item
$C_0 = \{0\}: E(C_0) = \Z_2$. Irreducible representations: $\pi_{0} = \{1,1\}$, $\pi_{1} = \{1,-1\}$.
\item
$C_1 = \{1\}: E(C_1) = \Z_2$. Irreducible representations: same as above.
\end{enumerate}

Hence, the toric code has 4 elementary excitations. In most literature, and in our paper, these quasi-particle types are labeled as follows: 

\begin{equation}
\begin{split}
(C_0, \pi_0) \rightarrow 1, \qquad (C_0, \pi_1) \rightarrow e \\
(C_1, \pi_0) \rightarrow m, \qquad (C_1, \pi_1) \rightarrow \epsilon
\end{split}
\end{equation}

\subsubsection{Gapped boundaries}

We now describe the gapped boundary Hamiltonian terms for the toric code. $G = \Z_2$ has two subgroups, namely the trivial subgroup $K = \{0\}$ and the group $K = G$ itself, which correspond to the ``rough'' and ``smooth'' boundaries in the literature, respectively. Let us first consider the former.

Since $K$ is trivial, the operators $A^K,L^K$ defined in Section \ref{sec:bd-hamiltonian} are equal to identity. The projectors $B^K,T^K$ restrict all edges to the state $\ket{0}$. 

Let us follow the method of Section \ref{sec:bd-excitations} to determine the excitations on the boundary, and the particles that can condense to vacuum. First, there are two double cosets:

\begin{enumerate}
\item
$r_{T_0} = 0: T_0 = \{0\}, K^{r_{T_0}} = \{0\} = Q({T_0}), S({T_0}) = \{0\}$ 
\item
$r_{T_1} = 1: T_1 = \{1\}, K^{r_{T_1}} = \{0\} = Q({T_1}), S({T_1}) = \{1\}$ 
\end{enumerate}

In each case, the only representation of $K^{r_{T_i}}$ is the trivial representation $R_0$. The particle $(T_0, R_0)$ corresponds to vacuum on the boundary.

By Theorem \ref{inverse-condensation-products}, the bulk elementary excitations that condense to $(T_0, R_0)$ are the particles $1,e$, and the bulk elementary excitations that condense to $(T_1, R_0)$ are the particles $m,\epsilon$. Because of this, we adopt the following notation for the elementary excitations on the boundary:

\begin{equation}
(T_0, R_0) \rightarrow 1, \qquad (T_1,R_0) \rightarrow m
\end{equation}

For the rest of the paper, we will call this boundary the $1+e$ boundary. In general, especially in Chapter \ref{sec:algebraic}, we will denote a boundary by the corresponding decomposition (\ref{eq:inverse-condensation}) for vacuum on that boundary. This is because the bulk excitations that can condense to vacuum will play the most important role in quantum computation using the boundary.

In a similar fashion, one can show that when we take $K = G$ for the toric code, we get a $1+m$ boundary: the pure fluxons $1,m$ condense to vacuum, and the particles $e,\epsilon$ condense to an excited state $e$.

\subsubsection{Defects between boundaries}
\label{sec:tc-defects-hamiltonian}

We may now consider the defects between $1+e$ and $1+m$ boundaries of the toric code, as an illustration of Section \ref{sec:defect-hamiltonian}. Let us examine the case where $K_1 = \{0\}$ and $K_2 = \Z_2$ By Eq. (\ref{eq:simple-defect-types}), the simple defect types are parameterized by pairs $(T,R)$, where $T \in K_1 \backslash G / K_2$, and $R$ is an irreducible representation of $(K_1, K_2)^{r_T}$. In this case, there is only one such pair $(T_0, R_0)$, namely the trivial double coset and the trivial representation. Furthermore, by Eq. (\ref{eq:simple-defect-qdim}), we have

\begin{equation}
\FPdim(T_0,R_0) = \sqrt{2}.
\end{equation}

The same defect type is obtained if we switch $K_1$ and $K_2$ in this case.

In fact, it is believed that this defect is topologically equivalent to (i.e. has the same projective braid statistics as) the Majorana zero mode. The above calculation for the quantum dimension of this defect is strong evidence for this claim. In principle, one can use techniques such as ribbon operators to compute other properties of the defect to fully verify it.

In general, this same calculation for the cyclic group $\Z_p$, $p \geq 3$ any prime, will show that the boundary defect between boundary types $K_1 = \{0\}$ and $K_2 = \Z_p$ gives a defect with the same projective braid statistics as the para-fermion zero mode.

\subsection{Example: $\mfD(S_3)$}
\label{sec:ds3-hamiltonian-example}

In this section, we present an example using the group $G = S_3 = \{r,s|r^3 = s^2 = srsr = 1 \}$, the permutation group on three elements, to illustrate our theory on the simplest non-abelian group. Since this group is already quite complicated, we will not explicitly write out the Hamiltonian in full, although the interested reader can easily obtain it from Eq. (\ref{eq:kitaev-hamiltonian}). Instead, we focus our attention on the elementary excitations, gapped boundaries, and boundary defects of this example.

\subsubsection{Elementary excitations}

To determine the elementary excitations of this model, we again only need to find the pairs $(C,\pi)$, as described in Eq. (\ref{eq:elementary-ribbon-basis}). The conjugacy classes $C$ of $S_3$, the corresponding centralizers $E(C)$ and their irreducible representations are:

\begin{enumerate}
\item
$C_0 = \{1\}: E(C_0) = S_3$. Three irreducible representations: trivial ($A$), sign ($B$), and the two-dimensional one ($C$).
\item
$C_1 = \{s, sr, sr^2\}: E(C_1) = \{1,s\} = \Z_2$. Two irreducible representations: trivial ($D$), sign ($E$).
\item
$C_2 = \{r, r^2\}: E(C_2) = \{1,r,r^2\} = \Z_3$. Three irreducible representations: trivial ($F$), $\{1, \omega, \omega^2\}$\footnote{Here $\omega = e^{2\pi i/3}$ is the third root of unity. $\{1, \omega, \omega^2\}$ means the representation where $1 \rightarrow 1$, $r \rightarrow \omega$, $r^2 \rightarrow \omega^2$.} ($G$),
$\{1, \omega^2, \omega^4\}$ ($H$).
\end{enumerate}

Hence, there are 8 anyon types for this model, namely $A-H$ as listed above.

\subsubsection{Gapped boundaries}

The group $G = S_3$ has 4 distinct subgroups up to conjugation, namely the trivial subgroup, $\Z_2$, $\Z_3$, and $G$ itself. In what follows, we solve for the 4 gapped boundaries corresponding to these subgroups. As before, we will follow the method of Section \ref{sec:bd-excitations} to determine the excitations on the boundary, and the bulk anyons that can condense to vacuum.

\vspace{2mm}
\begin{itemize}[wide, label={}, listparindent=1.5em, parsep=0.25mm, itemsep=3mm, labelindent=0pt]
\item
{\it \underline{Case I}:} $K = \{1\}$.

Since $K$ is trivial, there are 6 distinct double cosets, corresponding to each element $r_T \in G$. In each case, we have $K^{r_T} = Q(T) = \{1\}$, and $S(T) = \{r_T\}$, so the only representation of $K^{r_T}$ is the trivial one. There are hence 6 elementary excitations on the boundary; let us label each excitation by the corresponding choice of $r_T$.

By Theorem \ref{inverse-condensation-products}, the bulk elementary excitations that condense to the trivial excitation on the boundary are the particles corresponding to the trivial conjugacy class $C = \{1\}$, i.e. the chargeons. In general, for any finite group $G$, a simple argument shows that $K = \{1\}$ will always form the charge condensate boundary. More specifically, the ``boundary topological order'' corresponding to this boundary will always be described by the fusion category $\C[G]$.

More generally, we can use Theorem \ref{condensation-products} to determine the result of condensing each simple bulk anyon to the boundary:

\begin{multicols}{2}
\begin{enumerate}[label=(\roman*),leftmargin=0.5in]
\item
$A,B \rightarrow 1$
\item
$C \rightarrow 2 \cdot 1$
\item
$D,E \rightarrow s \oplus sr \oplus sr^2$
\item
$F,G,H \rightarrow r \oplus r^2$
\end{enumerate}
\end{multicols}

Similarly, by Theorem \ref{inverse-condensation-products}, we can determine the bulk anyons that result from pulling an elementary excitation out of the boundary:

\vspace{2.5mm}
\begin{enumerate}[label=(\roman*),leftmargin=0.5in]
\item
$1 \rightarrow A \oplus B \oplus 2C$
\item
$s, sr, sr^2 \rightarrow D \oplus E$
\item
$r, r^2 \rightarrow F \oplus G \oplus H$
\end{enumerate}
\vspace{2.5mm}

Following the convention in the previous section, we will say that this subgroup forms an $A+B+2C$ boundary.

\item
{\it \underline{Case II}:} $K = \Z_2 = \{1,s\}$.
\vspace{2mm}

In this case, we see that there are only 2 double cosets, which give 3 elementary boundary excitations:

\vspace{2.5mm}
\begin{enumerate}[label=(\roman*),leftmargin=0.5in]
\item
$r_{T_1} = 1: T_1 = \{1,s\} = K^{r_{T_1}}, Q(T_1) = S(T_1) = \{1\}$. 2 irreducible representations of $K^{r_{T_1}}$: the trivial one ($A$), the sign one ($B$).
\item
$r_{T_2} = r: T_2 = \{r,r^2,sr,sr^2\}, K^{r_{T_2}} = \{1\}, Q(T_2) = S(T_2) = \{1,s\}$. There is only one trivial representation ($C$) of $K^{r_{T_2}}$.
\end{enumerate}
\vspace{2.5mm}

Using ribbon operator techniques, it is possible to show that this is in fact a boundary topological order given by the fusion category $\Rep(S_3)$.

We apply Theorem \ref{condensation-products} to determine the result of condensing each simple bulk anyon to the boundary:

\begin{multicols}{2}
\begin{enumerate}[label=(\roman*),leftmargin=0.5in]
\item
$A \rightarrow {A}$
\item
$B \rightarrow {B}$
\item
$C \rightarrow {A} \oplus {B}$
\item
$D \rightarrow {A} \oplus {C}$
\item
$E \rightarrow {B} \oplus {C}$
\item
$F,G,H \rightarrow {C}$
\end{enumerate}
\end{multicols}

Similarly, by Theorem \ref{inverse-condensation-products}, we have

\vspace{2.5mm}
\begin{enumerate}[label=(\roman*),leftmargin=0.5in]
\item
${A} \rightarrow A \oplus C \oplus D$
\item
${B} \rightarrow B \oplus C \oplus E$
\item
${C} \rightarrow D \oplus E \oplus F \oplus G \oplus H$
\end{enumerate}
\vspace{2.5mm}

Hence, $K = \Z_2$ corresponds to the $A+C+D$ boundary. We would like to note that for this case, it does not matter which of the three $\Z_2$ subgroups we choose, since they are equivalent up to conjugation; in the end, they all give the same boundary condensation rules.

\item
{\it \underline{Case III}:} $K = \Z_3 = \{1,r,r^2\}$.

The subgroup $K = \Z_3$ gives 2 double cosets:

\vspace{2.5mm}
\begin{enumerate}[label=(\roman*),leftmargin=0.5in]
\item
$r_{T_1} = 1: T_1 = \{1,r,r^2\} = K^{r_{T_1}}, Q(T_1) = S(T_1) = \{1\}$. 3 irreducible representations of $K^{r_T}$: the trivial one (${1}$), $\{1,\omega,\omega^2\}$ (${r}$), and $\{1,\omega^2,\omega\}$ (${r^2}$).
\item
$r_{T_2} = s: T_2 = \{s,sr,sr^2\}, K^{r_{T_2}} = \{1,r,r^2\}, Q(T_2) = \{1\}, S(T_2) = \{s\}$. 3 irreducible representations of $K^{r_T}$: the trivial one (${s}$), $\{1,\omega,\omega^2\}$ (${sr}$), and $\{1,\omega^2,\omega\}$ (${sr^2}$).
\end{enumerate}
\vspace{2.5mm}

As in Case I, it is possible to use ribbon operators to show that this boundary has bordered topological order given by the fusion category $\C[S_3]$.

Applying Theorem \ref{condensation-products} gives

\begin{multicols}{2}
\begin{enumerate}[label=(\roman*),leftmargin=0.5in]
\item
$A \rightarrow {1}$
\item
$B \rightarrow {1}$
\item
$C \rightarrow {r} \oplus {r^2}$
\item
$D,E \rightarrow {s} \oplus {sr} \oplus {sr^2}$
\item
$F \rightarrow 2 \cdot {1}$
\item
$G,H \rightarrow {r} \oplus {r^2}$
\end{enumerate}
\end{multicols}

\noindent
and Theorem \ref{inverse-condensation-products} gives

\vspace{2.5mm}
\begin{enumerate}[label=(\roman*),leftmargin=0.5in]
\item
$1 \rightarrow A \oplus B \oplus 2F$
\item
$s, sr, sr^2 \rightarrow D \oplus E$
\item
$r, r^2 \rightarrow C \oplus G \oplus H$
\end{enumerate}
\vspace{2.5mm}

This is the $A+B+2F$ boundary. Note that it is exactly the same as the $A+B+2C$ boundary with $C,F$ switched. This is due to the duality of $C,F$ in the $\mfD(S_3)$ theory.

\item
{\it \underline{Case IV}:} $K = G = S_3$.

In general, for any finite group $G$, the subgroup $K = G$ yields only a single double coset, with $K^{r_T} = G, Q(T) = 1, S(T) = \{r_T\}$. The elementary excitations are given by the irreducible representations of $G$, and the resulting bordered topological order is described by the fusion category $\Rep(G)$. Because of this, the only anyons that condense to the boundary are pure fluxons. (More specifically, each fluxon appears at least once in the decomposition of vacuum on the boundary).

For the case of $G=S_3$, let ${A},{B},{C}$ denote the three irreducible representations of $G$ as in Case II. Theorem \ref{condensation-products} gives the following for the condensation products:

\begin{multicols}{2}
\begin{enumerate}[label=(\roman*),leftmargin=0.5in]
\item
$A \rightarrow {A}$
\item
$B \rightarrow {B}$
\item
$C \rightarrow {C}$
\item
$D \rightarrow {A} \oplus {C}$
\item
$E \rightarrow {B} \oplus {C}$
\item
$F \rightarrow {A}$
\item
$G,H \rightarrow {C}$
\end{enumerate}
\end{multicols}

\noindent
Theorem \ref{inverse-condensation-products} gives

\vspace{2.5mm}
\begin{enumerate}[label=(\roman*),leftmargin=0.5in]
\item
${A} \rightarrow A \oplus D \oplus F$
\item
${B} \rightarrow B \oplus C \oplus E$
\item
${C} \rightarrow C \oplus D \oplus E \oplus G \oplus H$
\end{enumerate}
\vspace{2.5mm}

\noindent
Hence, this is the $A+F+D$ boundary.

\end{itemize}

\subsubsection{Defects between boundaries}

Let us now consider the defects between different boundaries of the $\mfD(S_3)$ model. We first consider defects between the $A+C+D$ boundary and the $A+F+D$ boundary; this corresponds to $K_1 = \Z_2$, $K_2 = S_3$. By Eq. (\ref{eq:simple-defect-types}), the simple defect types are parameterized by pairs $(T,R)$, where $T \in K_1 \backslash G / K_2$, and $R$ is an irreducible representation of $(K_1, K_2)^{r_T}$. In this case, there are two possible pairs: both correspond to the choice $r_T = 1$, as there is only one double coset, but we may choose $R$ to be the trivial representation ($R_0$) or the sign representation ($R_1$) of $\Z_2$. By Eq. (\ref{eq:simple-defect-qdim}), we have

\begin{equation}
\FPdim(T_0,R_0) = \FPdim(T_0,R_1) = \sqrt{3}.
\end{equation}

\noindent
The same defect types are obtained by switching $K_1$ and $K_2$ in this case.

Similarly, we may also consider the simple defects between an $A+B+2C$ boundary and an $A+B+2F$ boundary. This corresponds to $K_1 = \{1\}$, $K_2 = \Z_3$. Here, there are two double cosets ($r_T = 1$ or $r_T = s$), but in both cases, $(K_1, K_2)^{r_T}$ is trivial and has only the trivial representation. Hence, both of these simple defects also have quantum dimension $\sqrt{3}$.

The same calculations may be repeated for all other choices of $K_1$ and $K_2$ as an exercise. However, we have chosen the above two examples due to the following observation:

\begin{remark}
In Section \ref{sec:tc-defects-hamiltonian}, we saw that the defect between the $1+e$ and $1+m$ boundaries of the toric code had quantum dimension $\sqrt{2}$, and claimed that it corresponds to the Majorana zero mode of the Ising theory. In this section, we see that there are two simple defect types between $A+C+D$/$A+F+D$ or between $A+B+2C$/$A+B+2F$, each with quantum dimension $\sqrt{3}$. We believe that in this case, the defects behave like the two particles of quantum dimension $\sqrt{3}$ in the $\SU(2)_4$ theory. It seems more than a mere coincidence that these correspond to the theories formed by gauging a $\Z_2$ symmetry in the original Dijkgraaf-Witten theory (the $e-m$ symmetry for the toric code, or the $C-F$ symmetry for the case of $\mfD(S_3)$). In general, it would be interesting to study whether there is any correspondence between these boundary defects and $\Z_2$ symmetry gauging.
\end{remark}

\subsection{Example: Genons in $\mfD(G \times G)$}
\label{sec:genon-hamiltonian}

In this section, we illustrate how the Hamiltonian $H_{\text{dft}}$ of Section \ref{sec:defect-hamiltonian} can be used to create genons in a bilayer TQFT based on an arbitrary finite group $G$. As discussed in Ref. \cite{Bark13a}, a genon in a bilayer TQFT is a defect in the $\Z_2$ symmetry. Genons have been studied extensively in Refs. \cite{Barkeshli11, Bark13a}, and can allow for universal quantum computation (see Refs. \cite{Bark13a, Barkeshli15}). However, these works have almost exclusively constructed genons for bilayer abelian group TQFTs. In this section, we will present a systematic Hamiltonian construction of the ``bare defect'' genon \cite{Barkeshli14} in bilayer Dijkgraaf-Witten TQFTs based on an arbitrary finite group $G$.

\begin{figure}
\centering
\includegraphics[width = 0.78\textwidth]{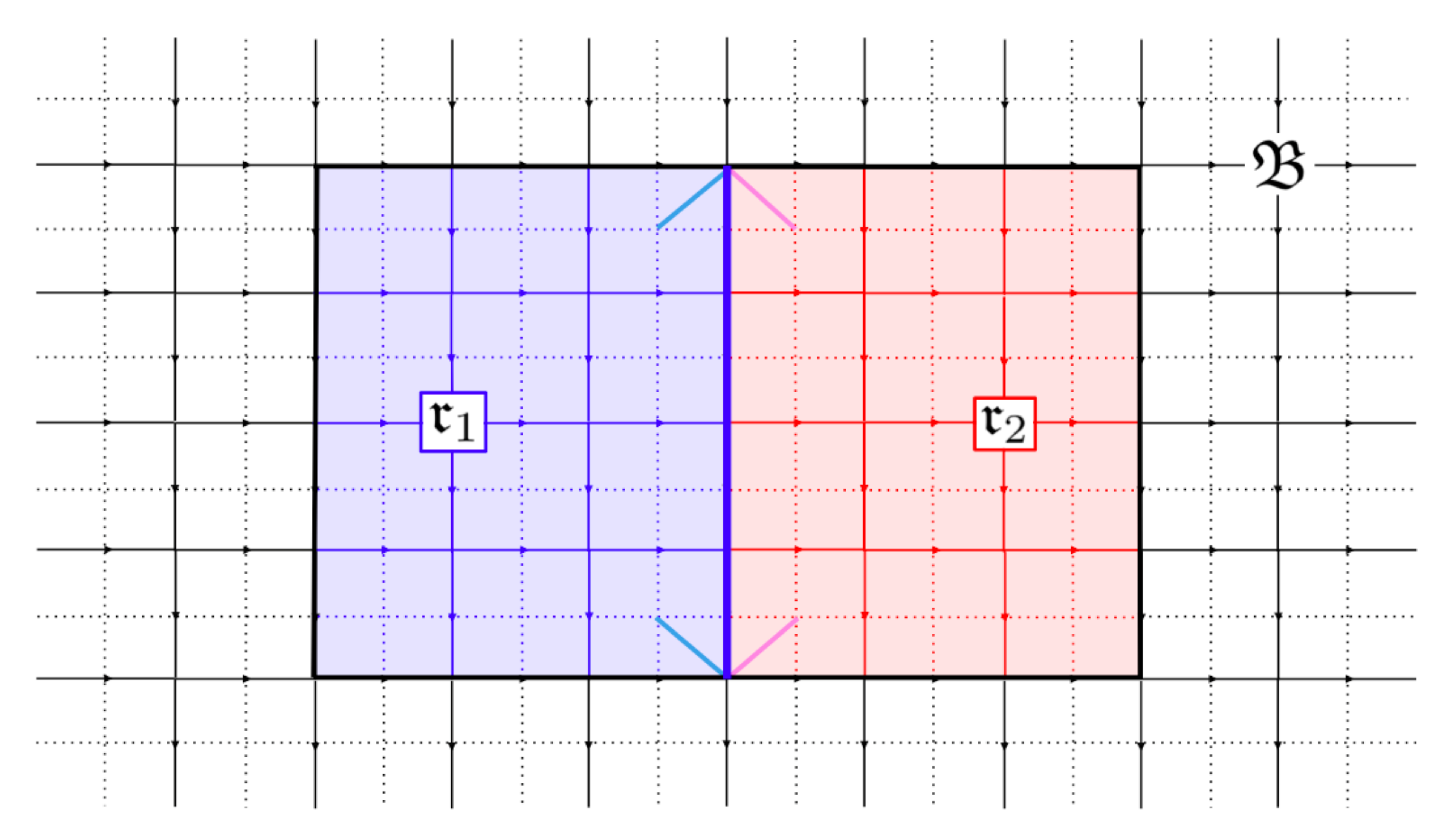}
\caption{Definition of the Hamiltonian $H_{\text{gn}}$. The region $\mfr_1$ consists of all vertices and plaquettes within and on all boundaries of the blue shaded rectangle, and all blue edges, including the ones on the thick blue dividing line between the two regions. The vertices and plaquettes of the region $\mfr_2$ are all those within the red shaded rectangle and on the upper, right, or lower boundaries of the rectangle; the edges of $\mfr_2$ are all red edges. The bulk $\mfB$ consists of everything else (black and white). Each resulting genon is shown by two cilia.} 
\label{fig:genon-hamiltonian}
\end{figure}

The Hamiltonian to create the bare defect genon is a special case of the general defect Hamiltonian $H_{\text{dft}}$, in the case where the input group is $G \times G$ (i.e. all data qudits on edges take values in the Hilbert space $\C[G \times G]$). In this situation, we begin with a hole in the lattice, divided into two regions, $\mfr_1$ and $\mfr_2$, as pictured in Fig. \ref{fig:genon-hamiltonian}. The region $\mfr_1$ consists of all vertices and plaquettes within and on all boundaries of the blue shaded rectangle, and all blue edges, including the ones on the thick blue dividing line between the two regions. The vertices and plaquettes of the region $\mfr_2$ are all those within the red shaded rectangle and on the upper, right, or lower boundaries of the rectangle; the edges of $\mfr_2$ are all red edges. To create bare defect genons, we will consider two specific subgroups of $G \times G$, namely the trivial subgroup $K_1 = \{1_G\} \times \{1_G\}$ and the subgroup $K_1 = G \times \{1_G\}$ ($1_G$ is the identity element of $G$). We apply the Hamiltonian $H^{(K_1,1)}_{(G \times G, 1)}$ to the region $\mfr_1$, and the Hamiltonian $H^{(K_2,1)}_{(G \times G, 1)}$ to the region $\mfr_2$; as always, the bulk Hamiltonian $H_{(G \times G, 1)}$ is applied to the bulk $\mfB$. Hence, the Hamiltonian to produce two bare defect genons is given by 

\begin{equation}
\label{eq:genon-hamiltonian}
H_{\text{gn}} = H_{(G \times G,1)} (\mfB) + H^{(K_1,1)}_{(G \times G,1)} (\mfr_1) + H^{(K_2,1)}_{(G \times G,1)} (\mfr_2).
\end{equation}

\noindent
By the analysis of Section \ref{sec:defect-hamiltonian}, this Hamiltonian is exactly solvable. Of course, the generalization to producing multiple genons on the same hole is simple; it directly follows from the generalization in Section \ref{sec:defect-hamiltonian}.

By Section \ref{sec:defect-hamiltonian}, the simple defect types that can be created by $H_{\text{gn}}$ are given by pairs $(T,R)$, where $T \in K_1 \backslash G / K_2$ is a double coset, and $R$ is an irreducible representation of $(K_1, K_2)^{r_T} = K_1 \cap r_T K_2 r_T^{-1}$ for some representative $r_T \in T$. Since $K_1$ is trivial, $(K_1, K_2)^{r_T}$ is always trivial, and there are exactly $|G|$ double cosets $T$, one corresponding to each $r_T = (1_G,g)$, $g \in G$. The bare defect genon corresponds to the choice $r_T = 1$. As discussed in Section \ref{sec:defect-hamiltonian}, this is also the only simple defect type which may be created in this way without a high energy cost due to confinement of excitations on and within the boundary. 

We may now compute the quantum dimension of the simple defect types of the Hamiltonian $H_{\text{gn}}$. By Eq. (\ref{eq:simple-defect-qdim}), we have

\begin{equation}
\label{eq:genon-qdim}
\FPdim(T,R) = \frac{\sqrt{|K_1| |K_2|}}{|K_1 \cap r_T K_2 r_T^{-1}|} \cdot \Dim(R) = \sqrt{|G|}
\end{equation} 

\noindent
for each simple defect type, including the bare defect genon. This agrees precisely with the prediction of Section X.H in Ref. \cite{Barkeshli14}, which states (for group-theoretical cases) that there should be exactly $|G|$ defects, and that the bare defect genon is a direct sum of all simple objects in the modular tensor category formed by giving $\C[G]$ a braided structure (and hence has quantum dimension $\sqrt{|G|}$).

\vspace{2mm}
\section{Algebraic model of gapped boundaries}
\label{sec:algebraic}

In this chapter, we present a mathematical model of gapped boundaries using Lagrangian algebras in a modular tensor category. Throughout the chapter, we will assume the reader is familiar with the concepts of a fusion category and a modular tensor category; for reference on these topics, see Ref. \cite{BakalovKirillov}.

\subsection{Topological order}
\label{sec:topological-order}

In this section, we briefly review the mathematical theory that describes elementary excitations and anyons.

As we saw in Section \ref{sec:ribbon-operators}, the topological charges/anyon types in the Kitaev model with group $G$ are given by irreducible representations of the Drinfeld double $\mfD(G) = \mZ(\text{Vec}_G)$, which are pairs $(C,\pi)$ of a conjugacy class of $G$ and an irreducible representation of the centralizer of $C$. 

If we think dually, we can also view the Kitaev model in terms of the representation category of $G$. In this case, every edge of the lattice will be labeled by an object in the unitary fusion category $\mC = \Rep(G)$, the complex linear representations of the group $G$. The elementary excitations in this model will be given by simple objects in the modular tensor category $\B = \mZ(\Rep(G))$, the Drinfeld double of the representation category. Because $\text{Vec}_G$ is Morita equivalent to $\Rep(G)$, these simple objects are given precisely by the same pairs $(C,\pi)$. In fact, this dualization exactly gives the same topological order as the Kitaev model, using the Levin-Wen Hamiltonian \cite{Levin04}.

By using ribbon operator techniques as presented in Chapter \ref{sec:hamiltonian}, one can in principle compute the twists and braidings of all of the elementary excitations in the Kitaev model. In doing so, one can determine all of the $\mathcal{S},\mathcal{T}$ matrix entries for this anyon system. It is conjectured that these two matrices uniquely determine a modular tensor category. If this conjecture holds, using such an analysis, one can show that the topological order of the Kitaev model is indeed described by the modular tensor category $\B = \mZ(\Rep(G))$.

In fact, it is widely believed that modular tensor categories can be used to describe not only the topological order of Kitaev models, but also of Levin-Wen models \cite{Levin04}. These models also use a lattice (similar to Fig. \ref{fig:kitaev}), with the modification that the lattice should be trivalent (e.g. the honeycomb lattice). Here, the label on each edge is given by a simple object in a unitary fusion category $\mC$. As shown in Ref. \cite{Levin04}, string operators may also be defined for this model, although it is not as simple to characterize the elementary excitations and anyon fusion. The topological order would be given by the Drinfeld center $\B = \mZ(\mC)$, although one must also at least compute the $\mathcal{S},\mathcal{T}$ matrices using string operators to verify this for each particular case.

In the rest of this chapter, we present an algebraic theory for gapped boundaries for any model whose topological order is given by a doubled theory $\B = \mZ(\mC)$ for some unitary fusion category $\mC$.

\subsection{Lagrangian algebras}
\label{sec:frobenius-algebras}

We will now describe the gapped boundaries in a theory with topological order given by $\B = \mZ(\mC)$. Let us first state a few definitions and theorems that will be crucial for the rest of the paper.

\begin{definition}
Let $\mC$ be a tensor category. A {\it (left) module category} is a category $\M$ with an exact bifunctor $\otimes: \mC \times \M \rightarrow \M$, with functorial associativity and unit isomorphisms $m_{X,Y,M}:(X \otimes Y) \otimes M \rightarrow X \otimes (Y \otimes M)$, $l_M: \one \otimes M \rightarrow \M$ ($\one$ is the tensor unit of $\mC$) for every $X,Y \in \Obj(\mC), M \in \Obj(\M)$ such that the following diagrams commute:
\begin{equation}
\vcenter{\hbox{\includegraphics[width = 0.66\textwidth]{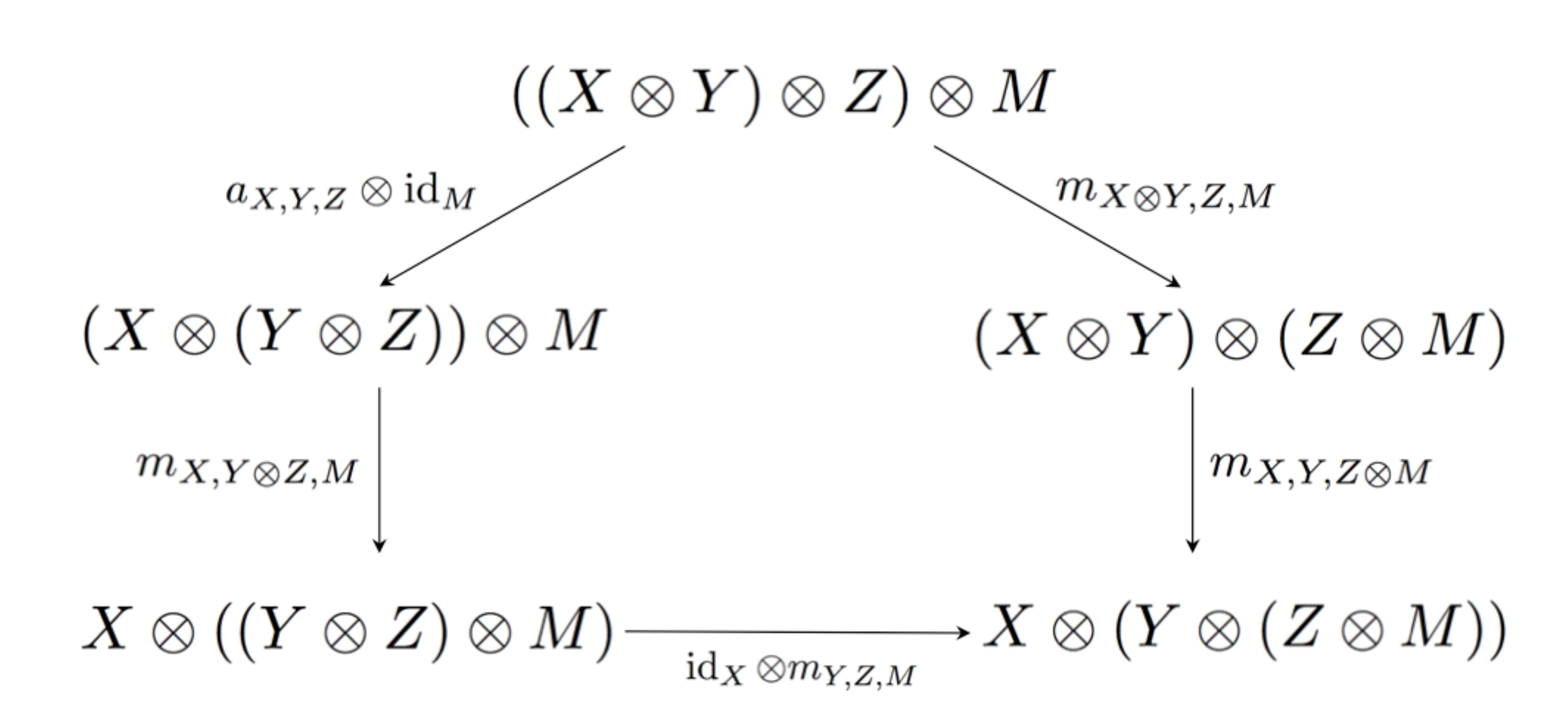}}}
\end{equation}
and
\begin{equation}
\vcenter{\hbox{\includegraphics[width = 0.48\textwidth]{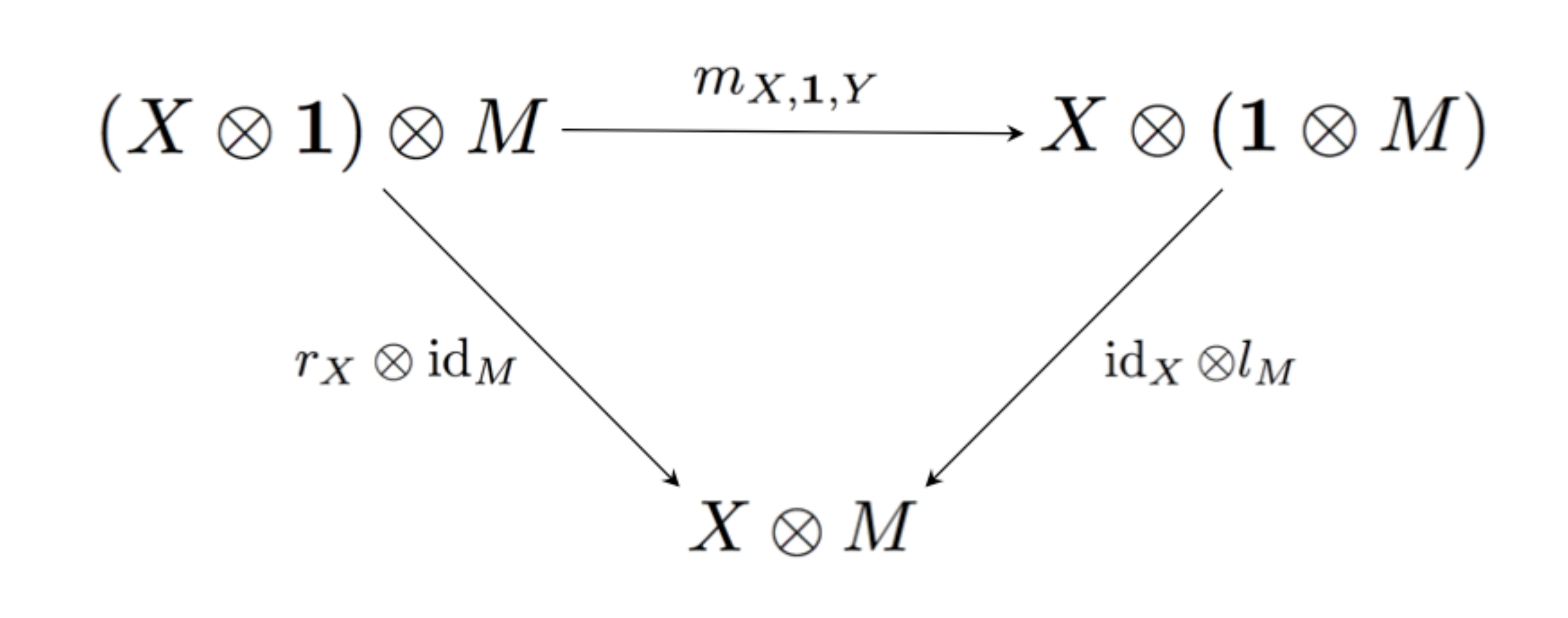}}}
\end{equation}
\noindent
Here, $a_{X,Y,Z}$ and $r_X$ are the functorial associativity and right unit isomorphisms from the monoidal category $\mC$, respectively. Right module categories are defined analogously.

$\M$ is said to be {\it indecomposable} if it is not the direct sum of two nontrivial module categories.
\end{definition}

\begin{theorem}
\label{indecomposable-module-repG}
Let $G$ be some finite group. There exists a one-to-one correspondence between the indecomposable module categories of $\Rep(G)$ and the pairs $(\{K\},\omega)$, where $\{K\}$ is an equivalence class of subgroups $K \subseteq G$ up to conjugation, and $\omega \in H^2(K,\C^\times)$ is a 2-cocycle of a representative $K \in \{K\}$.
\end{theorem}

\begin{proof}
See Ref. \cite{Ostrik03}, Theorem 2.
\end{proof}

\begin{definition}
\label{lagrangian-algebra-def}
A {\it Lagrangian algebra} $\A$ in a modular tensor category $\B$ is an algebra with a multiplication $m: \A \otimes \A \rightarrow \A$ such that:
\begin{enumerate}
\item
$\A$ is {\it commutative}, i.e. $\A \otimes \A \xrightarrow{c_{\A\A}} \A \otimes \A \xrightarrow{m} \A$ equals $\A \otimes \A \xrightarrow{m} \A$, where $c_{\A\A}$ is the braiding in the modular category $\B$.
\item
$\A$ is {\it separable}, i.e. the multiplication morphism $m$ admits a splitting $\mu:\A \rightarrow \A \otimes \A$ a morphism of $(\A,\A)$-bimodules.
\item
$\A$ is {\it connected}, i.e. $\Hom_\B(\one_\B, \A) = \C$, where $\one_\B$ is the tensor unit of $\B$.
\item
The Frobenius-Perron dimension (a.k.a. quantum dimension) of $\A$ is the square root of that of the modular tensor category $\B$,
\begin{equation}
\label{eq:lagrangian-algebra-dim}
\FPdim(\A)^2 = \FPdim(\B).
\end{equation}
\end{enumerate}
\end{definition}

\begin{remark}
We note that an algebra that satisfies conditions (2) and (3) in the above definition is often known in the literature as an {\it et\'ale} algebra.
\end{remark}

As seen in Section \ref{sec:bd-hamiltonian}, in a Kitaev model for the untwisted Dijkgraaf-Witten theory based on group $G$, every subgroup $K \subseteq G$ (up to conjugation) with a cocycle $\omega \in H^2(K,\C^\times)$ determines a distinct gapped boundary of the model (i.e. a unique boundary Hamiltonian). It follows from Theorem \ref{indecomposable-module-repG} that there is an injection from the indecomposable modules of the category $\mZ(\Rep(G))$ to the set of gapped boundaries of the Kitaev model.

By Proposition 4.8 of Ref. \cite{Davydov12}, we may state the following Theorem:

\begin{theorem}
\label{indecomposable-module-lagrangian-algebra}
Let $\mC$ be any fusion category, and let $\B = \mZ(\mC)$. There exists a one-to-one correspondence between the indecomposable modules of $\mC$ and the Lagrangian algebras of $\B$.
\end{theorem}

As a result, gapped boundaries of the Kitaev model may be determined by enumerating all Lagrangian algebras in $\mZ(\Rep(G))$.

In fact, it is well believed \cite{KitaevKong} that in any Levin-Wen model based on unitary fusion category $\mC$, the gapped boundaries are in one-to-one correspondence with the indecomposable modules $\M$ of $\mC$. In this case, by determining all the Lagrangian algebras of the Drinfeld center $\B = \mZ(\mC)$, we can also obtain gapped boundaries of the Levin-Wen model. In what follows, the theory we develop will be applicable to any model where gapped boundaries are given by indecomposable modules of the input fusion category.

\begin{remark}
The above definition of a Lagrangian algebra is the same as a special, symmetric Frobenius algebra, with an additional restriction on the quantum dimension of the algebra. This condition enforces that $\A$ has the maximal quantum dimension possible. Physically, this makes $\A$ into a gapped boundary (or equivalently, a domain wall between the $\B$ and the trivial category $\text{Vec}$), as we have discussed above. A special, symmetric Frobenius algebra with smaller quantum dimension would correspond to a domain wall between $\B$ and another topological phase.

We note that while this chapter deals purely with Lagrangian algebras and gapped boundaries, our work generalizes to the case of domain walls. In fact, a domain wall is mathematically equivalent to a gapped boundary, by using the ``folding'' technique discussed in Refs. \cite{Beigi11} and \cite{KitaevKong}.  
\end{remark}

To find all Lagrangian algebras of a modular tensor category, we will first state the following propositions:

\begin{prop}
\label{bosons}
$\A$ is a commutative algebra in a modular category $\B$ if and only if the object $\A$ decomposes into simple objects as $\A = \oplus_s n_s s$, with $\theta_s = 1$ (i.e. $s$ is bosonic) for all $s$ such that $n_s \neq 0$. 
\end{prop}

\begin{proof}
See Proposition 2.25 in Ref. \cite{Frohlich06}.
\end{proof}

\begin{prop}
\label{separability-prop}
$\A$ is a separable algebra in a unitary fusion category $\B$ if and only if for every $a,b \in \Obj(\B)$, there exists a partial isometry from $\Hom(a,\A) \otimes \Hom(b,\A) \rightarrow \Hom(a \otimes b, \A)$.\footnote{$\B$ is a fusion category, so all hom-spaces in $\B$ have vector space structure. The tensor product of hom-spaces is just the usual tensor product for vector spaces.}
\end{prop}

\begin{proof}
Fix $a,b \in \Obj(\B)$. Define a map $M$ from $\Hom(a,\A) \otimes \Hom(b,\A)$ to $\Hom(a \otimes b, \A)$ as follows:

By definition of a tensor category, there exists an injective map $\gamma: \Hom(a,\A) \otimes \Hom(b,\A) \rightarrow \Hom(a \otimes b, \A \otimes \A)$. Suppose we are given two morphisms $f \in \Hom(a,\A)$, $g \in \Hom(b,\A)$. Then $M(f \otimes g) = m \circ \gamma(f \otimes g)$ is a morphism in the hom-space $\Hom(a \otimes b, \A \otimes \A)$. We now show that this map $M$ is injective.

Suppose $M(f \otimes g) = 0$. Since $\B$ is a unitary fusion category, $M(f \otimes g) = 0$ if and only if the following trace is equal to 0: (note here that the pictures are read bottom-up)

\begin{equation}
\label{eq:M-trace}
\vcenter{\hbox{\includegraphics[width = 0.2\textwidth]{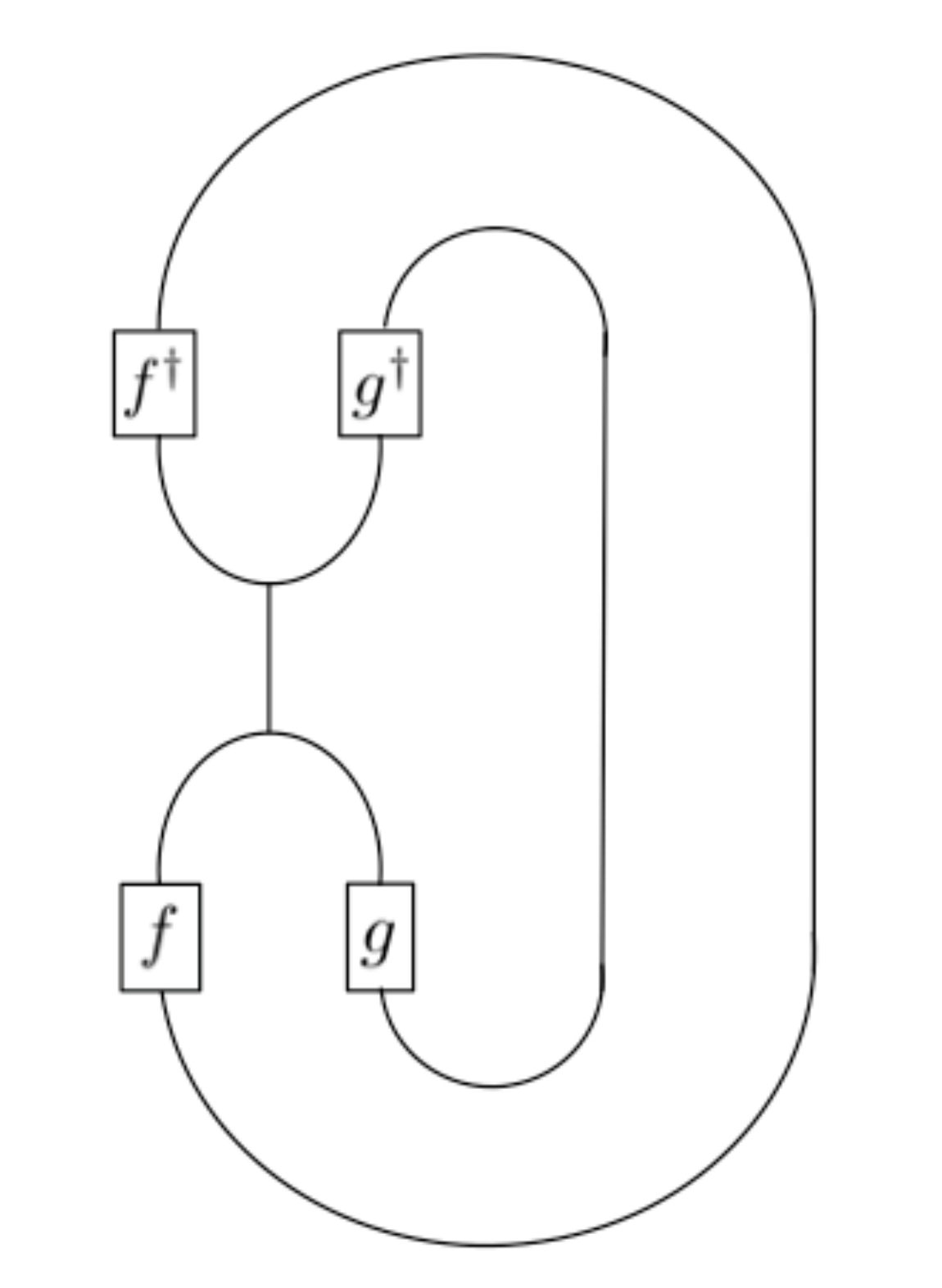}}} = 0.
\end{equation}

By definition of an $(A,A)$-bimodule, the condition (2) in Definition \ref{lagrangian-algebra-def} is equivalent to the following two conditions \cite{Muger12}

\begin{equation}
\label{eq:separability}
\vcenter{\hbox{\includegraphics[width = 0.55\textwidth]{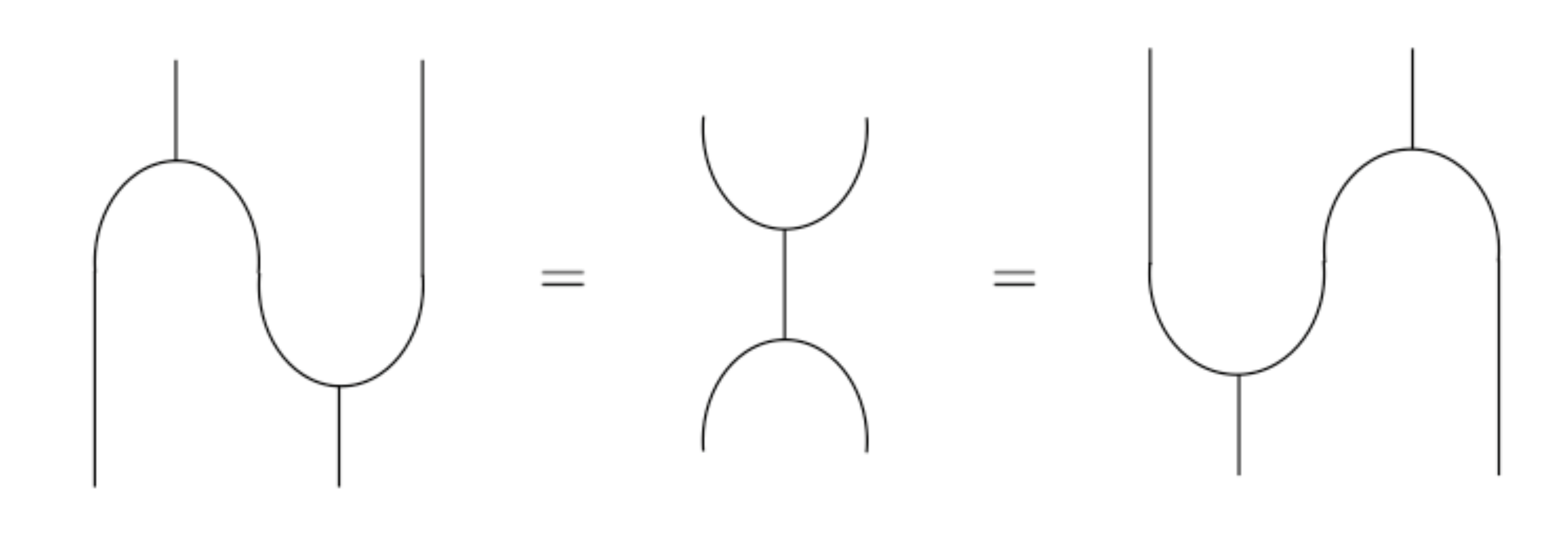}}}
\end{equation}

\noindent
and

\begin{equation}
\label{eq:separability-2}
\vcenter{\hbox{\includegraphics[width = 0.25\textwidth]{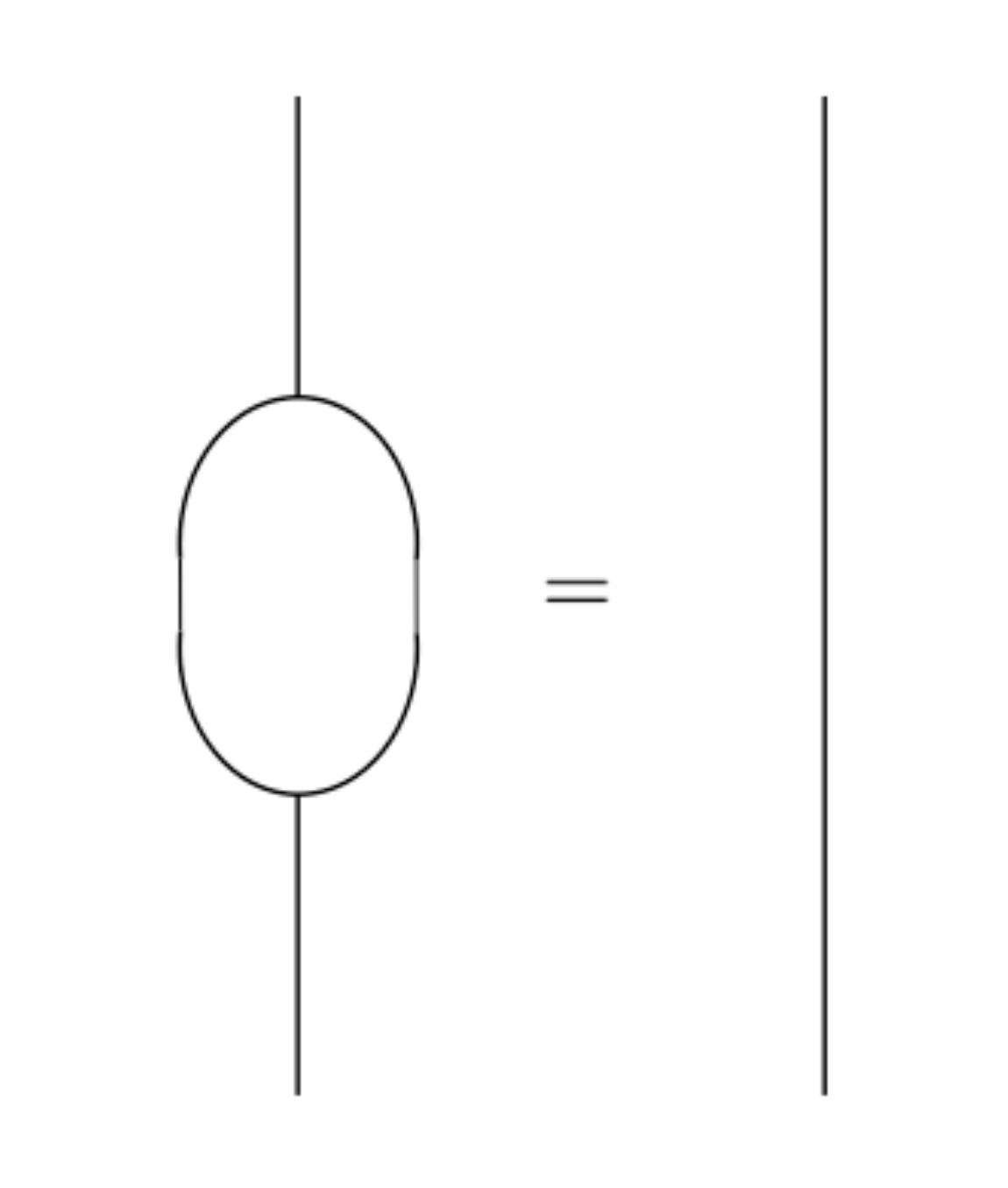}}}
\end{equation}

Hence, Eq. (\ref{eq:M-trace}) holds if and only if

\begin{equation}
\label{eq:M-trace-2}
\vcenter{\hbox{\includegraphics[width = 0.22\textwidth]{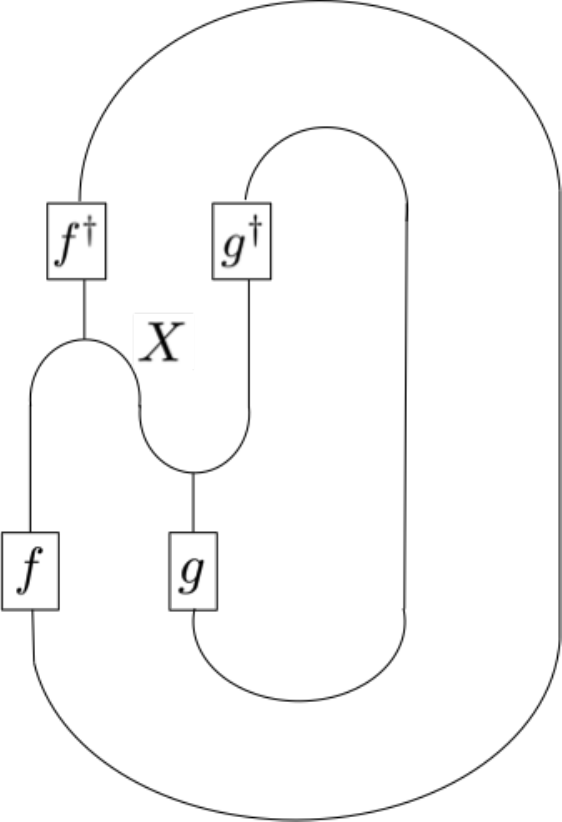}}} = 0.
\end{equation}

Since it is not possible to have tadpoles in such a diagram, the middle line connecting the $f$ and $g$ loops (labeled $X$) in Eq. (\ref{eq:M-trace-2}) must be equal to vacuum. We hence have the following picture:

\begin{equation}
\label{eq:M-trace-3}
\vcenter{\hbox{\includegraphics[width = 0.25\textwidth]{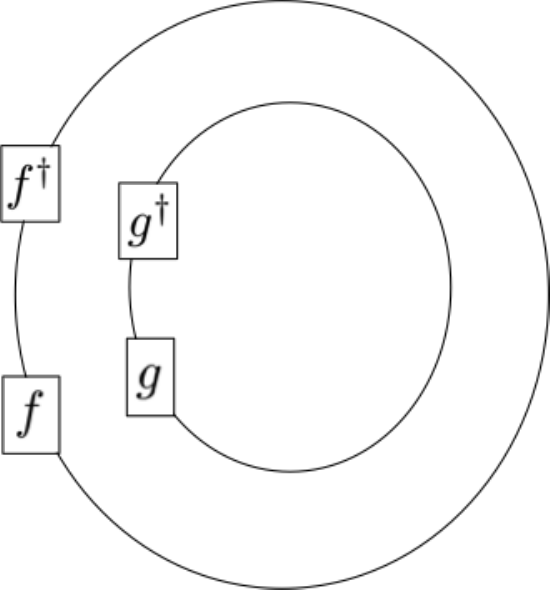}}} = 0.
\end{equation}

The left hand side of Eq. (\ref{eq:M-trace-3}) is precisely given by $\Tr(f)\Tr(g)$. Since $\B$ is a unitary fusion category, this equation holds if and only if $f = 0$ or $g = 0$, i.e. $f \otimes g = 0$.

Finally, $M$ is a partial isometry if and only if Eq. (\ref{eq:separability-2}) holds, which completes the forward direction of the proof.

Note that all steps in this proof were reversible, so that both directions of the Proposition hold.
\end{proof}

\begin{corollary}
A commutative connected algebra $\A = \oplus_s n_s s$ with $\FPdim(\A)^2 = \FPdim(\B)$ is a Lagrangian algebra in the unitary modular category $\B$ if and only if the following inequality holds for all $a,b \in \Obj(\B)$:

\begin{equation}
\label{eq:lagrangian-algebra-inequality}
n_a n_b \leq \sum_c N_{ab}^c n_c
\end{equation}

\noindent
where $N_{ab}^c$ are the coefficients given by the fusion rules of $\B$.
\end{corollary}

\begin{remark}
We would like to note that the algebra object $\A$ is not enough to uniquely identify the gapped boundary.  Let $G$ be the order$-64$ class $3$ group in Sec. IIIA of \cite{Davydov14}.  The standard Cardy Lagrangian algebra of $\mathcal{Z}(G\oplus G)$ has another different Lagrangian structure given by a soft braided auto-equivalence of $\mathcal{Z}(G)$.
\end{remark}

\subsection{Ground state degeneracy and qudit basis encoding}
\label{sec:algebraic-gsd}

In Section \ref{sec:hamiltonian-gsd-condensation}, we used ribbon operators to present the ground state degeneracy of the Kitaev model with gapped boundaries. In this section, we will present this same degeneracy using the algebraic model we have developed in this chapter.

Theorems \ref{condensation-products} and \ref{inverse-condensation-products} of Section \ref{sec:bd-excitations} described how an anyon in the bulk can condense to the boundary. Given any anyon $a$ in the bulk, the condensation space of $a$ to vacuum a boundary given by the Lagrangian algebra $\A$ can be modeled precisely by the hom-space $\Hom(a,\A)$. Specifically, as discussed in Section \ref{sec:bd-excitations}, the number of condensation channels in condensing to vacuum is equivalent to the number of times the particle $a$ appears in the decomposition of $\A$ into simple objects (obtained using Theorem \ref{inverse-condensation-products} on the boundary vacuum particle); since $\B$ is a unitary fusion category, this is exactly the dimension of the hom-space. As in previous sections, this is expected to hold for theories with a topological order $\B = \mZ(\mC)$, not just $\B = \mZ(\Rep(G))$.

More generally, we also can describe the ground state of the TQFT with boundaries using hom-spaces. Suppose we have a system with topological order given by the modular tensor category $\B$, and $n$ gapped boundaries given by Lagrangian algebras $\A_1, \A_2, ... \A_n$, as shown in Fig. \ref{fig:algebraic-gsd-n}. The outside boundary is taken to have total charge vacuum; one may alternatively view this as $n$ holes on a sphere.

\begin{figure}
\centering
\includegraphics[width = 0.65\textwidth]{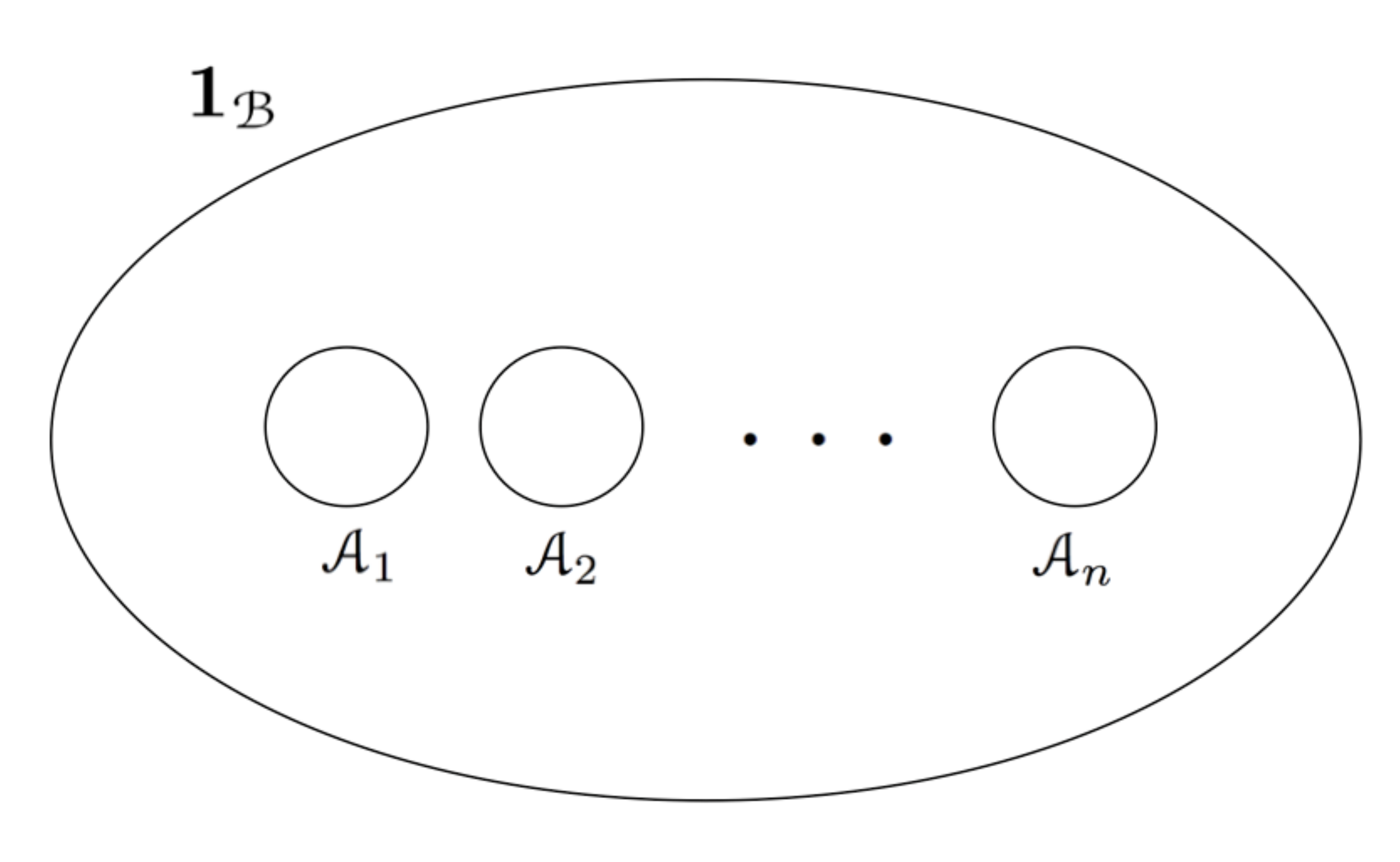}
\caption{Algebraic picture of the TQFT with $n$ gapped boundaries (holes). The outside charge is take to be vacuum, and each boundary type is given by a Lagrangian algebra $\A_i$.}
\label{fig:algebraic-gsd-n}
\end{figure}

The ground state degeneracy of this model is given by the number of ways we can create a pair of anyons from vacuum, admissibly split them into $n$ anyons, and condense all $n$ of them to vacuum onto the boundary, as shown in Fig. \ref{fig:algebraic-gsd-n-2}. Hence, the ground state of this model is given by

\begin{equation}
\label{eq:ground-state-algebraic}
\text{G.S.} = \Hom(\one_\B, \A_1 \otimes \A_2 \otimes ... \otimes \A_n)
\end{equation}

\noindent
where $1_\B$ is the tensor unit of $\B$ and represents the vacuum particle.

\begin{figure}
\centering
\includegraphics[width = 0.6\textwidth]{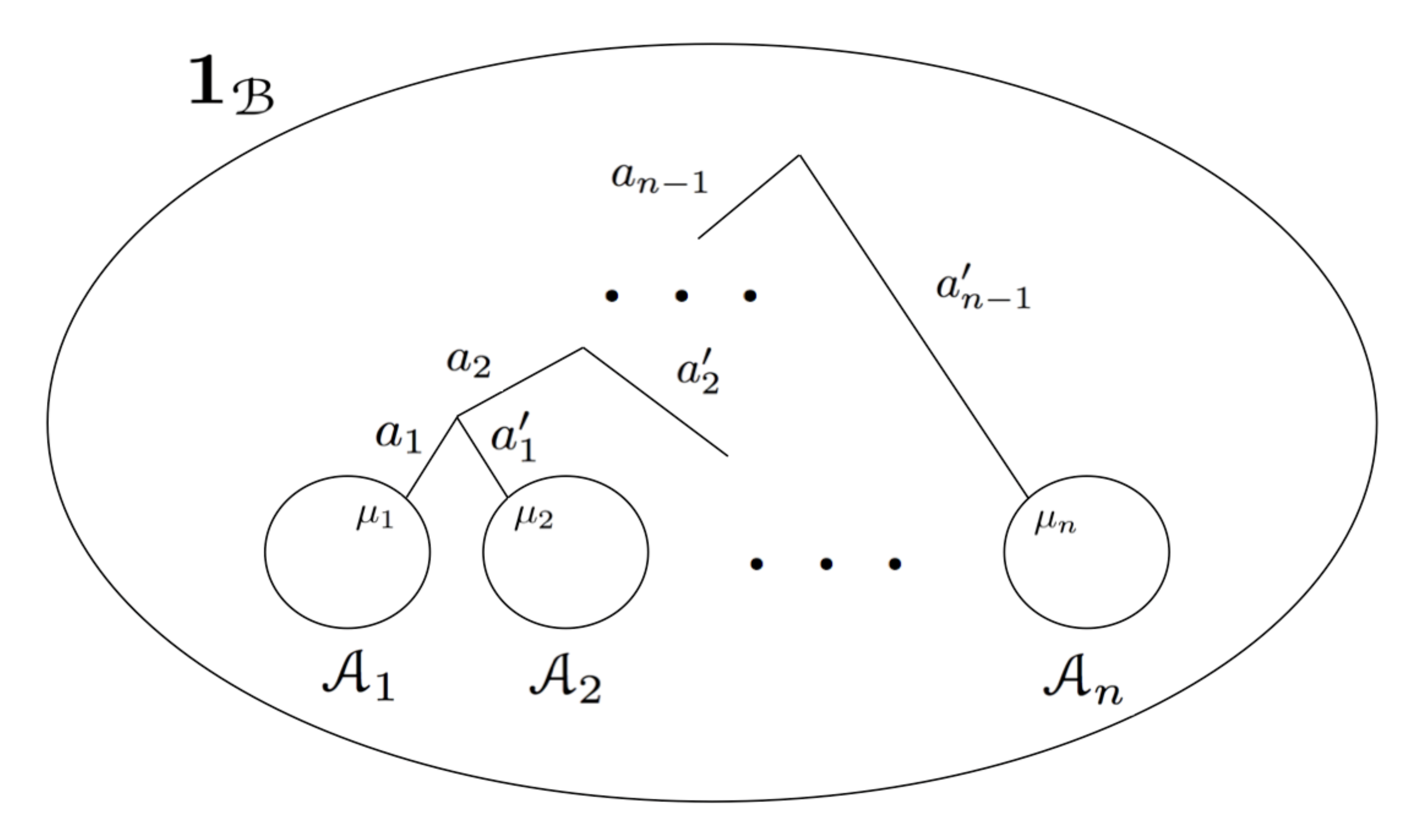}
\caption{Ground state of the model presented in Fig. \ref{fig:algebraic-gsd-n}. We assume here that all edges are directed to point downward. The $a_i$, $a_i'$ are simple bulk anyons such that the splittings are admissible, and the final products may all condense to vacuum on the respective boundaries. The $\mu_i$ correspond to multiplicities in condensing to vacuum on the boundary (see for instance Theorem \ref{condensation-products} of Section \ref{sec:bd-excitations}, and Section \ref{sec:multiple-condensation-channels}). The set of all such tuples $(a_1, a_1', ... a_{n-1}', \mu_1, ... \mu_n)$ form a basis for the ground state.}
\label{fig:algebraic-gsd-n-2}
\end{figure}

We encode our qudit in the special case where $n = 2$, $\A_1 = \oplus n_{1s} s$, $\A_2 = \oplus n_{2s} s$. This is the generalization of a qubit encoding chosen by Fowler et al. in Ref. \cite{Fowler12} for gapped boundaries of the $\Z_2$ toric code. In this case, the ground state is a Hilbert space with degeneracy given by 

\begin{equation}
\label{eq:gsd-formula}
\text{G.S.D.} = d = \sum_s n_{1s} n_{2s}.
\end{equation}

The basis for this Hilbert space is given by the action of ``tunneling operators'' $W_a(\gamma)_{\mu\nu}$, where $\gamma$ is a fixed ribbon joining the two gapped boundaries, $a$ is any particle such that $n_{1a} > 0$ and $n_{2\overbar{a}} > 0$ ($\overbar{a}$ is the dual particle of $a$), and $\mu,\nu$ denote the condensation channels for the condensation of $a$ and $\overbar{a}$, respectively. Graphically, this basis is given by the dumbbell picture shown in Fig. \ref{fig:gapped-boundary-basis}.

\begin{figure}
\centering
\includegraphics[width = 0.58\textwidth]{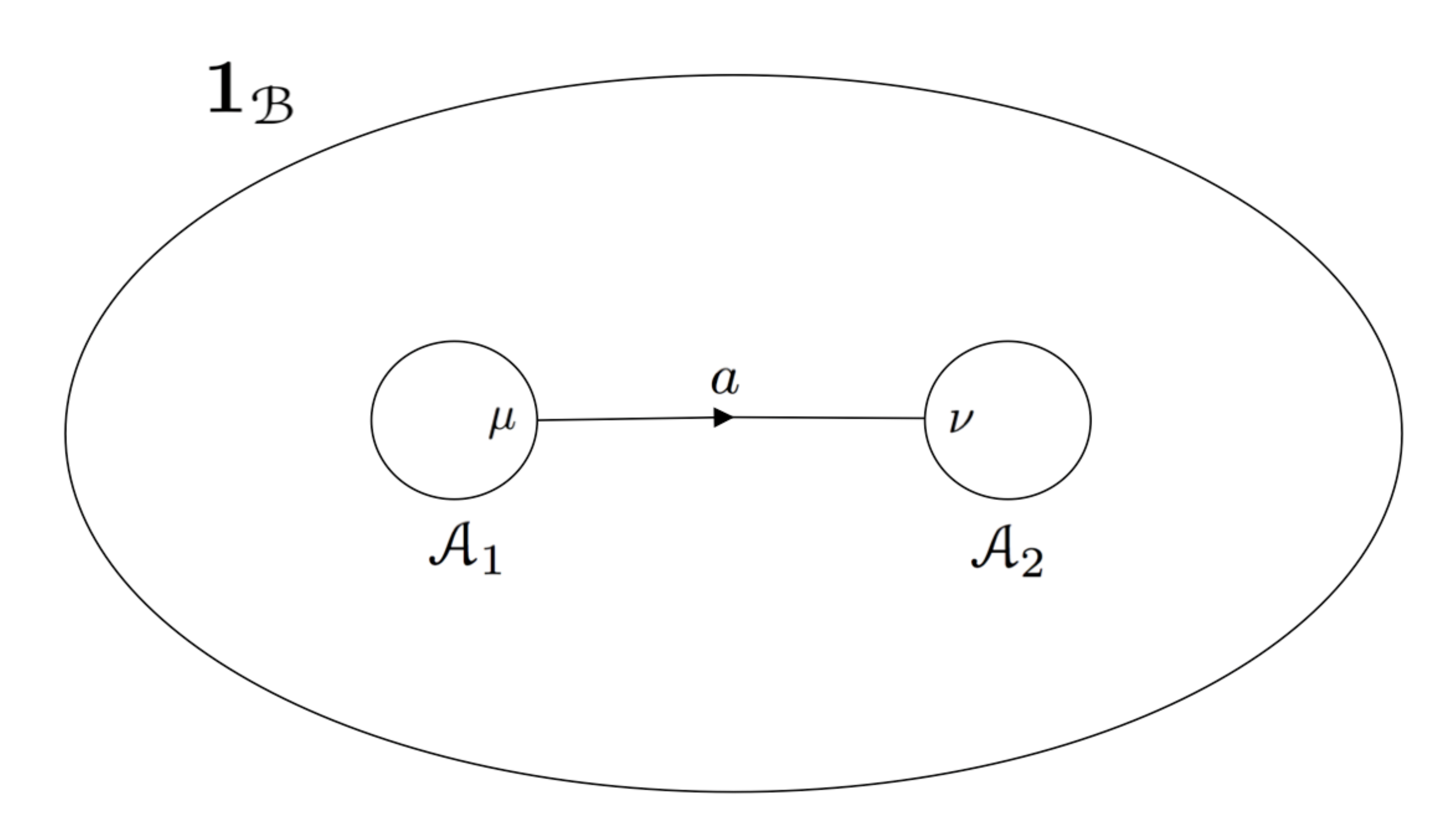}
\caption{Qudit encoding using gapped boundaries. $a$ is a bulk anyon such that the antiparticle $\overbar{a}$ condenses to vacuum on the $\A_1$ boundary, and $a$ condenses to vacuum on the $\A_2$ boundary. $\mu$, $\nu$ are the multiplicities corresponding to these condensations, respectively (see for instance Theorem \ref{condensation-products} of Section \ref{sec:bd-excitations}). For the case of Kitaev models, this is equivalent to the encoding presented in Fig. \ref{fig:degeneracy}.}
\label{fig:gapped-boundary-basis}
\end{figure}

This basis and the operators $W_a(\gamma)_{\mu\nu}$ will play a crucial role in defining the topologically protected operations in Chapter \ref{sec:operations}.

\subsection{Condensation and elementary excitations on the boundary}
\label{sec:algebraic-condensation}

In Section \ref{sec:bd-excitations}, we saw that the elementary excitations on a boundary of subgroup $K$ in the Kitaev model based on group $G$ are given by the irreducible representations of the group-theoretical quasi-Hopf algebra corresponding to $G,K$. In that section, we presented a method using the Hamiltonian to obtain the products of the condensation procedure and its reverse. In this section, we will present this procedure categorically, and show that these two views are actually the same in the case of group models. Let us first make the following definition \cite{MacLane}:

\begin{definition}
\label{quotient-cat-def}
Let $\B$ be a monoidal category, and let $\A$ be any object in $\B$. The {\it pre-quotient category} $\B/\A$ is the category such that:
\begin{enumerate}
\item
The objects of $\B/\A$ are the same as the objects of $\B$.
\item
The morphisms of $\B/\A$ are given by
\begin{equation}
\Hom_{\B/\A}(X,Y) = \Hom_{\B}(X,\A \otimes Y).
\end{equation}
\end{enumerate}
We denote by $\widetilde{F}:\B \rightarrow \B/\A$ the central functor that sends each object of $\B$ to the corresponding quotient object, and by $I: \B/\A \rightarrow \B$ its right adjoint.
\end{definition}

In our case, we would like to consider condensation of an anyon onto a gapped boundary. Here, the elementary excitations on the boundary seem to be given by the simple objects in the pre-quotient category $\widetilde{\mathcal{Q}} = \B / \A$. However, one minor problem with the above pre-quotient category definition is that the resulting category $\widetilde{\mathcal{Q}}$ may not be semisimple. As a result, the following definition/proposition \cite{Muger03} is needed to fully describe the condensation products:

\begin{definition}
\label{IC-def}
Let $\B$ be a strict braided tensor category and let $\A$ be a strongly separable Frobenius algebra in $\B$. Let $\widetilde{\mQ} = \B/\A$ be the pre-quotient category formed via Definition \ref{quotient-cat-def}. Let us form the canonical idempotent completion ${\mQ}$ of $\widetilde{\mQ}$ as follows:
\begin{enumerate}
\item
The objects of ${\mQ}$ are given by pairs $(X,p)$, where $X \in \Obj \widetilde{\mQ}$ and $p = p^2 \in \End_{\widetilde{\mQ}}(X)$.
\item
The morphisms of ${\mQ}$ are given by
\begin{equation}
\Hom_{{\mQ}}((X,p),(Y,q)) = \{ f \in \Hom_{\widetilde{\mQ}}(X,Y): f \circ p = p \circ f \text{ and } f \circ q = q \circ f\}.
\end{equation}
\end{enumerate}
Then, by Proposition 2.15 of Ref. \cite{Muger03}, the new category ${\mQ}$ is semisimple and the desired quotient.
\end{definition}

In general, condensation to a domain wall between two topological phases with topological orders $\B,\mD$ is mathematically described as the procedure

\begin{equation}
\label{eq:condensation-quotient-IC}
F: \mZ(\mC) = \B \xrightarrow{\text{quotient}} \B/\A = \widetilde{\mQ} \xrightarrow{\text{I.C.}}  {\mQ} = \msC \sqcup \msD.
\end{equation}

\noindent
(Here, I.C. is the idempotent completion). After condensation to a domain wall, all excitations in $\msC$ become {\it confined} to the domain wall, and all excitations in $\msD$ are {\it deconfined} and can enter the phase $\mD$. Physically, an excitation is said to be confined if there is an energy cost to move the excitation growing linearly with the distance of movement.

In the case of gapped boundaries, $\msD = \text{Vec}$ is vacuum, so there are no nontrivial deconfined excitations, and all excitations on the boundary are confined. In the case of Kitaev models for Dijkgraaf-Witten theories, the condensation procedure described in Eq. (\ref{eq:condensation-quotient-IC}) is then exactly the condensation procedure of Theorem \ref{condensation-products} of Section \ref{sec:bd-excitations}. Furthermore, the right adjoint $I$ of this procedure, which is given by the composition of the adjoint of the idempotent completion and the adjoint $\widetilde{I}$ of the quotient functor, is exactly the inverse condensation procedure of Theorem \ref{inverse-condensation-products}.

In Ref. \cite{KitaevKong}, Kitaev and Kong have claimed that in the Levin-Wen model based on input fusion category $\mC$, the excitations on a boundary given by the indecomposable module $\M$ are given by objects in the fusion category $\Fun_\mC(\M,\M)$. By Ref. \cite{Davydov12}, we know that this category is equivalent to the category ${\mQ}$ obtained through the procedure (\ref{eq:condensation-quotient-IC}). In what follows, we will prove the claim of Ref. \cite{KitaevKong} for the case of Kitaev models for Dijkgraaf-Witten theories.

In Section \ref{sec:bd-ribbon-operators}, we claimed that the elementary excitations on a boundary of type $K$ in a Kitaev model based on group $G$ are given by pairs $(T,R)$, where $T = K r_T K \in K\backslash G / K$ is a double coset in $G$, and $R$ is an irreducible representation of the group-theoretical quasi-Hopf algebra $\mZ = Z(G,1,K,1)$. We would now like to present the relationship between $Z(G,1,K,1)$ and the group-theoretical category $\mC(G,1,K,1)$ of Ref. \cite{Etingof05}, which is defined as follows:

\begin{definition}
Let $\text{Vec}_G^{\omega}$ be the category of finite-dimensional $G$-graded vector spaces with associativity $\omega$, where $G$ is a finite group and $\omega \in H^3(G, \C^\times)$. Let $K \in G$ be a subgroup of $G$ and $\psi \in H^2(K, \C^\times)$ be a 2-cocycle of $K$ such that $d \psi = \omega |_H$. Let $\text{Vec}_G^{\omega}(H)$ be the subcategory of $\text{Vec}_G^{\omega}$ of objects graded by $H$. The twisted group algebra $A = \C_\psi [H]$ is then an associative algebra in $\text{Vec}_G^{\omega}(H)$. The {\it group-theoretical category} $\mC(G,1,K,1)$ is defined as the category of $(A,A)$-bimodules in $\text{Vec}_G^{\omega}$. In particular, $\M$ is a fusion category with tensor product $\otimes_A$ and unit object $A$.
\end{definition}

The goal of this section is now to establish the equivalence between $\mC(G,1,K,1)$ and the representation category of $Z(G,1,K,1)$ as fusion categories.

To begin, we state the following theorem:

\begin{theorem}
\label{Z-irreps}
Let $G$ be a finite group, and let $K \subseteq G$ be a subgroup. The irreducible representations of $Z(G,1,K,1)$ are given by pairs $(T,R)$, where $T = K r_T K \in K\backslash G /K$ is a double coset and $R$ is an irreducible representation of the subgroup $K^{r_T} = K \cap r_T K r_T^{-1}$ of $K$.
\end{theorem}

\begin{proof}
See Refs. \cite{Zhu01} and \cite{Schauenburg02}.
\end{proof}

By Ref. \cite{Gelaki07}, have the following theorem:

\begin{theorem}
The pairs $(T,R)$, as described in Theorem \ref{Z-irreps}, are in one-to-one correspondence with the simple objects in the group-theoretical category $\mC(G,1,K,1)$.
\end{theorem}

The above theorem implies that the elementary excitations on the boundary in the group-theoretical case are indeed given by the simple objects in the fusion category $\Fun_\mC(\M,\M)$.

Finally, Refs. \cite{Zhu01,Schauenburg02} show the equivalence of $\Rep(Z(G,1,K,1))$ and $\mC(G,1,K,1)$ as fusion categories:

\begin{theorem}
The representation category of the group-theoretical quasi-Hopf algebra $Z(G,1,K,1)$ (or equivalently, the representation category of the coquasi-Hopf algebra $Y(G,1,K,1)$) is equivalent as a fusion category to the group-theoretical category $\mC(G,1,K,1)$.
\end{theorem}

\begin{proof}
See Refs. \cite{Zhu01} and \cite{Schauenburg02}.
\end{proof}

This shows that the ``bordered topological order'' introduced in Section \ref{sec:bd-excitations} is indeed given by the fusion category $\Fun_\mC(\M,\M)$.

\subsection{$M$ symbols}
\label{sec:m-symbols}

\subsubsection{The $M$-3$j$ symbol}
\label{sec:m3j}

Thus far, we have discussed the condensation of a single bulk anyon $a$ into a gapped boundary $\A$ in a system with topological order, and the tunneling of a single anyon from one gapped boundary to another. As we will see in Chapter \ref{sec:operations}, is equally important to consider the case where multiple bulk anyons all condense into the boundary.

\begin{figure}
\centering
\includegraphics[width = 0.65\textwidth]{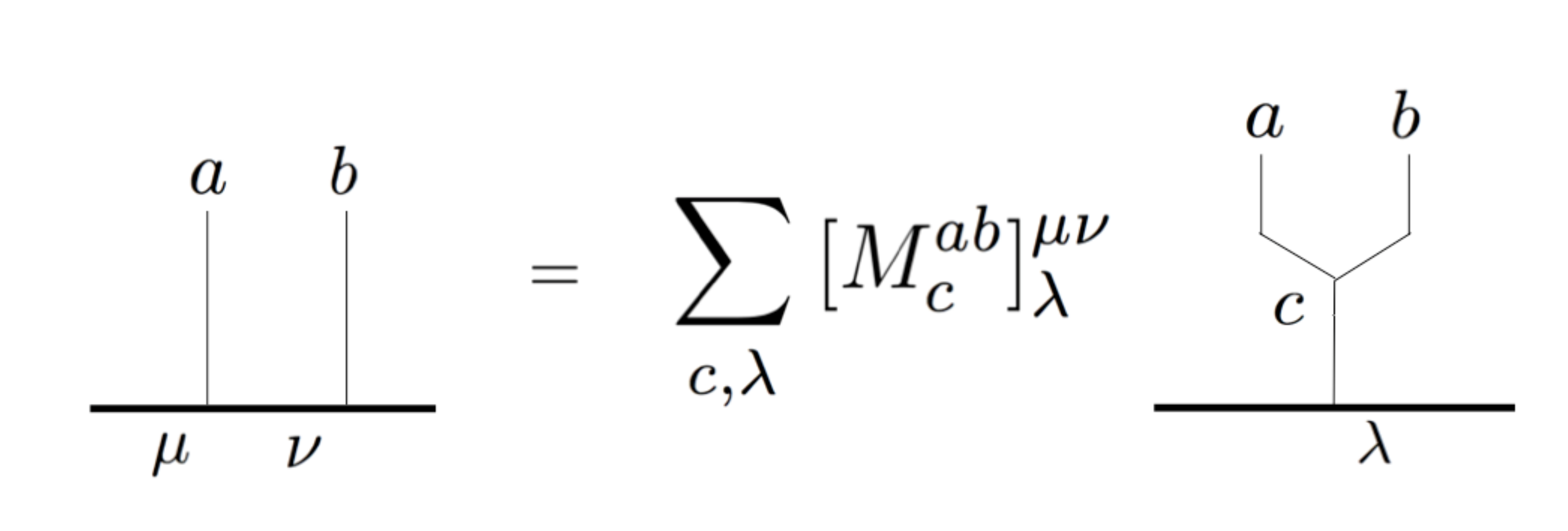}
\caption{Definition of the $M$-3$j$ symbol.}
\label{fig:m-3j}
\end{figure}

In a system with topological order, when three anyons fuse in the bulk, $F-6j$ symbols $F^{abc}_{d;ef}$ describe the associativity in the order of fusion. These 6$j$ symbols must satisfy certain pentagon and hexagon relations, corresponding to the pentagon and hexagon commutative diagrams for a modular tensor category. Similar associativity and braiding rules exist for the $M$ symbols, as we shall discuss below.

Let us first consider a relatively simple case, when the bulk anyons all condense to vacuum on the boundary. In Proposition \ref{separability-prop}, we showed that for any two anyons $a,b$ that can condense to vacuum on the boundary $\A$, there exists an injection $M$ from $\Hom(a,\A) \otimes \Hom(b,\A)$ to $\Hom(a \otimes b, \A)$. Physically, the $M$ operator corresponds to fusing the anyons $a,b$ in the bulk first, and then condensing to the boundary. The action of the $M$ operator is shown in Fig. \ref{fig:m-3j}.

In this case, each $M$ symbol has three topological charge indices, given by two original bulk anyons $a,b$ that condense to the boundary, and a third bulk anyon $c$ that results from the fusion of $a,b$ in the bulk. Furthermore, there will be local indices $\mu,\nu,\lambda$ corresponding to the condensation channels when $a,b,c$ condense to the boundary, respectively.

These $M$-3$j$ symbols are quite similar to the $\theta$-3$j$ symbols in a fusion category. If we start with three anyons in the bulk, the following pentagon diagram must commute:

\begin{equation}
\label{eq:m-3j-pentagon-fig}
\vcenter{\hbox{\includegraphics[width = 0.79\textwidth]{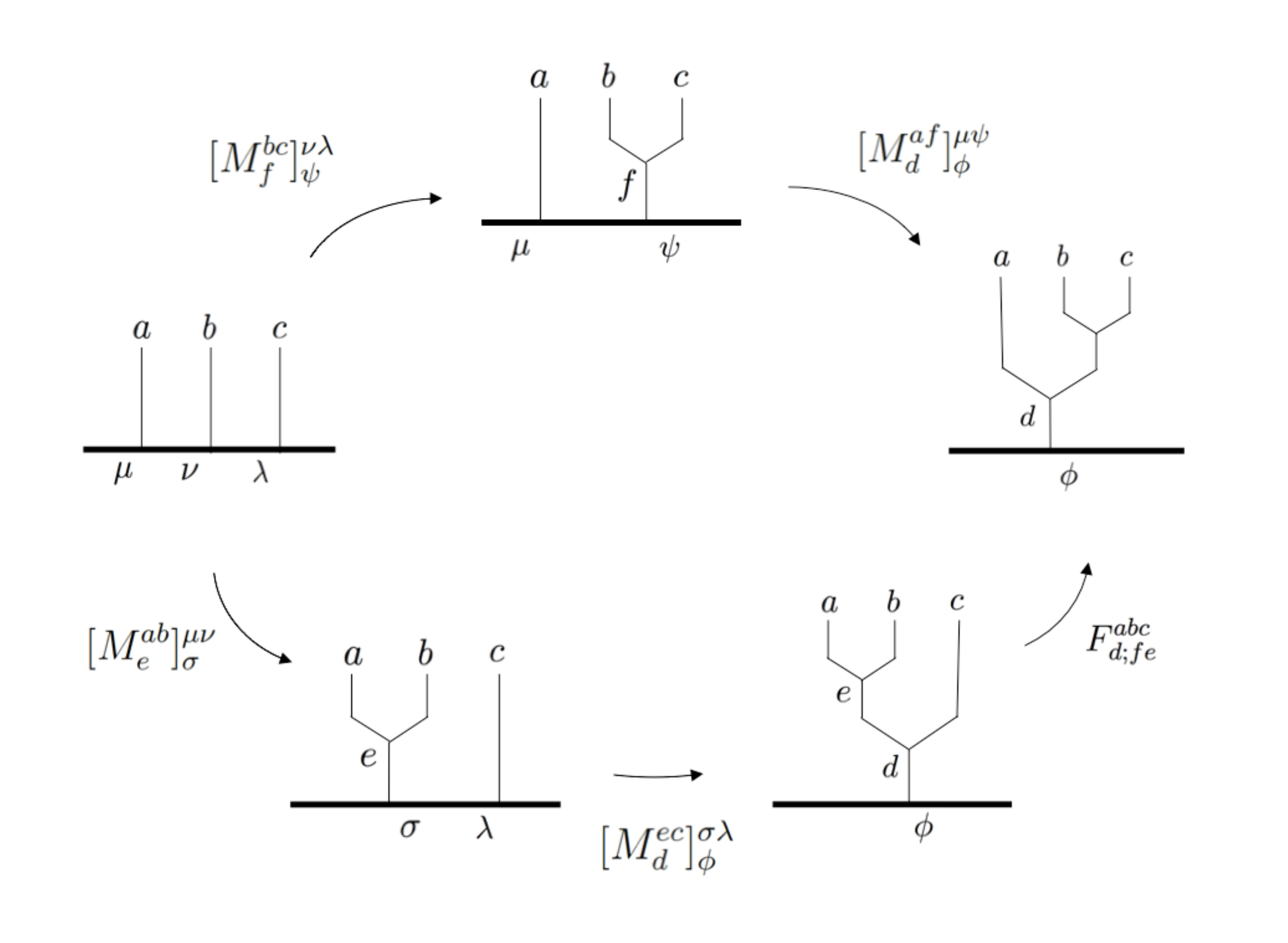}}}
\end{equation}

Algebraically, we can write Eq. \ref{eq:m-3j-pentagon-fig} as\footnote{In this and all associativity/braid relations that follow, we have assumed for simplicity of presentation that the anyon model has no fusion multiplicities. This is true in all of our examples, but the generalization is obvious.}:

\begin{equation}
\label{eq:m-3j-pentagon}
\sum_{e,\sigma} [M^{ab}_{e}]^{\mu\nu}_{\sigma} [M^{ec}_{d}]^{\sigma\lambda}_{\phi}F^{abc}_{d;fe} = \sum_\psi [M^{bc}_{f}]^{\nu\lambda}_{\psi} [M^{af}_{d}]^{\mu\psi}_{\phi}
\end{equation}

We note that the $M$ symbols will have gauge degrees of freedom, originating from the choice of basis for the condensation channels of each particle $a$. Specifically, we can define a unitary transformation $\Gamma^a_{\mu\nu}$ on the condensation space $V_a$: $\widetilde{\ket{a;\mu}} = \Gamma^a_{\mu\nu} \ket{a;\mu}$. These transformations yield new $M$ symbols, which are related to the original ones by the relation

\begin{equation}
\label{eq:m-3j-gauge-freedom}
[\tilde{M}^{ab}_c]^{\mu\nu}_\lambda =
\sum_{\mu',\nu',\lambda'} \Gamma^{a}_{\mu\mu'} \Gamma^{b}_{\nu\nu'}
[M^{ab}_{c}]^{\mu'\nu'}_{\lambda'}
[\Gamma^c]^{-1}_{\lambda\lambda'}
\end{equation}

The $M$ symbols will also be affected by gauge transformations of the bulk fusion space in the case of bulk fusion multiplicities.

Furthermore, when two bulk anyons $a,b$ condense to vacuum on the boundary, by the commutativity of the Frobenius algebra $\A$, it does not matter what order they condense in. Diagrammatically, this corresponds to

\begin{equation}
\label{eq:m-3j-braid-fig}
\vcenter{\hbox{\includegraphics[width = 0.5\textwidth]{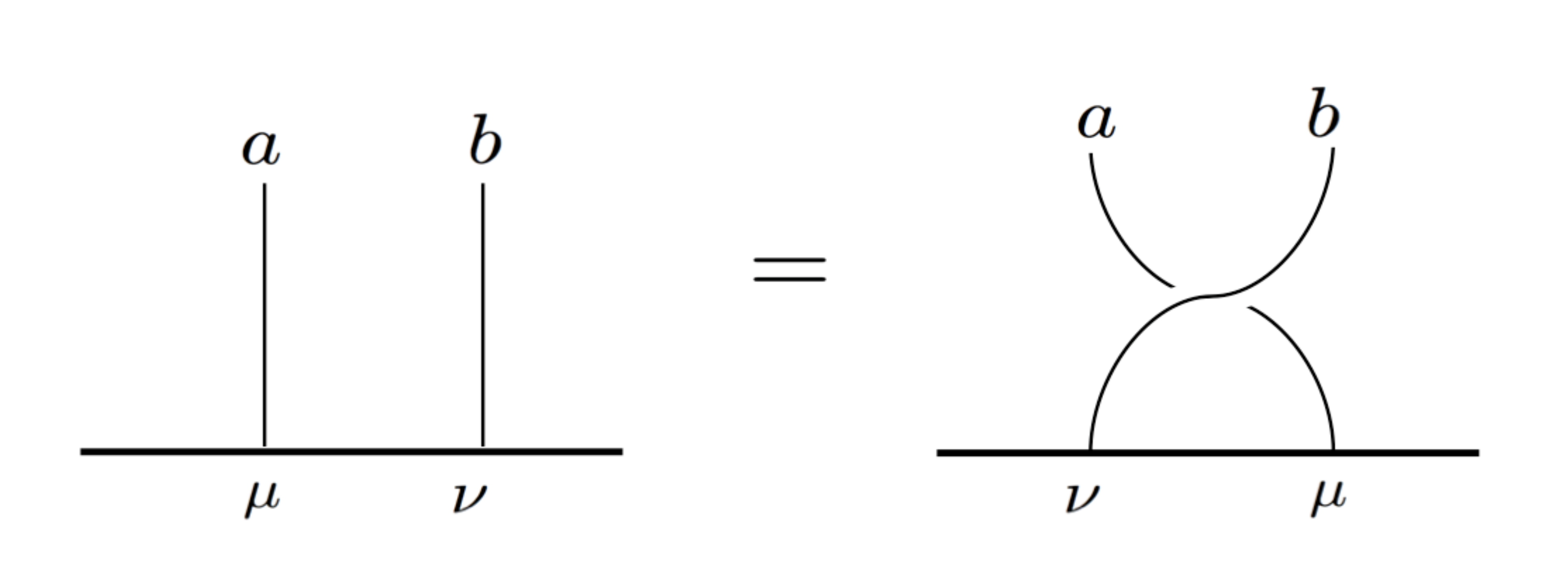}}}
\end{equation}

\noindent
This gives the following equation:

\begin{equation}
\label{eq:m-3j-braid}
\sum_c [M^{ba}_{c}]^{\nu\mu}_{\lambda} R^{ab}_c = \sum_c [M^{ab}_{c}]^{\mu\nu}_{\lambda}
\end{equation}

\noindent
Here, the sum is taken over all $c$ such that $\Hom(c,\A) \neq 0$.

Finally, we define some normalization conditions on these $M$-3$j$ symbols:

\begin{equation}
\label{eq:m-3j-normalization-1}
[M^{a\overbar{a}}_1]^{\mu\nu} = \frac{\delta_{\mu\nu}}{\sqrt{\Dim(\A)}}
\end{equation}

\begin{equation}
\label{eq:m-3j-normalization-2}
[M^{1a}_a]^\mu_\nu = [M^{a1}_a]^\mu_\nu = \delta_{\mu\nu}
\end{equation}

These normalizations are chosen so that $M$ becomes a partial isometry when we have the following basis vectors for the vector spaces $\Hom(a,\A)$, $\Hom(b,\A)$ and $\Hom(a \otimes b, \A)$:

\begin{equation}
\label{eq:m-3j-normalization-fig-1}
\vcenter{\hbox{\includegraphics[width = 0.58\textwidth]{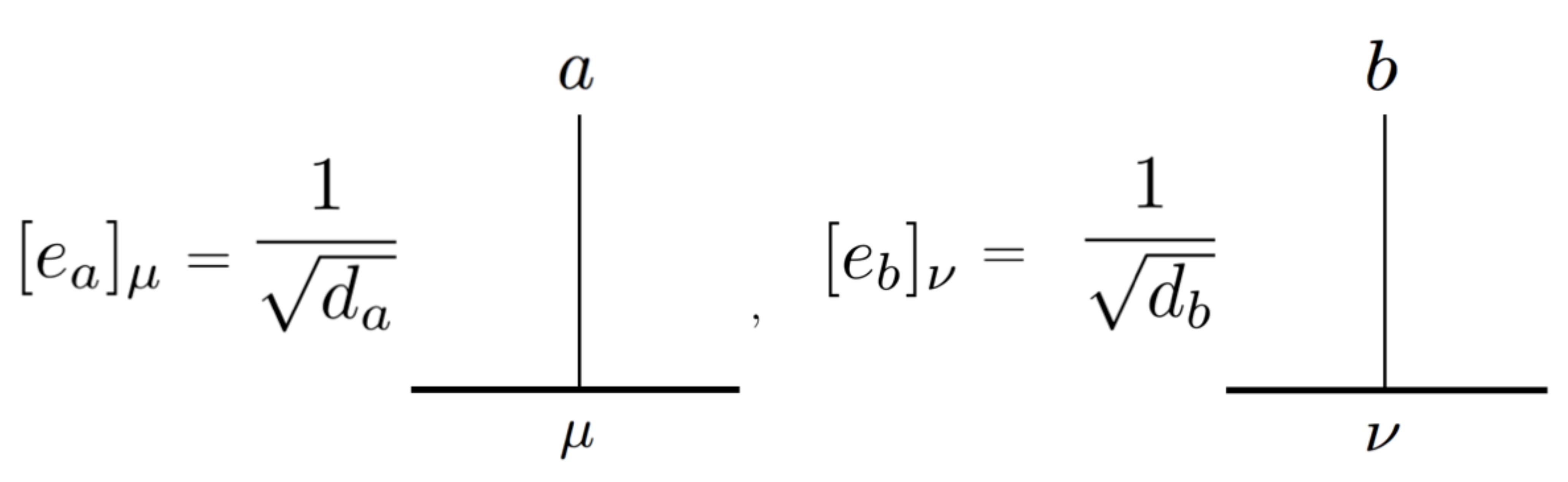}}}
\end{equation}

\begin{equation}
\label{eq:m-3j-normalization-fig-2}
\vcenter{\hbox{\includegraphics[width = 0.4\textwidth]{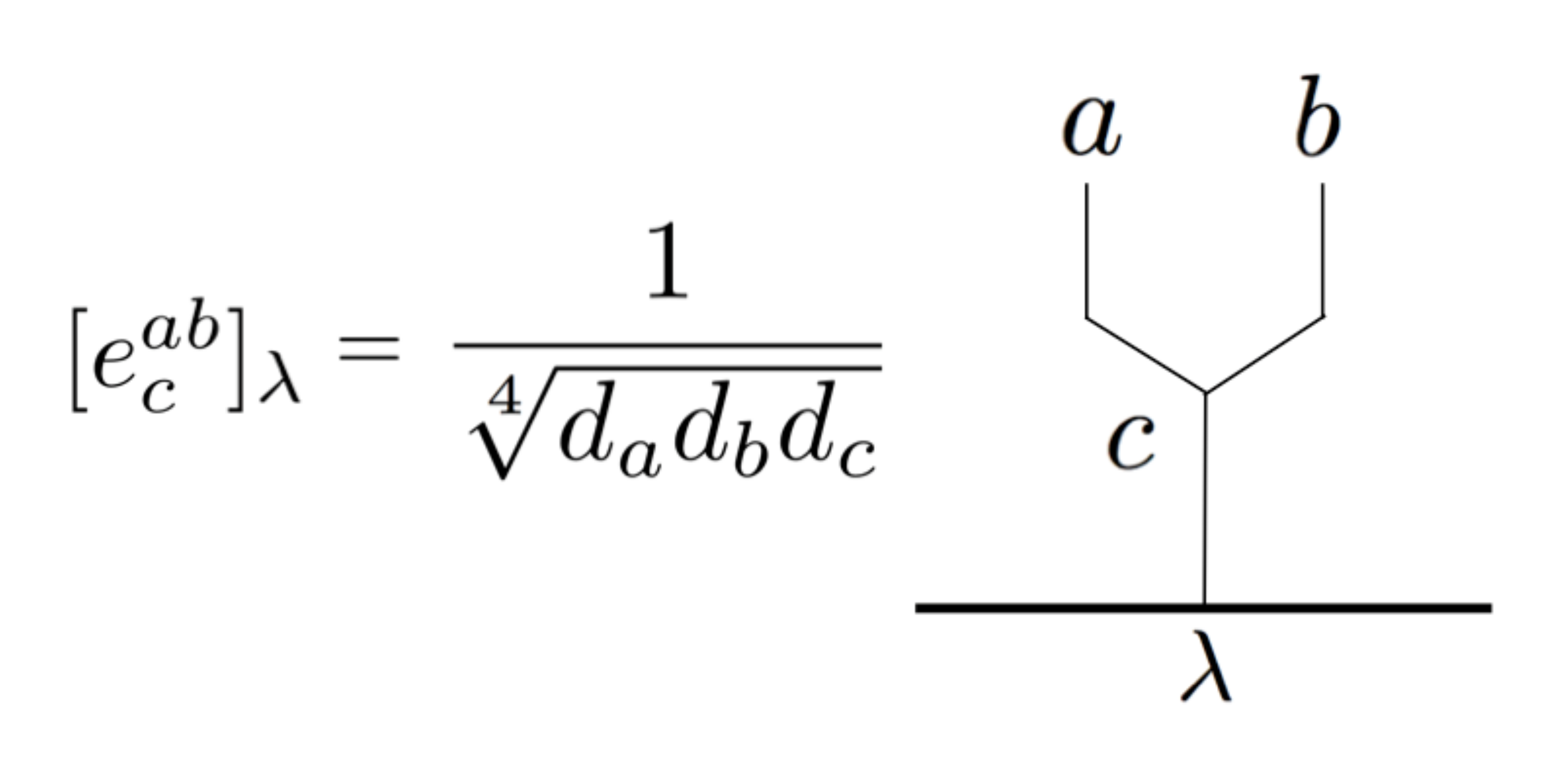}}}
\end{equation}

\noindent
The basis vectors for $\Hom(a,\A) \otimes \Hom(b,\A)$ are simply $[e_a]_\mu \otimes [e_b]_\nu$. These basis vectors are chosen because the following traces evaluate to 1, to provide orthonormal bases:

\begin{equation}
\vcenter{\hbox{\includegraphics[width = 0.28\textwidth]{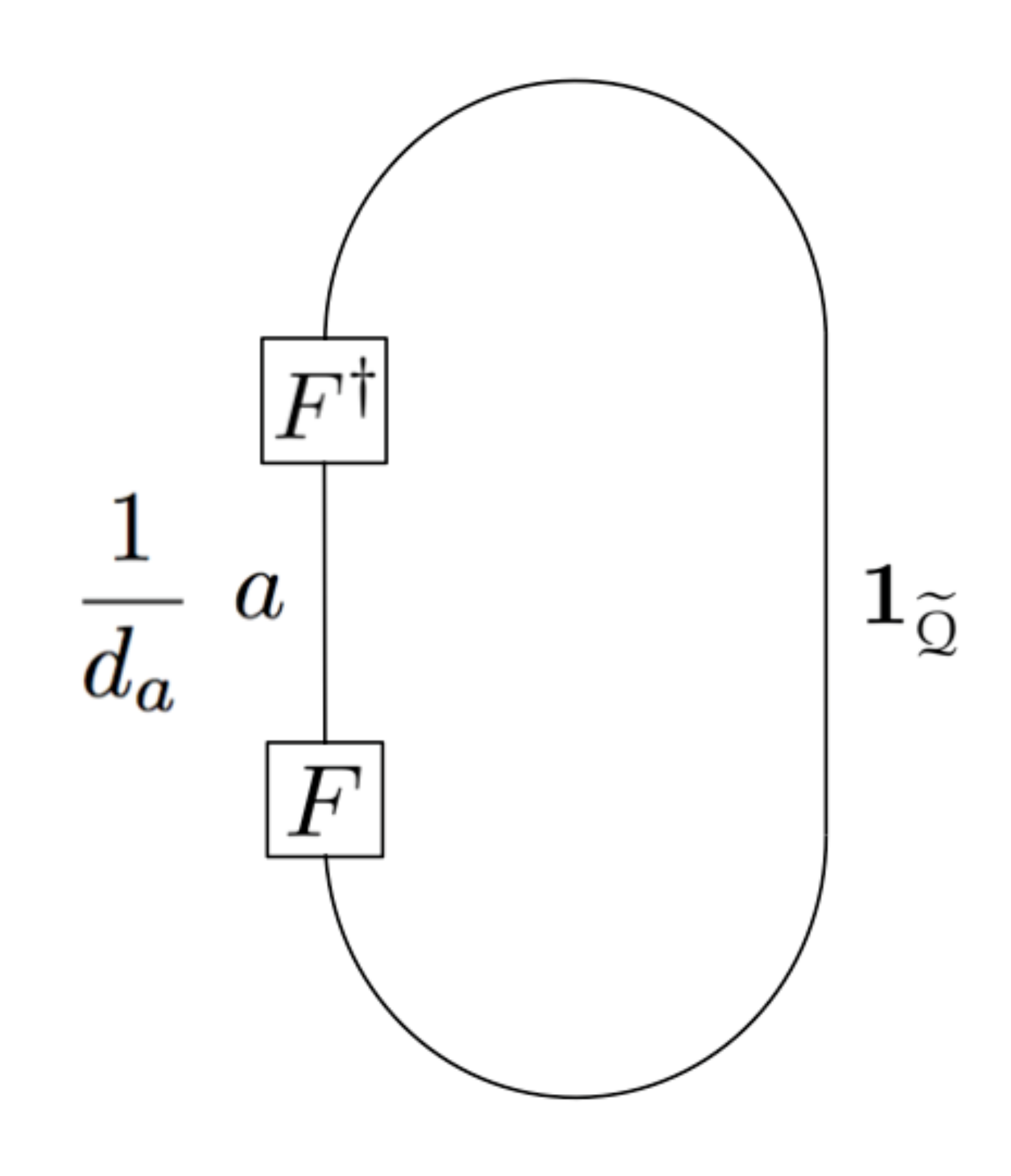}}} = 1,
\qquad
\vcenter{\hbox{\includegraphics[width = 0.35\textwidth]{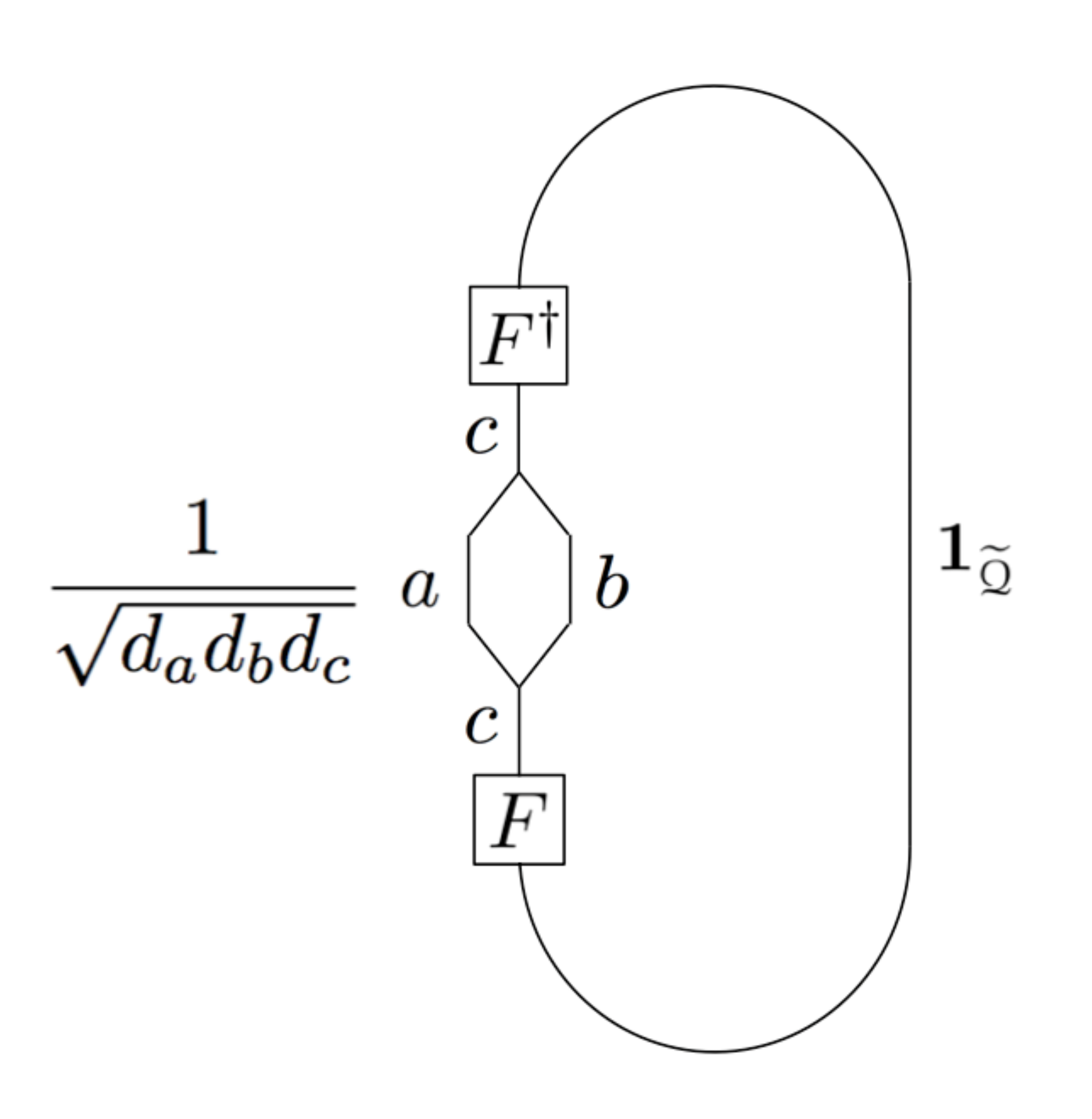}}} = 1
\end{equation}

\subsubsection{The $M$-6$j$ symbol}
\label{sec:m6j}

In the above discussion, we have considered only a special case, where all bulk anyons condense to vacuum on the boundary. More generally, we may consider a case where the bulk anyons become excitations on the boundary, as in Section \ref{sec:bd-excitations}. Here, we will define a $M$-6$j$ symbol, with six topological charge indices, whose action is shown in Fig. \ref{fig:m-6j}. We note that our $M-6j$ symbol is a symmetric version of the vertex lifting coefficients introduced in Ref. \cite{Eliens13}.

\begin{figure}
\centering
\includegraphics[width = 0.65\textwidth]{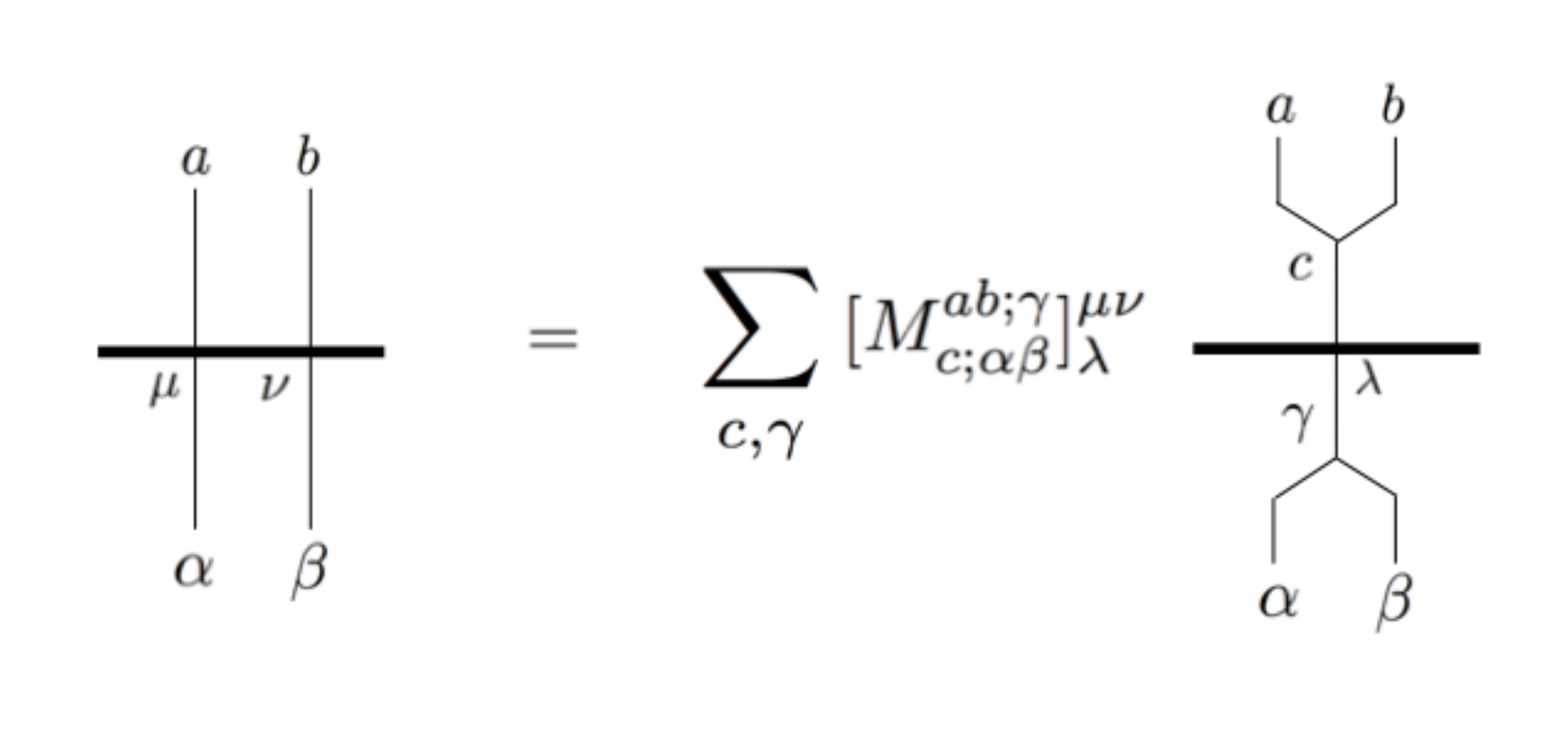}
\caption{Definition of the $M$-6$j$ symbol}
\label{fig:m-6j}
\end{figure}

The $M$-6$j$ symbol is indexed by bulk anyons $a,b,c$, and the boundary excitations $\alpha,\beta,\gamma$ they condense to. As before, we will also have condensation channel labels $\mu,\nu,\lambda$ for the multiplicity corresponding to the dimension of $\Hom(a,I(\alpha))$, etc.

As in the case of $M$-3$j$ symbols, these symbols must also satisfy a pentagon associativity relation. However, this relation will also depend on the $F-6j$ symbols of the fusion category $\Fun_\mC(\M, \M)$. The associativity is hence given by the following commuting pentagon:

\begin{equation}
\label{eq:m-6j-pentagon-fig}
\vcenter{\hbox{\includegraphics[width = 0.79\textwidth]{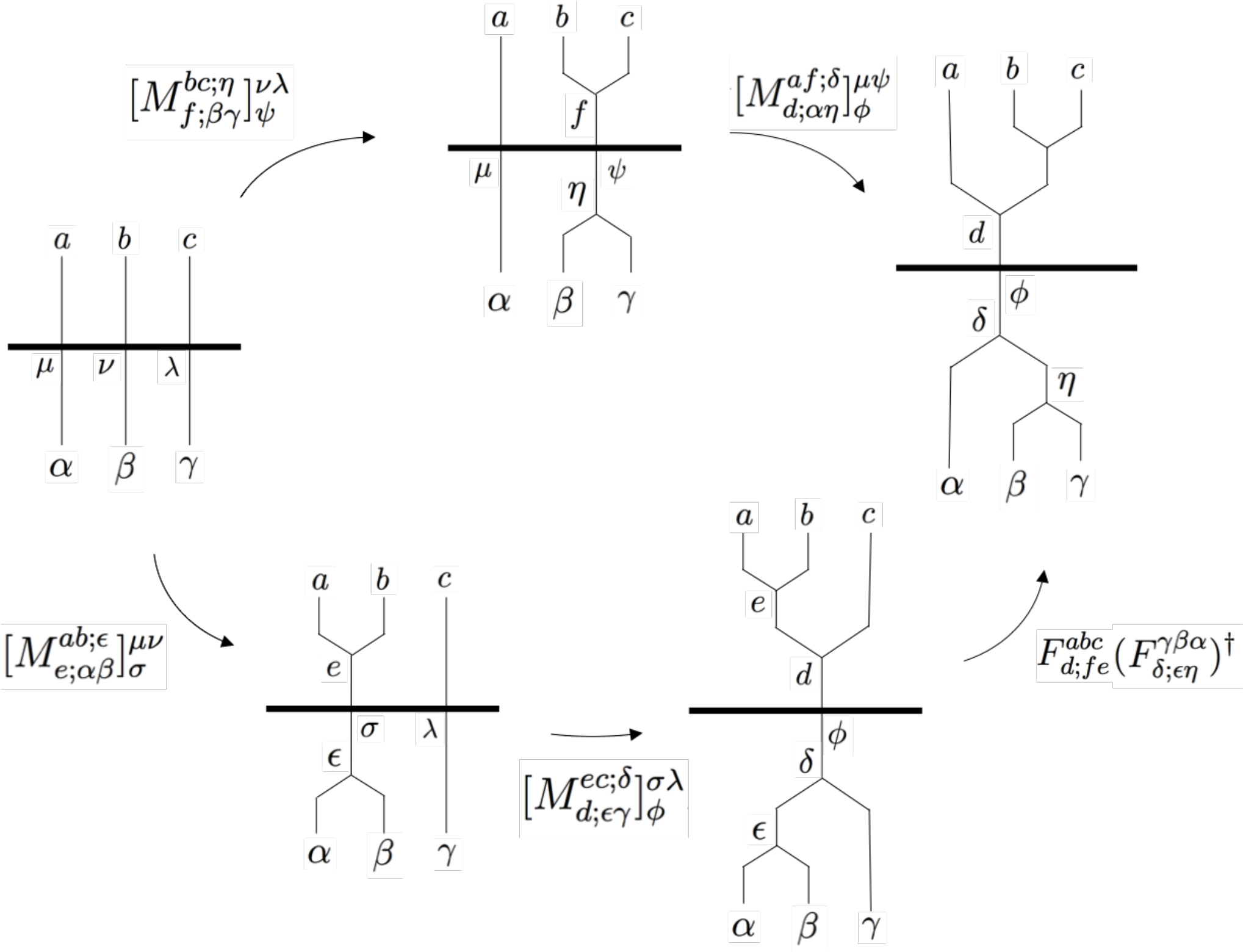}}}
\end{equation}

This is equivalent to the equation

\begin{equation}
\label{eq:m-6j-pentagon}
\sum_{e,\sigma,\epsilon} [M^{ab;\epsilon}_{e;\alpha\beta}]^{\mu\nu}_{\sigma} [M^{ec;\delta}_{d;\epsilon\gamma}]^{\sigma\lambda}_{\phi}F^{abc}_{d;fe} (F^{\gamma\beta\alpha}_{\delta;\epsilon\eta})^\dagger = \sum_\psi [M^{bc;\eta}_{f;\beta\gamma}]^{\nu\lambda}_{\psi} [M^{af;\delta}_{d;\alpha\eta}]^{\mu\psi}_{\phi}
\end{equation}

Similarly, the braiding relation is now given by the diagram

\begin{equation}
\label{eq:m-6j-braid-fig}
\vcenter{\hbox{\includegraphics[width = 0.5\textwidth]{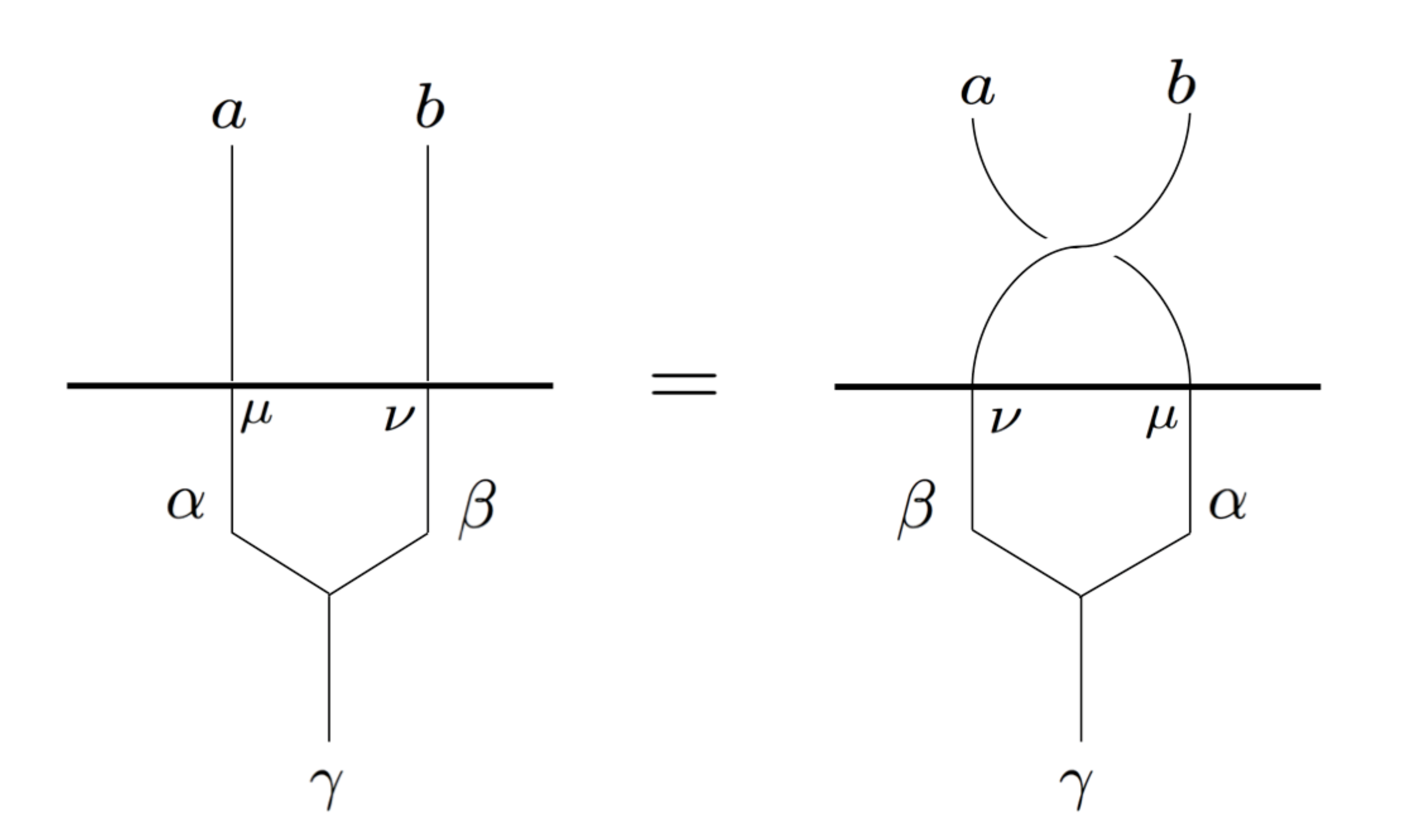}}}
\end{equation}

\noindent
which is equivalent to the equation

\begin{equation}
\label{eq:m-6j-braid}
\sum_c [M^{ba;\gamma}_{c;\beta\alpha}]^{\nu\mu}_{\lambda} R^{ab}_c = \sum_c [M^{ab;\gamma}_{c;\alpha\beta}]^{\mu\nu}_{\lambda}.
\end{equation}

\noindent
The sum is taken over all simple objects $c$ such that $\Hom(c,I(\gamma))\neq 0$, where $I$ is the right adjoint of the condensation procedure $F$.

Finally, we enforce $M$ to be a partial isometry with respect to the following choices of basis for $\Hom(a,I(\alpha)) \otimes \Hom(b,I(\beta))$ and $\Hom(a \otimes b, I(\alpha \otimes \beta))$:

\begin{equation}
\label{eq:m-6j-normalization-fig-1}
\vcenter{\hbox{\includegraphics[width = 0.59\textwidth]{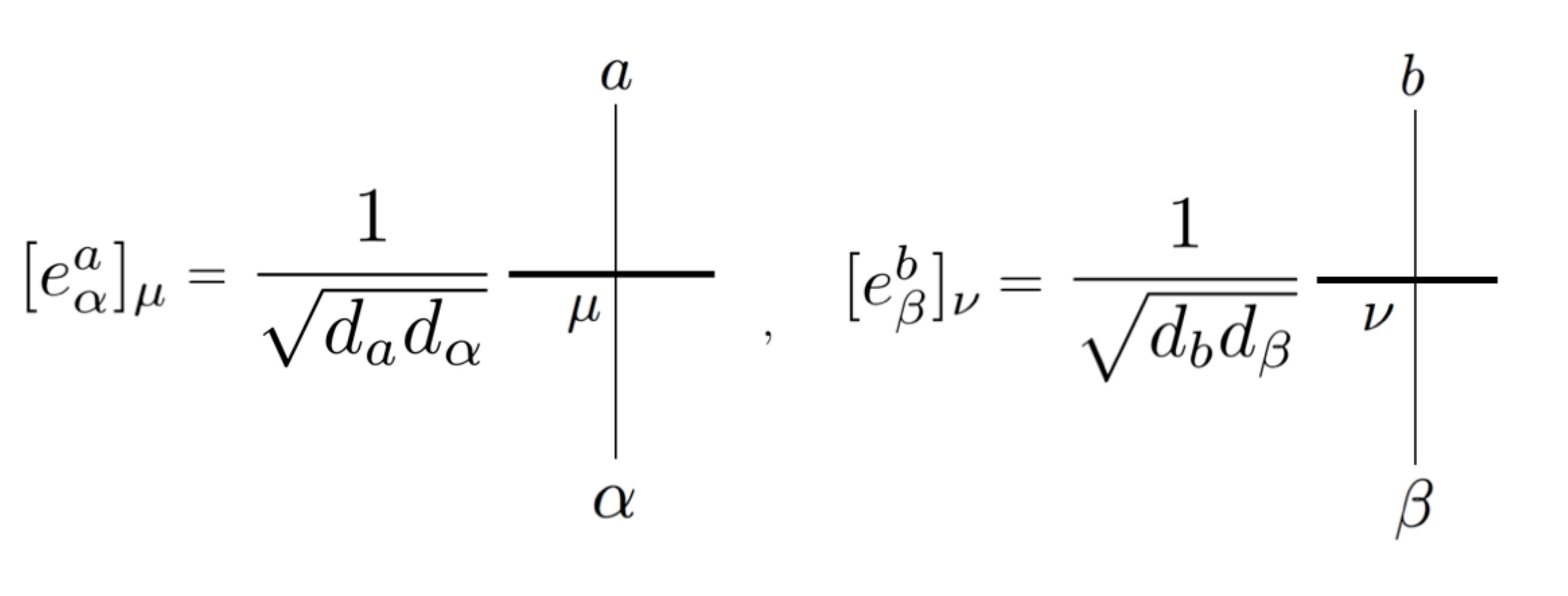}}}
\end{equation}

\begin{equation}
\label{eq:m-6j-normalization-fig-2}
\vcenter{\hbox{\includegraphics[width = 0.5\textwidth]{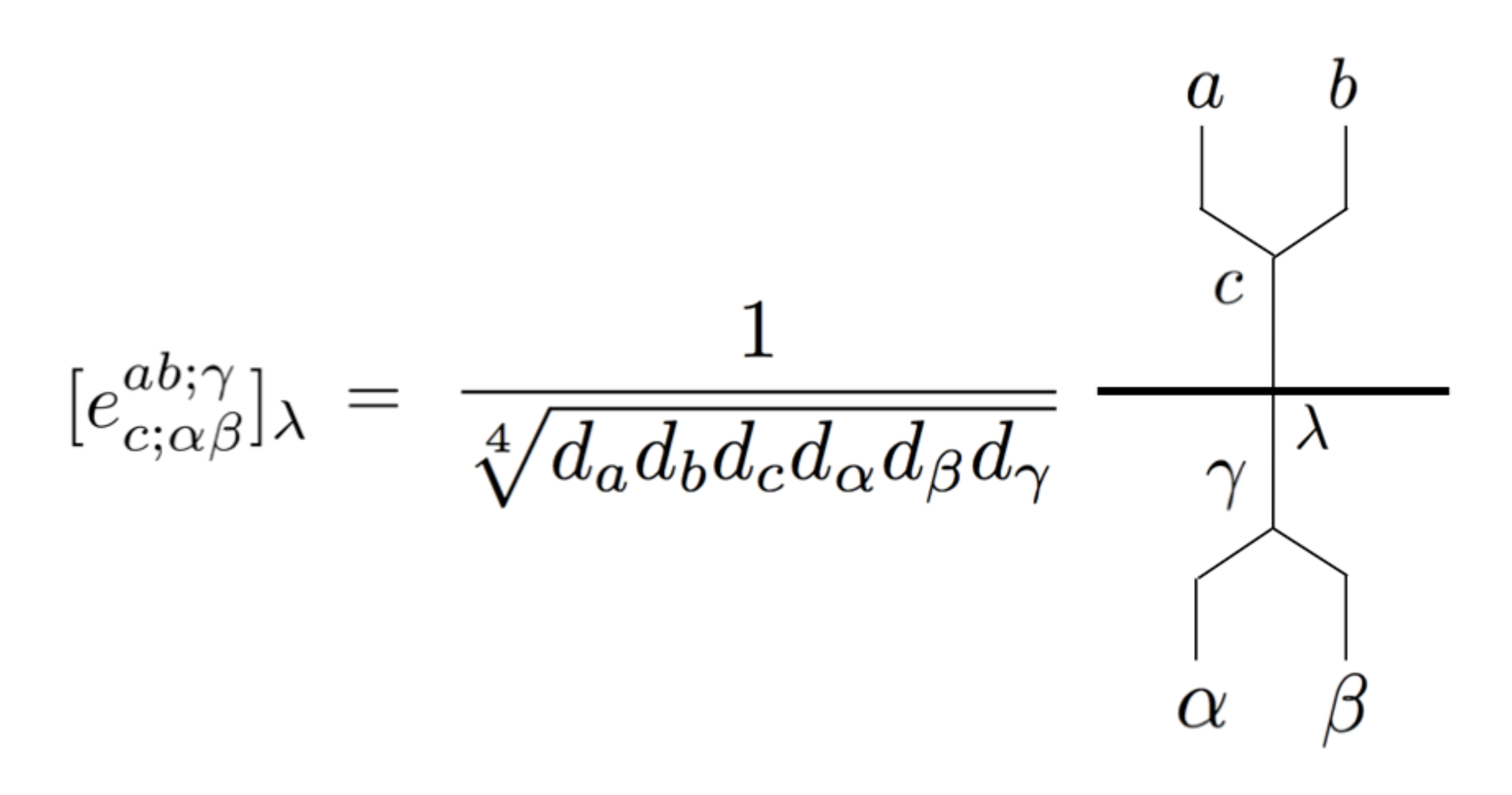}}}
\end{equation}

\noindent
As before, the basis vectors for $\Hom(a,I(\alpha)) \otimes \Hom(b,I(\beta))$ are given by $[e^a_\alpha]_\mu \otimes [e^b_\beta]_\nu$, and these bases are chosen because the following traces evaluate to 1:

\begin{equation}
\vcenter{\hbox{\includegraphics[width = 0.28\textwidth]{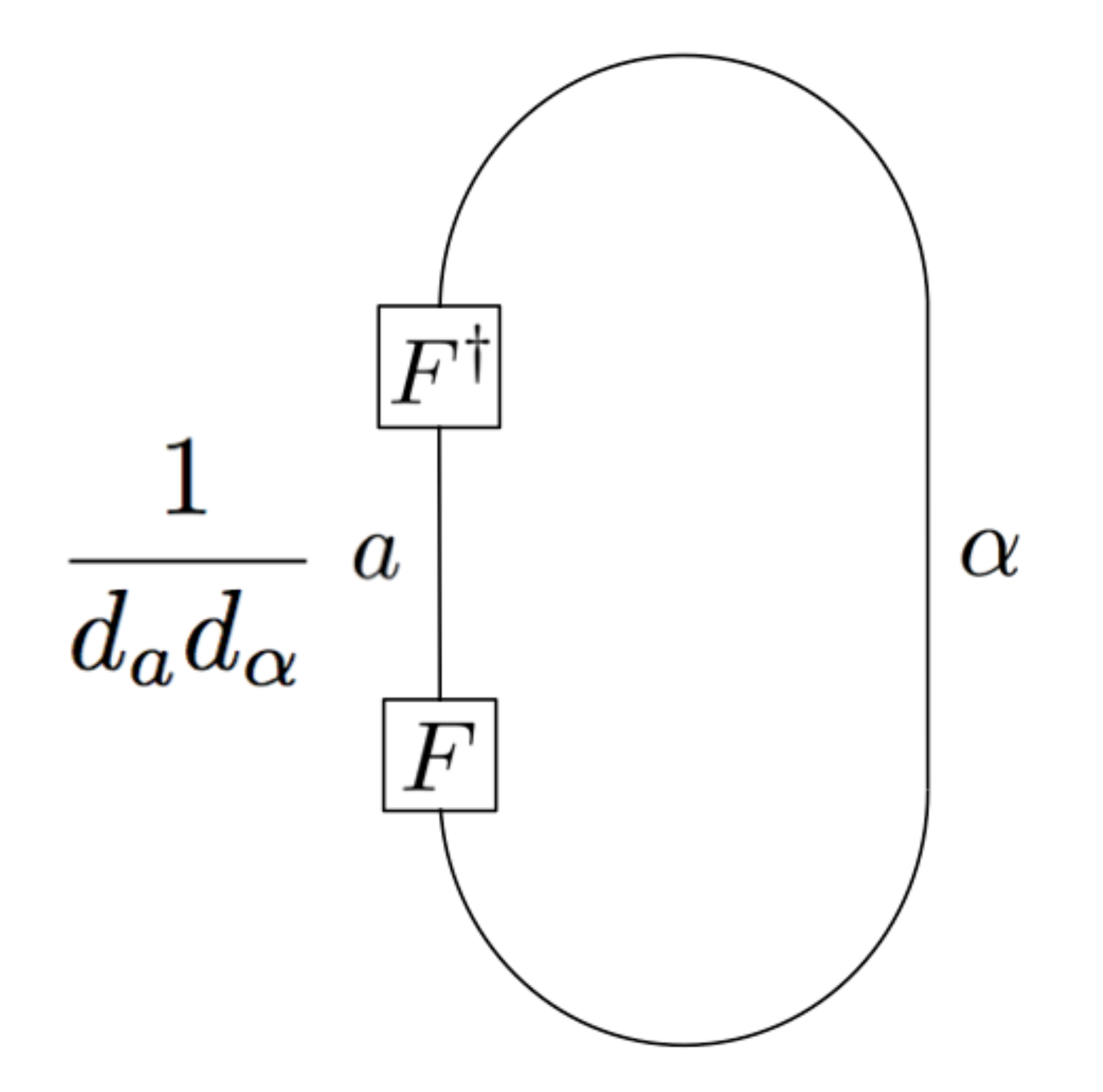}}} = 1,
\qquad
\vcenter{\hbox{\includegraphics[width = 0.41\textwidth]{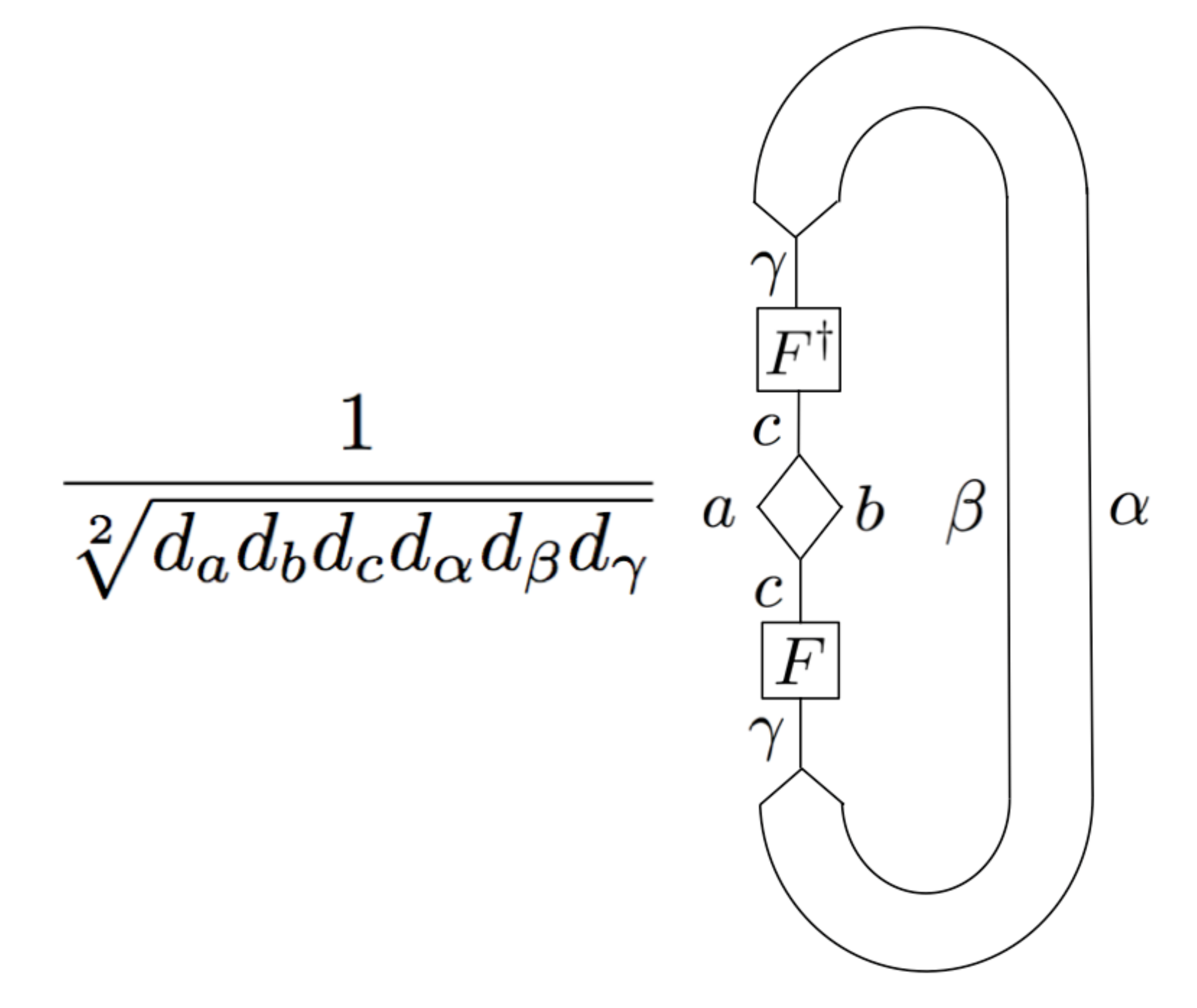}}} = 1
\end{equation}

This gives the normalization condition 

\begin{equation}
\label{eq:m-6j-normalization}
[M^{1a;\alpha}_{a;1\alpha}]^\mu_\nu = [M^{a1;\alpha}_{a;\alpha 1}]^\mu_\nu = \delta_{\mu\nu}.
\end{equation}

In general, $M$ symbols may be computed analytically using software packages. We note that given all data for the categories $\B = \mZ(\mC)$ and $\mQ = \Fun_\mC(\M,\M)$, these symbols are typically easier to calculate than the solutions to the standard $F-6j$ symbols, since Eq. (\ref{eq:m-6j-pentagon}) is at most quadratic.

\subsection{Defects between different boundary types}
\label{sec:defects-categorical}

In this section, we present the categorical model for the defects between different boundary types. These were introduced using the Hamiltonian $H_{\text{dft}}$ in Section \ref{sec:defect-hamiltonian} for the Kitaev model.

\subsubsection{Simple defect types}

Kitaev and Kong have claimed in Ref. \cite{KitaevKong} that in the Levin-Wen model based on input fusion category $\mC$, the defect types between two boundaries given by indecomposable modules $\M_i$ and $\M_j$ are given by objects in the functor category $\mC_{ij} = \Fun_\mC(\M_i, \M_j)$. In what follows, we will show this claim is true in the case of Kitaev models for Dijkgraaf-Witten theories.

In Section \ref{sec:defect-hamiltonian}, we stated that the simple defect types in a Kitaev model based on group $G$ between two boundaries given by subgroups $K_1, K_2 \subseteq G$ are given by pairs $(T,R)$, where $T \in K_1 \backslash G / K_2$ is a double coset, and $R$ is an irreducible representation of the subgroup $(K_1, K_2)^{r_T} = K_1 \cap r_T K_2 r_T^{-1}$. By Ref. \cite{Ostrik02}, we have the following Theorem:

\begin{theorem}
Let $G$ be a finite group, let $\mC = \Rep(G)$, and let $\B = \mZ(\mC)$ be its Drinfeld double. Let $\M_1$, $\M_2$ be indecomposable modules of $\mC$ given by subgroups $K_1$ and $K_2$, respectively, and trivial cocycles. The simple objects of the bimodule category $\mC_{12} = \Fun_\mC(\M_1, \M_2)$ are given by pairs $(T,R)$, where $T \in K_1 \backslash G / K_2$ is a double coset, and $R$ is an irreducible representation of the subgroup $(K_1, K_2)^{r_T} = K_1 \cap r_T K_2 r_T^{-1}$.
\end{theorem}

As a direct corollary of this theorem, we now see that the defect types between these two boundary types are exactly given by the objects of the functor category $\mC_{12}$.

In general, the category $\mC_{ij}$ is not a fusion category, as there is no canonical way to define a tensor product. However, if we consider the category of all such functor categories over the input fusion category $\mC$, we get a multi-fusion category $\mfC$.  This multi-fusion category is called an {\it $n\times n$ $2$-matrix} in Ref. \cite{Chang15}, where $n$ is the number of different indecomposable modules. Such a unitary multi-fusion category has well-defined quantum dimensions for all simple objects, which are given by formulas above.

\begin{figure}
\centering
\includegraphics[width = 0.55\textwidth]{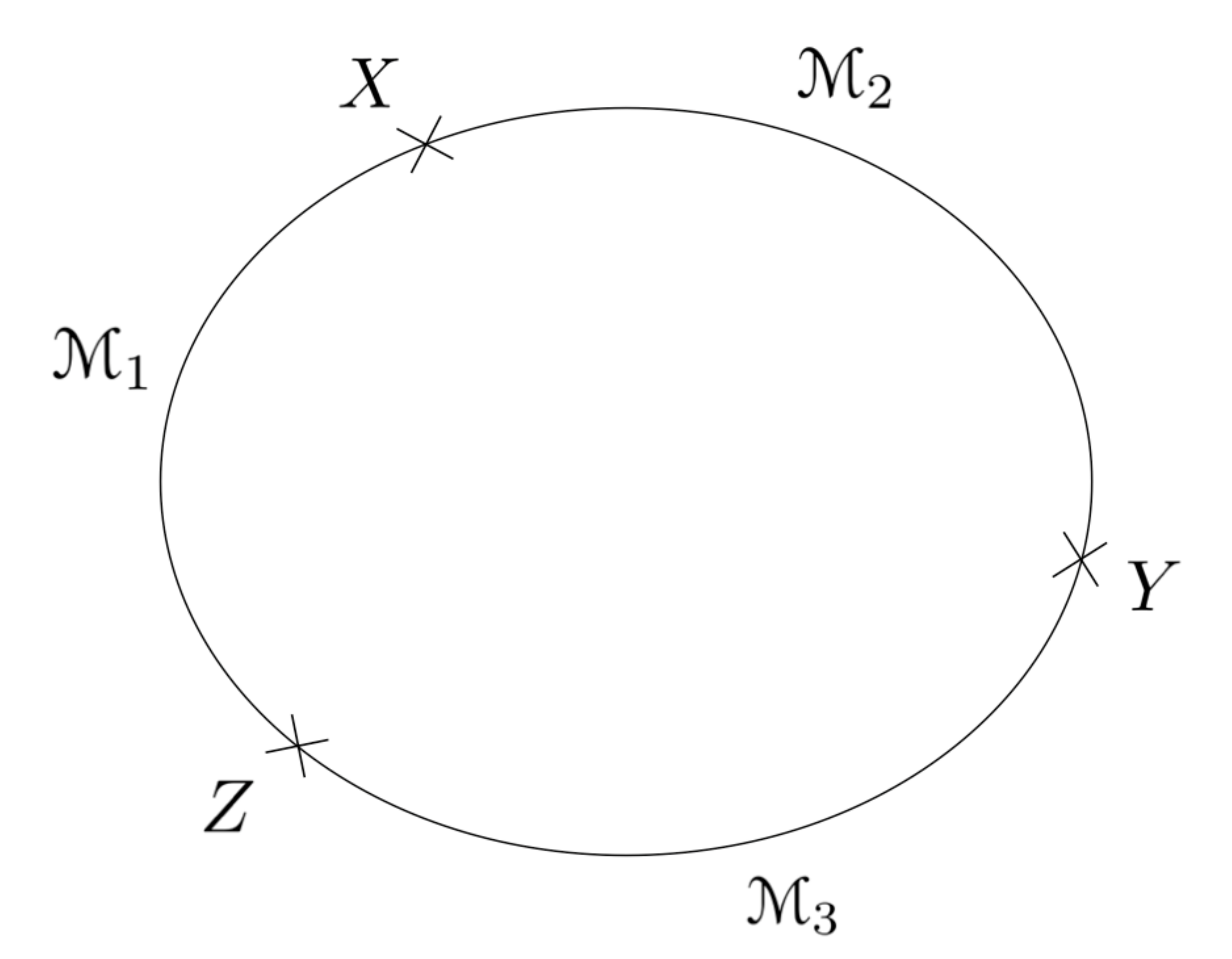}
\caption{Three boundary types and three defects on a hole. If we view all defects as going in the clockwise direction, we have $X \in \mC_{12}$, $Y \in \mC_{23}$, and $Z \in \mC_{31}$. If we view them counterclockwise, we have $X \in \mC_{12}^{\text{op}}$, $Y \in \mC_{23}^{\text{op}}$, and $Z \in \mC_{31}^{\text{op}}$.}
\label{fig:defects-fusion}
\end{figure}

\subsubsection{Fusion and braiding of boundary defects}

Using the multi-fusion category $\mfC$, we can mathematically formulate many physical processes that involve boundary defects. One such process that is most interesting to consider is the fusion of two boundary defects. Suppose we have a hole with three boundary types given by indecomposable modules $\M_1$, $\M_2$, and $\M_3$. We have three boundary defects, as shown in Fig. \ref{fig:defects-fusion}. If we view all defects as going in the clockwise direction, we have $X \in \mC_{12}$, $Y \in \mC_{23}$, and $Z \in \mC_{31}$. If we view them counterclockwise, we have $X \in \mC_{12}^{\text{op}}$, $Y \in \mC_{23}^{\text{op}}$, and $Z \in \mC_{31}^{\text{op}}$. This is because a defect $X_{ij}$ between two indecomposable modules $\M_i$ and $\M_j$ actually corresponds to two functors: one is from $\M_i$ to $\M_j$, and the other is from $\M_j$ to $\M_i$ (i.e. it goes in the opposite direction and hence lives in the category $\mC_{ij}^{\text{op}}$). The quantum dimension of a defect $X_{ij}$ depends only on the quantum dimension of a functor (object) in one of these categories, as the other functor is completely determined by this choice.

We would like to consider what happens when we fuse the two defects $X$ and $Y$, by contracting the boundary portion corresponding to $\M_2$. We expect to get a new defect $Z'$ between $\M_1$ and $\M_3$. As we have discussed, the category that each defect $X$, $Y$ belongs to depends on which way we move around the hole. Suppose we begin by viewing clockwise. When we fuse $X$ and $Y$, however, the fusion must occur in both of the categories $\mC_{12}$ and $\mC_{12}^{\text{op}}$ for $X$, and in both of the categories $\mC_{23}$ and $\mC_{23}^{\text{op}}$ for $Y$. As a result, the fusion of $X$ and $Y$ must be given by the procedure

\begin{equation}
\label{eq:defects-fusion}
\mC_{12} \otimes \mC_{23}
\rightarrow
(\mC_{12} \boxtimes \mC_{12}^{\text{op}}) \otimes (\mC_{23} \boxtimes \mC_{23}^{\text{op}})
\rightarrow
\mZ(\mfC)^{\otimes 2}
\rightarrow
\mZ(\mfC)
\rightarrow
\mC_{13} \boxtimes \mC_{13}^{\text{op}}
\rightarrow
\mC_{13}
\end{equation}

Hence, the fusion of two defects in $\mC_{ij}$ and $\mC_{jk}$ actually occurs in the Drinfeld center $\mZ(\mfC)$. The Drinfeld center is a modular tensor category, so it naturally gives a braiding structure to the defects.

As it turns out, Theorem 2.4 of Ref. \cite{Chang15} says that the categories $\mfC$, $\mC$, and $\mC_{ii}$ are all Morita equivalent:

\begin{equation}
\mZ(\mfC) = \mZ(\mC) = \mZ(\mC_{ii}).
\end{equation}

\noindent
This means that the bulk anyons in $\mZ(\mC)$ may also be braided with the boundary defects.

\begin{figure}
\centering
\includegraphics[width = 0.55\textwidth]{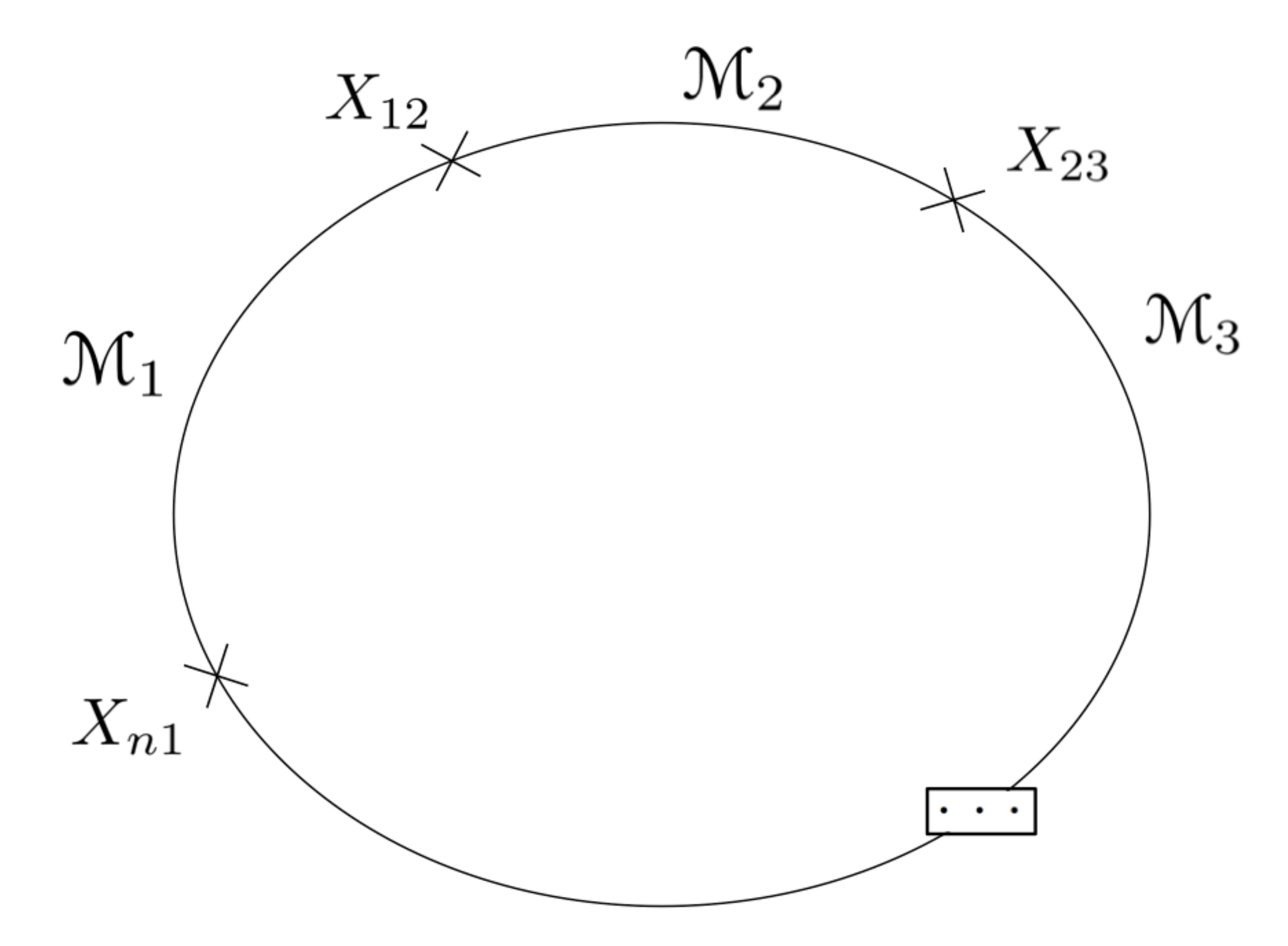}
\caption{Fusion of $n$ boundary defects on a hole with $n$ boundary times given by indecomposable modules $\M_1$, ... $\M_n$.}
\label{fig:defects-fusion-2}
\end{figure}

\subsubsection{Topological degeneracy}

As in the case of bulk anyons, topological degeneracy can also arise from the fusion of boundary defects. The most interesting case is when we have $n$ defects $X_{12} \in \mC_{12}$, $X_{23} \in \mC_{23}$, ... $X_{n1} \in \mC_{n1}$, and we would like to fuse them so that no boundary defect remains, and the resulting boundary is in the ground state of one of the boundaries (say the $\M_1$ boundary). The setup is illustrated in Fig. \ref{fig:defects-fusion-2}. Then, the topological degeneracy is given by the hom-space

\begin{equation}
\Hom(\one_{\M_1}, X_{12} \otimes X_{23} \otimes ... \otimes X_{n1})
\end{equation}

\noindent
where $\one_{\M_1}$ is the tensor unit in the fusion category $\mC_{11}$ (i.e. the trivial boundary excitation).

\subsection{Example: The Toric Code}
\label{sec:tc-algebraic}

\subsubsection{Topological order}

The topological order of the toric code is given by the modular tensor category $\B = \mfD(\Z_2)$. The $F$ symbols in this category are all trivial, and the $R$ symbols are given by $R^{a,b} = e^{\pi i a_2 b_1}$, where we write $a = e^{a_1} m^{a_2}$, $b = e^{b_1} m^{b_2}$ for $a_1,a_2,b_1,b_2 = 0,1$ ($\epsilon = em$) \cite{Barkeshli14}. The modular $\mathcal{S}$, $\mathcal{T}$ matrices are given by \cite{BakalovKirillov}:

\begin{equation}
\mathcal{S} = \frac{1}{4}
\begin{bmatrix}
1 & 1 & 1 & 1 \\
1 & 1 & -1 & -1 \\
1 & -1 & 1 & -1 \\
1 & -1 & -1 & 1
\end{bmatrix}
\end{equation}

\begin{equation}
\mathcal{T} = \text{diag}(1,1,1,-1).
\end{equation}

\noindent
Here, the rows/columns are all ordered as $(1,e,m,\epsilon)$.

\subsubsection{Lagrangian algebras}

In this section, we will calculate the Lagrangian algebras corresponding to the gapped boundaries of the toric code ($\B = \mZ(\Rep(\Z_2)$). There are 4 simple objects in this category, with fusion rules given by the multiplication of $\Z_2 \times \Z_2$. The calculation proceeds as follows:

\begin{enumerate}
\item
The Frobenius-Perron dimension of $\B$ is $\FPdim(\B) = 4$. All simple objects in $\B$ have dimension 1, so each Lagrangian algebra will have two simple objects in its decomposition.
\item
There are two bosons in $\B$, namely $e$ and $m$.
\item
By (1) and (2), the only possibilities for Lagrangian algebras would be $\A_1 = 1+e$, $\A_2 = 1+m$. We check to see that $\A_1,\A_2$ indeed satisfy the inequality (\ref{eq:lagrangian-algebra-inequality}). Hence the gapped boundaries of the toric code are exactly given by these two Lagrangian algebras.
\end{enumerate}

\subsubsection{Condensation procedure: $1+e$ boundary}

We now illustrate the condensation procedure described in Eq. (\ref{eq:condensation-quotient-IC}), for the $1+e$ boundary of the toric code. We first form the pre-quotient category ${\widetilde{\mQ}}$. By Definition \ref{quotient-cat-def}, we have

\begin{equation}
\begin{gathered}
\Hom_{\widetilde{\mQ}} (1,1) = \Hom_\B (1,1+e) \cong \C \\
\Hom_{\widetilde{\mQ}} (1,e) = \Hom_\B (1,e+1) \cong \C \\
\Hom_{\widetilde{\mQ}} (1,m) = \Hom_\B (1,m+\epsilon) = 0 \\
\Hom_{\widetilde{\mQ}} (1,\epsilon) = \Hom_\B (1,\epsilon+m) = 0 \\
\Hom_{\widetilde{\mQ}} (m,m) = \Hom_\B (m,m+\epsilon) \cong \C \\
\Hom_{\widetilde{\mQ}} (m,\epsilon) = \Hom_\B (m,\epsilon + m) \cong \C
\end{gathered}
\end{equation}

\noindent
Since we already have $\Hom_{\widetilde{\mQ}} (1,1) \cong \C$ and $\Hom_{\widetilde{\mQ}} (m,m) \cong \C$, there are no nontrivial splitting idempotents in the category ${\widetilde{\mQ}}$, and we have $\mQ = \widetilde{\mQ}$ (the completion is trivial). This gives the following condensation products:

\begin{equation}
1,e \rightarrow 1, \qquad m,\epsilon \rightarrow m
\end{equation}

\noindent
We would like to note that this is the exact same result as obtained in Section \ref{sec:tc-hamiltonian-example} using Theorem \ref{condensation-products}. The same procedure may be carried out with the $1+m$ boundary, and will also agree with the result of Section \ref{sec:tc-hamiltonian-example}.

\subsubsection{$M$ symbols}

As shown in Ref. \cite{Barkeshli14}, all $F$ symbols for the toric code are trivial. As a result, it is not hard to show using all of the equations from Section \ref{sec:m-symbols} that all $M$-3$j$ and $M$-6$j$ symbols are trivial for both boundaries of the toric code.

\subsection{Example: $\mfD(S_3)$}
\label{sec:ds3-algebraic-example}

\subsubsection{Topological order}

When $G = S_3$, the topological order of the resulting Kitaev model is given by the modular tensor category $\B = \Rep(D(S_3)) = \mZ(\Rep(S_3))$. As discussed in Section \ref{sec:ds3-hamiltonian-example}, there are 8 simple objects in this category, $A,B,...,H$. The fusion rules \cite{Cui15} are given in Fig. \ref{tab:DS3-fusion}, and the $\mathcal{S}$, $\mathcal{T}$ matrices are listed below. The $F$, and $R$ symbols for this category may be found in Appendix A of Ref. \cite{Cui15}.

{\tiny
\begin{table}\caption{Fusion rules of $\mfD(S_3)$}\label{tab:DS3-fusion}
\begin{tabular}{|c|c|c|c|c|c|c|c|c|}
\hline $\otimes$ &$A$ &$B$ &$C$ &$D$ &$E$ &$F$ &$G$ &$H$\\ \hline
$A$ &$A$ &$B$ &$C$ &$D$ &$E$& $F$ &$G$ &$H$\\ \hline
$B$ &$B$ &$A$ &$C$& $E$ &$D$ &$F$ &$G$ &$H$\\ \hline
$C$ &$C$ &$C$ &$A\oplus B\oplus C$& $D\oplus E$ &$D\oplus E$ & $G\oplus H$& $F\oplus H$ &$F\oplus G$\\ \hline
\multirow{2}{*}{$D$} &\multirow{2}{*}{$D$} &\multirow{2}{*}{$E$} &\multirow{2}{*}{$D\oplus E$}& $A\oplus C\oplus  F$ & $B\oplus C\oplus F$ & \multirow{2}{*}{$D\oplus E$} & \multirow{2}{*}{$D\oplus E$} & \multirow{2}{*}{$D\oplus E$} \\
& & & & $\oplus G\oplus H$ & $\oplus G\oplus H$ & & &  \\ \hline
\multirow{2}{*}{$E$} &\multirow{2}{*}{$E$}& \multirow{2}{*}{$D$}& \multirow{2}{*}{$D\oplus E$} & $B\oplus C\oplus F$ & $A\oplus C\oplus F$ & \multirow{2}{*}{$D\oplus E$} &\multirow{2}{*}{$D\oplus E$} & \multirow{2}{*}{$D\oplus E$} \\
& & & & $\oplus G\oplus H$ & $\oplus G\oplus H$ & & &  \\  \hline
$F$ &$F$ & $F$& $G\oplus H$& $D\oplus E$ & $D\oplus E$ & $A\oplus B\oplus F$ & $H\oplus C$ & $G\oplus C$ \\ \hline
$G$ &$G$ & $G$& $F\oplus H$ & $D\oplus E$ & $D\oplus E$ & $H\oplus C$ & $A\oplus B\oplus G$ & $F\oplus C$ \\ \hline
$H$ &$H$ & $H$& $F\oplus G$ & $D\oplus E$ & $D\oplus E$ & $G\oplus C$ & $F\oplus C$ & $A\oplus B\oplus H$\\\hline
\end{tabular}
\end{table}
}

The modular $\mathcal{S}$ and $\mathcal{T}$ matrices of $\mfD(S_3)$ are given by \cite{Cui15}:

\begin{equation}
\label{eq:ds3-S}
\mathcal{S} = \frac{1}{6}
\begin{bmatrix}
1 & 1 & 2 & 3 & 3 & 2 & 2 & 2 \\
1 & 1 & 2 & -3 & -3 & 2 & 2 & 2 \\
2 & 2 & 4 & 0 & 0 & -2 & -2 & -2 \\
3 & -3 & 0 & 3 & -3 & 0 & 0 & 0 \\
3 & -3 & 0 & -3 & 3 & 0 & 0 & 0 \\
2 & 2 & -2 & 0 & 0 & 4 & -2 & -2 \\
2 & 2 & -2 & 0 & 0 & -2 & -2 & 4 \\
2 & 2 & -2 & 0 & 0 & -2 & 4 & -2 \\
\end{bmatrix}
\end{equation}

\begin{equation}
\label{eq:ds3-T}
\mathcal{T} = \text{diag}(1,1,1,1,-1,1,\omega,\omega^2)
\end{equation}

\noindent
where all rows and columns are ordered alphabetically, $A-H$, and $\omega = e^{2 \pi i / 3}$ is the primitive third root of unity.

\subsubsection{Lagrangian algebras}

We now determine gapped boundaries of the $\mfD(S_3)$ model by computing all Lagrangian algebras in $\B$, using the procedure of Section \ref{sec:frobenius-algebras}:

First, the Frobenius-Perron dimension of $\B$ is $\FPdim(\B) = 36$, so each Lagrangian algebra $\A$ should have

\begin{equation}
\label{eq:ds3-dim}
\FPdim(\A) = 6.
\end{equation}

Second, by Eq. (\ref{eq:ds3-T}) the bosons in $\B$ are the anyons $A,B,C,D,F$. By Proposition \ref{bosons}, this means that

\begin{equation}
\label{eq:EGH-0}
n_E = n_G = n_H = 0.
\end{equation}

We now apply the inequality (\ref{eq:lagrangian-algebra-inequality}). By Eq. (\ref{eq:EGH-0}) and the fusion rule $C \otimes F \rightarrow G \oplus H$, we have

\begin{equation}
\label{eq:CF-0}
n_C n_F = 0.
\end{equation}

\noindent
We use casework to obtain the final decomposition $\A = \oplus n_s s$:

\vspace{2mm}
\begin{itemize}[wide, label={}, listparindent=1.5em, parsep=0.25mm, itemsep=3mm, labelindent=0pt]

\item
{\it \underline{Case I}}: $n_B > 0$.

Suppose $n_B > 0$. First, by the fusion rule $B \otimes B \rightarrow A$ and the inequality (\ref{eq:lagrangian-algebra-inequality}), we know that $n_B = 1$. By the fusion rule $B \otimes D \rightarrow E$ and Equations (\ref{eq:lagrangian-algebra-inequality}) and (\ref{eq:EGH-0}), we have $n_D = 0$. Finally, by Equations (\ref{eq:ds3-dim}) and (\ref{eq:CF-0}), we know that the only two possible Lagrangian algebras are 

\begin{equation}
\A_1 = A+B+2C, \qquad \A_2 = A+B+2F.
\end{equation}

By checking the inequality (\ref{eq:lagrangian-algebra-inequality}) on the rest of the fusion rules, we see that $\A_1$ and $\A_2$ are indeed both Lagrangian algebras in $\B$.

\noindent
{\it \underline{Case II}}: $n_B = 0$.

Suppose $n_B = 0$. By Eq. (\ref{eq:ds3-dim}), we know that $n_D = 1$, since $C,F$ have even quantum dimensions. Hence, the only two possible Lagrangian algebras are 

\begin{equation}
\A_3 = A+C+D, \qquad \A_4 = A+D+F.
\end{equation}

By checking the inequality (\ref{eq:lagrangian-algebra-inequality}) on the rest of the fusion rules, we see that $\A_3$ and $\A_4$ are indeed both Lagrangian algebras in $\B$.
\end{itemize}

We note that the Lagrangian algebras $\A_1, ... \A_4$ obtained here are exactly the same as the result we obtained by using Theorem \ref{inverse-condensation-products} in Section \ref{sec:ds3-hamiltonian-example}.

\subsubsection{Condensation procedure: $A+C+D$ boundary}

We now illustrate the condensation procedure of Eq. (\ref{eq:condensation-quotient-IC}) on the $A+C+D$ boundary of the $\mfD(S_3)$ theory. We first form the pre-quotient category ${\widetilde{\mQ}}$, which has the same objects as $\B = \mZ(\Rep(S_3))$. By Definition \ref{quotient-cat-def}, we have:

\begin{equation}
\Hom_{\widetilde{\mQ}}(A,C) = \Hom_\B(A, C+A+B+C+D+E) \cong \C
\end{equation}

\noindent
Similarly, 

\begin{equation}
\begin{gathered}
\Hom_{\widetilde{\mQ}}(A,D) \cong \C \qquad \Hom_{\widetilde{\mQ}}(C,D) \cong \C \\
\Hom_{\widetilde{\mQ}}(A,A) \cong \C \qquad \Hom_{\widetilde{\mQ}}(B,B) \cong \C \\
\Hom_{\widetilde{\mQ}}(B,C) \cong \C \qquad \Hom_{\widetilde{\mQ}}(B,E) \cong \C \\
\Hom_{\widetilde{\mQ}}(F,D) \cong \C \qquad \Hom_{\widetilde{\mQ}}(F,F) \cong \C \\
\Hom_{\widetilde{\mQ}}(F,G) \cong \C \qquad \Hom_{\widetilde{\mQ}}(F,H) \cong \C
\end{gathered}
\end{equation}

\noindent
Many other hom-sets in ${\widetilde{\mQ}}$ may be constructed from the above by composition (simply tensor product the corresponding hom-spaces). All other hom-sets in this category are zero. As in the case of the toric code, no idempotent completion is necessary, as all endomorphism spaces in $\widetilde{\mQ}$ for simple objects are one-dimensional and hence have no nontrivial splitting idempotents. It follows that the following rules describe the condensation of simple bulk particles onto the $A+C+D$ boundary:

\begin{multicols}{2}
\begin{enumerate}[label=(\roman*),leftmargin=0.5in]
\item
$A \rightarrow {A}$
\item
$B \rightarrow {B}$
\item
$C \rightarrow {A} \oplus {B}$
\item
$D \rightarrow {A} \oplus {F}$
\item
$E \rightarrow {B} \oplus {F}$
\item
$F,G,H \rightarrow {F}$
\end{enumerate}
\end{multicols}

We note that this is exactly the same result we obtained in Section \ref{sec:ds3-hamiltonian-example}, if we identify the boundary excitation label $F$ with the label $C$ from that section.

The $A+F+D$ boundary is easily shown to have the same condensation rules and properties, with all instances of $C$ and $F$ switched.

\subsubsection{Condensation procedure: $A+B+2C$ boundary}

We now illustrate the condensation procedure of Eq. (\ref{eq:condensation-quotient-IC}) on the $A+B+2C$ boundary of the $\mfD(S_3)$ theory, as this example will require a nontrivial idempotent completion. As before, we first use Definition \ref{quotient-cat-def} to construct the pre-quotient category $\widetilde{\mQ}$. This gives the following hom-sets:

\begin{equation}
\begin{gathered}
\Hom_{\widetilde{\mQ}}(A,A) \cong \C \qquad \Hom_{\widetilde{\mQ}}(A,B) \cong \C \\
\Hom_{\widetilde{\mQ}}(A,C) \cong \C^2 \qquad \Hom_{\widetilde{\mQ}}(D,D) \cong \C^3 \\
\Hom_{\widetilde{\mQ}}(D,E) \cong \C^3 \qquad \Hom_{\widetilde{\mQ}}(E,E) \cong \C^3 \\
\Hom_{\widetilde{\mQ}}(F,F) \cong \C^2 \qquad \Hom_{\widetilde{\mQ}}(F,G) \cong \C^2 \\
\Hom_{\widetilde{\mQ}}(F,H) \cong \C^2 \qquad \Hom_{\widetilde{\mQ}}(G,G) \cong \C^2 \\
\Hom_{\widetilde{\mQ}}(G,H) \cong \C^2 \qquad \Hom_{\widetilde{\mQ}}(H,H) \cong \C^2 \\
\end{gathered}
\end{equation}

All other hom-sets in $\widetilde{\mQ}$ between simple objects in $\B$ are either tensor products of the above (in case of composition), or zero. It follows that the quotient functor acts as follows on the simple objects of $\B$:

\begin{multicols}{2}
\begin{enumerate}[label=(\roman*),leftmargin=0.5in]
\item
$A,B \rightarrow {A}$
\item
$C \rightarrow {2 \cdot A}$
\item
$D,E \rightarrow D$
\item
$F,G,H \rightarrow F$
\end{enumerate}
\end{multicols}

However, we would now like to note that the pre-quotient category $\widetilde{\mQ}$, with simple objects given by $A$,$D$, and $F$, is not semisimple, as we see $\Hom_{\widetilde{\mQ}}(D,D) \cong \C^3$ and $\Hom_{\widetilde{\mQ}}(F,F) \cong \C^2$, when they should be one-dimensional. This tells us that we must perform the canonical idempotent completion of $\widetilde{\mQ}$ to ${\mQ}$, which transforms the simple objects of $\widetilde{\mQ}$ as follows:

\begin{enumerate}[label=(\roman*),leftmargin=0.5in]
\item
$A \rightarrow A$
\item
$D \rightarrow (D,p_1) \oplus (D,p_2) \oplus (D,p_3)$
\item
$F \rightarrow (F,q_1) \oplus (F,q_2)$
\end{enumerate}

\noindent
where each $p_i$ is a splitting idempotent in $\Hom_{\widetilde{\mQ}}(D,D)$, and each $q_j$ is a splitting idempotent in $\Hom_{\widetilde{\mQ}}(F,F)$. (In general, if $\Hom_{\widetilde{\mQ}}(X,X)$ is $n$-dimensional, there are $n$ splitting idempotents).

Hence, we have the following rules for the overall condensation procedure of simple bulk anyons of $\mfD(S_3)$ to the $A+B+2C$ boundary:

\begin{multicols}{2}
\begin{enumerate}[label=(\roman*),leftmargin=0.5in]
\item
$A,B \rightarrow {A}$
\item
$C \rightarrow {2 \cdot A}$
\item
$D,E \rightarrow (D,p_1) \oplus (D,p_2) \oplus (D,p_3)$
\item
$F,G,H \rightarrow (F,q_1) \oplus (F,q_2)$
\end{enumerate}
\end{multicols}

We note that this is exactly the same result we obtained in Section \ref{sec:ds3-hamiltonian-example}, if we identify the above boundary excitation labels with those of Section \ref{sec:ds3-hamiltonian-example} as follows:

\begin{multicols}{2}
\begin{enumerate}[label=(\roman*),leftmargin=0.5in]
\item
$A \rightarrow 1$
\item
$(F,q_1) \rightarrow r$
\item
$(F,q_2) \rightarrow r^2$
\item
$(D,p_1) \rightarrow s$
\item
$(D,p_2) \rightarrow sr$
\item
$(D,p_3) \rightarrow sr^2$
\end{enumerate}
\end{multicols}

The $A+B+2F$ boundary is easily shown to have the same condensation rules and properties, with all instances of $C$ and $F$ switched.

\subsubsection{$M$-3$j$ symbols of the $A+C+D$ boundary}
\label{sec:ds3-m3j}

We have computed the $M$-3$j$ symbols of the $A+C+D$ boundary, up to a sign in a few cases. By Equations (\ref{eq:m-3j-pentagon}) and (\ref{eq:m-3j-braid}-\ref{eq:m-3j-normalization-2}), we have:

\begin{equation}
\begin{gathered}
M^{AX}_X = 1, \text{ } X = A,C,D \\
M^{CC}_A = \frac{1}{\sqrt{6}} \qquad M^{CC}_C = \pm \frac{i}{\sqrt{2}} \\
M^{DD}_A = \frac{1}{\sqrt{6}} \qquad M^{DD}_C = \pm i \sqrt{\frac{2}{3}} \\
M^{CD}_D = M^{DC}_D = \mp i
\end{gathered}
\end{equation}

\noindent
Other $M-3j$ symbols for this boundary are all zero.

\subsection{Example: Genons}
\label{sec:categorical-genons}

In this section, we focus on an example of the broad categorical theory of defects between two types of gapped boundaries that was developed in Section \ref{sec:defects-categorical}. Specifically, we would like to consider genons, in the case where the input fusion category is a bilayer $\D = \mC \boxtimes \mC$ for some unitary modular tensor category $\mC$. Then, the topological order of the system is given by $\B = \mZ(\D) = \D \boxtimes \overbar{\D} = \mC_1 \boxtimes {\mC_2} \boxtimes \overbar{\mC_3} \boxtimes \overbar{\mC_4}$, where $\mC_1, ... \mC_4$ are four disjoint copies of $\mC$. Hence, we have two identical, non-interacting copies of a TQFT with topological order $\mathcal{E} = \mC \boxtimes \overbar{\mC}$, namely $\mathcal{E}_1 = \mC_1 \boxtimes \overbar{\mC_3}$ and $\mathcal{E}_2 = \mC_2 \boxtimes \overbar{\mC_4}$. Defects in this system arise when two gapped boundaries (Lagrangian algebras) of $\B$ meet; genons are special types of these defects that break the $\Z_2$ symmetry in $\B$ (which interchanges the two layers $\mE_1, \mE_2$) \cite{Barkeshli14}. 

In Section \ref{sec:genon-hamiltonian}, we presented a Hamiltonian implementation of genons in the case where $\mC = \C[G]$ for an arbitrary finite group $G$, and we have chosen a particular modular structure for the fusion category $\C[G]$. We saw that genons were given by a defect between the boundary type given by the trivial subgroup $K_1 = \{1_G\} \times \{1_G\}$, and the subgroup $K_2 = G \times \{1_G\}$. Categorically, these subgroups correspond to indecomposable module categories $\M, \mN$ in $\B = \mZ(\C[G \times G])$, and the simple defect types (genons) correspond to the simple objects in $\Fun_\mC(\M,\mN)$.

We will now generalize this result to provide a construction of genons from any system with topological order $\B$ given by the Drinfeld center of a bilayer $\D$ of a modular tensor category $\mC$. This can be done as follows: Since $\mE = \mC \boxtimes \overbar{\mC}$ is already a modular tensor category, let us find all Lagrangian algebras of $\mE$. Let $\A$ be a Lagrangian algebra such that the functor category $\Fun_\mC(\M,\M) \cong \mC$, where $\M$ is the indecomposable module category corresponding to algebra $\A$ (see Theorem \ref{indecomposable-module-lagrangian-algebra}); in general, $\Fun_\mC(\M,\M)$ is Morita equivalent to $\mC$). We now construct two gapped boundaries (Lagrangian algebras) $\A_{1234}$, $\A_{1423}$ in the MTC $\B$: 

\begin{equation}
\begin{gathered}
\A_{1324} = \A_{13} \boxtimes \A_{24} \\
\A_{1423} = \A_{14} \boxtimes \A_{23}
\end{gathered}
\end{equation}

\noindent
Here, $\A_{ij}$ the Lagrangian algebra corresponding to $\A$ when considered in the modular tensor category $\mE_{ij} = \mC_i \boxtimes \overbar{\mC_j} \cong \mE$. As indecomposable module categories of $\B$, we have

\begin{equation}
\begin{gathered}
\M_{1324} = \M_{13} \boxtimes \M_{24} \\
\M_{1423} = \M_{14} \boxtimes \M_{23}
\end{gathered}
\end{equation}

\noindent
where $\M_{ij}$ is the indecomposable module category in $\mE_{ij}$ corresponding to the Lagrangian algebra $\A_{ij}$. To generalize the language used by Ref. \cite{Bark13a}, $\A_{1324}$ represents the {\it intralayer} gapped boundary, while $\A_{1423}$ represents the {\it interlayer} gapped boundary. Then, the genon corresponds to a simple defect between these two gapped boundaries, i.e. it is a simple object in the functor category $\Fun_\D(\M_{1324}, \M_{1423})$. The ``bare defect'' discussed in Section \ref{sec:genon-hamiltonian} is one of these simple objects.

Genons play a very important role in quantum computation, as their braiding, when combined with anyons, has the power to give universal quantum computation (see Ref. \cite{Barkeshli16}).

\vspace{2mm}
\section{Implementation using surface codes}
\label{sec:circuits}

In this chapter, we will present one possible method to physically realize gapped boundaries of Kitaev models. Specifically, we generalize the $\Z_2$ surface code methodology of Fowler et al. in Ref. \cite{Fowler12}, to perform the necessary operations on gapped boundaries of any quantum double model.

This chapter should be regarded as only a summary/outline for new works to come. Since we would like to keep our theory as general as possible, we will simply present a possible surface code implementation that, in principle, works for any untwisted Dijkgraaf-Witten theory based on some finite group $G$, while omitting error rate analysis and software checking details. These details are of course essential to studying any quantum error correction code, and should be studied for each specific example of interest.

Throughout the chapter, the qudits implementing the stabilizer code will be referred to as the {\it physical} qudits (and will be further sub-classified as data/syndrome qudits in the following section). The qudit encoded by two gapped boundaries will be referred to as a {\it logical} qudit, to follow the notation of Ref. \cite{Fowler12}.

\subsection{Generalized surface codes and stabilizer circuits}
\label{sec:stabilizers}

The concept of the surface code was first introduced by Dennis et al. in \cite{Dennis02}, to physically realize Kitaev's $\Z_2$ toric code \cite{Kitaev97} on a planar surface and implement a topological quantum memory. In that publication, local stabilizer circuits were developed to implement the Hamiltonian (\ref{eq:tc-hamiltonian}). More recently, the same circuits for the $\Z_2$ surface code have been studied by Fowler et al. in \cite{Fowler12}, partly in the context of gapped boundaries. In this section, we present the generalized version of these circuits and the surface code, to implement the Hamiltonian $H_{(G,1)}$ of Eq. (\ref{eq:kitaev-hamiltonian}) based on arbitrary group $G$. We then make one further generalization, to implement the boundary Hamiltonian $H_{(G,1)}^{(K,1)}$ of Eq. (\ref{eq:bd-hamiltonian-K}).

\begin{figure}
\centering
\includegraphics[width = 0.7\textwidth]{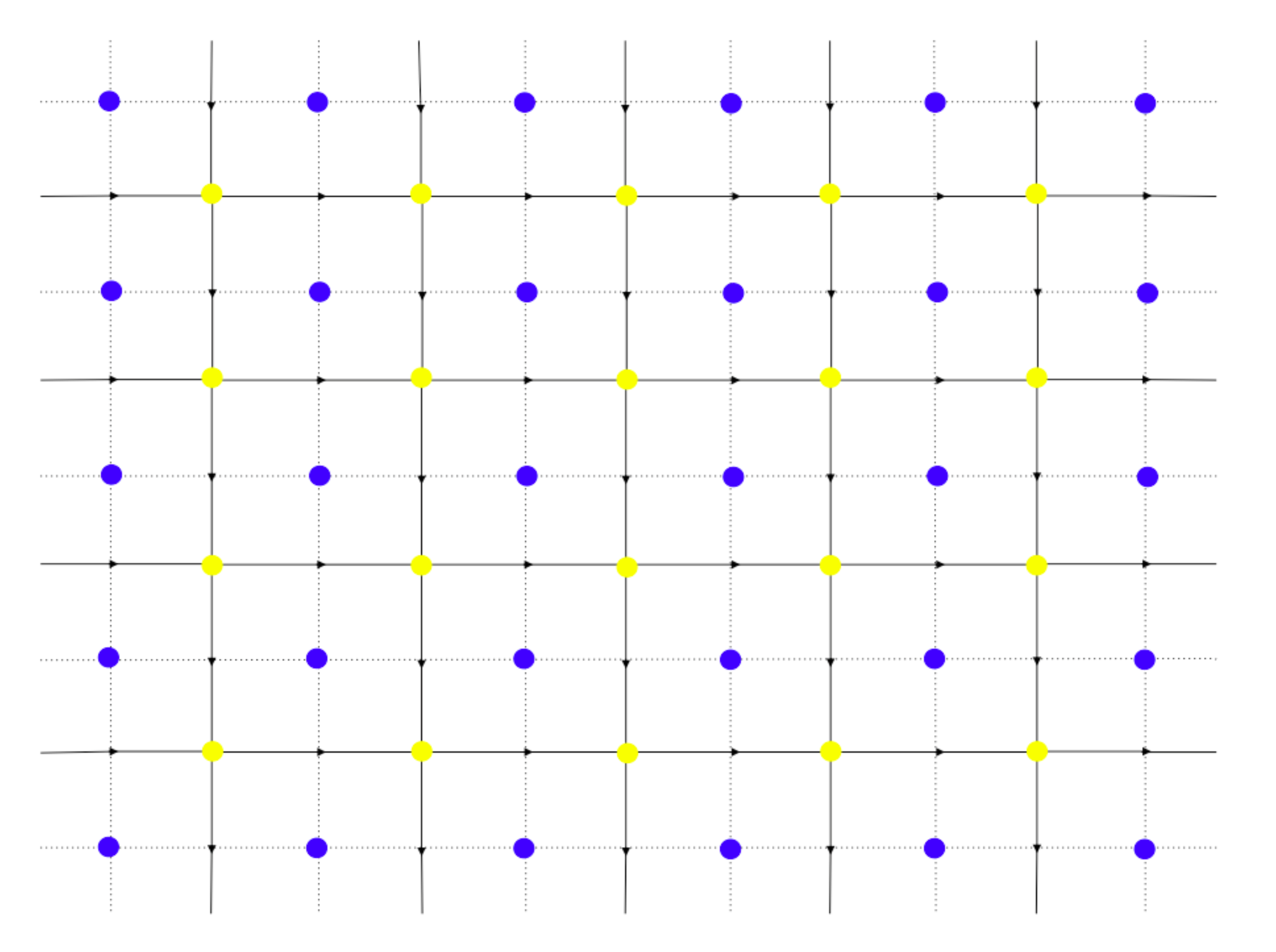}
\caption{Surface code layout that implements the Kitaev Hamiltonian $H_{(G,1)}$ of Eq. (\ref{eq:kitaev-hamiltonian}). As in Section \ref{sec:kitaev-hamiltonian}, we consider a square lattice in the plane, where a data qudit in $\C[G]$ is placed at the center of each edge. Data qudits are drawn as black arrows at the centers of edges, as they need to have an orientation for the definition of $H_{(G,1)}$ (see Chapter \ref{sec:hamiltonian}). We place a vertex syndrome qudit $\ket{v}$ in $\C[G]$ at each vertex (yellow dots) to project the system onto the simultaneous eigenstates of all $A(v)$ operators, and we place a plaquette syndrome qudit $\ket{p}$ in $\C[G]$ at each plaquette (blue dots) to project the system onto the simultaneous eigenstates of all $A(v)$ operators.}
\label{fig:surface-code-layout}
\end{figure}

\subsubsection{Bulk Hamiltonian surface code}

We first present an implementation of the bulk Hamiltonian $H_{(G,1)}$ of Eq. (\ref{eq:kitaev-hamiltonian}). Following Refs. \cite{Dennis02} and \cite{Fowler12}, we will present the generalized surface code on a planar lattice. For simplicity of presentation and illustration, we will assume that we have a square lattice; however, this is certainly not necessary and our work may be easily generalized. Suppose we are given an arbitrary finite input group $G$. As introduced in Chapter \ref{sec:hamiltonian}, we place a {\it data qudit} at the center of each edge. We then place a {\it vertex syndrome qudit} at each vertex, and a {\it plaquette syndrome qudit} at the center of each plaquette. In this model, all qudits takes values in the Hilbert space $\C[G]$ with orthonormal basis $\{ \ket{g}: g \in G \}$\footnote{Although it is highly improbable to find a qudit that naturally has this structure for general groups $G$, one can certainly implement this by placing multiple physical qubits at each edge, and entangling them appropriately for the desired group.}. The setup is shown in Fig. \ref{fig:surface-code-layout}. Notations required to define the stabilizer circuits are shown in Fig. \ref{fig:v-p-def}.

\begin{figure}
\centering
\includegraphics[width=0.62\textwidth]{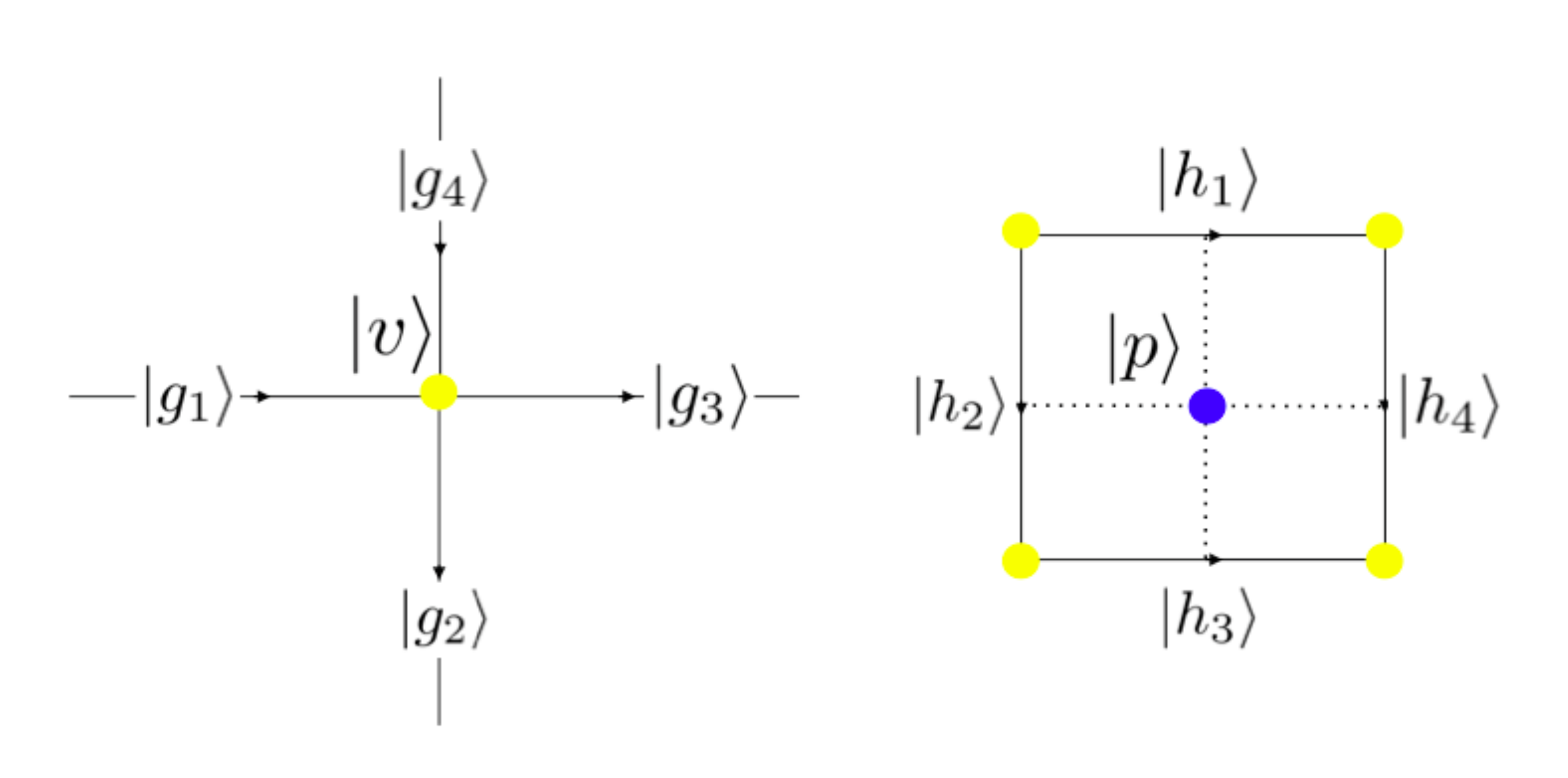}
\caption{Notations for defining the stabilizer circuits. On the square lattice, each vertex syndrome $\ket{v}$ is surrounded by four data qudits, $\ket{g_1}$ through $\ket{g_4}$, which are numbered counterclockwise as shown. Similarly, each plaquette syndrome $\ket{p}$ is surrounded by four data qudits, $\ket{h_1}$ through $\ket{h_4}$, which are numbered counterclockwise as shown.}
\label{fig:v-p-def}
\end{figure}

\begin{figure}
\centering
\includegraphics[width=0.7\textwidth]{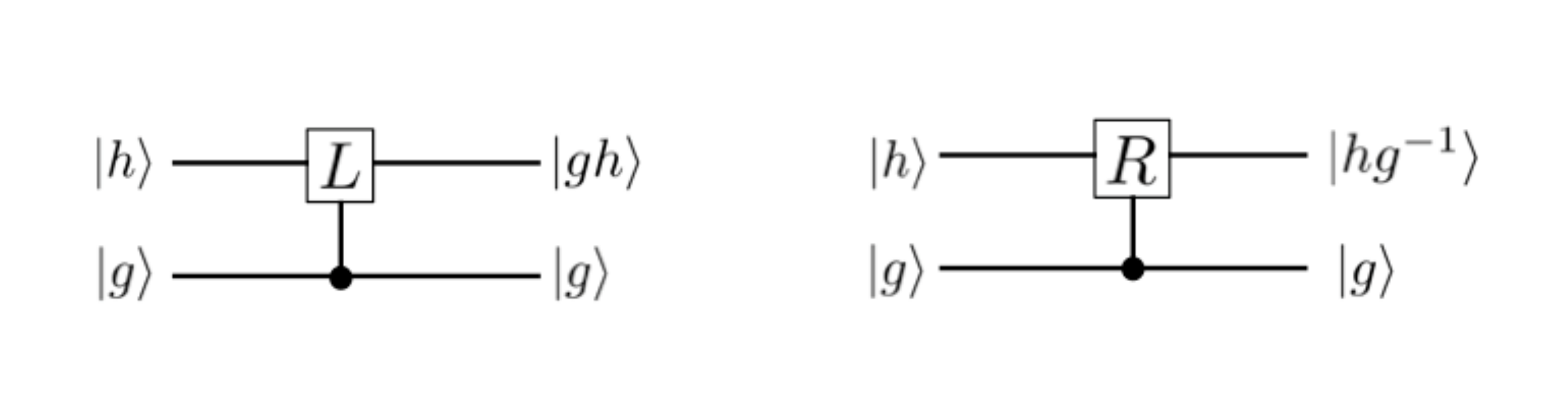}
\caption{Definition of the entangling left-multiplication and right-multiplication gates for two qudits in $\C[G]$, $G$ an arbitrary finite group. These are generalizations of the CNOT gate for qubits.}
\label{fig:gate-def}
\end{figure}

\floatsetup[figure]{style=plain,subcapbesideposition=bottom}
\begin{figure}[ht]
     \centering
        \sidesubfloat[]{%
            \includegraphics[width=0.62\textwidth]{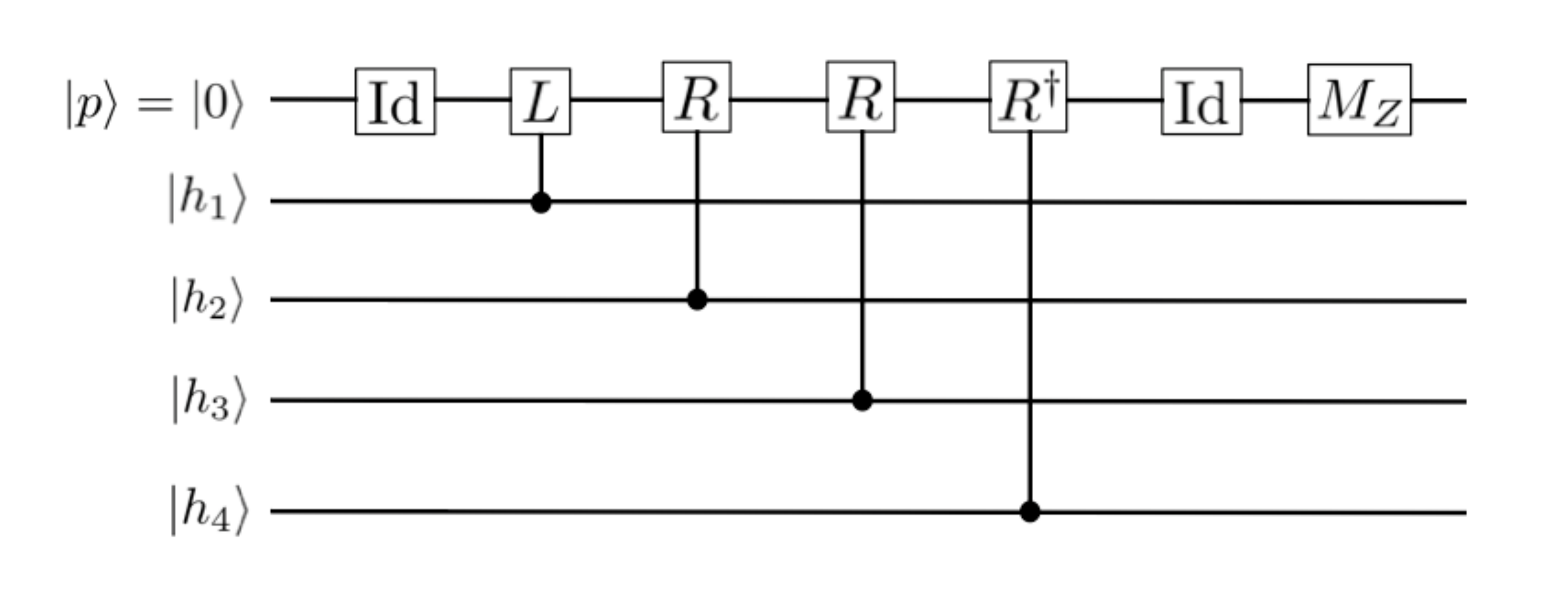}\label{fig:stabilizer-B}
        }\\
        \sidesubfloat[]{%
           \includegraphics[width=0.62\textwidth]{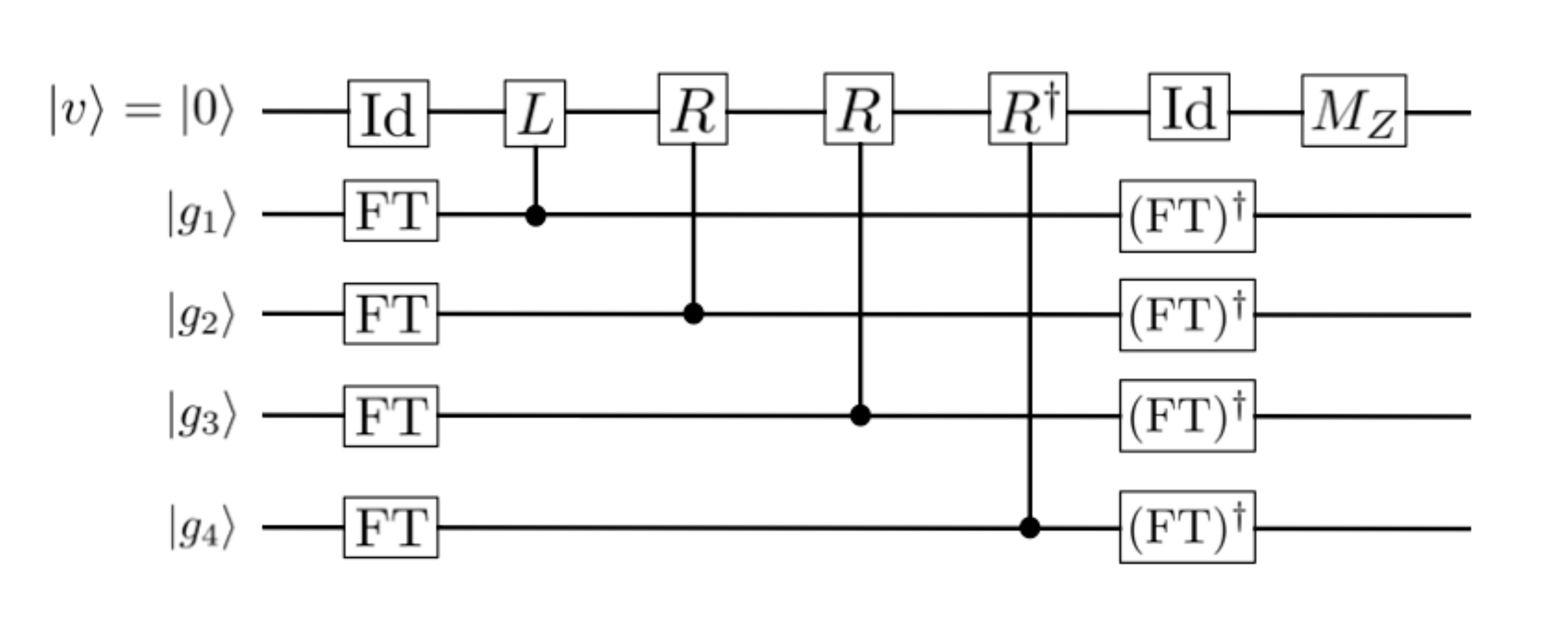}\label{fig:naive-stabilizer-A}
        }\\ 
		\sidesubfloat[]{%
           \includegraphics[width=0.62\textwidth]{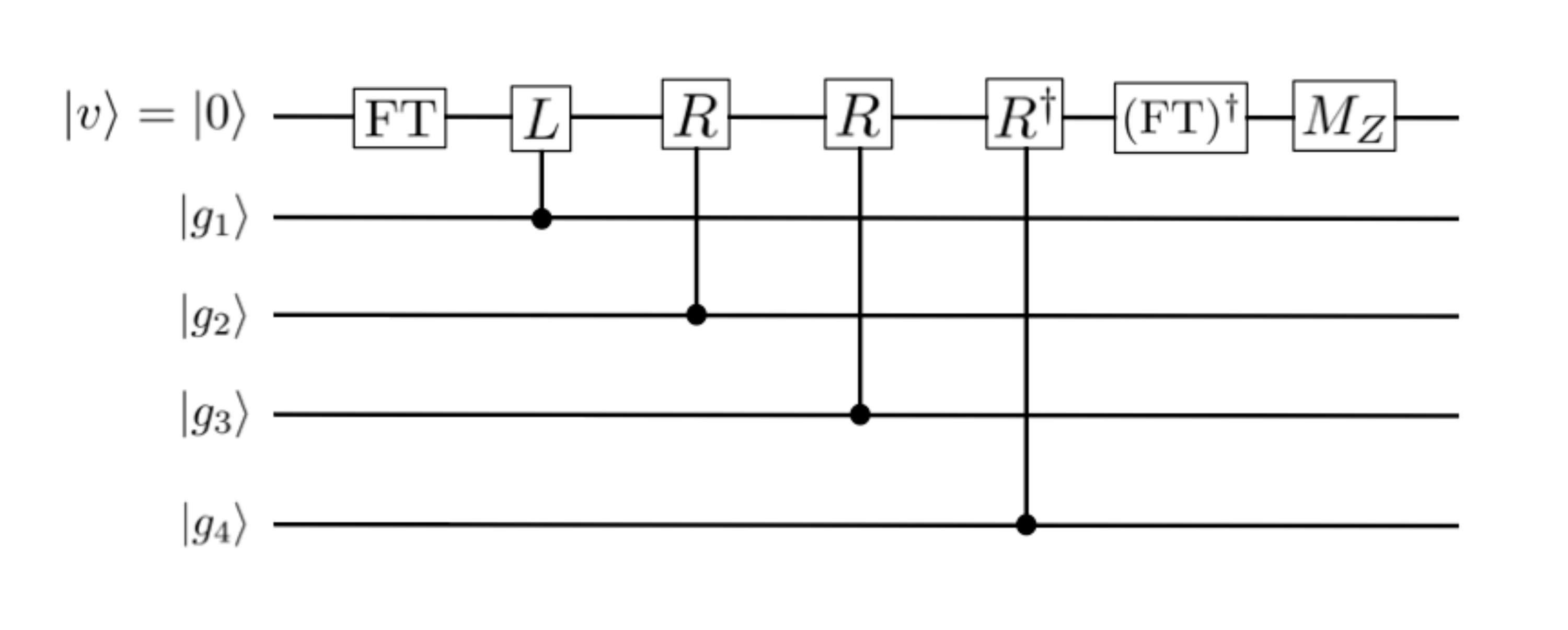}\label{fig:stabilizer-A}
        }
    \caption{(A) Stabilizer circuit for the $B(p)$ operator. (B) A naive implementation of a stabilizer circuit for the $A(v)$ operator, which simply takes conjugates all data qudits by Fourier transforms (denoted FT). It would not be possible to implement this in the same surface code cycle as the $B(p)$ operators, which do not conjugate the data qudits. (C) Stabilizer circuit for the $A(v)$ operator, which has an action equivalent to the one in (B). This one can be implemented in the same surface code cycle as the $B(p)$ circuit. The convention for the labels $\ket{g_i}$ and $\ket{h_i}$ is given in Fig. \ref{fig:v-p-def}, and the definitions of the controlled gates $L,R$ are given in Fig. \ref{fig:gate-def}.}%
\end{figure}

At each vertex and plaquette, we can now define short circuits (Figs. \ref{fig:stabilizer-B} and \ref{fig:stabilizer-A}) to project onto the eigenstates of the operators $B(p)$ and $A(v)$, respectively, as defined in Section \ref{sec:kitaev-hamiltonian}.

We begin with the operator $B(p)$ in Fig. \ref{fig:stabilizer-B}, as it is the simpler of the two. Suppose we are given a plaquette syndrome qudit $\ket{p}$ surrounded by data qudits $\ket{g_1}, \ket{g_2}, \ket{g_3}, \ket{g_4}$. We first initialize the syndrome qudit to $\ket{p} = \ket{1},$ corresponding to the identity element in $G$. We can then apply the generalized controlled-multiplication gates (defined in Fig. \ref{fig:gate-def}) to obtain the product $\ket{p} = \ket{g_1 g_2^{-1} g_3^{-1} g_4}$. Finally, we measure the $\ket{p}$ qudit in the $Z$ basis to project the data qudits onto an eigenstate of $B(p)$.

The operator $A(v)$ is implemented by a very similar circuit, as shown in Fig. \ref{fig:stabilizer-A}. As discussed in Section \ref{sec:kitaev-hamiltonian}, the $A(v)$ operator projects each vertex to a trivial representation sector, while $B(p)$ projects to trivial flux. Physically, this simply corresponds to conjugating all data qudits by the Fourier transform before the measurement cycle. The details on efficiently implementing generic quantum Fourier transforms can be found in Ref. \cite{Moore06}. Naively, this would result in the circuit shown in Fig. \ref{fig:naive-stabilizer-A}; however, a problem with this circuit is that it cannot be implemented in the same surface code cycle as the $B(p)$ circuit, which does not operate on conjugated data qudits. Fortunately, it is simple to show that by conjugating both the control and target qubits by a Fourier transform, one effectively switches the role of ``control'' and ``target'' in a controlled-multiplication gate. Hence, the circuit of Fig. \ref{fig:naive-stabilizer-A} is equivalent to the circuit in Fig. \ref{fig:stabilizer-A}, which can operate in the same surface code cycle as the $B(p)$ stabilizer circuit. 

Since any two of these stabilizer circuits commute, and since all Hamiltonian terms of $H_{(G,1)}$ commute, we can simultaneously apply all stabilizer circuits to project to a simultaneous eigenstate of all $A(v)$ and $B(p)$ terms.

\subsubsection{Boundary Hamiltonian surface code}

\begin{figure}
\centering
\includegraphics[width = 0.7\textwidth]{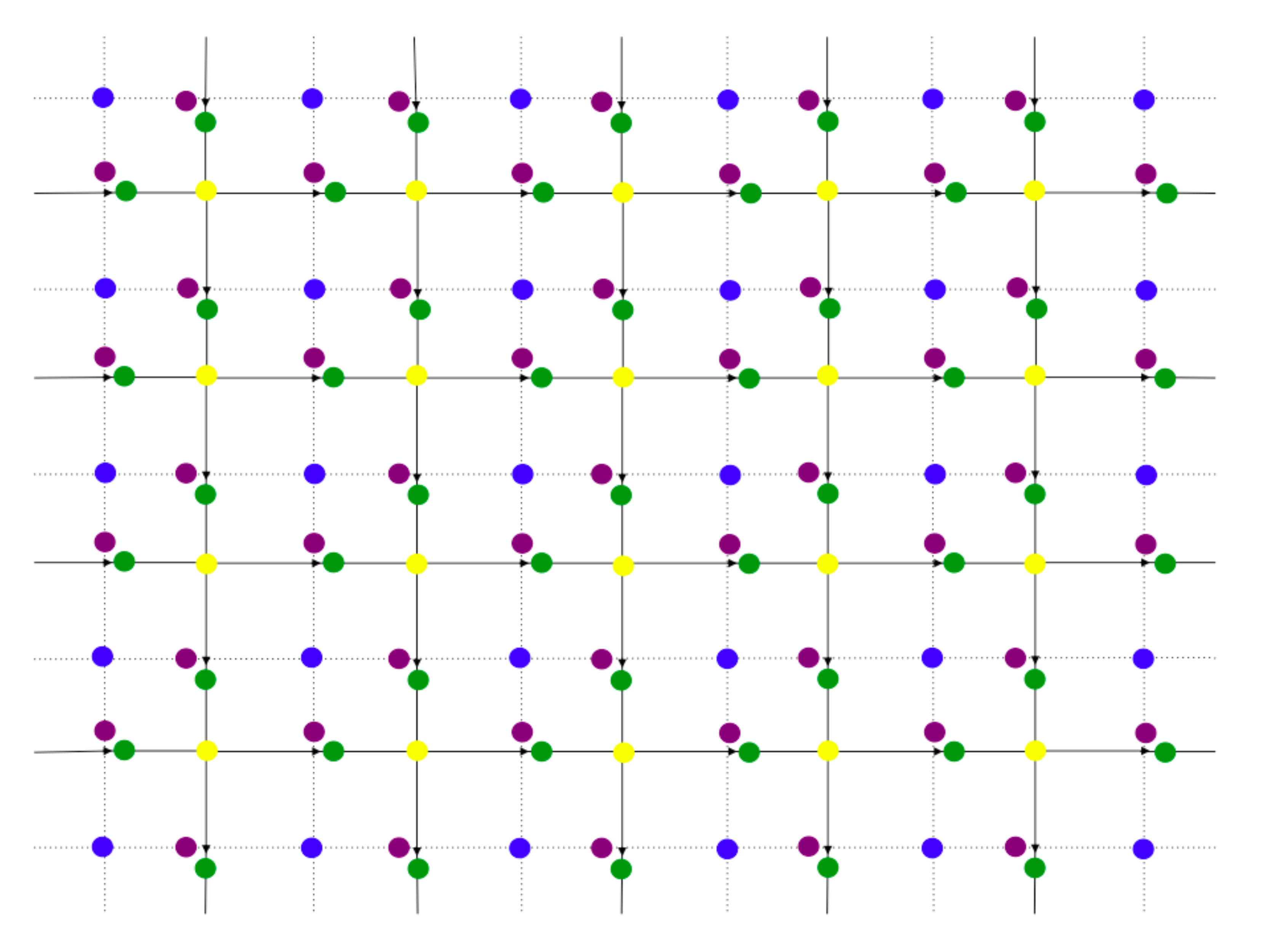}
\caption{Surface code layout that implements the boundary Hamiltonian $H_{(G,1)}^{(K,1)}$ of Eq. (\ref{eq:bd-hamiltonian-K}), for a given subgroup $K \subseteq G$. As before, data qudits and plaquette syndrome qudits take values in $\C[G]$. Vertex syndrome qudits are now modified to take values in $\C[K]$. We also add new $L$ and $T$ syndrome qudits to implement the string tension terms of $H_{(G,1)}^{(K,1)}$; the $L$ syndrome qudit takes values in $\C[K]$ while the $T$ syndrome qudit takes values in $\C[G]$.}
\label{fig:bd-surface-code-layout}
\end{figure}

We now present a surface code implementation of the boundary Hamiltonian (\ref{eq:bd-hamiltonian-K}). In this implementation, the data and plaquette syndrome qudits will take values in $\C[G]$ as before; however, the new vertex syndrome qudits take values in $\C[K]$. Furthermore, we must add two new syndrome qudits on each edge of the lattice: we add one {\it $L$ syndrome qudit}, which projects the system onto the eigenstate of the $L^K$ operator, and we add one {\it $T$ syndrome qudit} to project onto the eigenstate of the $T^K$ operator. The $L$ syndrome qudit takes values in $\C[K]$, while the $T$ syndrome qudit takes values in $\C[G]$. The new layout is illustrated in Fig. \ref{fig:bd-surface-code-layout}.

In this layout, the vertex and plaquette operators will project the system to eigenstates of the operators $A^K(v)$ and $B^K(p)$, respectively, instead of $A(v)$ and $B(p)$. The $T$ and $L$ stabilizer circuits are very similar, but are in fact much simpler than the vertex and plaquette stabilizers. Specifically, the $T^K$ operator enforces that the edge value is within the subspace spanned by $\{\ket{k}: k \in K\}$, and the $L^K$ operator projects each edge to a trivial representation when restricted to $K$. The stabilizer circuits for all of the boundary Hamiltonian terms are shown in Figs. \ref{fig:stabilizer-B-K}-\ref{fig:stabilizer-L}. All stabilizer circuits operate in lockstep.

\floatsetup[figure]{style=plain,subcapbesideposition=bottom}
\begin{figure}[ht]
     \centering
        \sidesubfloat[]{%
            \includegraphics[width=0.62\textwidth]{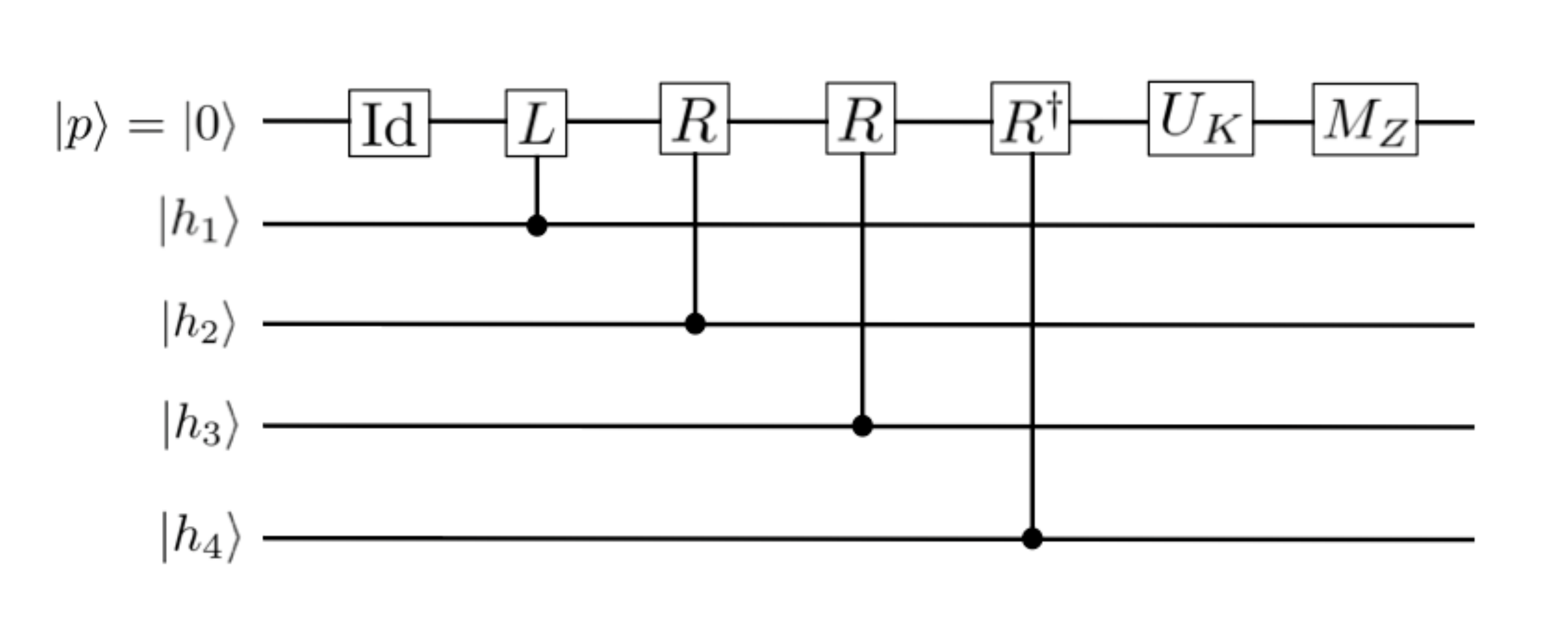}\label{fig:stabilizer-B-K}
        }\\%
        \sidesubfloat[]{%
            \includegraphics[width=0.62\textwidth]{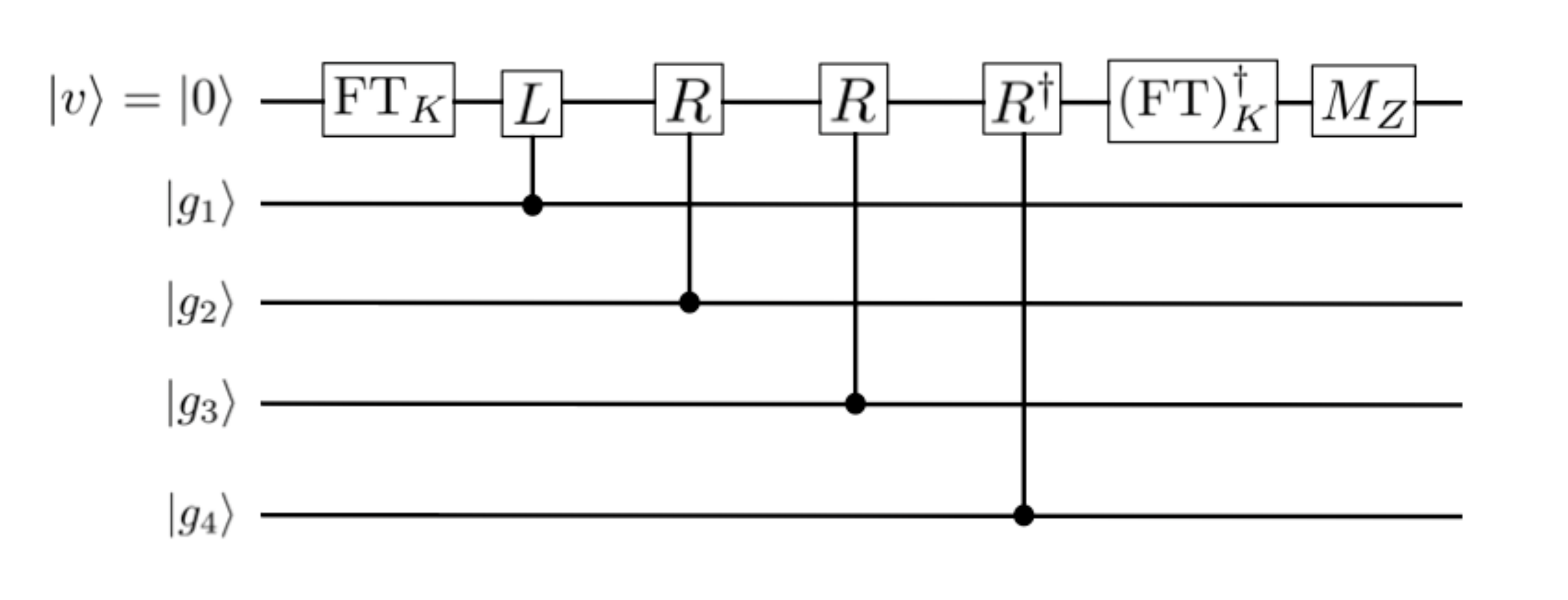}\label{fig:stabilizer-A-K}
        }\\%
        \sidesubfloat[]{%
            \includegraphics[width=0.4\textwidth]{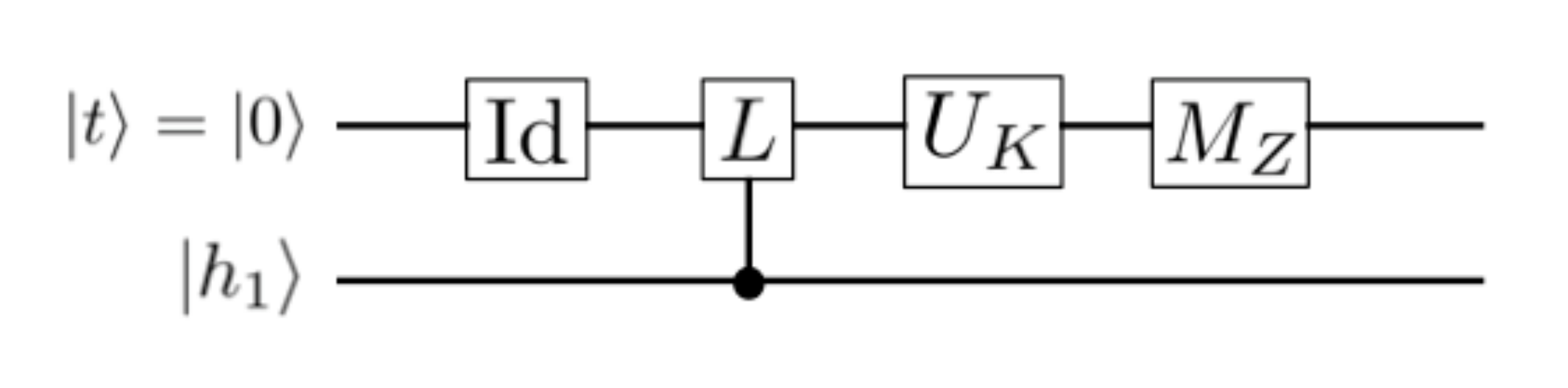}\label{fig:stabilizer-T}
        }\\%
        \sidesubfloat[]{%
           \includegraphics[width=0.4\textwidth]{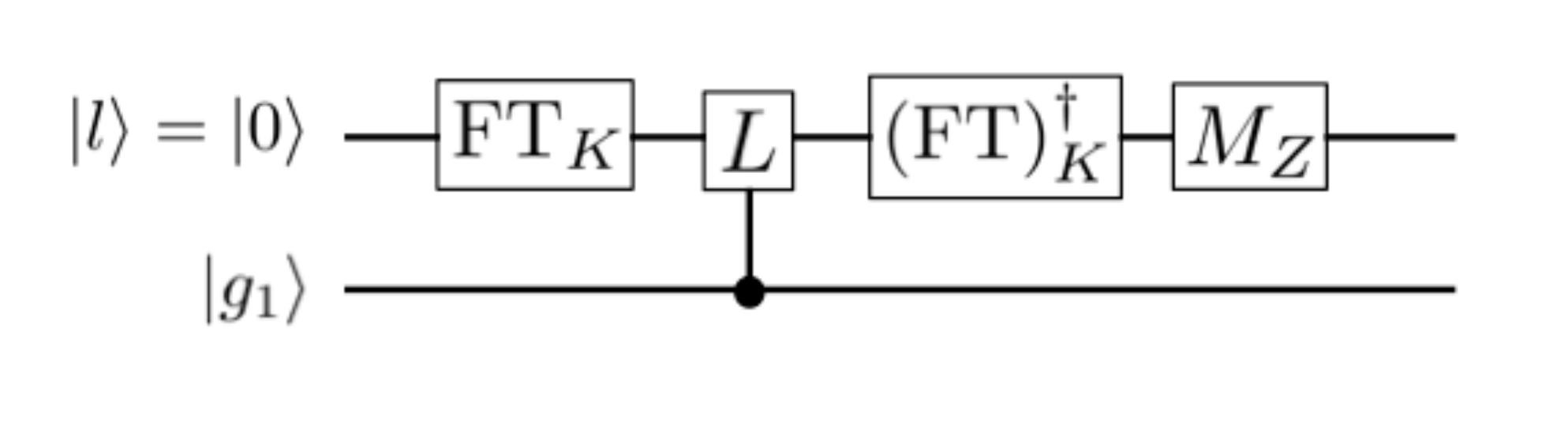}\label{fig:stabilizer-L}
        } 
    \caption{(A) Stabilizer circuit for the $B^K(p)$ operator. (B) Stabilizer circuit for the $A^K(v)$ operator. (C) Stabilizer circuit for the $T^K(e)$ operator. (D) Stabilizer circuit for the $L^K(e)$ operator. The operator $U_K$ is defined to have the following action on a qudit in $\C[G]$: $U_K \sum_{g \in G} a_g \ket{g} = \sum_{i=1}^{m} \sqrt{\sum_{j \in h_iK} a_j^2} \ket{h_i}$, where $\{h_i\}_{i=1}^m$ form a set of representatives of the left cosets $G/K$.}
\end{figure}

\subsection{Anyons in the surface code model}

\begin{figure}
\centering
\includegraphics[width = 0.7\textwidth]{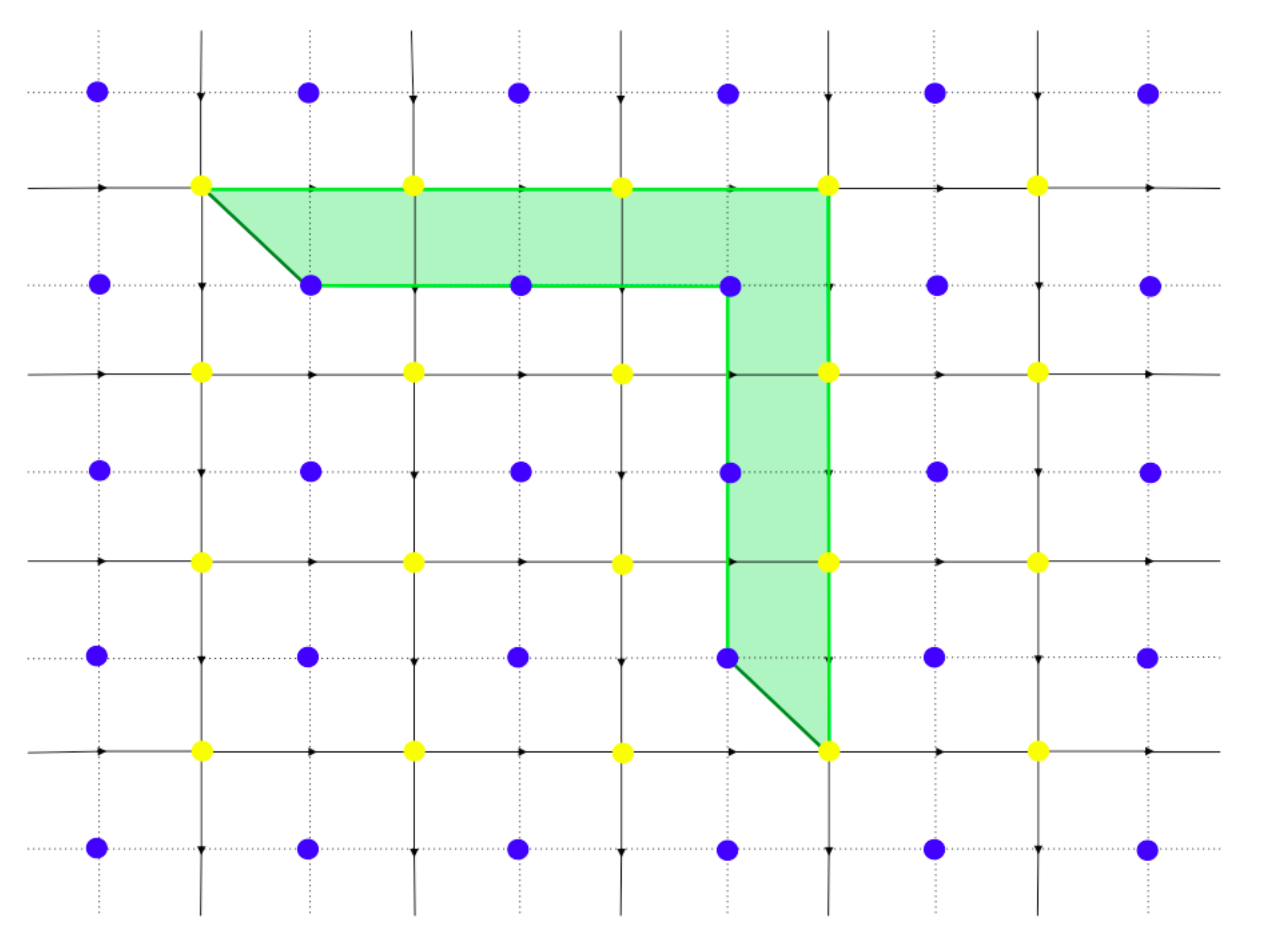}
\caption{Excitations in the surface code model for the bulk Hamiltonian $H_{(G,1)}$. Excited states can result from measuring a syndrome qudit in an excited eigenstate, or from applying a ribbon operator on data qudits. The latter case is illustrated, where we apply a ribbon operator on the light green ribbon to create a pair of excitations at the cilia marked by the dark green lines.}
\label{fig:surface-code-excitations}
\end{figure}

In both the bulk and boundary surface code models, the ground state of the Hamiltonians of Equations (\ref{eq:kitaev-hamiltonian}) and (\ref{eq:bd-hamiltonian-K}) correspond to measuring all syndrome qudits in the state $\ket{0}$. If one of the syndromes is consistently measured to be in a different state, an excitation is located at the corresponding vertex or plaquette (or more generally, the cilium). For instance, suppose one plaquette syndrome $p$ is consistently measured to have value $\ket{h}$, and one of its neighboring vertex syndromes $v$ is consistently measured to have value $\ket{g}$. From the value of $p$, by Eq. (\ref{eq:bulk-ribbon-FT}), we may deduce that there is an excitation at the cilium $(v,p)$ whose conjugacy class $C$ is the conjugacy class of $h$. Likewise, from the value of $v$, we may deduce the representation $\pi$ of the centralizer of $C$ corresponding to the excitation. 

Excitations may also be created in this model by applying one of the ribbon operators from Section \ref{sec:ribbon-operators} on all data qudits on a specific ribbon of the surface code lattice. This is illustrated in Fig. \ref{fig:surface-code-excitations}.

These similar properties also hold for the case of the boundary Hamiltonian $H^{(K,1)}_{(G,1)}$: excitations can result from syndromes in excited eigenstates, or from applying a boundary ribbon operator from Section \ref{sec:bd-excitations}.

\subsection{Creation and annihilation of gapped boundaries}
\label{sec:circuit-creation-annihilation}

We now discuss how to create and annihilate gapped boundaries on the surface code. This provides a physical implementation of the Hamiltonian $H_{\text{G.B.}}$ (\ref{eq:gapped-bds-hamiltonian}), and will allow us to perform quantum computation in the following chapters. The methods here are generalizations of those developed in Ref. \cite{Fowler12} to create the $1+e$ and $1+m$ boundaries of the $\Z_2$ toric code.

\begin{figure}
\centering
\includegraphics[width = 0.65\textwidth]{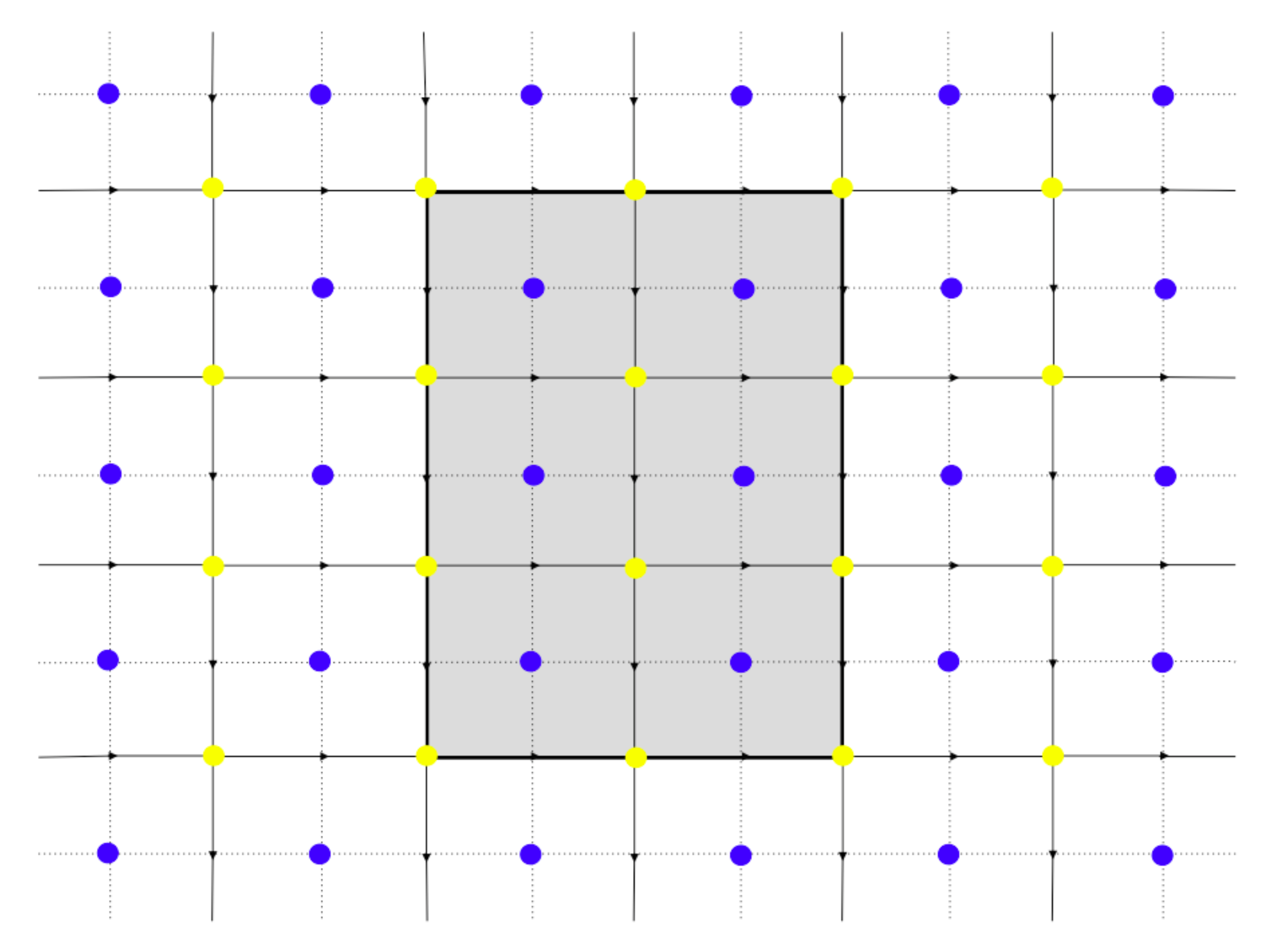}
\caption{Gapped boundary in the surface code context. We would like to be able to arbitrarily create or annihilate a subgroup $K$ gapped boundary corresponding to the pictured hole.}
\label{fig:bd-creation-1}
\end{figure}

\begin{figure}
\centering
\includegraphics[width = 0.65\textwidth]{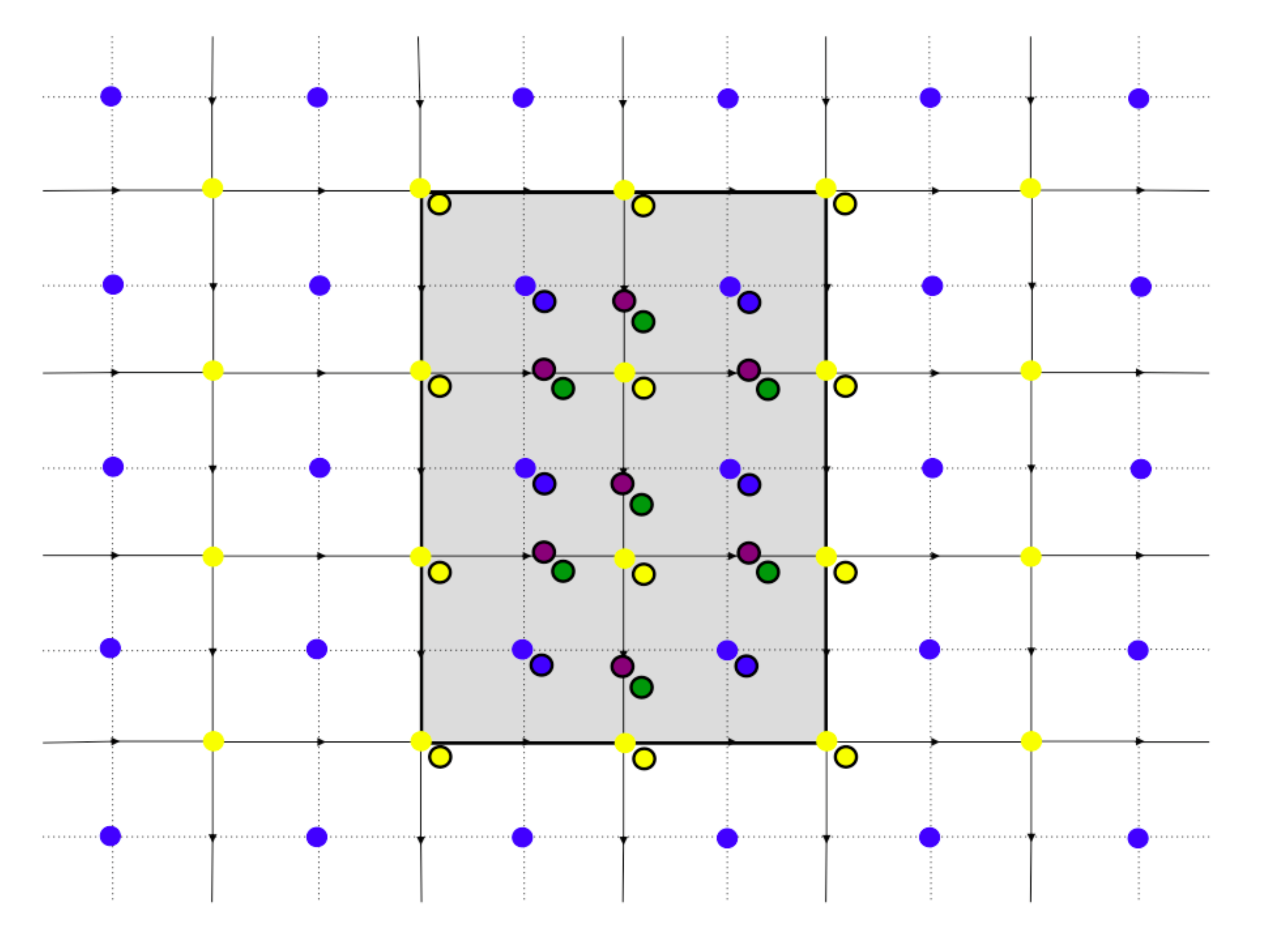}
\caption{Circuit layout that allows for arbitrary creation and annihilation of the subgroup $K$ gapped boundary of Fig. \ref{fig:bd-creation-1}. We place an extra set of syndrome qudits and stabilizer circuits on the shaded area, corresponding to the Hamiltonian $H^{(K,1)}_{(G,1)}$ (shown with black borders, slightly displaced from their usual position as in Fig. \ref{fig:bd-surface-code-layout}). These circuits are turned on to create the boundary, and turned off when we annihilate it.}
\label{fig:bd-creation-2}
\end{figure}

\begin{figure}
\centering
\includegraphics[width = 0.82\textwidth]{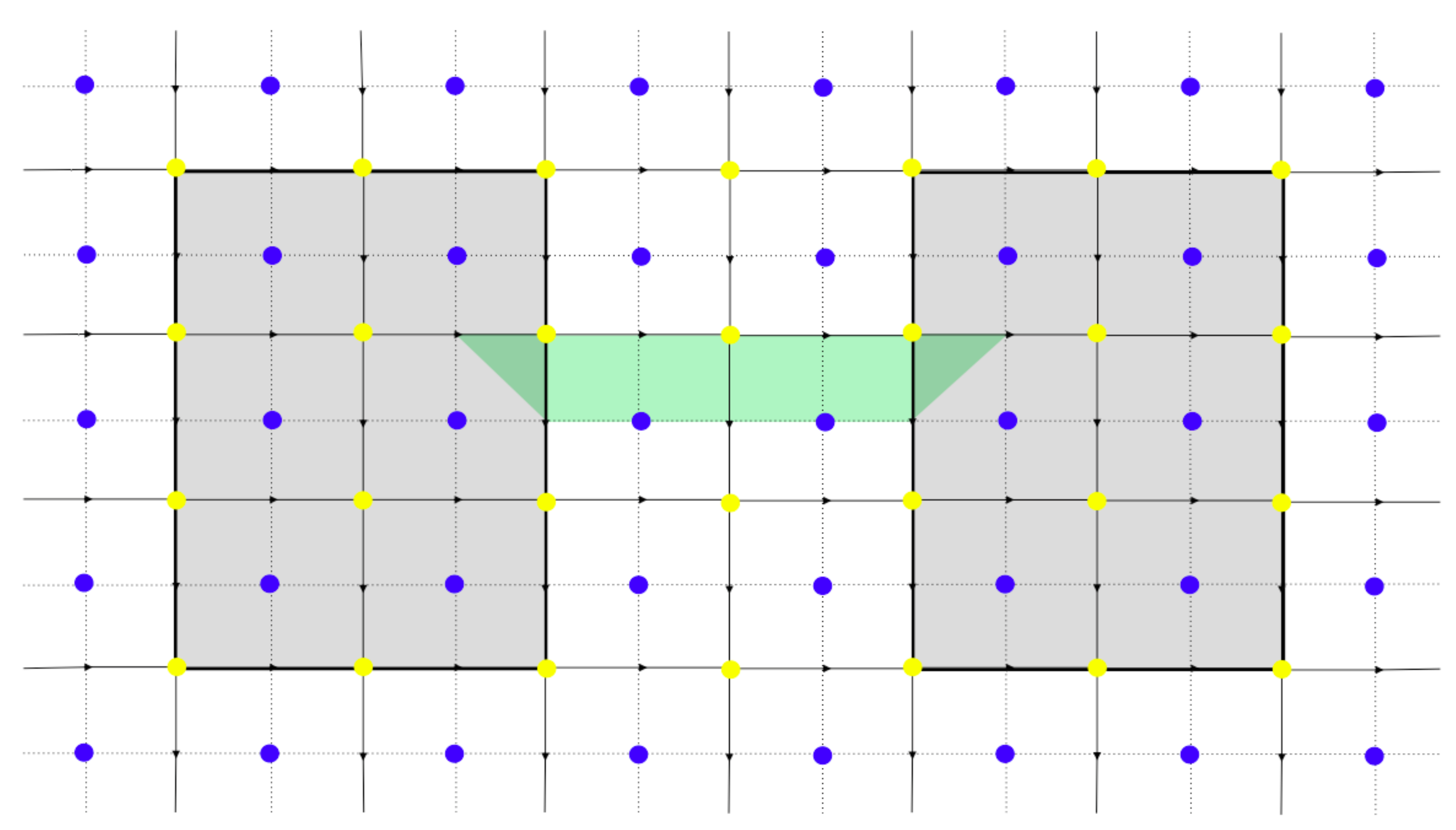}
\caption{A gapped boundary qudit in the context of surface code implementation. Ribbon operators on the green ribbon correspond to Wilson operators $W_a(\gamma)$, which were introduced in Chapters \ref{sec:hamiltonian} and \ref{sec:algebraic}.}
\label{fig:gapped-boundary-qudit}
\end{figure}

Let us first introduce the circuit layout required to arbitrarily create or annihilate a gapped boundary. Suppose we are in the surface code state shown in Fig. \ref{fig:bd-creation-1}, and would like to, at some point, create and annihilate a gapped boundary corresponding to the subgroup $K \subseteq G$ along the boldfaced edges (i.e. create a hole within the shaded square). Then, we must be able to implement the Hamiltonian $H_{(G,1)}^{(K,1)}$ on all data qudits within and on the boundary of this square, or at least on all data qudits along the boundary ribbon corresponding to this boundary (see Remark \ref{bd-ribbon-def}).

To do this while ensuring maximal flexibility, we will place two sets of stabilizer circuits at the qudits of every ribbon along which we may have a subgroup $K$ boundary: one for the bulk Hamiltonian $H_{(G,1)}$, and one for the boundary Hamiltonian $H_{(G,1)}^{(K,1)}$. This is shown in Fig. \ref{fig:bd-creation-2}. (If we potentially want to create multiple types of boundaries, we will place multiple sets of stabilizer circuits, one for each subgroup.) Initially, we start in the state where only the bulk Hamiltonian circuits are turned on. This creates a lattice in which we have no internal boundaries.

\subsubsection{Gapped boundary initialization}


We would like to initialize the gapped boundary shown in Fig. \ref{fig:bd-creation-2}. To do so, we can simply turn off all stabilizer circuits within and on the boldface square. We then turn on all (subgroup $K$) boundary Hamiltonian stabilizer circuits, on exactly these qudits. This effectively implements the Hamiltonian (\ref{eq:gapped-bds-hamiltonian}), in the case where we have precisely this single hole. By Remark \ref{bd-ribbon-def}, we may equivalently only turn on stabilizers that touch the boundary ribbon, if that is preferable.

Suppose we would now like to initialize a logical qudit, which is comprised of two gapped boundaries (see Fig. \ref{fig:gapped-boundary-basis}). This is shown in the surface code context in Fig. \ref{fig:gapped-boundary-qudit}. We would like to initialize the qudit in the ground state $\ket{0}$, where no particles have condensed into either boundary. To do so, we first use ribbon operators to move any excitations within the boundary outside, into the bulk. (Alternatively, one can track this data in surface code software.) Next, we apply the procedure to create both gapped boundaries. This would implement the Hamiltonian $H_{\text{G.B.}}$ of Eq. (\ref{eq:gapped-bds-hamiltonian}) in a case where we have exactly those two holes.

It is clear now that the Wilson line operators $W_a(\gamma)$ introduced in Chapters \ref{sec:hamiltonian} and \ref{sec:algebraic} simply correspond to applying ribbon operators on the green ribbon (or a topologically equivalent one) in Fig. \ref{fig:gapped-boundary-qudit}. Similarly, Wilson loop operators, which will be introduced in Section \ref{sec:loop}, correspond to applying a closed ribbon operator on a ribbon encircling one of the holes.

We note that this procedure initializes a hole whose sides are on the direct lattice, which corresponds to the $X$-cut hole in Ref. \cite{Fowler12}. The generalization to the case where sides can be on the dual lattice is obvious.


\subsubsection{Gapped boundary measurement}

We now discuss how to measure a gapped boundary to annihilate it. In general, if the gapped boundary is not being used for computation purposes, we may simply turn off the stabilizers corresponding to $H_{(G,1)}^{(K,1)}$ and turn on those corresponding to the bulk Hamiltonian.

If instead we have a qudit as in Fig. \ref{fig:gapped-boundary-qudit}, it is also useful to measure the value of the qudit before destroying the holes. This can be done as follows: We take two ancilla qudits, $\ket{a},\ket{b} \in \C[G]$, and initialize them to identity. We then turn on all stabilizer circuits for the bulk Hamiltonian, but omit the final measurement step in each circuit. We then take $\ket{a}$ to be the product of all vertex stabilizers for $H_{(G,1)}$, and we take $\ket{b}$ to be the product of all plaquette stabilizers for $H_{(G,1)}$. We now measure $\ket{a}$ and $\ket{b}$ in the $Z$ basis, which will tell us which particles have condensed to the boundary. This gives us a way to measure the $Z$ eigenstate of the logical qudit. (Specifically, $\ket{b}$ gives the magnetic flux $C$, and $\ket{a}$ gives the electric charge $\pi$). Finally, we completely turn on the bulk Hamiltonian stabilizer circuits, and then turn off the stabilizers for $H_{(G,1)}^{(K,1)}$.

\subsection{Moving gapped boundaries}
\label{sec:gapped-bd-moving}

In this section, we describe how to arbitrarily move a gapped boundary qudit hole. This is an essential task: it is required for gapped boundary braiding, which gives entangling gates. This procedure generalizes but closely follows the procedure described by Fowler et al. in Ref. \cite{Fowler12}. As in that case, we will consider the Heisenberg picture, and focus on the transformations of the Wilson operators and the measure $Z$ operator.

\begin{figure}
\centering
\includegraphics[width = 0.7\textwidth]{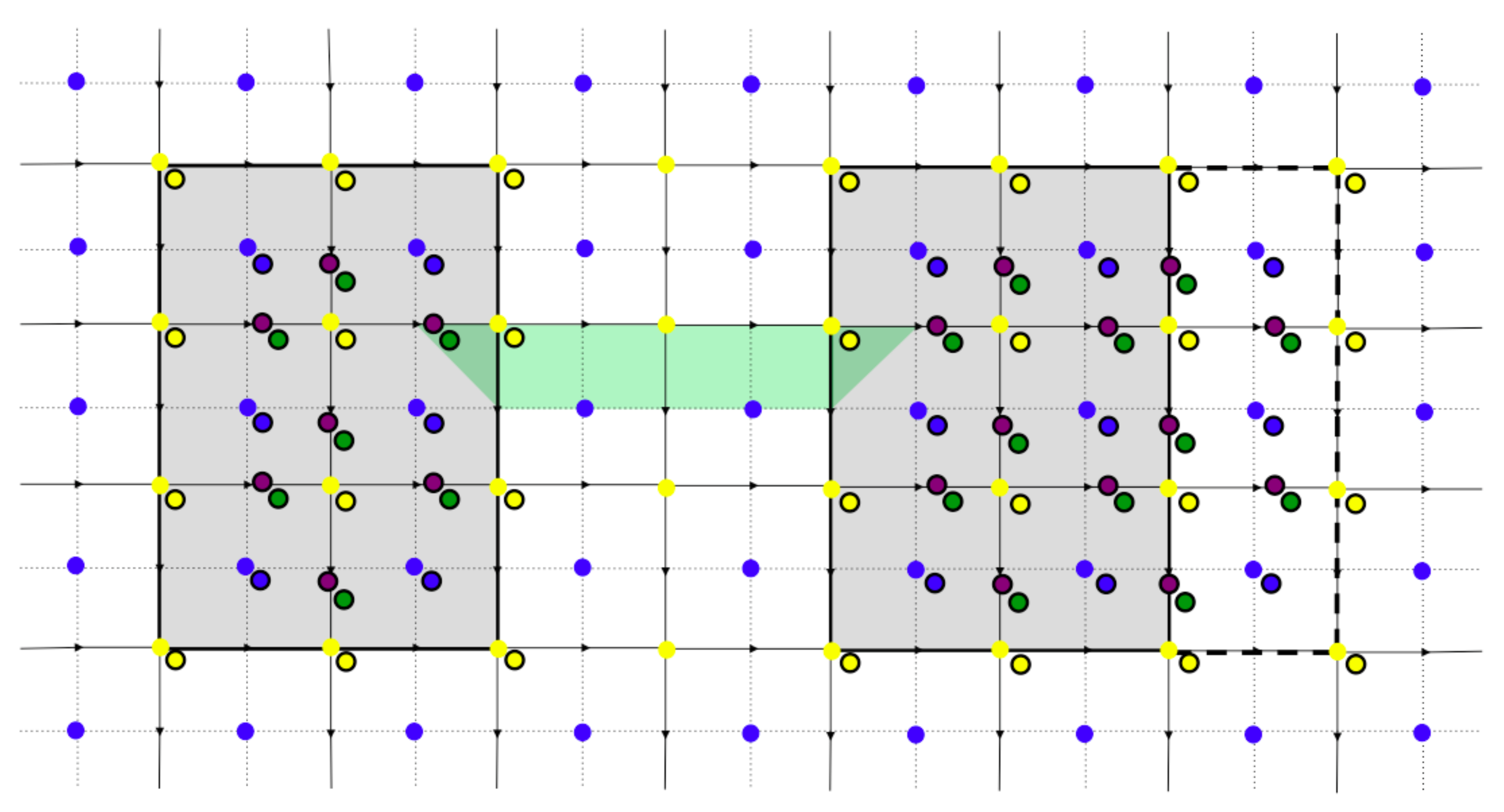}
\caption{Enlarging a gapped boundary in the surface code context. See text for implementation details.}
\label{fig:gapped-bd-enlarging}
\end{figure}

\begin{figure}
\centering
\includegraphics[width = 0.7\textwidth]{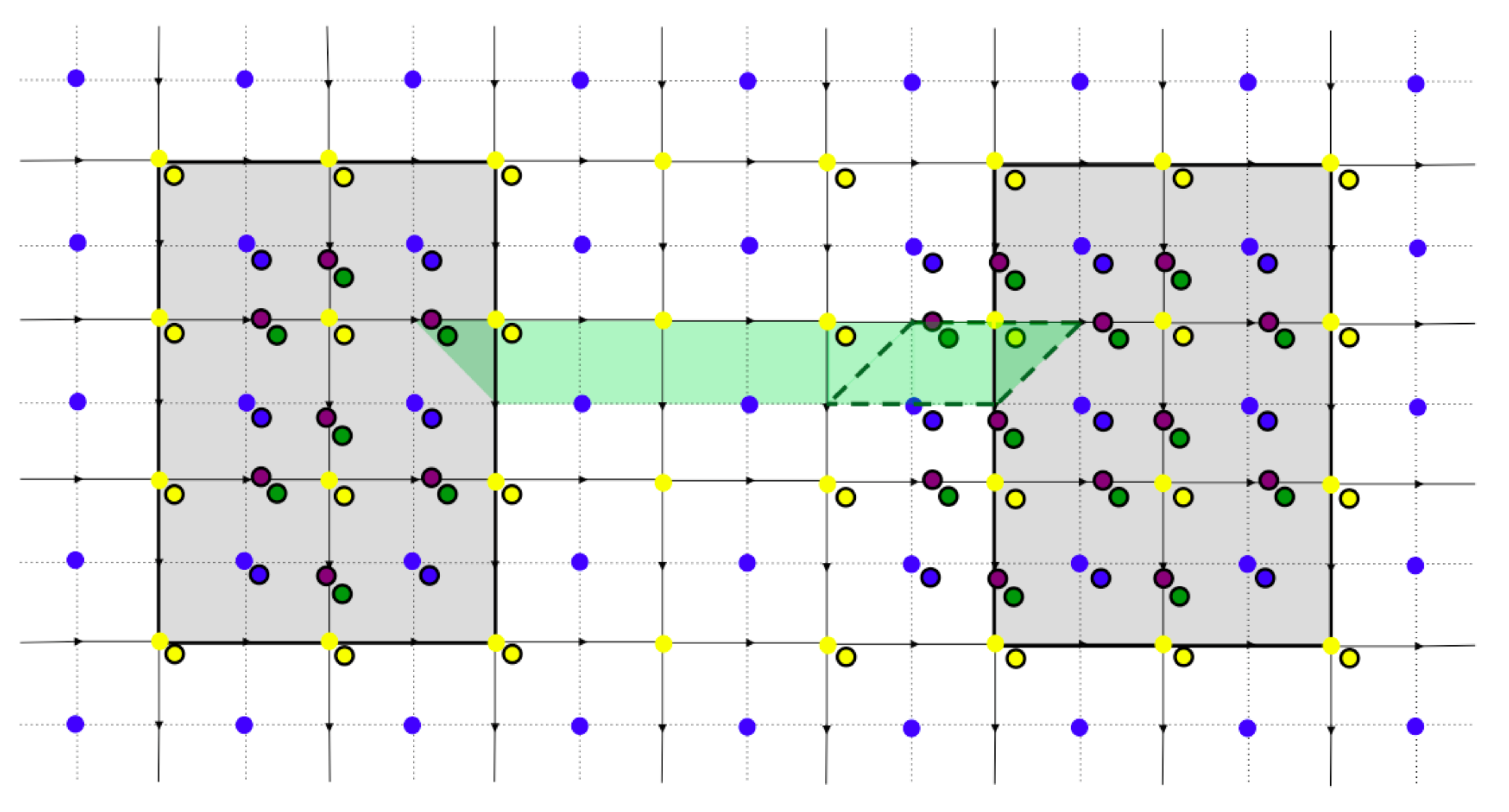}
\caption{Shrinking a gapped boundary in the surface code context. See text for implementation details.}
\label{fig:gapped-bd-shrinking}
\end{figure}

The procedure to move a logical qudit hole horizontally to the right by one cell is depicted in Figs. \ref{fig:gapped-bd-enlarging} and \ref{fig:gapped-bd-shrinking}; it consists of first enlarging the hole by one cell width, and then shrinking it by one cell width. We begin by extending all loop operators (including the measure-$Z$ operator) to enclose all cells that will be in the enlarged hole (dotted line in Fig. \ref{fig:gapped-bd-enlarging}). We then wait until the current surface code cycle completes. For the next cycle, we turn off the $H_{(G,1)}$ stabilizer circuits for the cells that will now be in the hole, and turn on the $H_{(G,1)}^{(K,1)}$ circuits for these cells. This is a general step that can always be taken to enlarge a gapped boundary in the surface code implementation.

The next step is to shrink the gapped boundary, as shown in Fig. \ref{fig:gapped-bd-shrinking}. To do this, we define the new tunneling ribbon operator by extending it by one cell (shown in dotted green lines in Fig. \ref{fig:gapped-bd-shrinking}). In the language of Chapter \ref{sec:hamiltonian}, this is exactly one dual and one direct triangle. In the next surface code cycle, we then turn off the $H_{(G,1)}^{(K,1)}$ stabilizer circuits and turn on the $H_{(G,1)}$ circuits for all of the qudits that will no longer be in the hole. To establish the stabilizer values in time, we have to wait another $d-1$ surface code cycles, where $d$ is the separation distance between the two gapped boundaries \cite{Fowler12}.

As shown in Refs. \cite{Fowler12} and \cite{Raussendorf03}, \lq\lq byproduct operators" may result from performing this movement, and should be accounted for in the surface code control software. These are relatively simple generalizations of Refs. \cite{Fowler12} and \cite{Raussendorf03} and will not be presented here.

\vspace{2mm}
\section{Topologically protected operations}
\label{sec:operations}

In this chapter, we present the topologically protected operations on the qudit encoded in the basis of Fig. \ref{fig:gapped-boundary-basis}. There are 5 such types of operations, which are presented in Sections \ref{sec:tunnel}-\ref{sec:physically-implementable}. In Sections \ref{sec:tc-operations} and \ref{sec:ds3-operations}, we demonstrate the gates we can obtain in this way for the toric code and $\mfD(S_3)$ examples.

Throughout most of the chapter, the operations will be presented in the categorical language of Chapter \ref{sec:algebraic} for simplicity, since we have already shown that it is equivalent to the Hamiltonians of Chapter \ref{sec:hamiltonian} or the quantum circuits of Chapter \ref{sec:circuits} in the cases of interest (group models).

For simplicity of presentation, we assume throughout the chapter that there are no condensation multiplicities. The generalizations are obvious.

\subsection{Tunnel-$a$ operations}
\label{sec:tunnel}

The first topological operation we consider is to tunnel an elementary excitation $a$ from one gapped boundary to another. Physically, this corresponds to applying the $a$ ribbon operator to a ribbon connecting the two gapped boundaries. We will denote this operation by $W_a(\gamma)$, where $\gamma$ is this ribbon. $W_a(\gamma)$ is often known as a {\it Wilson line operator}.

Suppose we have two gapped boundaries given by Lagrangian algebras $\A_1,\A_2$, which encode a qudit with orthonormal basis as in Fig. \ref{fig:gapped-boundary-basis} of Section \ref{sec:algebraic-gsd}. We would like to compute the result of applying each $W_a(\gamma)$ on each basis element of the ground state $\Hom(1, \A_1 \otimes \A_2)$, and express the result in terms of the original basis.

Let us consider an arbitrary basis element $W_b(\gamma)\ket{0}$ as described in Section \ref{sec:algebraic-gsd}. Diagrammatically, after applying the $W_a(\gamma)$ operator, we have arrived in the following state:

\begin{equation}
\label{eq:tunnel-1}
\vcenter{\hbox{\includegraphics[width = 0.55\textwidth]{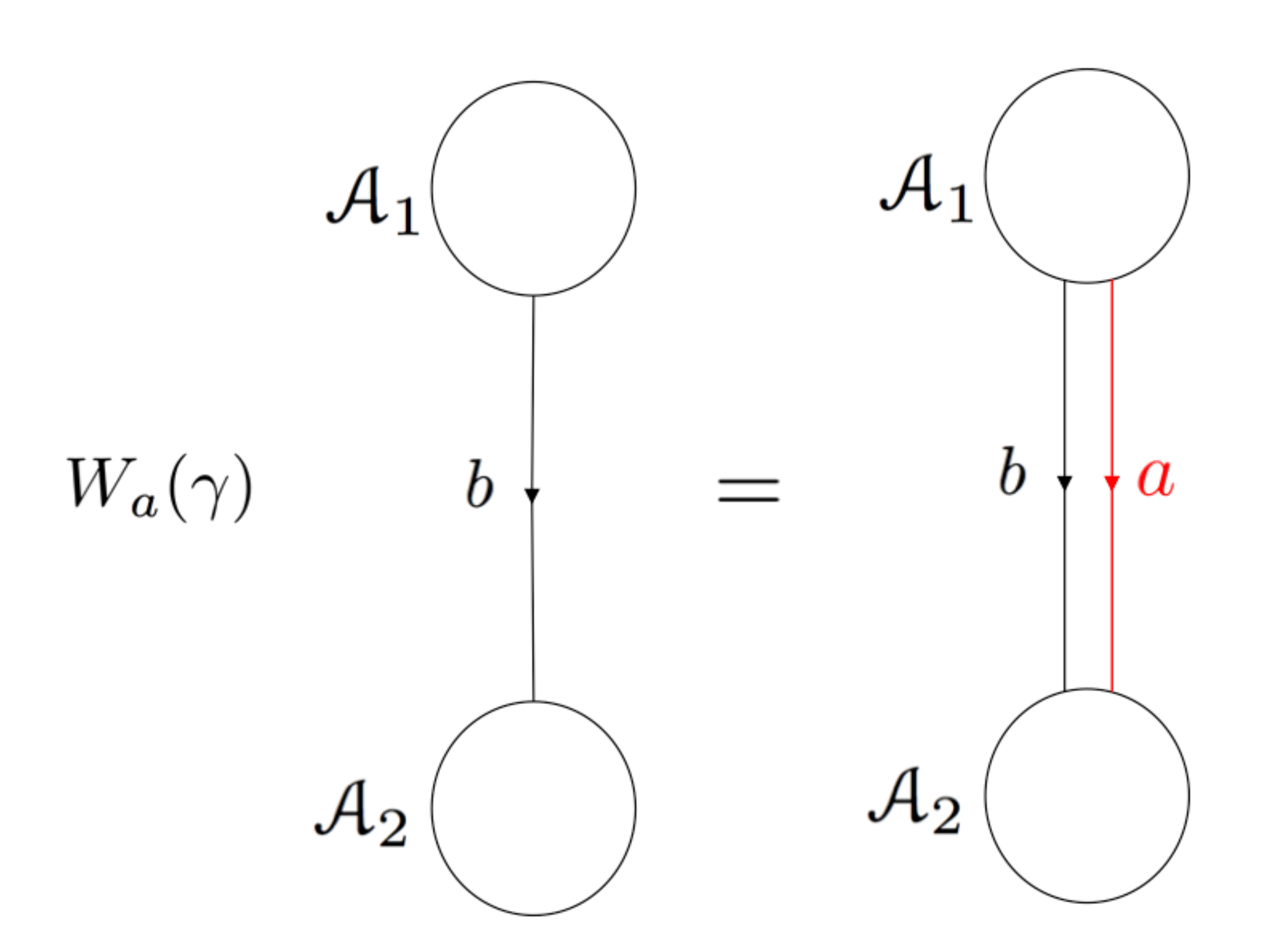}}}
\end{equation}

\noindent
Here, and for the rest of the chapter, solid black lines are used to indicate a basis element of the hom-space that describes the ground state, while solid red lines are used to denote Wilson operators.

To express this in terms of our original basis, we must convert the two anyon-tunneling ribbon operators into one. To do this, we can first apply the $M$-3$j$ operator and its Hermitian conjugate to the bottom and top boundaries of (\ref{eq:tunnel-1}), respectively, to get:\footnote{In this analysis, we will drop the multiplicity indices $\mu,\nu,\lambda$ from the $M$ symbols for concision; the generalization is obvious.}

\begin{equation}
\label{eq:tunnel-2}
\vcenter{\hbox{\includegraphics[width = 0.75\textwidth]{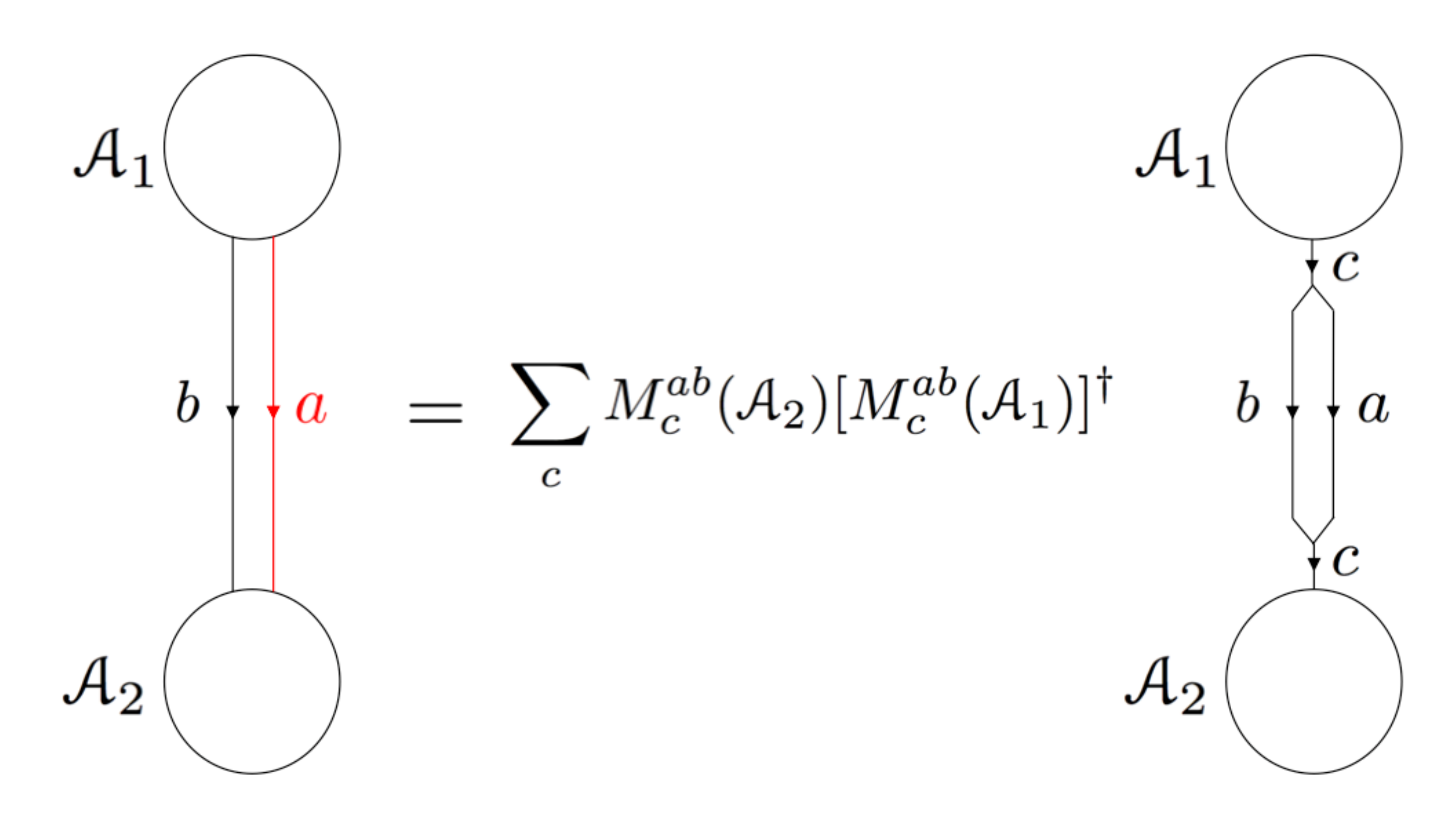}}}
\end{equation}

\noindent
Here, $M^{ab}_c(\A_i)$ indicates that the $M$-3$j$ symbol is for the gapped boundary given by the Lagrangian algebra $\A_i$.

We are now left with a bubble in the bulk. This can be eliminated using $\theta$ symbols of the bulk modular tensor category, by the following relation:

\begin{equation}
\label{eq:tunnel-3}
\includegraphics[width = 0.6\textwidth]{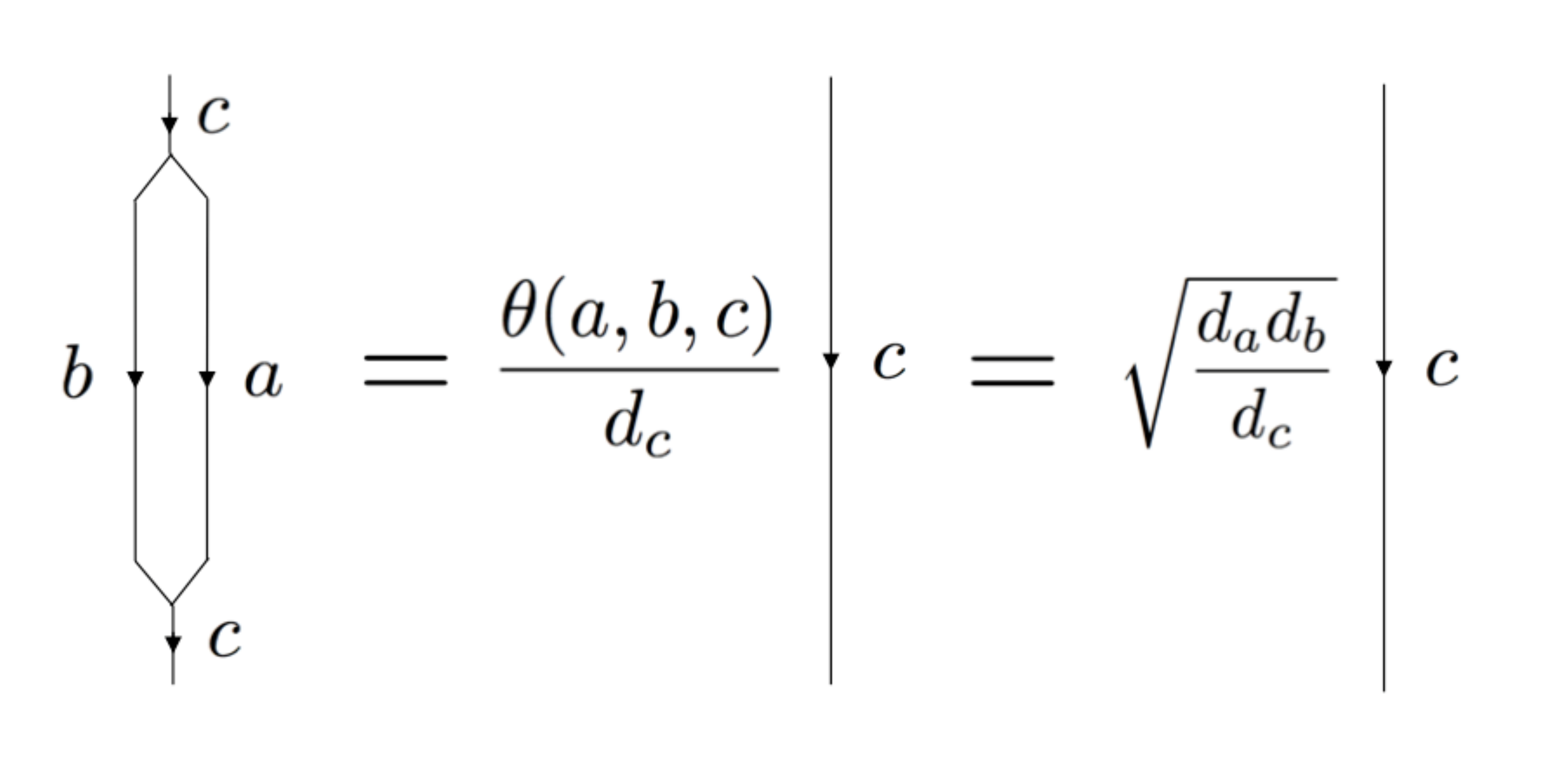}
\end{equation}

Hence, we have the following equation:

\begin{equation}
\label{eq:tunnel-formula}
W_a(\gamma) W_b(\gamma) \ket{0} =
W_a(\gamma) \vcenter{\hbox{\includegraphics[width = 0.12\textwidth]{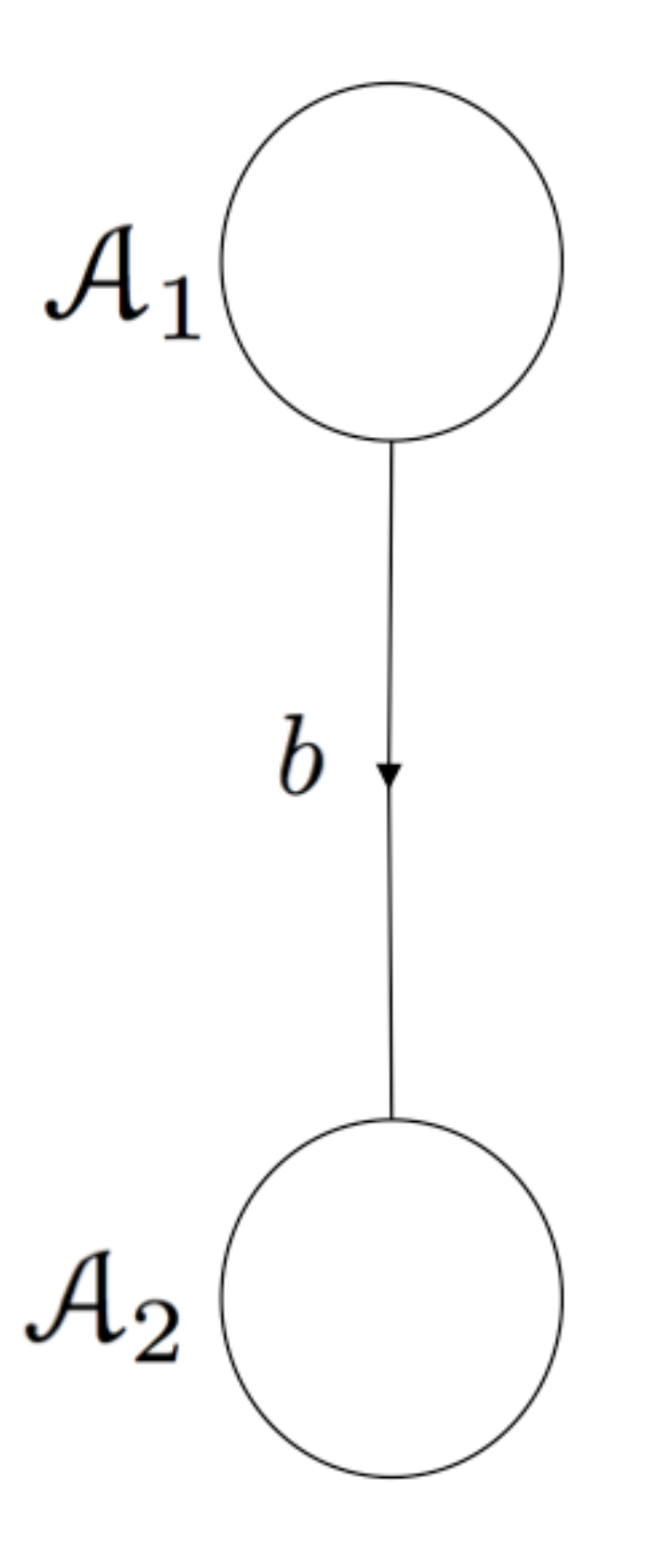}}} 
= \sum_c M^{ab}_c (\A_1) [M^{ab}_c]^\dagger (\A_2) \sqrt{\frac{d_a d_b}{d_c}}
\vcenter{\hbox{\includegraphics[width = 0.12\textwidth]{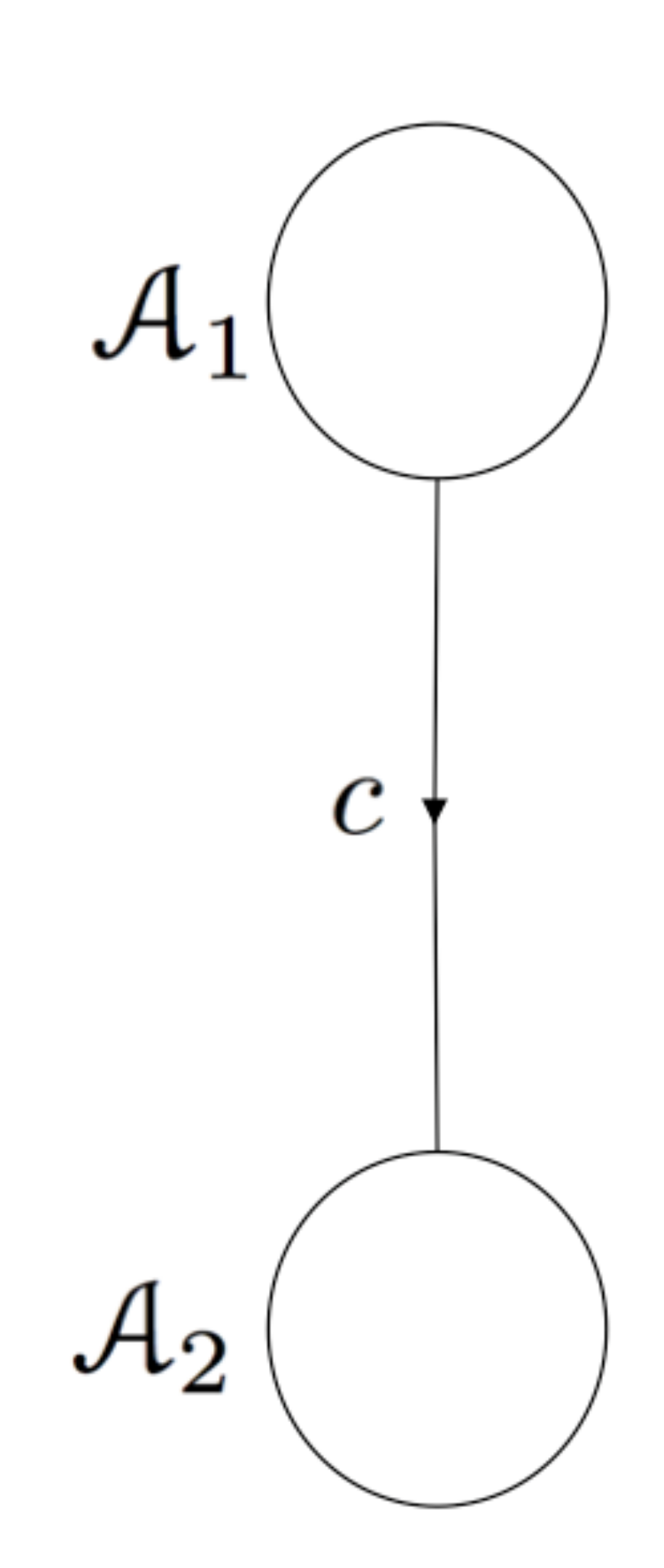}}} 
\end{equation}

In general, especially when we consider specific examples, we will  express $W_a(\gamma)$ as a $d\times d$ matrix that acts on the ground state $\Hom(1, \A_1 \otimes \A_2)$.

\begin{remark}
\label{tunnel-rmk}
We would like to note that in general, the tunneling operator $W_a(\gamma)$ need not be Hermitian or unitary. Specifically, $W_a(\gamma)$ is Hermitian if and only if the two boundaries are given by the same Lagrangian algebra and $a$ is self-dual. $W_a(\gamma)$ is unitary if and only if $n^a_\alpha(\A_i) = 0$ for all $\alpha$ not equal to vacuum. ($n^a_\alpha(\A_i) = 0$ is the coefficient of $\alpha$ in the decomposition of $a$ after condensation on boundary $\A_i$.)
\end{remark}

\subsection{Loop-$a$ operations}
\label{sec:loop}

The next topological operation we can consider is to create a pair of anyons $a,\overbar{a}$ in the bulk, move $a$ counter-clockwise around a gapped boundary, and come back and annihilate the pair to vacuum. In the language of Chapter \ref{sec:hamiltonian}, this corresponds to applying the $a$ ribbon operator to a counter-clockwise closed ribbon encircling the gapped boundary. We will denote this operation by $W_a(\alpha_i)$, where $\alpha_i$ is the closed ribbon encircling boundary $i$. This operator is often known as the {\it Wilson loop operator}.

Suppose we have two gapped boundaries given by Lagrangian algebras $\A_1,\A_2$, which encode a qudit with orthonormal basis as in Fig. \ref{fig:gapped-boundary-basis} of Section \ref{sec:algebraic-gsd}. As before, we would like to compute the result of applying each $W_a(\alpha_i)$ on each basis element of the ground state $\Hom(1, \A_1 \otimes \A_2)$, and express the result in terms of the original basis.

Suppose we start as an arbitrary basis element $W_b(\gamma)\ket{0}$. Diagrammatically, the operator $W_a(\alpha_2)$ transforms this basis element into the following state:

\begin{equation}
\label{eq:loop-1}
\vcenter{\hbox{\includegraphics[width = 0.47\textwidth]{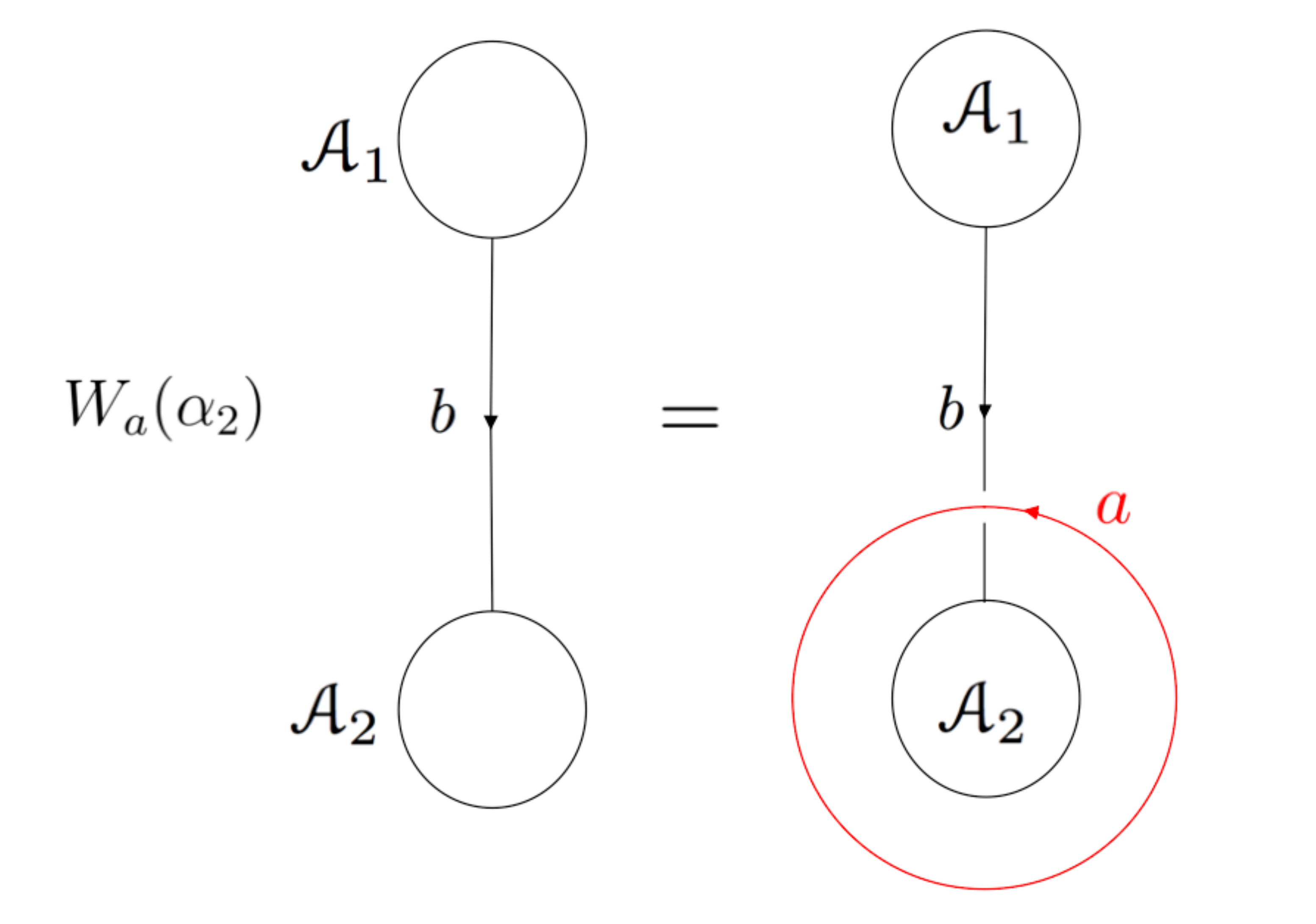}}}
\end{equation}

\noindent
Since we are working in a model where the total charge is vacuum, we may consider this picture as two holes on a sphere. We can hence push the anyon loop back through infinity, to get

\begin{equation}
\label{eq:loop-2}
\vcenter{\hbox{\includegraphics[width = 0.47\textwidth]{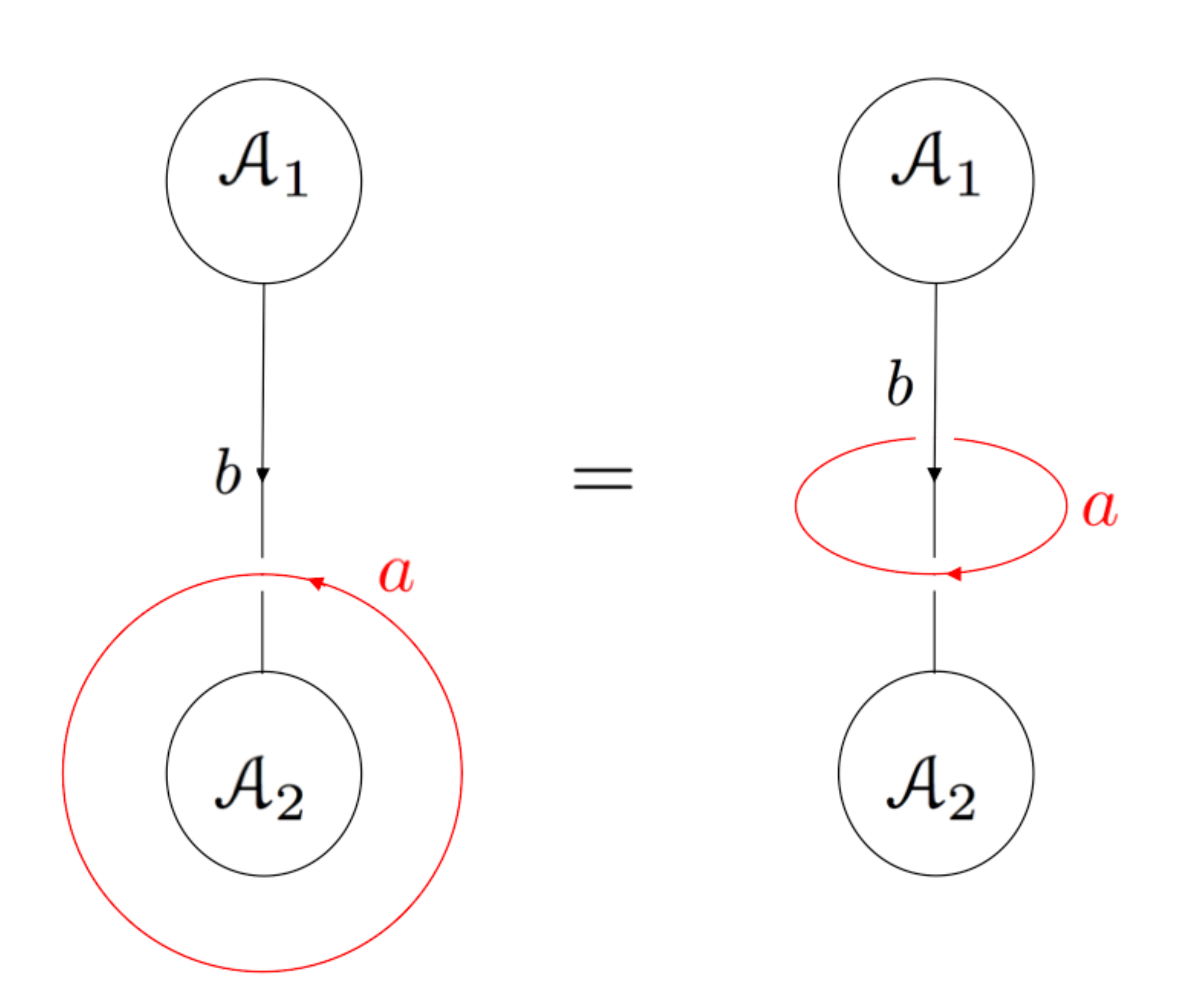}}}
\end{equation}

\noindent
The right hand side of Eq. (\ref{eq:loop-2}) may be simplified using the definition of the $\mathcal{S}$ matrix as follows: Suppose

\begin{equation}
\label{eq:loop-3}
\vcenter{\hbox{\includegraphics[width = 0.28\textwidth]{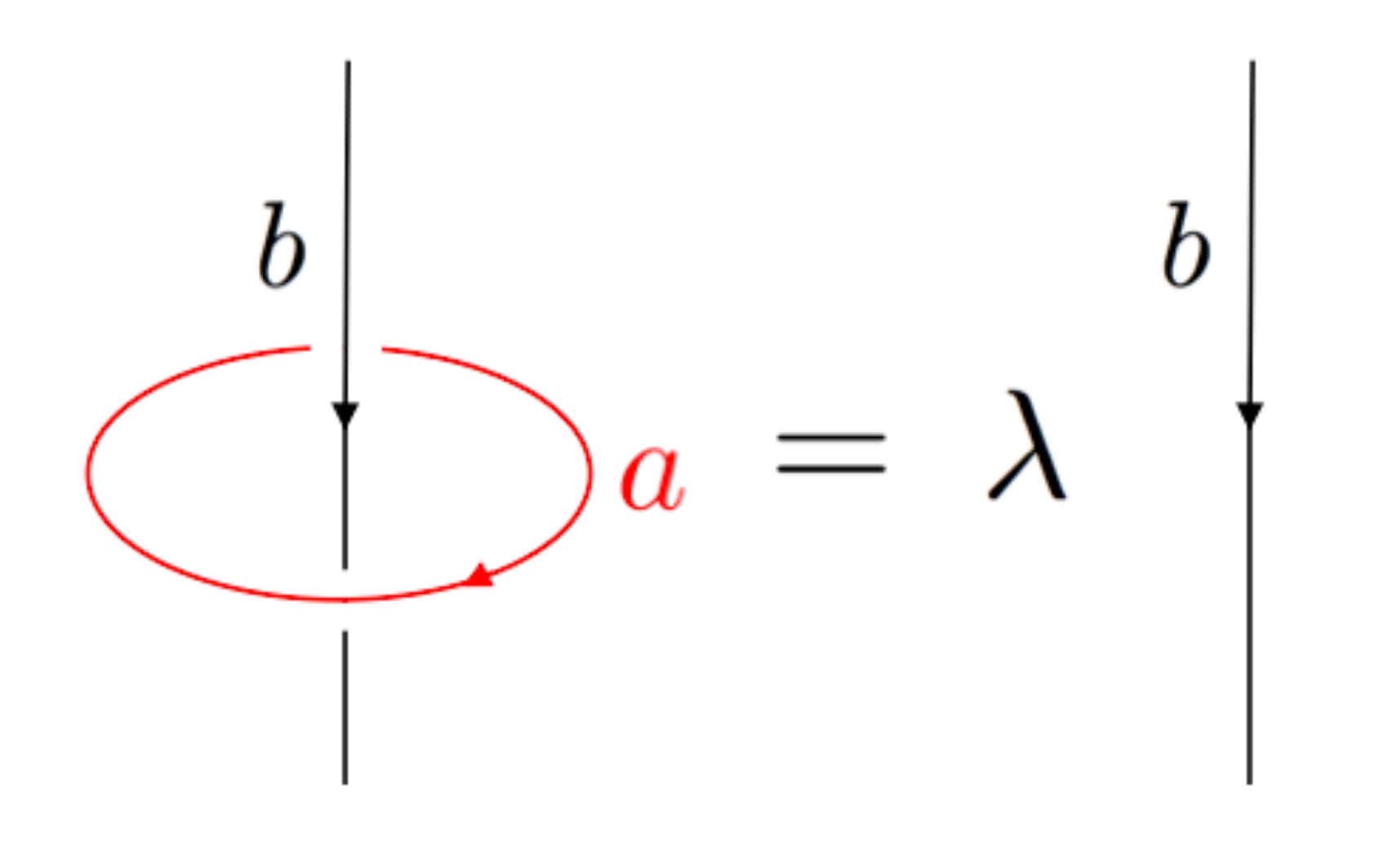}}}.
\end{equation}

\noindent
Then, taking traces on both sides, we get

\begin{equation}
\label{eq:loop-4}
\vcenter{\hbox{\includegraphics[width = 0.42\textwidth]{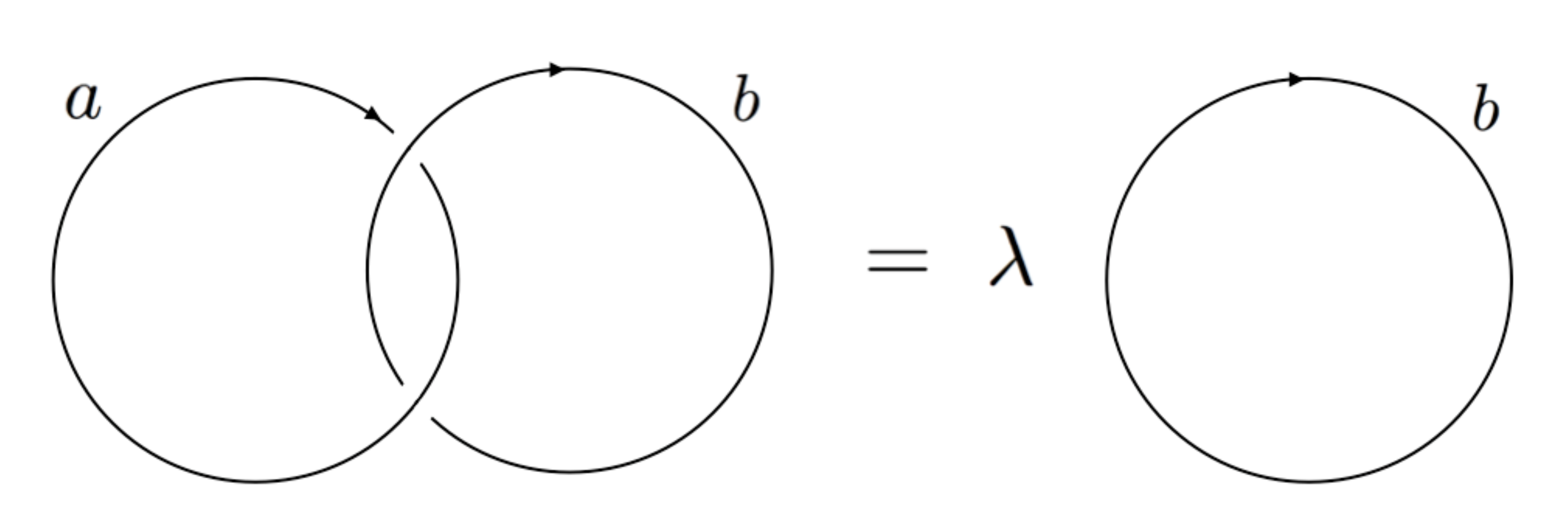}}}.
\end{equation}

By definition of the $\mathcal{S}$ matrix of the modular tensor category $\B$ \cite{BakalovKirillov}, we have $\lambda = \frac{S_{ab}}{d_b}$. Hence, we have the formula

\begin{equation}
\label{eq:loop-formula}
\vcenter{\hbox{\includegraphics[width = 0.45\textwidth]{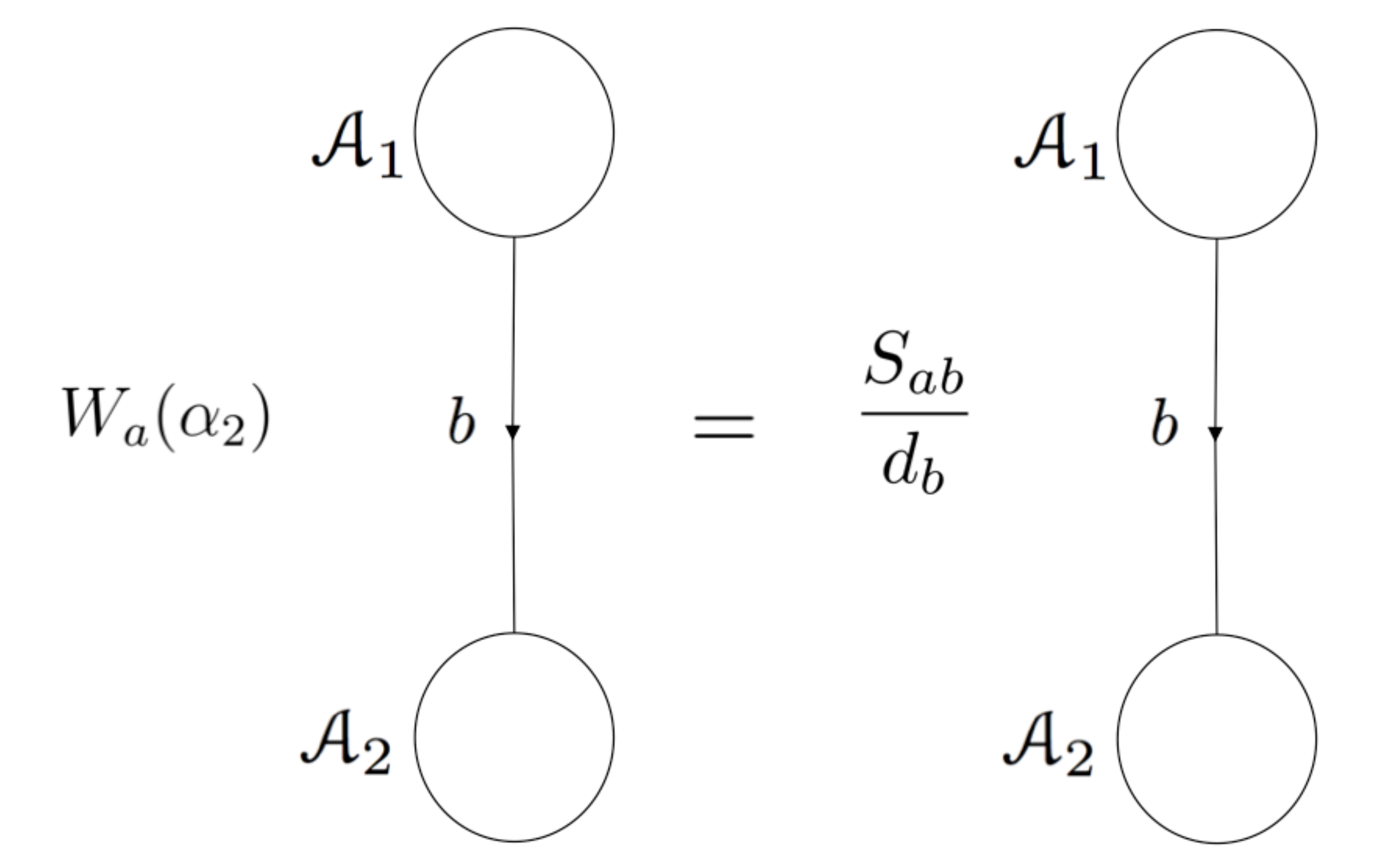}}}.
\end{equation}

As before, $W_a(\alpha_i)$ gives a $d\times d$ matrix that acts on the ground state $\Hom(1, \A_1 \otimes \A_2)$.

As in the case of the tunneling operator, the loop operator $W_{a}(\alpha_i)$ also need not be Hermitian or unitary. In general, it is just a Wilson loop operator, which is a holonomy.

\subsection{Braiding gapped boundaries}
\label{sec:braiding}

\begin{figure}
\centering
\includegraphics[width = 0.4\textwidth]{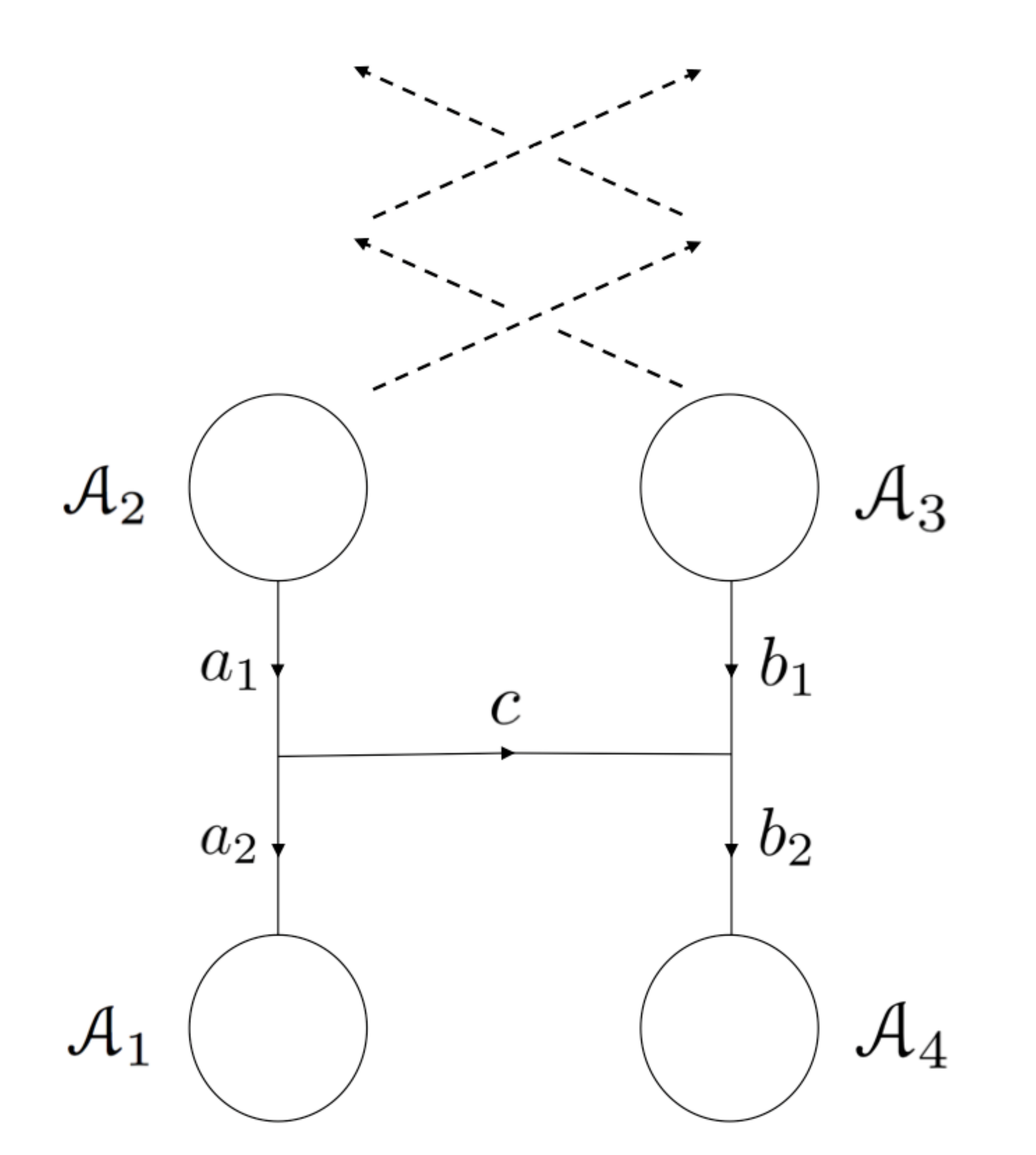}
\caption{Illustration of the braiding of two gapped boundaries, for the generator $\sigma_2^2$ of the four-strand pure braid group $P_4$. We note that the solid lines indicate the tunneling operators from the basis vectors (i.e. they do {\it not} signify motion of the holes), while the dotted lines indicate how the holes move in the braiding process.}
\label{fig:braiding}
\end{figure}

In this section, we discuss how to braid gapped boundaries around each other. This gives multiple-qudit operations that can be used to obtain entangling gates. Physically, braiding corresponds to moving gapped boundaries around each other, which can be done by tuning the Hamiltonian $H_{\text{G.B.}}$ adiabatically or by using the procedure described in Chapter \ref{sec:circuits}.

In general, suppose we have $n$ gapped boundaries, given by Lagrangian algebras $\A_1, ... \A_n$, in a model which has total charge vacuum. By Eq. (\ref{eq:ground-state-algebraic}), the ground state of this model is given by the vector space $\Hom(\one_\B, \A_1 \otimes \A_2 \otimes ... \otimes \A_n)$. By the procedure of Chapter \ref{sec:circuits}, we may arbitrarily braid the gapped boundaries around each other to obtain a unitary transformation on this hom-space, so long as we return each boundary to its original position. Mathematically, this means that the braiding operation should be a representation of the $n$-strand pure braid group $P_n$. 

For most purposes of universal quantum computation, it is sufficient for us to consider 2-qudit encodings, where we have $n = 4$, as shown in Fig. \ref{fig:braiding}. In general, it is necessary to compute all 6 generators of $P_4$. Here, as an example, we focus on the computation of the generator $\sigma_2^2$. It is simple for the interested reader to generalize this computation to all other cases, and even to higher-strand braid groups.

Consider an arbitrary basis element of $\Hom(\one_\B, \A_1 \otimes \A_2 \otimes \A_3 \otimes \A_4)$ as our start state. After applying $\sigma_2^2$, we have:\footnote{For the rest of this section, we assume for simplicity of illustration and computation that all anyons are self-dual. The generalization is obvious, but one just needs to be more careful in drawing orientations for each edge and using $F$ and $R$ symbols.}

\begin{equation}
\label{eq:braid-1}
\vcenter{\hbox{\includegraphics[width = 0.65\textwidth]{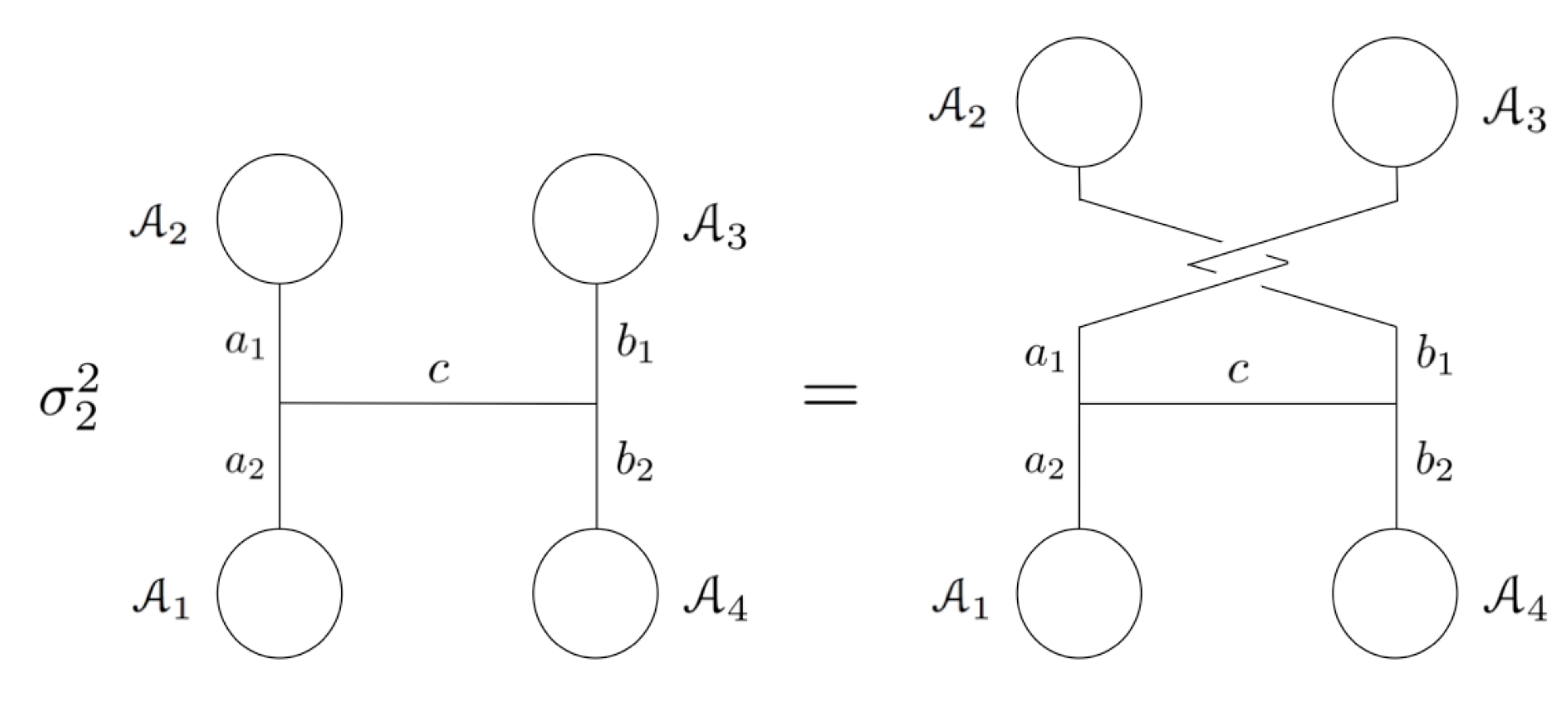}}}
\end{equation}

\noindent
As with the earlier cases, we would like to express the right hand side of Eq. (\ref{eq:braid-1}) in terms of the original basis. We can apply an $F$-move to get:

\begin{equation}
\label{eq:braid-2}
\vcenter{\hbox{\includegraphics[width = 0.7\textwidth]{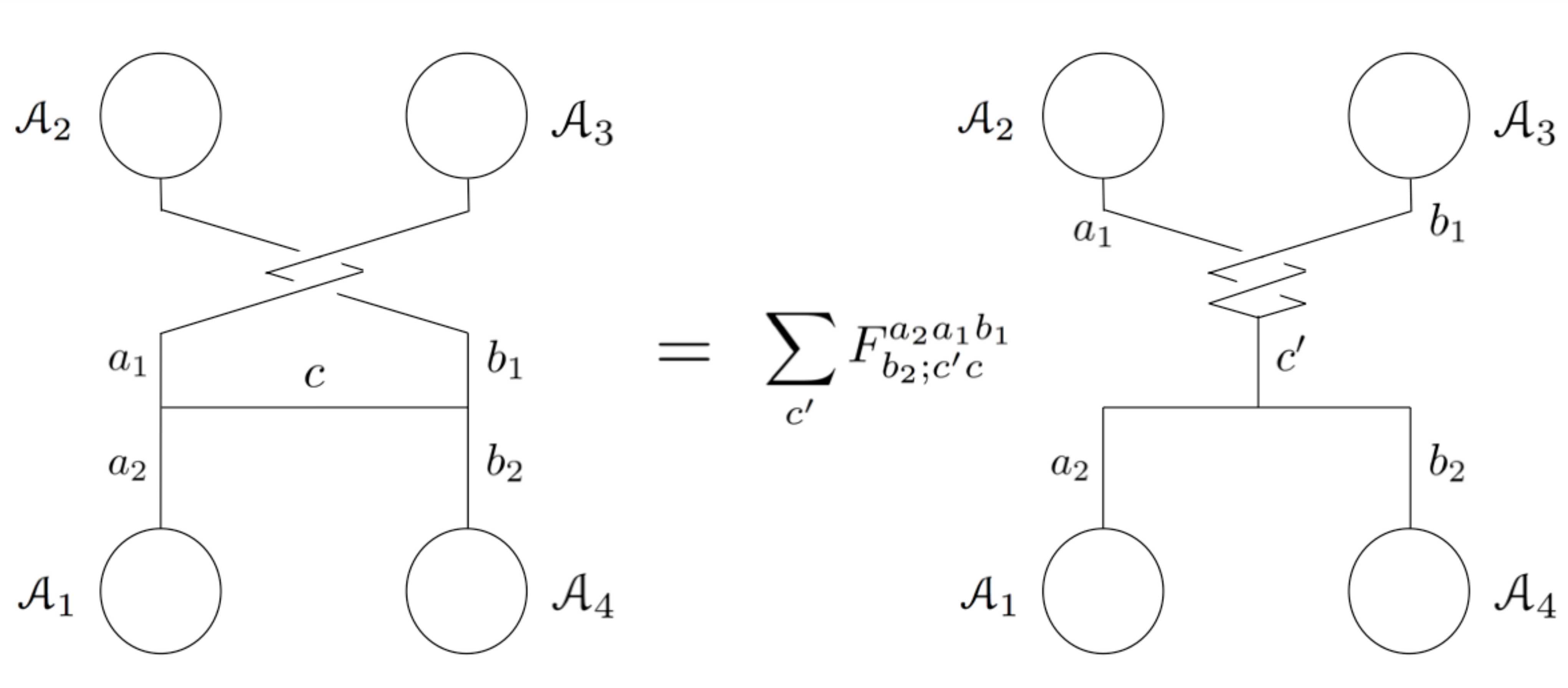}}}
\end{equation}

\noindent
Next, applying two $R$-moves gives:

\begin{equation}
\label{eq:braid-3}
\vcenter{\hbox{\includegraphics[width = 0.7\textwidth]{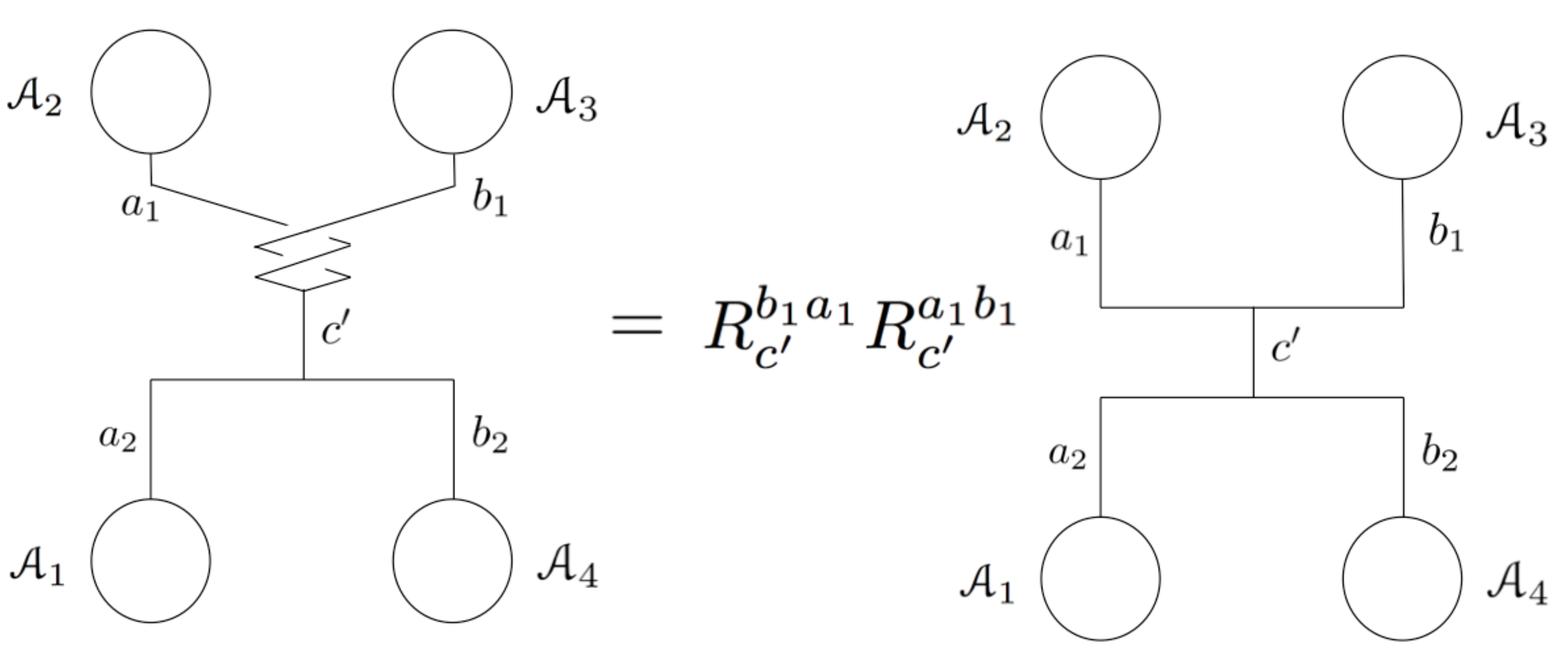}}}
\end{equation}

\noindent
Finally, we apply one more $F$-move, which gives:

\begin{equation}
\label{eq:braid-4}
\vcenter{\hbox{\includegraphics[width = 0.65\textwidth]{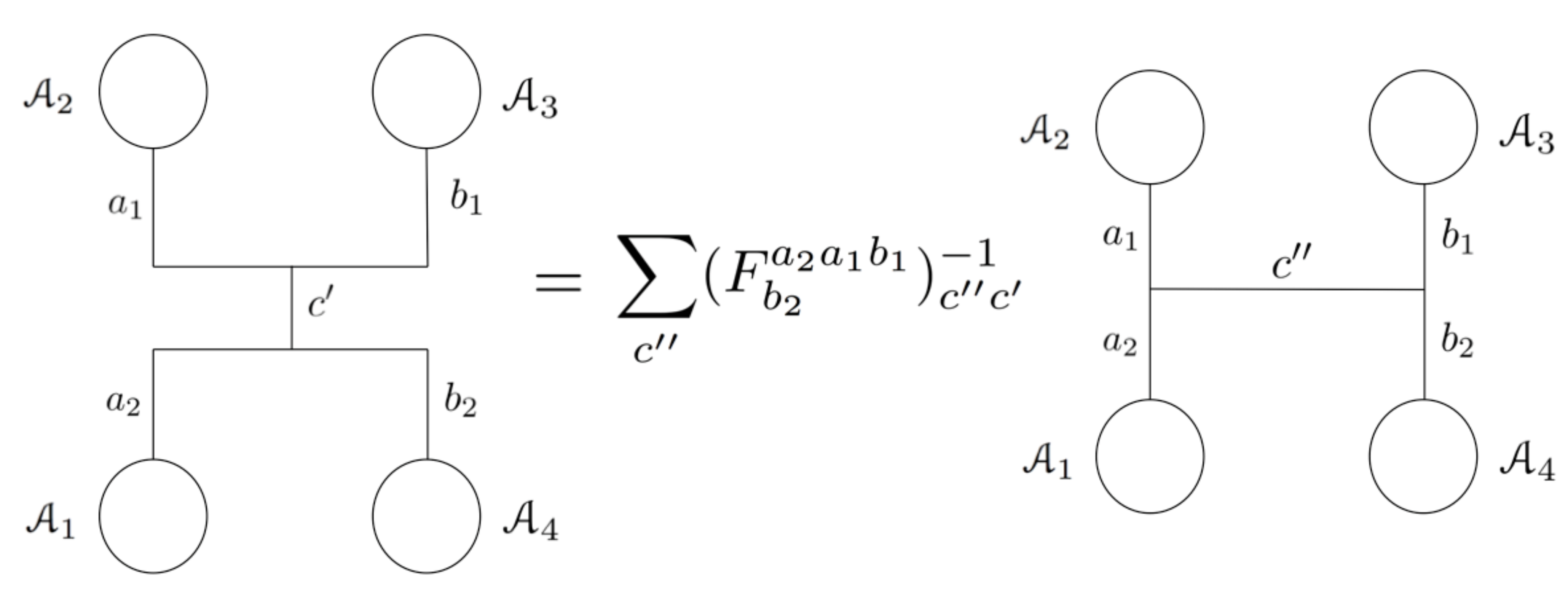}}}
\end{equation}

\noindent
Hence, we see that action of the pure braid group generator $\sigma_2^2$ on an arbitrary basis vector of $\Hom(\one_\B, \A_1 \otimes \A_2 \otimes \A_3 \otimes \A_4)$ is given by the following formula:

\begin{equation}
\label{eq:braid-formula}
\vcenter{\hbox{\includegraphics[width = 0.93\textwidth]{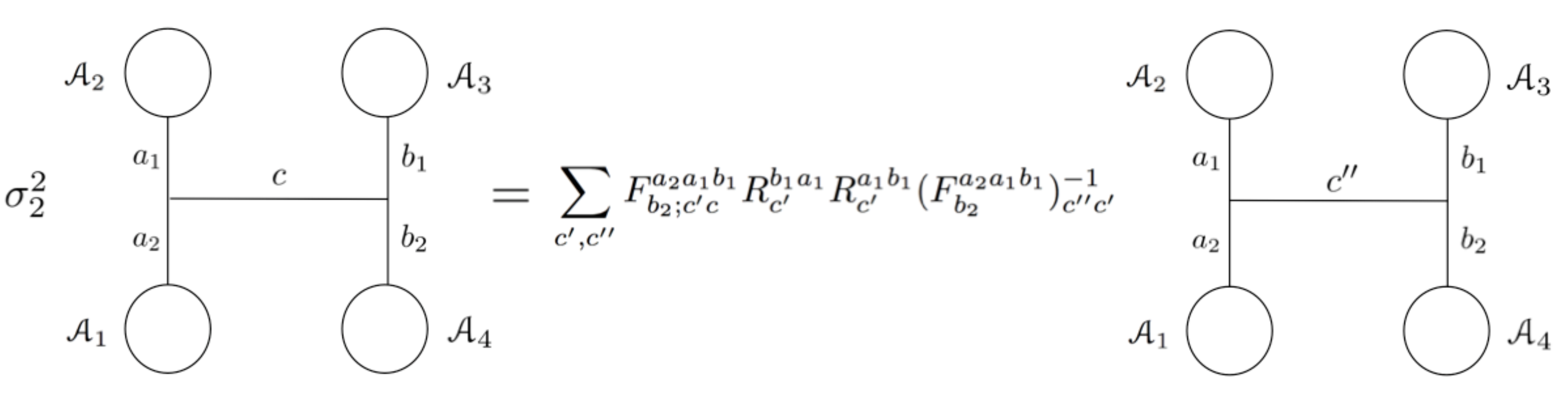}}}
\end{equation}

In general, $\sigma_2^2$ (or any other generator of $P_4$) gives a unitary transformation on the hom-space $\Hom(\one_\B, \A_1 \otimes \A_2 \otimes \A_3 \otimes \A_4)$. However, for the purposes of quantum computation, we would like to work in a computational subspace of this hom-space which corresponds to two qudits in the encoding of Fig. \ref{fig:gapped-boundary-basis}. As a result, to avoid leakage into non-computational subspace, one should find a braid such that the computational subspace is preserved.

We note that we have only considered the pure braid group $P_4$ in this context, simply to ensure that the boundary type of each hole does not change after braiding. If any two $\A_i$ and $\A_j$ are equal, it is certainly valid to interchange these two holes. In fact, this will be done as an example in Section \ref{sec:ds3-operations}.

\begin{remark}
In fact, the gapped boundary braiding presented in this section gives more than just a representation of the four-strand pure braid group. It is actually a representation of the spherical four-strand pure braid group, since we assume that the total charge around the four holes is vacuum (so the model is equivalent to four holes on a sphere). In general, it would be an interesting mathematical question to compute the image of this representation. This would have very useful applications to achieving universal quantum computation.
\end{remark}

\subsection{Topological charge measurement}
\label{sec:measurement}

\begin{figure}
\centering
\includegraphics[width = 0.65\textwidth]{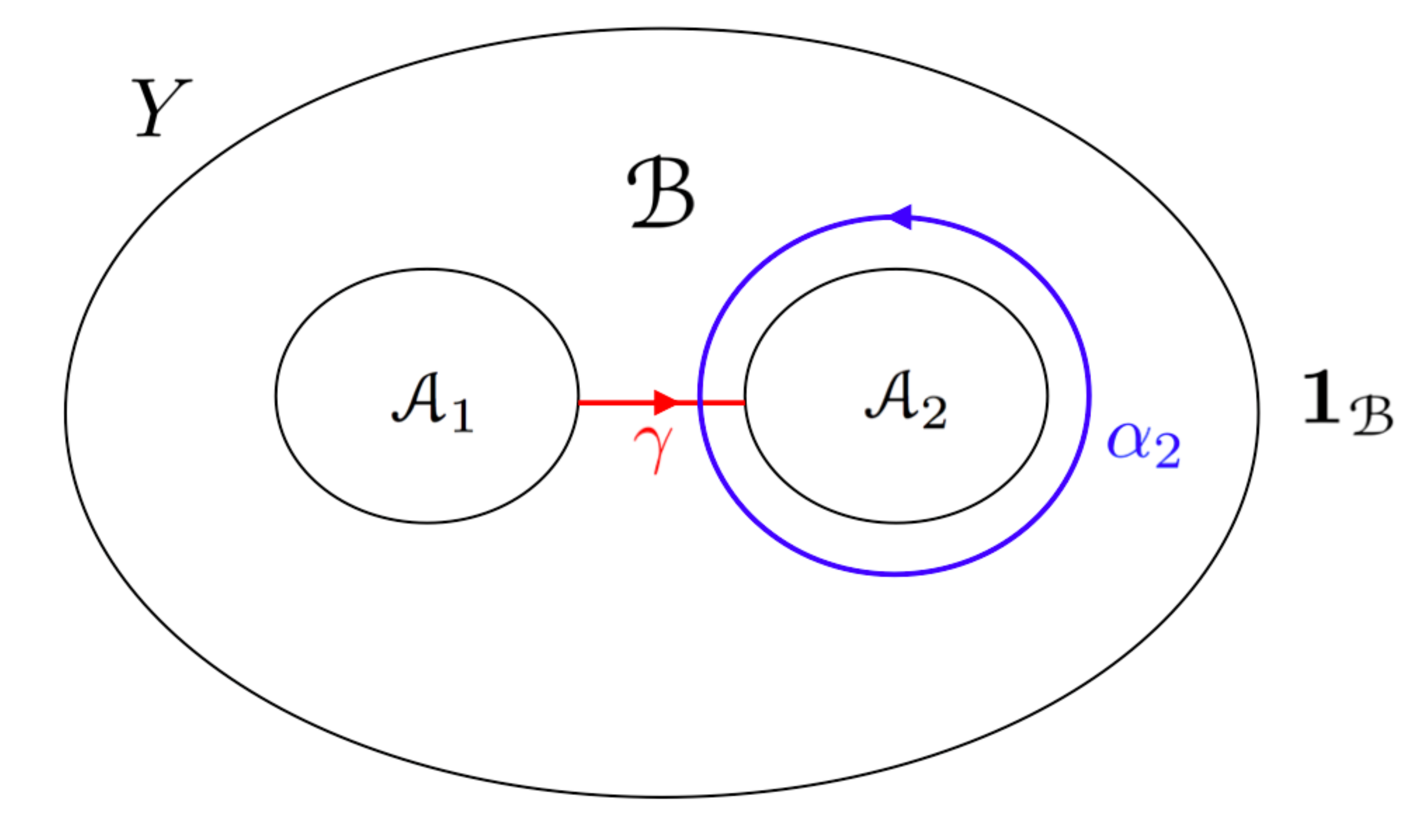}
\caption{Illustration of topological charge measurement protocol, in the case of two inside gapped boundaries. The total charge of the system is vacuum. $\alpha_2$ is a counter-clockwise loop that encircles the hole given by $\A_2$. There is only one simple arc $\gamma_1$ used to define the basis; it connects $\A_1$ to $\A_2$ and is labeled here as $\gamma$.}
\label{fig:tcm}
\end{figure}

\subsubsection{Motivations}
Suppose we are given a topological order $\B = \mZ(\mC)$ on a planar region $Y$ with $n$ holes, $\mfh_i,i=1,...,n$, each labeled by a Lagrangian algebra $\A_i$. Let us assume that the total charge of the system is given by vacuum, the tensor unit $\one_B$ of $\B$. As discussed in Section \ref{sec:algebraic-gsd}, the ground state manifold of this setup is given by the hom-space $\Hom(\one_B, \otimes_i \A_i)$. An example of this setup is shown in Fig. \ref{fig:tcm}, in the special case where $n = 2$.

If $\B$ is a Dijkgraaf-Witten theory (e.g. the TQFTs realized by the Hamiltonians of Chapter \ref{sec:hamiltonian}), the unitary gates from tunnel-$a$ operators, loop-$a$ operators, and gapped boundary braiding (introduced in the preceding sections of this chapter) generate only a finite group. Hence, it is not possible to use these operations alone to form universal gate sets. This leads us to consider other physically reasonable topological protocols to obtain more gates. In this section, inspired by the results of Ref. \cite{Barkeshli16}, we discuss one such protocol, which is topological charge measurement.

\subsubsection{General topological charge measurement}
\label{sec:general-tcm}

To introduce this protocol, let us first define an operator algebra $\mathcal{W}(\B,\{\A_i\})$ for the symmetries of the theory at low energy, which will be known as the {\it Wilson operator algebra}. We first construct a set $\Gamma(Y)$ of simple loops and arcs in $Y$ as follows:

\begin{enumerate}
\item
For each $i = 1,2, ... n$, let $\alpha_i$ be a simple loop that encircles hole $i$. We define all loops $\alpha_i$ to be oriented counter-clockwise. Then $\alpha_i \in \Gamma(Y)$.
\item
For each $i = 1,2, ... n-1$, let $\gamma_i$ be a simple arc that connects hole $i$ and hole $i+1$. We define $\gamma_i$ so that it is always oriented to point from $i$ to $i+1$. Then $\gamma_i \in \Gamma(Y)$.
\end{enumerate}

Examples of loops and arcs in $\Gamma(Y)$ are shown in Fig. \ref{fig:tcm}, for the case where $n=2$. By definition of $\Gamma(Y)$, each knot diagram in $Y$ can be resolved using the graphical calculus to a linear combination of loop and arc operators in $\Gamma(Y)$.

We can now construct a basis $L_\mW (\B,\{\A_i\})$ for the Wilson operator algebra $\mathcal{W}(\B,\{\A_i\})$:

\begin{enumerate}
\item
For each simple object $a \in \B$, and for each loop $\alpha_i \in \Gamma(Y)$, the Wilson loop operator $W_{a}(\alpha_i)$ (defined in Section \ref{sec:loop}) is a basis element in $L_\mW (\B,\{\A_i\})$.
\item
For each $i = 1,2,... n-1$, let $A_{i,i+1}$ be the set of all anyon types $a \in \B$ such that the anti-particle $\overbar{a}$ has a nonzero coefficient in the Lagrangian algebra $\A_i$, while $a$ has a nonzero coefficient in the Lagrangian algebra $\A_{i+1}$. Then, for each $a \in A_{i,i+1}$, the Wilson line operator $W_{\gamma_i}(a)$ is a basis element in $L_\mW (\B,\{\A_i\})$.
\end{enumerate}

We posit that any Hermitian operator $\mathcal{O}\in \mathcal{W}(\B,\{\A_i\}) = \Span (L_\mW (\B,\{\A_i\}))$ can be measured. Such operators $\mathcal{O}$ will be called {\it topological charge measurement operators}.  The corresponding projective measurements $P_\mathcal{O}$ will be called {\it topological charge measurements}. In general, any operator $h \in \mathcal{W}(\B,\{\A_i\})$ can be used to form a topological charge measurement operator $h+{h}^{\dagger}$, although it may be unphysical.

To be more physical, we will consider only operators which are monomials of basis operators. Let $H \subseteq L_\mW$ be the collection of all Hermitian basis operators of $\mathcal{W}(\B,\{\A_i\})$, and let $O = L_\mW \backslash H \subseteq L_\mW$ be the collection of all non-Hermitian basis vectors. Each basis vector $H_i \in H$ gives rise to a topological charge measurement that has a definite physical meaning.

In general, certain linear combinations of basis vectors $O_i \in O$ should also give rise to topological charge measurements. However, it is more complicated to determine which coefficients in the linear combination will give rise to physical measurements. One special case of this is the topological charge projections of Ref. \cite{Barkeshli16}. These are discussed below.

\subsubsection{Topological charge projection}

\begin{figure}
\centering
\includegraphics[width = 0.65\textwidth]{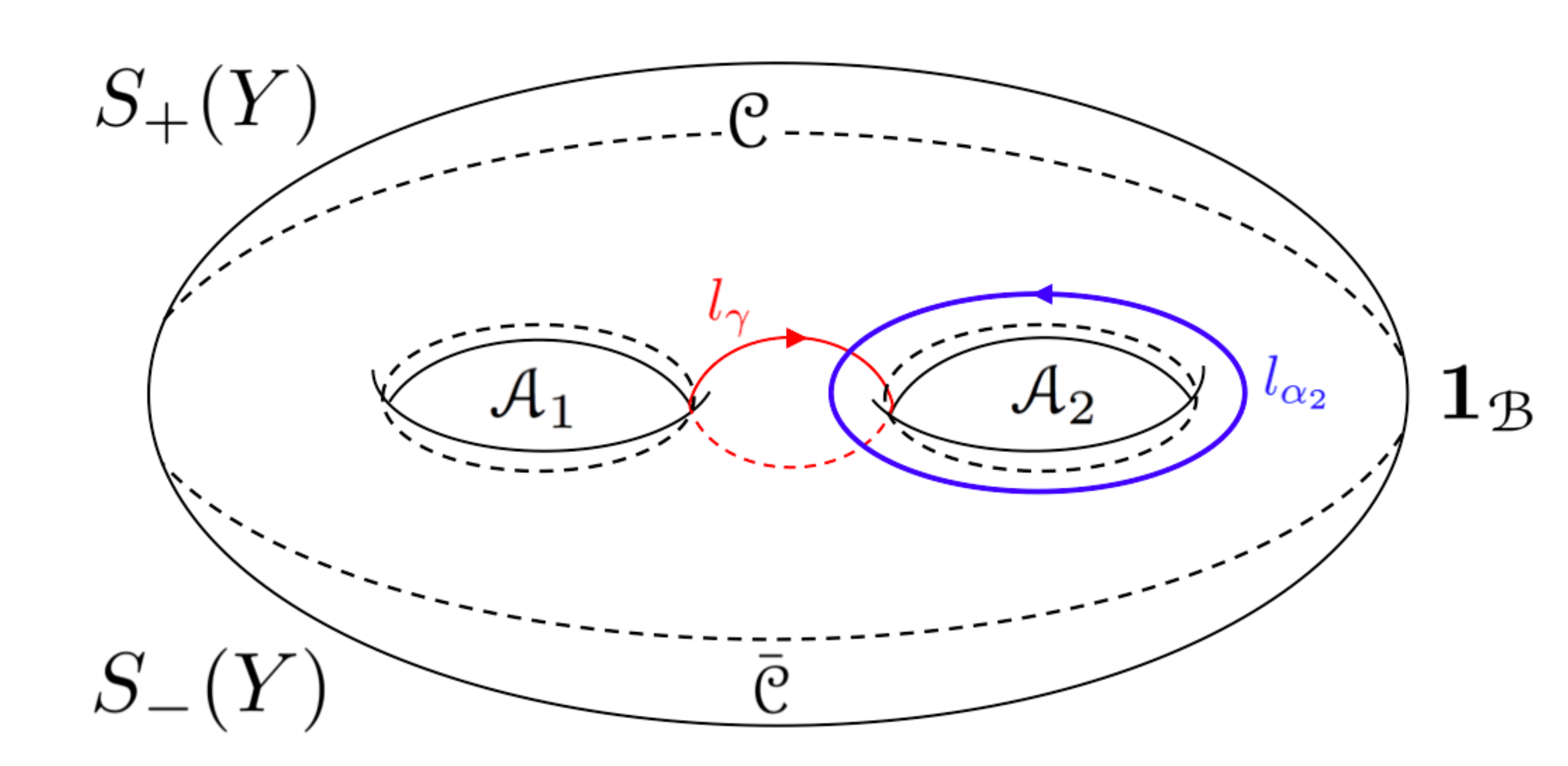}
\caption{Illustration of topological charge projection protocol. Because $\mC$ is modular, the original bulk TQFT $\B$ splits into two layers $\mC$ and $\overbar{\mC}$, which are completely separate in the bulk, but are glued together at the original boundaries (black dotted lines). Each loop $\alpha_i \in \Gamma(Y)$ now becomes a loop $l_{\alpha_i}$ in one of the layers $S_+ (Y)$ or $S_- (Y)$, and each arc $\gamma_i \in \Gamma(Y)$ now lifts to a loop $l_{\gamma_i}$ that goes around both layers.}
\label{fig:tcp}
\end{figure}

Suppose $\B=\ZZ(\CC)$ represents a topological order, where $\CC$ is also a modular tensor category. Because $\mC$ is modular, the category $\mC$ also gives a TQFT. As a result, $\B = \mZ(\mC) = \mC \boxtimes \overbar{\mC}$ is a doubled theory which splits into two topological orders $\mC$ and $\overbar{\mC}$. The two theories $\mC$ and $\overbar{\mC}$ do not interact with each other in the bulk, but are ``stuck together'' at the original boundaries $\A_1, ... \A_n$ of $\B$. In this case, the planar region $Y$ (such as the one portrayed in Fig. \ref{fig:tcm}) splits into two mirror layers, $S_{+}(Y)$ and $S_{-}(Y)$, such that $S_{+}(Y)$ now forms the bulk of the TQFT $\mC$, while $S_{-}(Y)$ now forms the bulk of the TQFT $\overbar{\mC}$. The two layers $S_{+}(Y)$, $S_{-}(Y)$ are completely disjoint, except they are also ``stuck together'' at the original boundaries of $Y$. This configuration is illustrated in Fig. \ref{fig:tcp} in the case where $n=2$.

Let us now consider the Wilson operator algebra $\mathcal{W}(\B,\{\A_i\})$ that we constructed in the previous subsection. First, in the new context, each loop $\alpha_i \in \Gamma(Y)$ now becomes a loop $l_{\alpha_i}$ in one of the layers $S_{+}(Y)$ or $S_{-}(Y)$, while each arc $\gamma_i \in \Gamma(Y)$ now lifts to a loop $l_{\gamma_i}$ that goes around both layers. Hence, physically, a Wilson loop operator along $\alpha_i$ now measures charges through the loop $l_{\alpha_i}$. The Wilson operator measuring charge $a$ through $l_{\alpha_i}$ leads to a topological charge measurement given by \cite{Barkeshli16}

\begin{equation}
P^{(a)}_{l_{\alpha_i}}=\sum_{x\in \CC}S_{0a}S_{xa}^{*}W_x(\alpha_i)
\end{equation}

\noindent
where the sum runs over the simple objects $x$ of the modular tensor category $\mC$, and $S_{ab}$ is the modular $\mathcal{S}$-matrix of $\CC$. We emphasize that the projector measures topological charges $a$ in the topological order $\mC$, not the original doubled order $\B$. Likewise, the Wilson loop operators $W_x(\alpha_i)$ are computed simply using the loop formula (\ref{eq:loop-formula}) with the $\mathcal{S}$-matrix data of $\CC$.

Similarly, a Wilson line operator along $\gamma_i$ now measures charges through the loop $l_{\gamma_i}$. A Wilson operator measuring charge $a$ through $l_{\gamma_i}$ leads to a topological charge measurement given by \cite{Barkeshli16}

\begin{equation}
P^{(a)}_{l_{\gamma_i}}=\sum_{x\in \CC}S_{0a}S_{xa}^{*}W_{x\overbar{x}}(\gamma_i).
\end{equation}

\noindent
As before, $P^{(a)}_{l_{\gamma_i}}$ measures topological charges $a$ in the topological order $\mC$. However, this time, the is computed using Wilson line operators $W_{x\overbar{x}}(\gamma_i)$ by evaluating the tunneling formula (\ref{eq:tunnel-formula}) with the data from the {\it doubled} topological order $\B$.

In Ref. \cite{Barkeshli16}, it is shown that topological charge projections such as $P^{(a)}_{l_{\alpha_i}}$ and $P^{(a)}_{l_{\gamma_i}}$ can be used to generate all of the mapping class group representations $V_{\CC}(Y)$ of a closed surface $Y$ from the modular category $\CC$.

This theory can be generalized partially to the case where $\CC$ is not modular but abelian \cite{Bark13b}. We note that our postulate is a generalization of all of these results.

\subsubsection{Topological charge measurement in the surface code context} 
\label{sec:surface-code-tcm}
Finally, we conclude the discussion on topological charge measurement by considering a type of measurement that may be implemented through the surface code methods of Chapter \ref{sec:circuits}. Given a product of some basis Hermitian operators $H_{i_j}$ of the Wilson operator algebra such that $H_{i_1}\cdots H_{i_k}=\alpha H_{i_k}\cdots H_{i_1}$ for some phase $\alpha$, the Hermitian operator

\begin{equation}
P_{(i_1, ... i_k)} = \sqrt{\alpha}H_{i_1}\cdots H_{i_k}
\end{equation}

\noindent
is a topological charge measurement operator.

In general, topological charge measurements may be implemented on the surface code by introducing ancilla qudits and suitably entangling them with the data qudits using a sequence of available Clifford gates. Then, by measuring the ancilla qudit (e.g. in the $\sigma^z$ basis), we can project the entire system into an eigenstate of some topological charge measurement operator, which is determined by the sequence of gates applied.

\subsection{Physically implementable gates}
\label{sec:physically-implementable}

In Ref. \cite{Fowler12}, Fowler et al. provide a detailed discussion of how to implement the Hadamard gate on the gapped boundary basis in the $\Z_2$ surface code. The same procedure may be used to obtain the generalized Hadamard gate (i.e. Fourier transform) with the $\Z_p$ surface code, for any $p$ prime. Specifically, one can simply replace all instances of the qubit Hadamard $H_2$ with the generalized Hadamard gate $H_p$ for qupits, all instances of the qubit CNOT with the generalized entangling gate $\SUM_p$, and all instances of the qubit Pauli-$X$ and Pauli-$Z$ operators with the generalized Pauli-$X$ and Pauli-$Z$ for qupits. (These operations are defined in Section \ref{sec:universal-gate-set}). Due to time and space constraints, the detailed procedure will not be presented here.

\subsection{Example: The Toric Code}
\label{sec:tc-operations}

In this section, we present a qubit encoding using gapped boundaries of the $\Z_2$ toric code, and use this to illustrate the topological operations discussed in the previous sections. Sections \ref{sec:tc-qubit-encoding}-\ref{sec:tc-braiding} and \ref{sec:tc-hadamard} are adapted from the work by Fowler et al. \cite{Fowler12} to our categorical presentation, and Section \ref{sec:tc-measurement} presents a new way to obtain the phase gate using topological charge measurement.

\subsubsection{Qubit encoding}
\label{sec:tc-qubit-encoding}

\begin{figure}
\centering
\includegraphics[width = 0.45\textwidth]{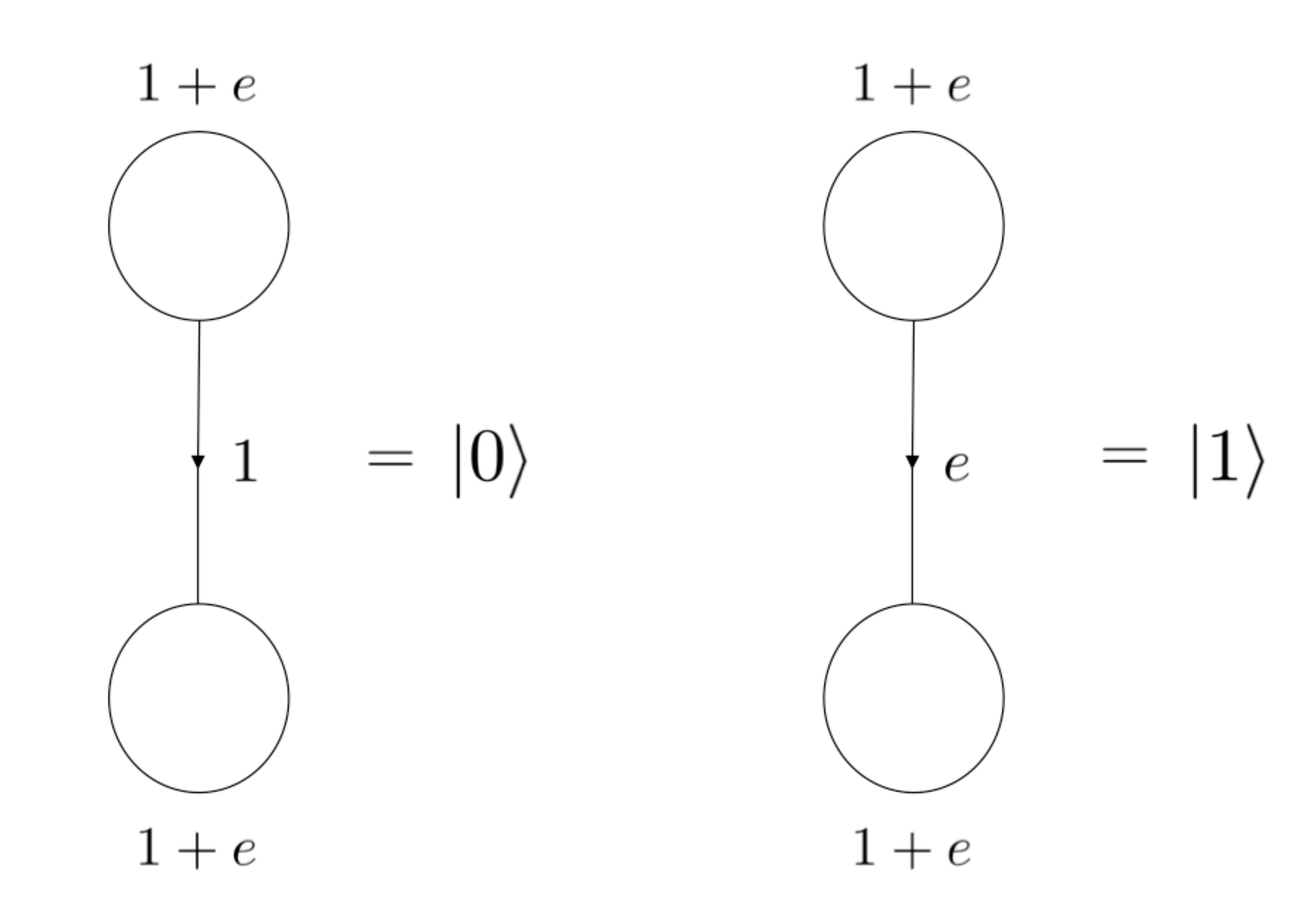}
\caption{Qubit encoding using $1+e$ boundaries of the $\Z_2$ toric code.}
\label{fig:tc-qubit-encoding}
\end{figure}

As shown in Fig. \ref{fig:tc-qubit-encoding}, we will mainly use two $1+e$ boundaries to encode a computational qubit, but will occasionally use two $1+m$ boundaries to encode an ancillary qubit. We associate the degeneracy of two $1+e$ boundaries with the qudit basis as follows: Let $\ket{0}$ be state that would be at zero energy even if the holes were not present (i.e. no particle tunneling), and let $\ket{1} = W_e(\gamma)\ket{0}$. (The same encoding will be used for $1+m$ boundaries when necessary, with $\ket{1} = W_m(\gamma)\ket{0}$). When we present a quantum gate, we will present in the standard way such that all rows and columns are ordered as $(\ket{0},\ket{1})$.

\subsubsection{Tunnel-$e$ operators}
\label{sec:tc-tunnel}

The first operation we can consider on our qubit is to tunnel an $e$ particle from one of the holes to the other. This is also presented in Ref. \cite{Fowler12} using the stabilizer code language of Chapter \ref{sec:circuits}. Using our methods from Chapter \ref{sec:algebraic} and this chapter, by Eq. (\ref{eq:tunnel-formula}), since all $M$ symbols of the toric code are trivial, we find that the matrix corresponding to this operation is

\begin{equation}
W_e(\gamma) =
\begin{bmatrix}
0 & 1 \\
1 & 0
\end{bmatrix}
= \sigma^x.
\end{equation}

\subsubsection{Loop-$m$ operators}

The next topological operation we consider is to loop an $m$ particle around one of the holes. Ref. \cite{Fowler12} also presents this using the stabilizer code language of Chapter \ref{sec:circuits}. By Eq. (\ref{eq:loop-formula}), since the modular $\mathcal{S}$ matrix of the $\mfD(\Z_2)$ theory has $S_{em} = -1$, we find that the matrix corresponding to this operation is

\begin{equation}
W_m(\alpha_2) =
\begin{bmatrix}
1 & 0 \\
0 & -1
\end{bmatrix}
= \sigma^z.
\end{equation}

\subsubsection{Braiding}
\label{sec:tc-braiding}

\begin{figure}
\centering
\includegraphics[width = 0.35\textwidth]{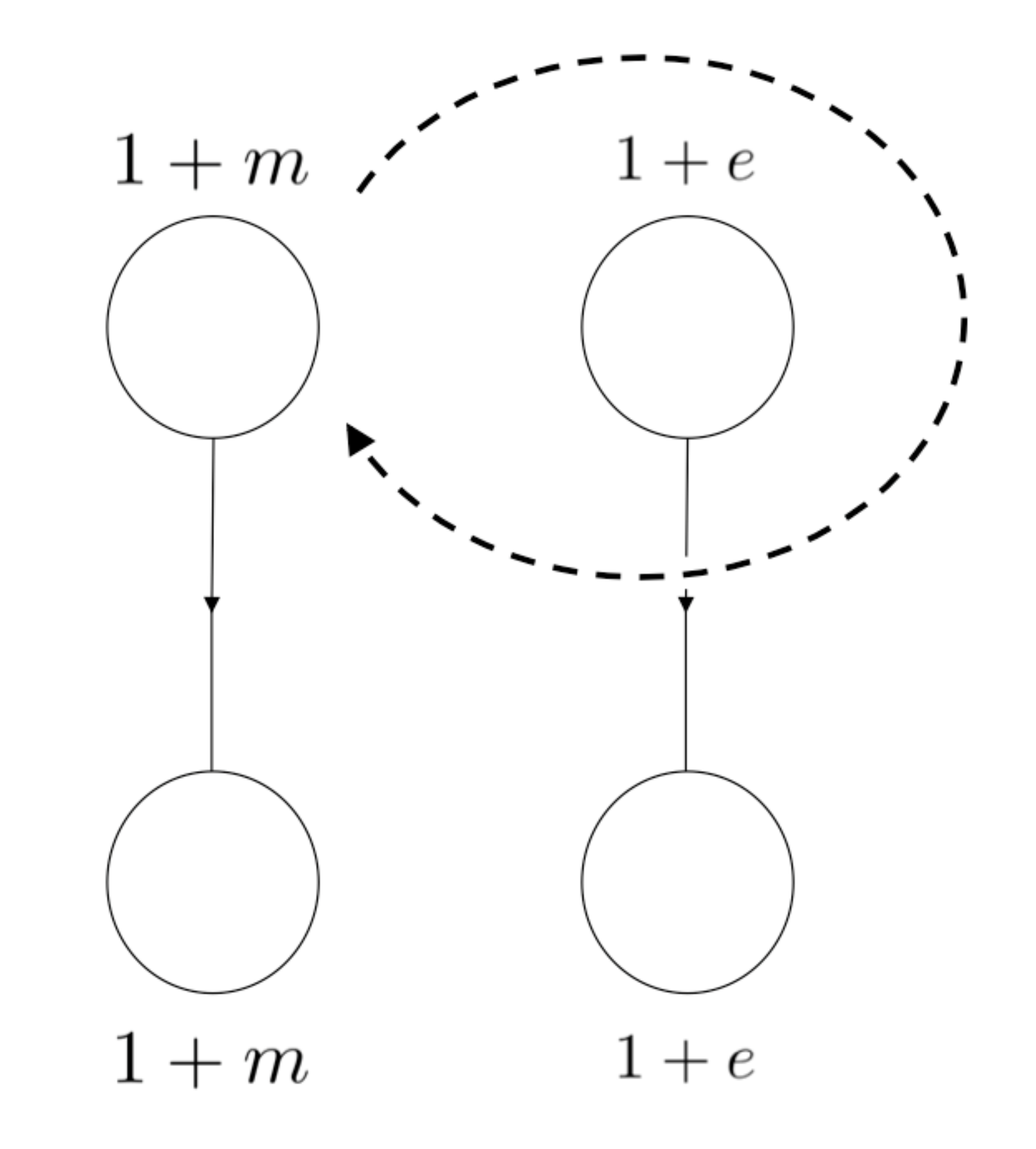}
\caption{Braid for the $\Z_2$ toric code to obtain a CNOT gate. As discussed in Section \ref{sec:braiding}, the dotted line indicates motion of a hole, while the solid lines specify the basis element of the hom-space.}
\label{fig:tc-braiding}
\end{figure}

We now discuss how to implement the controlled-$\sigma^z$ gate by braiding gapped boundaries. As in Ref. \cite{Fowler12}, we will first present how this is done in the case where the control qubit is encoded in two $1+m$ boundaries (i.e. $\A_1 = \A_2 = 1+m$) and the target qubit is encoded in two $1+e$ boundaries (i.e. $\A_3 = \A_4 = 1+e$). In this case, we may simply braid one of the control qubit boundaries around a target qubit boundary (i.e. apply $\sigma_2^2$), as shown in Fig. \ref{fig:tc-braiding}. By Eq. (\ref{eq:braid-formula}) and the $F$, $R$ symbols of $\mfD(\Z_2)$ (see Section \ref{sec:tc-algebraic}), we have

\begin{equation}
\sigma_2^2 =
\begin{bmatrix}
1 & 0 & 0 & 0 \\
0 & 1 & 0 & 0 \\
0 & 0 & 1 & 0 \\
0 & 0 & 0 & -1
\end{bmatrix}
= \wedge\sigma^z.
\end{equation}

\begin{figure}
\centering
\includegraphics[width=0.7\textwidth]{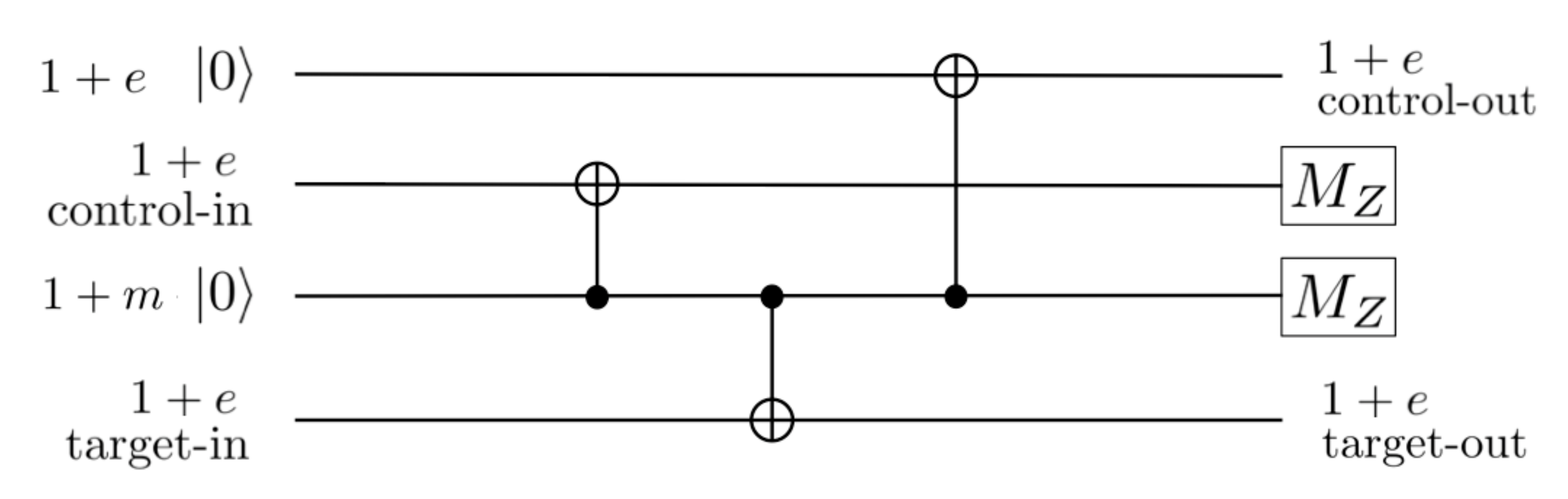}
\caption{Short circuit presented in Ref. \cite{Fowler12} to use one ancilla $1+e$ qubit, one ancilla $1+m$ qubit, and 3 topological CNOTs between $1+m$ and $1+e$ qubits to implement a topological CNOT between 2 $1+e$ qubits.}
\label{fig:tc-cnot}
\end{figure}

In Section \ref{sec:tc-hadamard}, we show how to implement the Hadamard on $1+e$ or $1+m$ encoded qubits. If we conjugate the target qubit by Hadamard gates, we obtain a topological CNOT gate.

Ref. \cite{Fowler12} presents a simple and short circuit for a CNOT between two $1+e$ encoded qubits. We reproduce this circuit in Fig. \ref{fig:tc-cnot}. As discussed in \cite{Fowler12}, the measurement of the ancilla $1+m$ qubit must be interpreted to obtain the CNOT. Specifically, there are two possible outcomes:

\begin{enumerate}
\item
If the measurement of the $1+m$ qubit yields the state $\ket{0}$, no further action is necessary, and we have effectively implemented a topological CNOT between the two $1+e$ qubits.
\item
If the measurement of the $1+m$ qubit yields the state $\ket{1}$, we must perform a topological $\sigma^x$ on the control-out qubit (described in Section \ref{sec:tc-tunnel}).
\end{enumerate}

\begin{remark}
In Fig. \ref{fig:tc-cnot}, we have presented method to implement a topological CNOT between two $1+e$ qubits, using a simple quantum circuit with topological CNOT gates between $1+e$ and $1+m$ encoded qubits. Alternatively, this circuit may be viewed algebraically as the following procedure:
\begin{enumerate}
\item
To begin, we embed the four-dimensional Hilbert space of two logical qubits into a 16-dimensional Hilbert space of four logical qubits:
\begin{equation}
\begin{gathered}
\C^4 = \Hom(\one, (1+e)^{\otimes 2})^{\otimes 2}\\
\hookrightarrow
\Hom(\one, (1+e)^{\otimes 2})^{\otimes 2}
\otimes \Hom(\one, (1+m)^{\otimes 2})
\otimes \Hom(\one, (1+e)^{\otimes 2}) = \C^{16}
\end{gathered}
\end{equation}
\item
We now perform an eight-strand pure braid $p \in P_8$ corresponding to the three CNOTs between $1+e$ and $1+m$ qubits.
\item
We project back into the four-dimensional Hilbert space of two logical qubits via measurement of logical qubits:
\begin{equation}
\begin{gathered}
\C^{16} = \Hom(\one, (1+e)^{\otimes 2})^{\otimes 2}
\otimes \Hom(\one, (1+m)^{\otimes 2})
\otimes \Hom(\one, (1+e)^{\otimes 2}) \\
\rightarrow\!\!\!\!\!\rightarrow
\Hom(\one, (1+e)^{\otimes 2}) \otimes  \Hom(\one, (1+e)^{\otimes 2}) = \C^4 
\end{gathered}
\end{equation}
\end{enumerate}
\end{remark}

\subsubsection{Topological charge measurement}
\label{sec:tc-measurement}

\begin{figure}
\centering
\includegraphics[width = 0.6\textwidth]{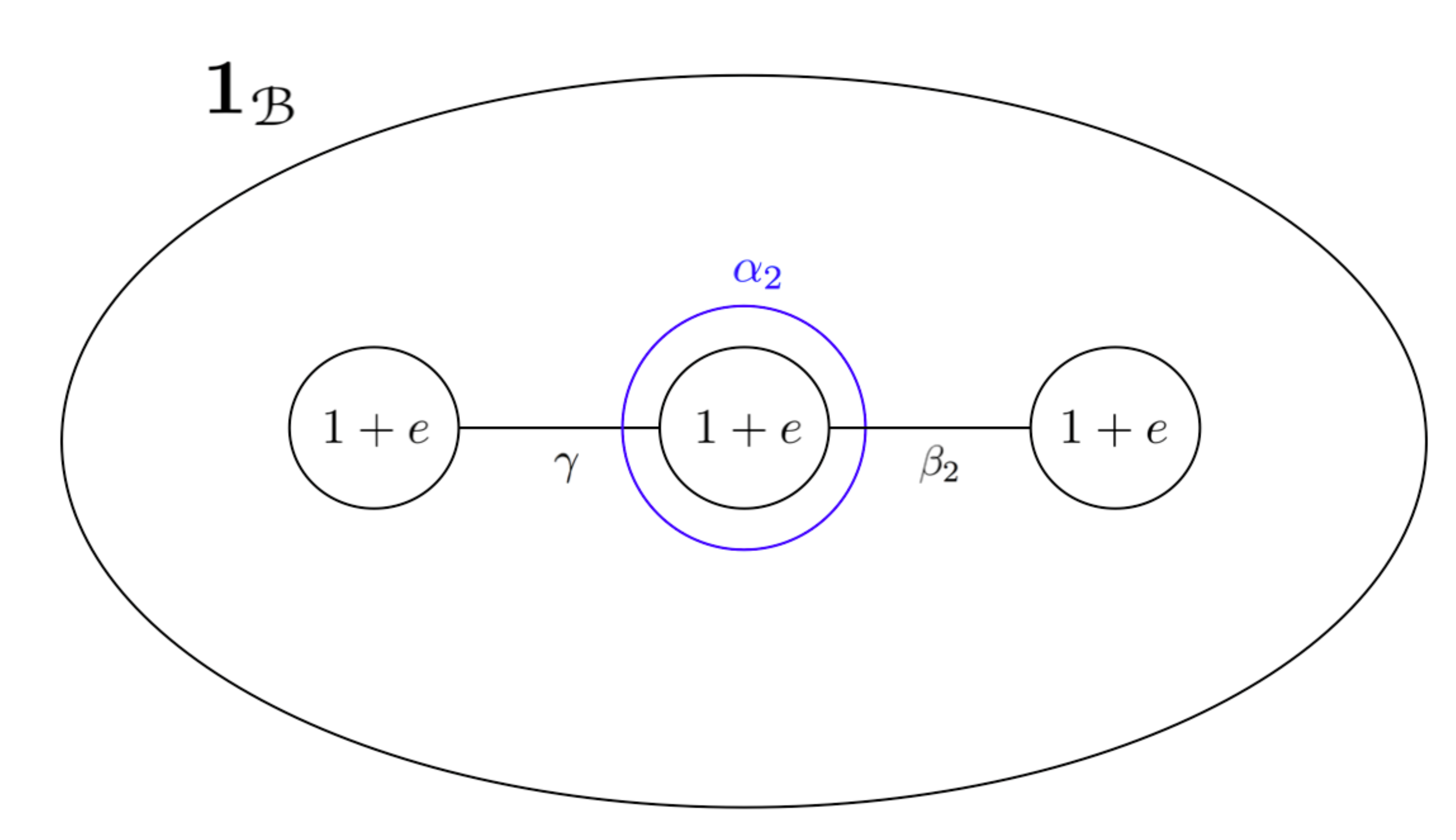}
\caption{Labeling of arcs for topological charge measurement for the toric code gapped boundaries. Our qubit is encoded in the two leftmost holes. We use an ancilla $1+e$ hole for the procedure.}
\label{fig:tc-phase-gate}
\end{figure}

In Ref. \cite{Fowler12}, Fowler et al. used magic state distillation to implement the phase gate and $\pi/8$ gates for the $\Z_2$ surface code. However, that procedure heavily relies on injecting external states, such as $\frac{1}{\sqrt{2}}(\ket{0} + i\ket{1})$ and $\frac{1}{\sqrt{2}}(\ket{0} + e^{i\pi/4}\ket{1})$. These states would not be topologically protected, and hence may arise as a source of error.

In this section, we will present an improved method in obtaining the phase gate, by using the topological charge measurement methods described in Section \ref{sec:measurement}.

In this protocol, we use a single auxiliary $1+e$ gapped boundary to implement the phase gate on a qubit encoded using two $1+e$ boundaries. The setup is shown in Fig. \ref{fig:tc-phase-gate}. To provide compatibility with all of the surface code gate implementations presented in Ref. \cite{Fowler12}, we will use the protocol for topological charge measurement in the surface code context, as discussed in Section \ref{sec:surface-code-tcm}. The overview of the procedure is as follows:

\begin{enumerate}
\item
Initialize the system by projecting to the eigenstate of $W_e(\beta_2)$.
\item
Project to the eigenstate of $W_m(\alpha_2)$.
\item
Project to the eigenstate of $iW_m(\alpha_2)W_e(\gamma)$.
\item
Finalize by projecting again to the eigenstate of $W_e(\beta_2)$.
\end{enumerate}

In Steps (1) and (2), we may simply perform measurements on each individual data qubit that comprises the ribbons $\beta_2$ and $\alpha_2$. To measure $W_e$, one can just measure the eigenvalue of the operator that performs $\sigma^z$ at each qubit along the ribbon, and to measure $W_m$, one measures the eigenvalue of the operator that performs $\sigma^x$ at each qubit along the ribbon.

Step (3) is slightly more complicated. In this step, we are trying to measure an operator on a graph, which consists of ribbons, $\alpha_2$ and $\gamma$, which intersect at a data qubit, which we call $q$. As discussed in Section \ref{sec:surface-code-tcm}, we would like to perform a projection

\begin{equation}
P_{(i_1, ... i_k)} = \sqrt{\alpha}H_{i_1}\cdots H_{i_k}
\end{equation}

\noindent
for some Hermitian Wilson operators $H_{i_1}, ... H_{i_k}$ such that

\begin{equation}
H_{i_1}\cdots H_{i_k}=\alpha H_{i_k}\cdots H_{i_1}.
\end{equation}

In our case, we simply have $k=2$, $H_{i_1} = W_m(\alpha_2)$ and $H_{i_2} = W_e(\gamma)$, which gives $\alpha = -1$ (hence the factor $i = \sqrt{\alpha}$).

Let us now provide some physical motivation for this projector using the surface code quantum circuits. To implement $P_{(i_1, i_2)}$ in this case, we use a single ancilla qubit $\ket{a}$, initialized to the state $\ket{0}$, that will act like a syndrome qubit. We begin by performing controlled-$\sigma^z$ operations on $\ket{a}$ for each data qubit along $\gamma$ except for $q$. We then perform CNOTs on $\ket{a}$ for each data qubit along $\alpha_2$ except for $q$. Finally, we perform a controlled-$U$ on $\ket{a}$, where $U$ is the Clifford operator that changes basis from $\sigma^z$ to $\sigma^y$. In all controlled operations, the data qubit of the surface code will act as control, while the ancilla $\ket{a}$ will act as target. Finally, measuring $\ket{a}$ in the standard $\sigma^z$ basis will give the desired projection.

We can now show that our procedure indeed implements the phase gate:

After Step (1), we may assume that the system is in a state 

\begin{equation}
\ket{\psi(s_1)} = \frac{\ket{0} + \ket{1}}{\sqrt{2}} \otimes \frac{\ket{0} + s_1 \ket{1}}{\sqrt{2}},
\end{equation}

\noindent
where $s_1 = \pm 1$ is the eigenvalue corresponding to the projective measurement. We then apply the rest of the measurements to obtain:

\begin{multline}
\begin{aligned}[t]
\frac{1+s_1 W_e(\beta_2)}{2}
\frac{1+s_3 i W_m(\alpha_2)W_e(\gamma)}{2}
\frac{1+s_2 W_m(\alpha_2)}{2}
\ket{\psi(s_1)}
\\ = \frac{1 + i s_1 s_2 s_3}{4} \ket{\psi(s_1)}.
\end{aligned}
\end{multline}

\noindent
In this equation, $s_2, s_3 = \pm 1$ are the eigenvalues of the measurements in Steps (2) and (3). It follows that the relative phase between $\ket{\psi(+)}$ and $\ket{\psi(-)}$ is now $e^{i \pi s_2 s_3/2}$.

Hence, this procedure is able to produce either the phase gate or its inverse, and using the measured values of $s_2$ and $s_3$, we can find out which one is produced. Repeating the procedure many times corresponds to an unbiased random walk on the integers, starting from $0$. It is well known that the probability of eventually reaching any integer in such a random walk is $1$. Therefore, after some amount of time, we will have obtained the phase gate itself.

\subsubsection{Physically implementable gates}
\label{sec:tc-hadamard}

Finally, the Hadamard may be implemented using stabilizer code techniques of Chapter \ref{sec:circuits}. The details of these have all been described in Ref. \cite{Fowler12}, and will not be repeated here.

\subsection{Example: $\mfD(S_3)$}
\label{sec:ds3-operations}

In this section, we present two possible ways to use gapped boundaries of $\mfD(S_3)$ to perform quantum computation, using the operations described in this chapter. 

\begin{figure}
\centering
\includegraphics[width = 0.45\textwidth]{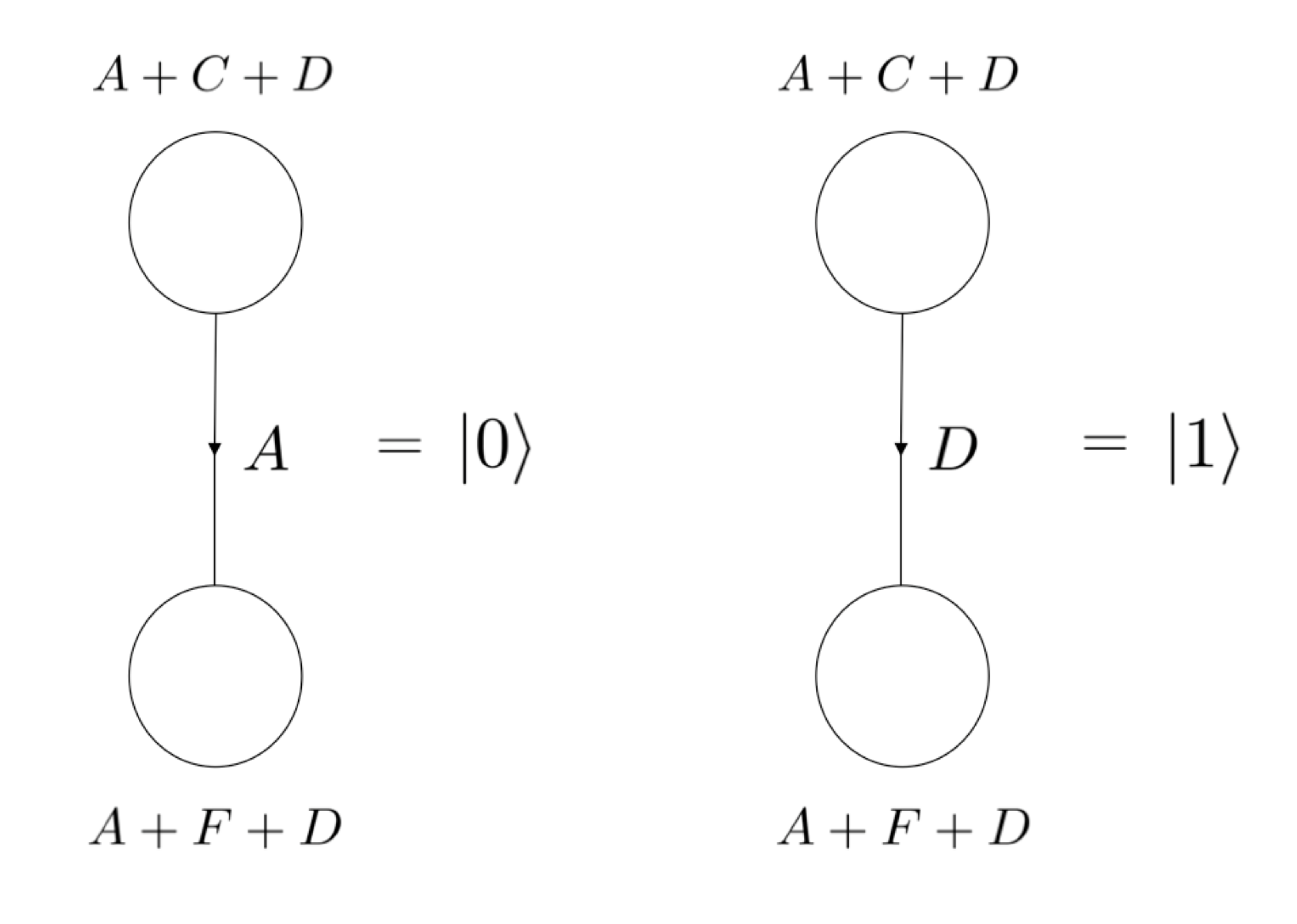}
\caption{Qubit encoding using gapped boundaries of $\mfD(S_3)$.} 
\label{fig:ds3-qubit}
\end{figure}

\begin{figure}
\centering
\includegraphics[width = 0.6\textwidth]{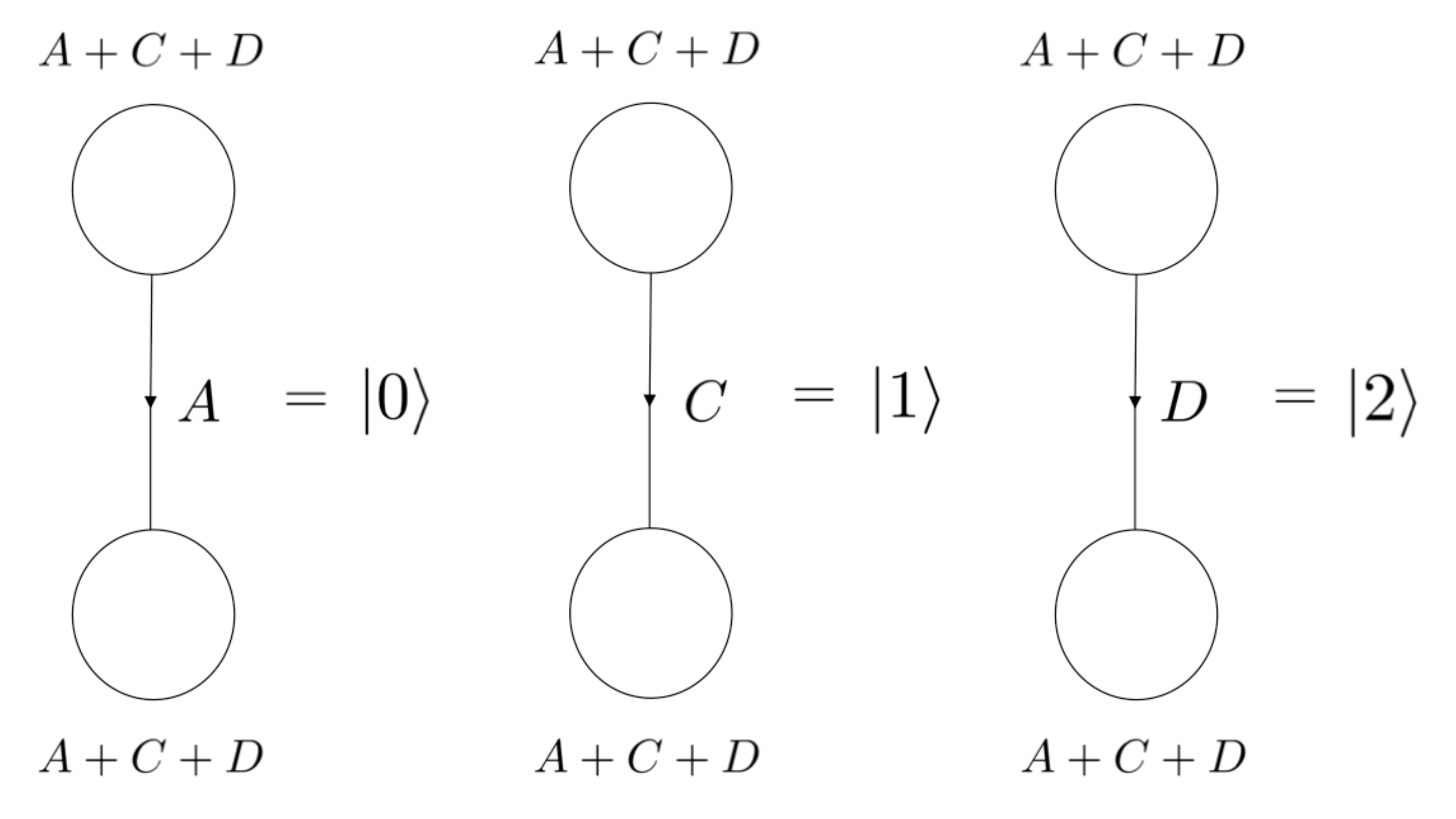}
\caption{Qutrit encoding using gapped boundaries of $\mfD(S_3)$.} 
\label{fig:ds3-qutrit}
\end{figure}

\subsubsection{Qubit and qutrit encodings}
\label{sec:ds3-qubit-encoding}

We consider two possible encoding schemes. The first is a qubit encoding, which uses one $A+C+D$ boundary, and one $A+F+D$ boundary. The encoding is illustrated in Fig. \ref{fig:ds3-qubit}. The second is a qutrit encoding, which uses two $A+C+D$ boundaries, and is shown in Fig. \ref{fig:ds3-qutrit}. When we present a quantum gate, we will present in the standard way such that all rows and columns are ordered as $(\ket{0},\ket{1})$, or $(\ket{0},\ket{1},\ket{2})$.

\subsubsection{Tunnel-$a$ operators}

In this section, we consider the tunneling operators $W_C (\gamma)$ and $W_D (\gamma)$ in the qutrit encoding, as an example for Section \ref{sec:tunnel}. By the general formula (\ref{eq:tunnel-formula}) and the $M$-3j symbols presented in Section \ref{sec:ds3-m3j}, we have the following matrices for these operators:

\begin{equation}
\label{eq:ds3-tunnel-C}
W_C(\gamma) =
\begin{bmatrix}
0 & 1 & 0 \\
1 & 1/\sqrt{2} & 0 \\
0 & 0 & \sqrt{2}
\end{bmatrix}
\end{equation}

\begin{equation}
\label{eq:ds3-tunnel-D}
W_D(\gamma) =
\begin{bmatrix}
0 & 0 & 1 \\
0 & 0 & \sqrt{2} \\
1 & \sqrt{2} & 0
\end{bmatrix}
\end{equation}

We note that by Remark \ref{tunnel-rmk}, $W_C(\gamma)$ and $W_D(\gamma)$ are not unitary here since $C \rightarrow A \oplus B$ and $D \rightarrow A \oplus C$ when condensing to the boundary (i.e. they may condense to excitations that are not vacuum). However, by the same Remark, they are Hermitian because $\A_1 = \A_2$ and all anyons in $\mfD(S_3)$ are self-dual. 

\subsubsection{Loop-$a$ operators}

We now compute the loop operator $W_B(\alpha_2)$ for both the qubit and the qutrit encoding, as an example for Section \ref{sec:loop}. For the qubit encoding, by Eq. (\ref{eq:loop-formula}) and the $\mathcal{S}$ matrix entries for $\mfD(S_3)$ (see Section \ref{sec:ds3-algebraic-example}), we have:

\begin{equation}
W_B(\alpha_2) =
\begin{bmatrix}
1 & 0 \\
0 & -1
\end{bmatrix}
= \sigma^z,
\end{equation}


For the qutrit encoding, we have:

\begin{equation}
W_B(\alpha_2) =
\begin{bmatrix}
1 & 0 & 0 \\
0 & 1 & 0 \\
0 & 0 & -1
\end{bmatrix}.
\end{equation}

\noindent
In Chapter \ref{sec:uqc}, we will see that this $\text{diag}(1,1,-1)$ matrix becomes very important for universal quantum computation.

\subsubsection{Braiding}
\label{sec:ds3-braiding}

We now discuss braiding in the context of the two encodings. In the qutrit encoding, we consider the case where $\A_1 = \A_2 = \A_3 = \A_4 = A+C+D$. By the fusion rules of $\mfD(S_3)$ (see Section \ref{sec:ds3-algebraic-example}), we see that 

\begin{equation}
\Hom(\one_\B, \A_1 \otimes \A_2 \otimes \A_3 \otimes \A_4) \cong \C^{49}.
\end{equation}

\noindent
Here, since all gapped boundaries are given by the same Lagrangian algebra, we may consider all generators of the four-strand braid group $B_4$, rather than just the pure braid group. In fact, we have numerically computed all generators of $B_4$ using the method described in Section \ref{sec:braiding}. Each is a $49 \times 49$ matrix, and we used a C program to verify that these matrices indeed satisfy all spherical braid group conditions. Unfortunately, due to large leakage issues, we were not able to get any interesting gates within the computational subspace. However, it is highly nontrivial for matrices of this size to satisfy the spherical braid group relations, and this gives very good evidence that the braiding we described in Section \ref{sec:braiding} will actually give a representation of the four-strand braid group. 

\begin{figure}
\centering
\includegraphics[width = 0.45\textwidth]{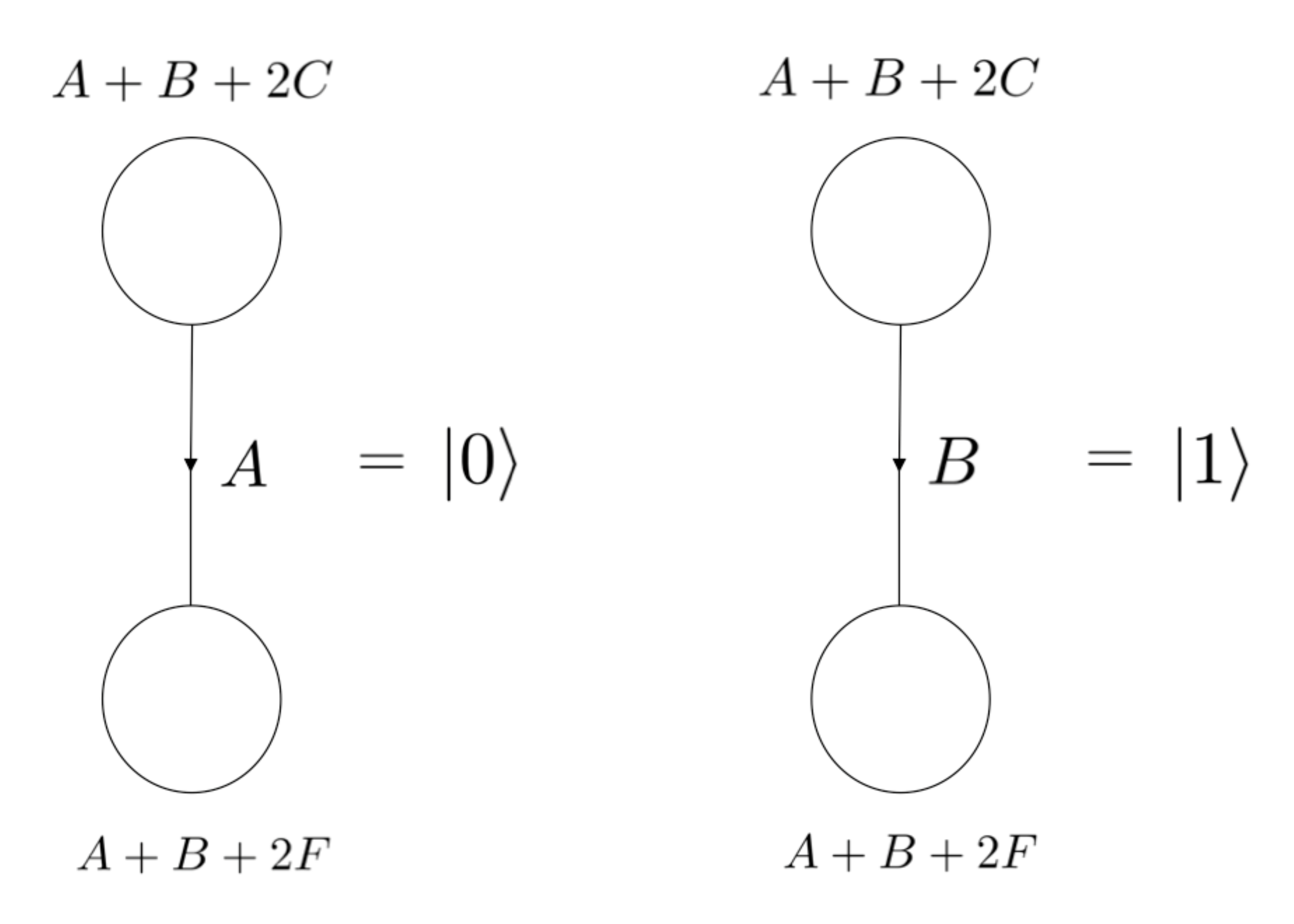}
\caption{Ancilla qubit encoding scheme for $\mfD(S_3)$.}
\label{fig:ds3-qubit-2}
\end{figure}

\begin{figure}
\centering
\includegraphics[width = 0.35\textwidth]{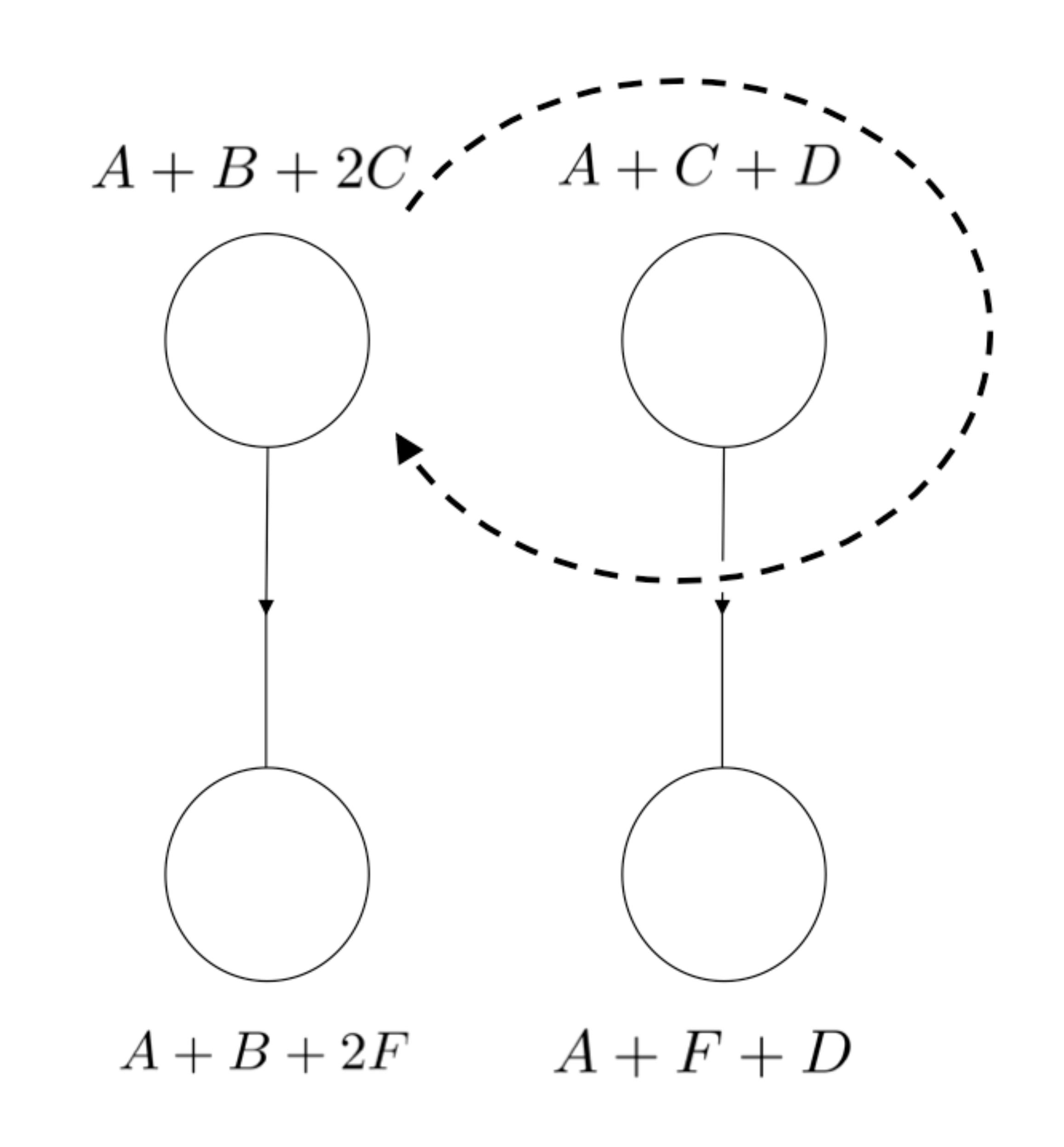}
\caption{Braid for the $\mfD(S_3)$ theory to obtain a controlled $\sigma_Z$ gate. As discussed in Section \ref{sec:braiding}, the dotted line indicates motion of a hole, while the solid lines specify the basis element of the hom-space.}
\label{fig:ds3-braiding}
\end{figure}

In the case of the qubit encoding, we can get an interesting entangling braid if we introduce an ancilla qubit using the encoding scheme shown in Fig. \ref{fig:ds3-qubit-2}. Specifically, we use the $A+C+D/A+F+D$ encoding for our target qubit, and the $A+B+2C/A+B+2F$ encoding for the control qubit. Then, we braid one of the control qubit holes around one of the target qubit holes (i.e. apply $\sigma_2^2$) as shown in Fig. \ref{fig:ds3-braiding}. By the general formula (\ref{eq:braid-formula}) and the $F$, $R$ symbols of $\mfD(S_3)$, we obtain:

\begin{equation}
\sigma_2^2 =
\begin{bmatrix}
1 & 0 & 0 & 0 \\
0 & 1 & 0 & 0 \\
0 & 0 & 1 & 0 \\
0 & 0 & 0 & -1
\end{bmatrix}
= \wedge\sigma^z.
\end{equation}

If we can implement the Hadamard in this theory, this will allow us to use a short circuit similar to that of Fig. \ref{fig:tc-cnot} to obtain a topological CNOT between two $A+C+D/A+F+D$ encoded qubits, as in the case of the $\Z_2$ toric code gapped boundaries.

\vspace{2mm}
\section{Universal quantum computation with gapped boundaries}
\label{sec:uqc}

In this chapter, we will illustrate how gapped boundaries and boundary defects may be combined with anyons to perform universal quantum computation. We begin in Section \ref{sec:universal-gate-set} by presenting target universal gate sets for qubit and qutrit quantum computation models. In Section \ref{sec:ds3-uqc}, we demonstrate how to use gapped boundaries of $\mfD(S_3)$ to potentially achieve the universal qubit gate set or the universal qutrit gate set.
Finally, in Section \ref{sec:dz3}, we illustrate how to use the braiding and topological charge measurement for gapped boundaries of $\mfD(\Z_3)$ to achieve universal quantum computation with qutrits.

\subsection{Universal gate sets for qubits and qutrits}
\label{sec:universal-gate-set}

In this section, we will present three universal gate sets for qubit and qutrit quantum circuit models. The first will be for the qubit computation model, and the rest will be for the qutrit computation model. In what follows, we adopt the following conventional notations for the standard gates:

\begin{enumerate}
\item
The generalized Hadamard gate for qudits will be denoted as \cite{Cui15-m}
\begin{equation*}
H_d \ket{j} = \frac{1}{\sqrt{d}} \sum_{i=0}^{d-1} \omega_d^{ij}\ket{i}, \text{ } j = 0,1, ... d-1.
\end{equation*}
\noindent
where $\omega_d = e^{2\pi i/d}$ is the $d^{\text{th}}$ root of unity.
\item
The generalized CNOT gate for two qudits is the SUM gate, which will be denoted as \cite{Cui15-m}
\begin{equation*}
\SUM_d \ket{i}\ket{j} = \ket{i} \ket{(i+j) \text{ mod } d}, \text{ } i,j = 0,1,...d-1.
\end{equation*}
\item
The generalized Pauli-X gate for qudits is
\begin{equation*}
\sigma^x_d \ket{j} = \ket{(j+1) \text{ mod } d}, \text{ } j = 0,1,...d-1.
\end{equation*}
\item
The generalized Pauli-Z gate for qudits is
\begin{equation*}
\sigma^z_d \ket{j} =  \omega_d^j \ket{j}, \text{ } j = 0,1,...d-1.
\end{equation*}
\item
The generalized two-qudit controlled $\sigma^z$ gate is
\begin{equation*}
\wedge \sigma^z_d \ket{i} \ket{j} = (I_d \otimes H_d) \SUM_d (I_d \otimes H_d) \ket{i} \ket{j}  = \omega_d^{ij} \ket{i}\ket{j}
\end{equation*}
\item
The well-known single-qubit phase gate is
\begin{equation*}
P =
\begin{bmatrix}
1 & 0 \\
0 & i
\end{bmatrix}.
\end{equation*}
\item
The well-known single-qubit $\pi/8$ gate is
\begin{equation*}
T =
\begin{bmatrix}
1 & 0 \\
0 & e^{i\pi/4}
\end{bmatrix}.
\end{equation*}
\end{enumerate}



The first gate set we present is a new gate set for the qubit computation model. This will be the target gate set for $\mfD(S_3)$ qubit theories:

\begin{theorem}
\label{order6-qubit-set}
The following set of qubit gates are sufficient for universal quantum computation:
\begin{enumerate}
\item
The single-qubit Hadamard gate $H_2$.
\item
The two-qubit entangling gate $\text{CNOT} = \SUM_2$.
\item
The single-qubit $\pi/6$ phase gate
\begin{equation}
P_{\pi/6} =
\begin{bmatrix}
1 & 0 \\
0 & e^{i\pi/3}
\end{bmatrix}.
\end{equation}
\end{enumerate}
\end{theorem}

\begin{proof}
Let $M = H_2 P_{\pi/3} H_2 P_{\pi/3}^\dagger$, $N = H_2 P_{\pi/3}^\dagger H_2 P_{\pi/3}$. We have:

\begin{equation}
M = \frac{1}{2}
\begin{bmatrix}
1-\omega & -1-\omega^2 \\
1+\omega & 1-\omega^2
\end{bmatrix}, 
\qquad
N = \frac{1}{2}
\begin{bmatrix}
1-\omega^2 & -1-\omega \\
1+\omega^2 & 1-\omega
\end{bmatrix}.
\end{equation}

\noindent
where $\omega = e^{2 \pi i /3}$ is the third root of unity. A simple calculation shows that

\begin{equation}
\label{eq:MN-not-commuting}
MN \neq NM.
\end{equation}

Let us now prove the following Lemma:

\begin{lemma}
\label{M-N-infinite-order}
$M,N$ are both of infinite order.
\end{lemma}

\begin{proof}
Since $M$ and $N$ are both Hermitian, let us diagonalize them. Simple calculation shows that the eigenvalues of $M$ and $N$ are given by $\lambda = \frac{3 \pm \sqrt{7} i}{4}$. We hence have:

\begin{equation}
M^n = A 
\begin{bmatrix}
\left(\frac{3 + \sqrt{7} i}{4} \right)^n & 0 \\
0 & \left(\frac{3 - \sqrt{7} i}{4}\right)^n
\end{bmatrix}
A^\dagger
\end{equation}

\begin{equation}
N^n = B 
\begin{bmatrix}
\left(\frac{3 + \sqrt{7} i}{4} \right)^n & 0 \\
0 & \left(\frac{3 - \sqrt{7} i}{4}\right)^n
\end{bmatrix}
B^\dagger
\end{equation}

\noindent
for some (unitary) matrices $A,B$. It follows that $M^n = I$ (or $N^n = I$) if and only if $\left(\frac{3 + \sqrt{7} i}{4} \right)^n = 1$ and $\left(\frac{3 - \sqrt{7} i}{4} \right)^n = 1$. De Moivre's theorem says that

\begin{equation}
\left(\frac{3 + \sqrt{7} i}{4} \right)^n = \cos (n \theta) + i \sin (n \theta)
\end{equation}

\noindent
where $\theta = \arccos(3/4)$. Hence, $M$, $N$ are of finite order only if $n \arccos(3/4) = 2 \pi$ for some positive integer $n$. However, by Ref. \cite{Varona06}, the number $\arccos(3/4) / \pi$ is irrational. Hence, $M$ and $N$ are both of infinite order.
\end{proof}

By Lemma 1 of Ref. \cite{Cui15-m}, we see that Eq. (\ref{eq:MN-not-commuting}) and Lemma \ref{M-N-infinite-order} imply that this set of gates is indeed universal.
\end{proof}

\begin{corollary}
\label{order6-qubit-set-cor}
The following set of qubit gates are sufficient for universal quantum computation:
\begin{enumerate}
\item
The single-qubit Pauli-Z gate $\sigma^z_2$.
\item
The single-qubit Hadamard gate $H_2$.
\item
The two-qubit entangling gate $\text{CNOT} = \SUM_2$.
\item
The single-qubit $\pi/3$ phase gate
\begin{equation}
P_{\omega} =
\begin{bmatrix}
1 & 0 \\
0 & \omega
\end{bmatrix}.
\end{equation}
\noindent
where $\omega = e^{2 \pi i /3}$ is a primitive third root of unity.
\end{enumerate}
\end{corollary}

We now present universal gate sets for the qutrit computation model. The universality of these gate sets was proven in Ref. \cite{Cui15-m}. These are often known as the {\it metaplectic gate sets}.

\begin{theorem}
\label{qutrit-set-2}
The following set of qutrit gates are sufficient for universal quantum computation:
\begin{enumerate}
\item
The single-qutrit Hadamard gate $H_3$
\item
The two-qutrit entangling gate $\SUM_3$.
\item
The single-qutrit generalized phase gate
\begin{equation}
Q_3 = 
\begin{bmatrix}
1 & 0 & 0 \\
0 & 1 & 0 \\
0 & 0 & \omega
\end{bmatrix},
\end{equation}
\noindent
where $\omega = e^{2 \pi i /3}$ is a primitive third root of unity.
\item
One of the following two options:
\begin{enumerate}
\item
The single-qutrit sign-flip gate
\begin{equation}
\text{Flip}_3 = 
\begin{bmatrix}
1 & 0 & 0 \\
0 & 1 & 0 \\
0 & 0 & -1
\end{bmatrix}
\end{equation}
\item
Any nontrivial single-qutrit classical (i.e. Clifford) gate not equal to $H_3^2$, AND a projection $M$ of a state in the qutrit space $\C^3$ to $\Span\{ \ket{0} \}$ and its orthogonal complement $\Span\{ \ket{1}, \ket{2} \}$, so that the resulting state is coherent if projected into $\Span\{ \ket{1}, \ket{2} \}$.
\end{enumerate}
\end{enumerate}
\end{theorem}

By Ref. \cite{Cui15-m}, the projection $M$ and the nontrivial Clifford gate can be used to probabilistically construct the sign-flip gate.

\subsection{Example: $\mathfrak{D}(S_3)$}
\label{sec:ds3-uqc}

Using purely the topological operations presented in Chapter \ref{sec:operations}, we have not been able to implement a full universal gate set for either the qubit or the qutrit encoding scheme for $\mathfrak{D}(S_3)$ (see Section \ref{sec:ds3-operations} for the definitions of these encodings). However, we would like to make the following notes regarding each of the two computation models:

\subsubsection{Qubit model for $\mathfrak{D}(S_3)$}

For the qubit model, we have implemented an entangling gate $\wedge\sigma^z$ between an $A+C+D/A+F+D$ qubit and an ancilla $A+B+2C/A+B+2F$ qubit. If we are able to implement a Hadamard on the $A+C+D/A+F+D$ qubit, we will have the first three gates from Corollary \ref{order6-qubit-set-cor}. Although we do not currently have an implementation of the single-qubit $2\pi/3$ phase gate gate, we believe it is very likely that some other topologically protected operation (e.g. some topological charge measurement we have not tried, or perhaps another operation not covered by Chapter \ref{sec:operations}) may be able to implement this gate. This is because the number $\omega = e^{\pi i/3}$ appears very often in the $F$ and $R$ symbols and in the $\mathcal{T}$ matrix of the $\mfD(S_3)$ theory. Hence, we believe it is promising to pursue this qubit model to obtain universal quantum computation.

\subsubsection{Qutrit model for $\mathfrak{D}(S_3)$}

While we are missing many Clifford gates for the qutrit model of $\mathfrak{D}(S_3)$, we would like to point out that we have found a very simple implementation of the only non-Clifford in the gate set of Theorem \ref{qutrit-set-2}. Specifically, by simply braiding a $B$ particle around one of the $A+C+D$ qutrit holes, we can easily obtain the gate 

\[
Flip_3 = 
\begin{bmatrix}
1 & 0 & 0 \\
0 & 1 & 0 \\
0 & 0 & -1
\end{bmatrix}.
\]

\noindent
We believe that the Clifford gates may be easier to implement at high fidelity locally, and there may be some method to transport the high fidelity local gates into our topological encoding, which would give universality.  

Alternatively, Cui et al. present a way in Ref. \cite{Cui15-m} to achieve universal computation using pure anyonic braiding in $\mfD(S_3)$. In that work, all of the Clifford gates were implemented using simple anyon braids; however, it was necessary to use a probabilistic procedure (i.e. (4b) in Theorem \ref{qutrit-set-2}) to construct the non-Clifford gate $Flip_3$. It is interesting to see whether the computation power of anyonic braiding may be combined with that of gapped boundaries. If this is the case, we would have a better implementation of the universal qutrit gate set, as there would be no probabilistic procedure involved.

\subsection{Example: $\mfD(\Z_3)$}
\label{sec:dz3}

\begin{figure}
\centering
\includegraphics[width = 0.6\textwidth]{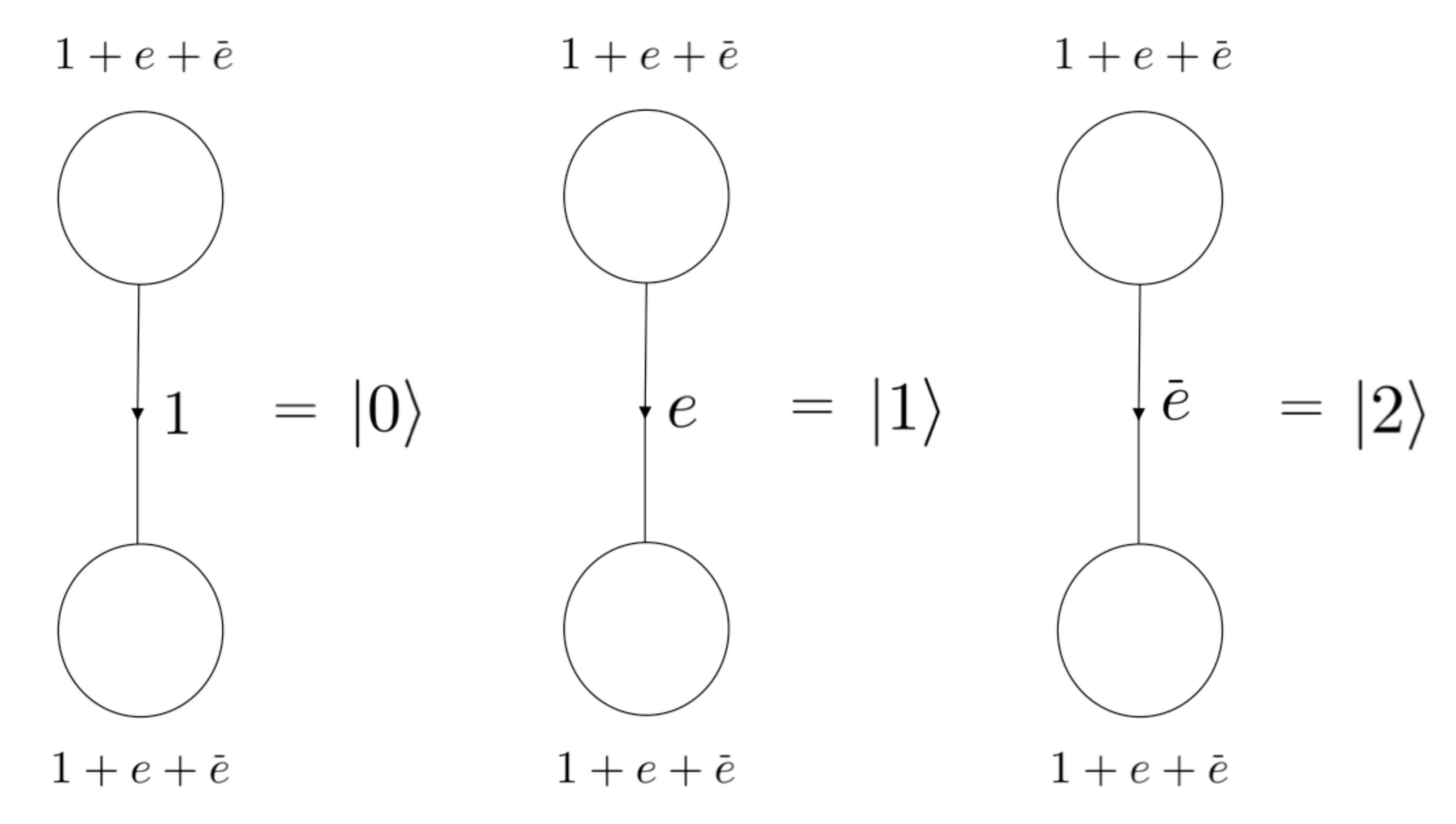}
\caption{Qutrit encoding using gapped boundaries of $\mfD(\Z_3)$. We would specifically like to note that since $\mfD(\Z_3)$ is {\it not} a self-dual anyon theory, it is very important to draw the arrow on the tunneling operator that describes the hom-space between the two boundaries and vacuum.} 
\label{fig:dz3-qutrit}
\end{figure}

We now demonstrate how gapped boundaries of $\mfD(\Z_3)$ can be used to obtain universal quantum computation. Our target gate set is the metaplectic gate set presented in Theorem \ref{qutrit-set-2}, using the option (4b). The $\mfD(\Z_3)$ theory has two gapped boundaries: the pure charge condensate $1+e+\overbar{e}$ corresponding to the trivial subgroup, and the pure flux condensate $1+m+\overbar{m}$ corresponding to the full subgroup. First, we encode a qutrit using two $1+e+\overbar{e}$ boundaries of $\mfD(\Z_3)$, as shown in Fig. \ref{fig:dz3-qutrit}. For most of our computation, we will use this qutrit, but occasionally, we will use an ancilla qutrit encoded by two $1+m+\overbar{m}$ boundaries. (That qutrit encoding is exactly the same as Fig. \ref{fig:dz3-qutrit}, with all instances of $e$,$\overbar{e}$ replaced by $m$,$\overbar{m}$, respectively).

\subsubsection{Topological order}

Before we begin finding gates, let us first review the topological order given by the modular tensor category $\mfD(\Z_3)$. By Ref. \cite{Barkeshli14}, the fusion rules are given by multiplication in $\Z_3 \times \Z_3$, the $F$ symbols in this category are all trivial, and the $R$ symbols are given by $R^{a,b} = e^{\pi i a_2 b_1}$, where we write $a = e^{a_1} m^{a_2}$, $b = e^{b_1} m^{b_2}$ for $a_1,a_2,b_1,b_2 = 0,1,2$. The modular $\mathcal{S} = [S_{ab}]$ and $\mathcal{T} = [T_{ab}]$ matrices for this category are given by \cite{BakalovKirillov}:

\begin{equation}
S_{ab} = \omega^{-a_2 b_1 - a_1 b_2}
\end{equation}

\begin{equation}
T_{ab} = \omega^{a_1a_2}\delta_{ab}.
\end{equation}

A simple calculation shows that all $M$ symbols of $\mfD(\Z_3)$ are trivial.

We would also like to mention that the $\mfD(\Z_3)$ topological order can be realized in bilayer fractional quantum Hall systems: Ref. \cite{Bark16} considers an electron-hole bilayer FQH system, with a $1/3$ Laughlin state in each layer. The topological order in this system can be described as $\mathcal{Z}(\mathrm{SU}(3)_1)=\mathrm{SU}(3)_1\times\overline{\mathrm{SU}(3)_1}$ (together with physical fermions). It is easy to see that $\mathcal{Z}(\mathrm{SU}(3)_1)$ is equivalent to $\mfD(\Z_3)$, so the discussion can be directly applied to the setting of Ref. \cite{Bark16} as well. It would also be convenient to view $\mfD(\Z_3)$ with a fixed type of boundaries (i.e. $1+e+\bar{e}$) as a higher genus surface of a single layer of $\mathrm{SU}(3)_1$, so that we can apply the results of Ref. \cite{Barkeshli16} directly.

\subsubsection{The Hadamard gate $H_3$}
First, by Section \ref{sec:physically-implementable}, the generalized Hadamard gate $H_3$ may be implemented by generalizing the method of Fowler et al. in Ref. \cite{Fowler12} to the $\Z_3$ surface code. 
Alternatively, $H_3$ is equal to the modular $\mathcal{S}$ matrix of the modular category $\mathrm{SU}(3)_1$. This matrix is in the representation of mapping class group of the torus. Hence, by Ref. \cite{Barkeshli16}, it may be implemented via a sequence of topological charge projections (see also Section \ref{sec:measurement}).

\begin{figure}
\centering
\includegraphics[width = 0.35\textwidth]{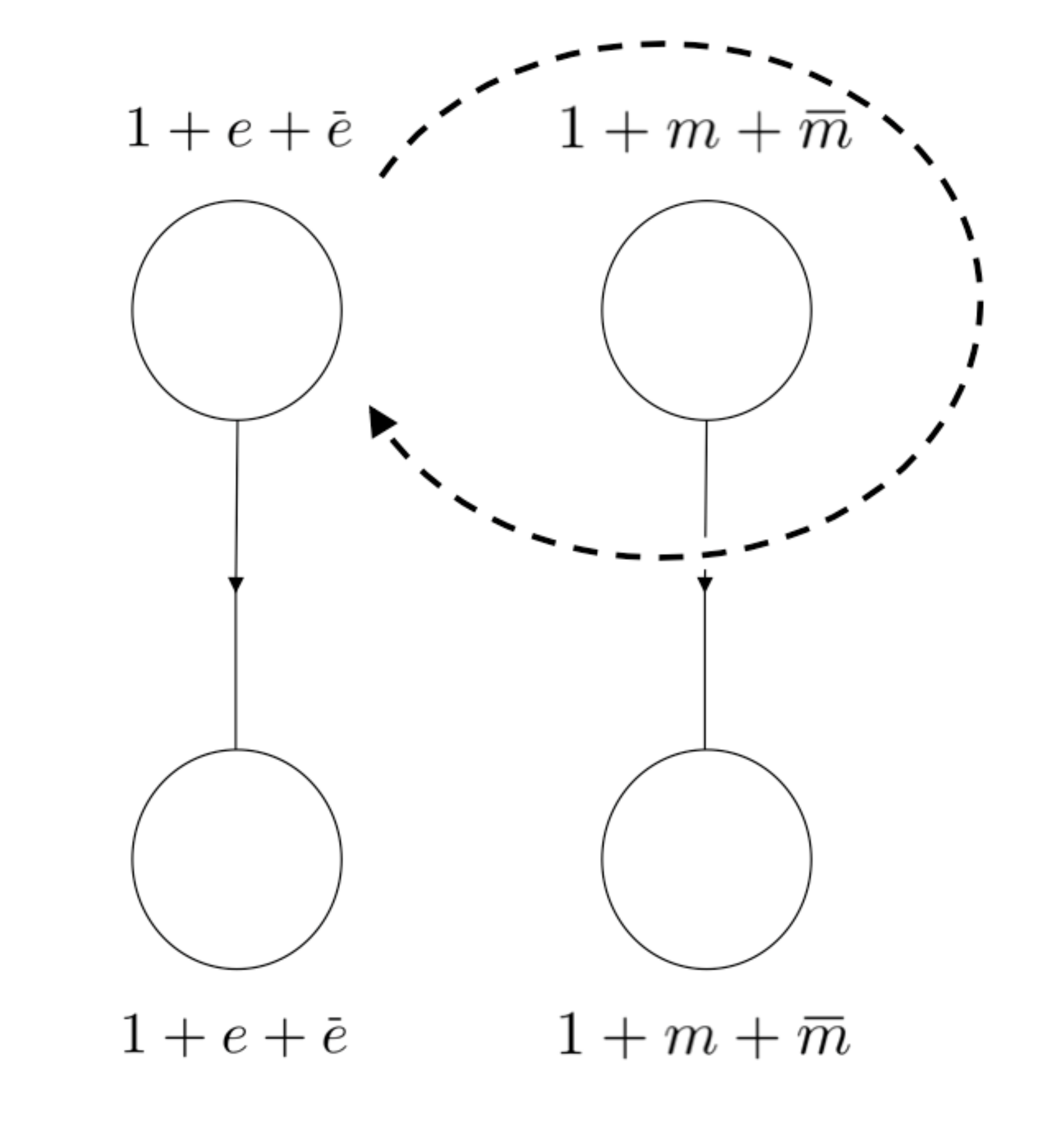}
\caption{Braid for the $\mfD(\Z_3)$ theory to obtain a generalized controlled $\sigma^z_3$ gate. As discussed in Section \ref{sec:braiding}, the dotted line indicates motion of a hole, while the solid lines specify the basis element of the hom-space.}
\label{fig:dz3-braiding}
\end{figure}

\subsubsection{The generalized Pauli-X gate $\sigma^x_3$}
By the general formula (\ref{eq:tunnel-formula}) of Section \ref{sec:tunnel}, we see that the tunneling operator $W_e (\gamma)$ precisely implements the single-qutrit Pauli-X gate $\sigma^x_3$. In fact, for the general theory $\mfD(\Z_p)$, $p$ any prime, the tunneling operator $W_e (\gamma)$ implements the single-qupit Pauli-X gate $\sigma^x_p$.

\subsubsection{The entangling gate $\SUM_3$}

We now consider the braiding of a $1+e+\overbar{e}$ qutrit with a $1+m+\overbar{m}$ qutrit. Specifically, we use the $1+e+\overbar{e}$ encoding for our target qutrit, and the $1+m+\overbar{m}$ encoding for the control qutrit. Then, if we braid one of the control qutrit holes around one of the target qutrit holes (i.e. apply $\sigma_2^2$) as shown in Fig. \ref{fig:dz3-braiding}, by the general formula (\ref{eq:braid-formula}) and the $F$, $R$ symbols of $\mfD(\Z_3)$, we obtain:

\begin{equation}
\sigma_2^2 = \text{diag} (1,1,1,1,\omega,\omega^2,1,\omega^2,\omega) = \wedge \sigma^z_3.
\end{equation}

\begin{figure}
\centering
\includegraphics[width=0.71\textwidth]{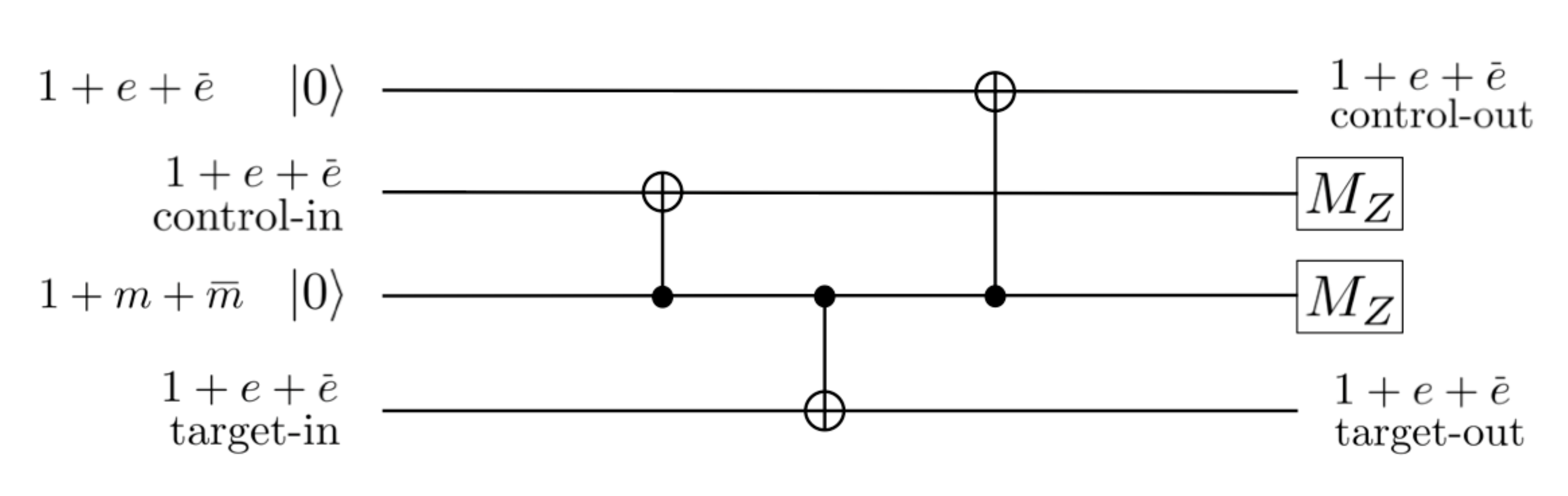}
\caption{Short circuit that generalizes the one presented in Ref. \cite{Fowler12} to use one ancilla $1+e+\overbar{e}$ qutrit, one ancilla $1+m+\overbar{m}$ qutrit, and 3 topological SUM gates between $1+m+\overbar{m}$ and $1+e+\overbar{e}$ qutrits to implement a topological SUM gate between 2 $1+e+\overbar{e}$ qutrits. All entangling gates drawn in this circuit are the SUM gate. The circuit generalizes completely to the case of gapped boundary qupits in $\mfD(\Z_p)$, $p$ any prime.}
\label{fig:dz3-sum}
\end{figure}

Because we have an implementation of the Hadamard, and because we have the relation $\wedge \sigma^z_d  = (I_d \otimes H_d) \SUM_d (I_d \otimes H_d)$ for any positive integer $d$, by conjugating the target qutrit by Hadamards, we may obtain the two-qutrit SUM gate between a $1+e+\overbar{e}$ qutrit and a $1+m+\overbar{m}$ qutrit. As in the case of the $\Z_2$ toric code, we have a short circuit that uses these SUM gates to implement a SUM gate between two $1+e$ boundaries. This circuit is shown in Fig. \ref{fig:dz3-sum}. In fact, this circuit can be generalized to $\Z_p$, $p$ any prime, where the gapped boundaries are replaced by $1+e+...+e^{p-1}$ and $1+m+...+m^{p-1}$, and we use the generalized gate $\SUM_p$.

As before, one must interpret the measurement outcome of the ancilla $1+m+\overbar{m}$ qutrit to obtain the correct result for the SUM gate. In general, in the case of $\Z_p$ and $\SUM_p$ gates, if we measure the ancilla $1+m+...+m^{p-1}$ qupit in the state $\ket{j}$, we must apply the operator $(\sigma^x_p)^j$ to the $1+e+...+e^{p-1}$ control-out qupit.In general, this can be done by applying the Wilson tunneling operator $W_{e^j}(\gamma)$. 

\subsubsection{The generalized phase gate $Q_3 = \mathrm{diag}(1,1,\omega)$}

By Ref. \cite{Barkeshli16}, topological charge measurements can be used to implement the phase gate $\mathrm{diag}(\mathrm{1,\omega, \omega})$ since it is the Dehn twist of the $\mathrm{SU}(3)_1$ theory. We follow this by a generalized Pauli-Z gate to obtain the single-qutrit $Q_3$ gate.

\subsubsection{Coherent projection}

We now need to implement the coherent projection that can allow us to go beyond Clifford gates and achieve universal quantum computation. As we have mentioned, a planar $\mfD(\Z_3)$ with two $1+e+\bar{e}$ holes can be effectively viewed as double layers of $\mathrm{SU}(3)_1$ connected via two handles, so the curve $\gamma$ connecting the two $1+e+\overbar{e}$ qutrit holes lifts to a loop in this perspective. We can hence project to a specific topological charge $a$ within this loop. By Ref. \cite{Barkeshli16}, this projector is given by:

\begin{equation}
P_\gamma^{(a)} = \sum_{x \in \mC} S_{0a} S_{xa}^* W_{x \overbar{x}} (\gamma).
\end{equation}

\noindent
In this case where $a = e$, by the $\mathcal{S}$ matrix of $\mfD(\Z_3)$, we have

\begin{equation}
P_\gamma^{(e)} = \frac{1}{3}
\begin{bmatrix}
1 & \omega & \overbar{\omega} \\
\overbar{\omega} & 1 & \omega \\
\omega & \overbar{\omega} & 1
\end{bmatrix}.
\end{equation}

Simple linear algebra shows that the eigenvalues and corresponding eigenspaces of $P_\gamma^{(e)}$ are given by:

\begin{equation}
\lambda = 0: V_\lambda = \Span \left\{ 
\begin{bmatrix}
1 \\
1 \\
1
\end{bmatrix},
\begin{bmatrix}
1 \\
\omega \\
\overbar{\omega}
\end{bmatrix}
\right\}
\end{equation}

\begin{equation}
\lambda = 1: V_\lambda = \Span \left\{ 
\begin{bmatrix}
1 \\
\overbar{\omega} \\
\omega
\end{bmatrix}
\right\}
\end{equation}

It follows that one can obtain the coherent projection $M$ of Theorem \ref{qutrit-set-2} by conjugating the orthogonal projector $1-P_\gamma^{(e)}$ with the Hadamard and Pauli-X, i.e. $(\sigma^x_3)^\dagger H^{\dagger}_3 (1-P_\gamma^{(e)})H_3 \sigma^x_3$.  While $P_\gamma^{(e)}$ is a topological charge projection as in \cite{Barkeshli16}, $1-P_\gamma^{(e)}$ is a general topological charge measurement (as introduced in Section \ref{sec:general-tcm}).  It would be interesting to know whether or not $1-P_\gamma^{(e)}$ has more computational power than $P_\gamma^{(e)}$, and whether or not there is a physically reasonable implementation of $1-P_\gamma^{(e)}$.

The coherent projection may also be obtained by conjugating the topological charge measurement $1-P_\gamma^{(1)}$ by Hadamard gates $H_3$, where $P_\gamma^{(1)}$ is a topolgoical charge projection onto trivial topological charge

\begin{equation}
P_\gamma^{(1)} = \frac{1}{3}
\begin{bmatrix}
1 & 1 & 1 \\
1 & 1 & 1 \\
1 & 1 & 1
\end{bmatrix}.
\end{equation}

By Theorem \ref{qutrit-set-2}, we now have universal quantum computation using gapped boundaries of $\mfD(\Z_3)$. This is a very significant result, as we demonstrate that purely topological methods (i.e. it does not use external high-fidelity state injection, as in Ref. \cite{Fowler12}) can achieve universal quantum computation model using only an abelian TQFT (all anyon braidings in $\mfD(\Z_3)$ are projectively trivial). 

It still remains an important open problem to design a surface code implementation or an experimentally practical way to realize this particular topological charge measurement. We will leave these questions for future works.

\vspace{2mm}
\section{Conclusions}
\label{sec:conclusions}

Based on Kitaev's quantum double models and Bombin and Martin-Delgado's two parameter generalization, we find exactly solvable Hamiltonian realizations of both gapped boundaries and the defects between them in Dijkgraaf-Witten theories. By combining with an algebraic model, we develop a microscopic theory for gapped boundaries and defects between boundaries that allows us to compute topological operations such as braiding on these new topological degeneracy.  We design qubit and qutrit topological quantum computing models using $\mathfrak{D}(S_3)$. We then develop a universal qutrit topological quantum computing model using only gapped boundaries of $\mathfrak{D}(\Z_3)$.

We would like to conclude by considering several potential areas to generalize our work. First, many physics papers have studied gapped domain walls between different topological phases. While gapped boundaries are often considered as a special case of gapped domain walls, by the folding trick \cite{KitaevKong}, they also completely cover the domain wall theory mathematically. Physically, however, it is still interesting to analyze the general gapped domain walls following our work.

Another direction is to generalize our theory to the Levin-Wen model, using quantum groupoids.

The most interesting question that we have not touched on is the stability of the topological degeneracy in our model.  Once our Hamiltonian moves off the fixed-point, finite-size splitting of the degeneracy would occur.  It would be interesting to study the energy spitting of the ground state degeneracies of our Hamiltonians $H_{\text{G.B.}}$ and $H_{\text{dft}}$ numerically under small perturbations.

Gapped boundaries and symmetry defects significantly enrich the physics of topological phases of matter in two spacial dimensions.  Their higher dimensional generalizations would be also very interesting.  As in the case of non-abelian anyons, experimental confirmation of non-abelian defects such as parafermion zero modes would be a landmark in condensed matter physics.

\vspace{2mm}
\begin{appendix}

\vspace{2mm}

\section{Notations}
\label{sec:notations}

In this Appendix, we list all of the notations that are used throughout the paper.

In general, if $G$ is any finite group, we denote the set of irreducible representations of $G$ by $(G)_{\text{ir}}$.

We will adopt the following conventions for labeling qudits, anyons, gapped boundaries, and their excitations:

\begin{enumerate}
\item
The data qudits in the bulk will be labeled as $g_1, g_2, ...$, $h_1, h_2, ...$ for the Kitaev model, where they are members of a finite group $G$.
\item
The data qudits on the boundary will be labeled as $k_1, k_2, ...$ for the Kitaev model, where they are members of a subgroup $K \subseteq G$.
\item
In the more general case where data qudits are simple objects in a unitary fusion category $\CC$, we will label the bulk data qudits as $x_1, x_2,...$, $y_1, y_2,...$.
\item
In this same general case, the boundary data qudits will be labeled as $r_1, r_2, ...$, $s_1, s_2, ...$.
\item
Bulk excitations (a.k.a. anyons or topological charges), which are the simple objects within the modular tensor category $\B = \ZZ(\text{Rep}(G))$ or $\B = \ZZ(\CC)$ will be labeled by $a,b,c...$. Their dual excitations are labeled by $\overbar{a}, \overbar{b}, \overbar{c}, ...$, respectively.
\item
The gapped boundary will be given as a Lagrangian algebra $\A$ which is an object in $\B$.
\item
Excitations on the boundary will be labeled as $\alpha, \beta, \gamma, ...$. When necessary, the local degrees of freedom during condensation will be labeled as $\mu, \nu, \lambda, ...$.
\item
Defects between different boundary types will be labeled as $X_1, X_2, ...$, $Y_1, Y_2, ...$

\end{enumerate}

Furthermore, when using any $F$ symbols and $R$ symbols for a fusion category or a modular tensor category, we will adopt the following conventions for indices:

\begin{equation}
\vcenter{\hbox{\includegraphics[width = 0.66\textwidth]{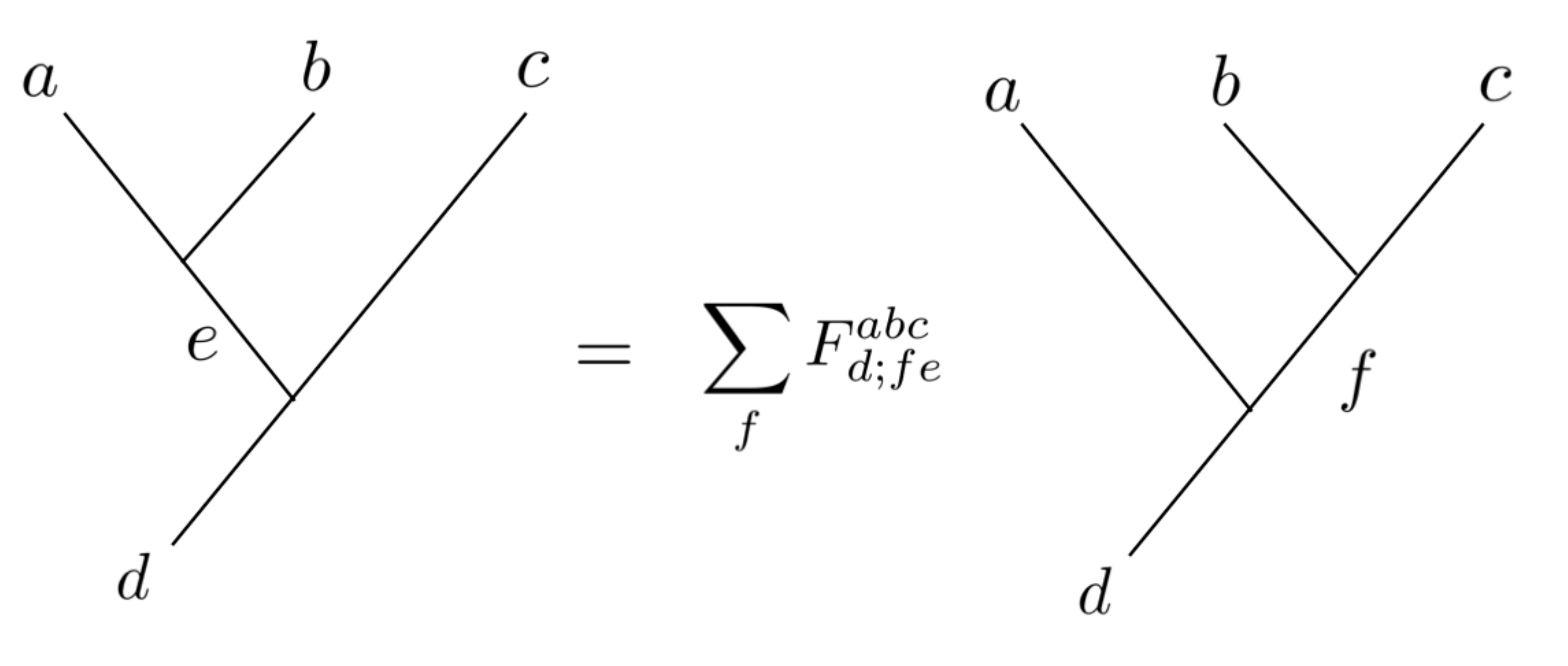}}}
\end{equation}

\begin{equation}
\vcenter{\hbox{\includegraphics[width = 0.7\textwidth]{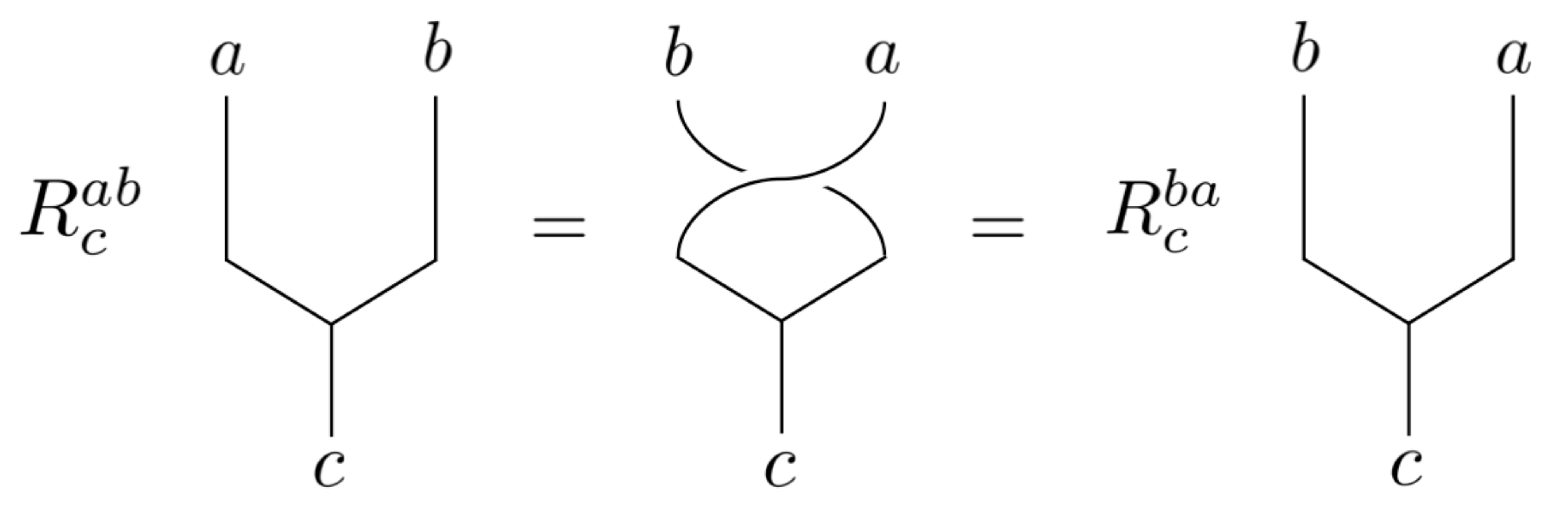}}}.
\end{equation}

\vspace{2mm}
\section{Hopf structures of $\mathcal{Z}$ and $\mathcal{Y}$}
\label{sec:quasi-hopf-algebra}

In this Appendix, we present the quasi-Hopf structures of the local operator algebras $\mathcal{Z} = Z(G,1,K,1)$ and the coquasi-Hopf structures of the ribbon algebras $\mY = Y(G,1,K,1)$ for the boundary Hamiltonian $H^{(K,1)}_{(G,1)}$. These structures have appeared in Refs. \cite{Zhu01} and \cite{Schauenburg02}, but we will discuss them in the context of these local and ribbon operators of the quantum double model with boundary. We only need to consider the case where $G$ is a finite group. This construction is similar to the construction of the quasi-triangular Hopf structures of the quantum double $\mathcal{D} = D(G)$ in Ref. \cite{Kitaev97}. 

Let us first begin by defining a quasi-Hopf algebra and a coquasi-Hopf algebra:

\begin{definition}
\label{quasi-bialgebra-def}
A {\it quasi-bialgebra} is a unital associative algebra $(A,m,\eta)$, together with a not necessarily coassociative coalgebra structure $(A,\Delta,\epsilon)$, and an invertible element $\Phi \in A \otimes A \otimes A$ such that the following axioms are satisfied \cite{Drinfeld89}:

\begin{equation}
\begin{gathered}
(1 \otimes \Delta)(\Delta(a)) = \Phi \cdot (\Delta \otimes 1)(\Delta(a)) \cdot \Phi^{-1} \\
(1 \otimes 1 \otimes \Delta)(\Phi)(\Delta \otimes 1 \otimes 1)(\Phi) = (1 \otimes \Phi)(1 \otimes \Delta \otimes 1)(\Phi)(\Phi \otimes 1) \\
(\epsilon \otimes 1) (\Delta a) = (1 \otimes \epsilon) \circ \Delta = 1 \\
(1 \otimes \epsilon \otimes 1)(\Phi) = 1 \otimes 1 
\end{gathered}
\end{equation}
\noindent
for all $a \in A$. $\Phi$ is often known as the {\it Drinfeld associator}.
\end{definition}

\begin{definition}
\label{coquasi-bialgebra-def}
A {\it coquasi-bialgebra} is a counital coassociative coalgebra $(A,\Delta, \epsilon)$, together with a not necessarily associative algebra structure $(A,m,\eta)$ and a convolution invertible element $\Phi \in (A \otimes A \otimes A)^*$ such that the following axioms are satisfied \cite{Drinfeld89}:

\begin{equation}
\begin{gathered}
h_1(g_1 k_1)\Phi(h_2,g_2,k_2) = \Phi(h_1,g_1,k_1)(h_2 g_2)k_2 \\
1_A h = h 1_A = h \\
\Phi(h_1, g_1, k_1 l_1) \Phi(h_2 g_2, k_2, l_2) = \Phi(g_1, k_1, l_1) \Phi(h_1, g_2 k_2, l_2) \Phi(h_2, g_3, k_3) \\
\Phi(h, 1_A, g) = \epsilon(h) \epsilon(g)
\end{gathered}
\end{equation}
\noindent
for all $h,g,k,l \in A$.
\end{definition}

\begin{definition}
A {\it quasi-Hopf algebra} is a quasi-bialgebra $A$ with an antipode map $S: A \rightarrow A$, and elements $\alpha, \beta \in A$, such that $S$ is an algebra anti-morphism and for all $a \in A$ with $\Delta a = a_{(1)} \otimes a_{(2)}$ (using Sweedler notation), the following axioms are satisfied \cite{Drinfeld89}:

\begin{equation}
\sum S(a_{(1)}) \alpha a_{(2)} = \epsilon(a) \alpha
\qquad
\sum a_{(1)} \beta (a_{(2)}) = \epsilon(a) \beta
\end{equation}

\noindent
for all $a \in A$, and 

\begin{equation}
\sum x_i \beta S(y_i) \alpha z_i = 1
\qquad
\sum S(x_i) \alpha y_i \beta S(z_i) = 1
\end{equation}

\noindent
where $\Phi = \sum x_i \otimes y_i \otimes z_i$ and $\Phi^{-1} = \sum x_i \otimes y_i \otimes z_i$.
\end{definition}

\begin{definition}
A {\it coquasi-Hopf algebra} is a coquasi-bialgebra $A$ with an antipode map $S: A \rightarrow A$ with elements $\alpha, \beta \in A^*$ such that the following axioms are satisfied \cite{Drinfeld89}:

\begin{equation}
\begin{gathered}
S(h_1) \alpha(h_2) h_3 = \alpha(h) 1_A \\
h_1 \beta(h_2) S(h_3) = \beta(h) 1_A \\
\Phi(h_1 \beta(h_2), S(h_3), \alpha(h_4) h_5) = \Phi^{-1}(S(h_1), \alpha(h_2) h_3 \beta(h_4), S(h_5)) = \epsilon(h)
\end{gathered}
\end{equation}
\noindent
for all $h \in A$.
\end{definition}

We now present all of these structures for the group-theoretical quasi-Hopf algebra $\mZ = Z(G,1,K,1)$ of local operators of the Hamiltonian $H^{(K,1)}_{(G,1)}$. As discussed in Section \ref{sec:bd-hamiltonian}, the basis vectors of $\mZ$ are of form

\begin{equation}
Z^{(hK,k)} = B^{hK}A^k
\end{equation}

\noindent
for some $k \in K$, $hK \in G / K$. For convenience, let us first define the following notations. Let $R$ be a set of representatives of the left cosets $G / K$. Then, every $g \in G$ can be written uniquely as $g = r(g) \{g \}^{-1}$ for some $r(g) \in R$, $\{g\} \in K$.

By definition of the operators $B^{hK}$, $A^k$, we have the following rule for multiplication in $\mZ$ (where we assume without loss of generality $h_1, h_2 \in R$):

\begin{align}
\begin{split}
Z^{(h_1 K, k_1)} Z^{(h_2 K, k_2)} &=
\sum_{\substack{h \in R \\
				j,k \in K}}
\delta_{h_1 K, hK} \delta_{k_1, k} \delta_{h_2 K, r(k^{-1}h) K} \delta_{k_2, k^{-1} j} Z^{(hK,j)}\\
&= \delta_{h_1 K, k_1 h_2 K} Z^{(h_1 K, k_1 k_2)}
\end{split}
\end{align}

\noindent
(Here, we note that certain Kronecker deltas may be taken between left cosets, e.g. $\delta_{h_1 K, hK}$.) 
Similarly, the following rule is used to define comultiplication in $\mZ$:

\begin{equation}
\Delta(Z^{(hK, k)}) =
\sum_{h_1  \in R}
(Z^{(r({h_1^{-1} h)}K, \{k^{-1} h_1\})} Z^{(h_1 K, k)})
\end{equation}

\noindent
The antipode in $\mZ$ is given by:

\begin{equation}
S(Z^{(hK,k)}) = Z^{( r({r({k^{-1}h})^{-1}}) K, \{ k^{-1}h \}^{-1} )} 
\end{equation}

\noindent
In both equations above, we assume $h \in R$. The specific elements $\alpha, \beta \in A$ corresponding to the quasi-Hopf structure of the antipode are:

\begin{equation}
\alpha = Z^{(K,1)}, \qquad \beta = \sum_{h \in R} Z^{(hK, \{h\}^{-1})}
\end{equation}

\noindent
Finally, the Drinfeld associator is given by:

\begin{equation}
\Phi = \sum_{h_1, h_2, h_3 \in R} Z^{(h_1 K, 1)} Z^{(h_2 K, 1)} Z^{(h_3 K, \{ h_1 h_2 \})}
\end{equation}

The construction of these structures for the coquasi-Hopf algebra $\mathcal{Y}$ is completely analogous and dual to the above construction. Details may be found in Refs. \cite{Zhu01} and \cite{Schauenburg02}, and will not be presented here.

We will note in particular that the comultiplication of $\mathcal{Y}$ exactly corresponds to the multiplication of $\mathcal{Z}$. By the formulas presented in Refs. \cite{Zhu01,Schauenburg02}, we see that it is indeed given by the gluing formula (\ref{eq:Y-gluing-formula}), as claimed in Section \ref{sec:bd-hamiltonian}.

\end{appendix}

\end{document}